\documentclass[12pt]{amsart}
\usepackage{amsfonts}
\usepackage{amssymb}
\usepackage{amsmath}
\usepackage{tikz-cd}
\usepackage{enumerate}
\usepackage{relsize}
\usepackage{scalerel}
\usepackage{bm}
\usepackage{amsaddr}
\usepackage{amsrefs}
\usepackage[makeroom]{cancel}
\usepackage{appendix}

\usepackage{fullpage}

\newtheorem{theorem}{Theorem}[subsection]
\newtheorem{lemma}[theorem]{Lemma}
\newtheorem{proposition}[theorem]{Proposition}
\newtheorem{corollary}[theorem]{Corollary}

\newtheorem*{theorem*}{Theorem}

\theoremstyle{definition}
\newtheorem{definition}[theorem]{Definition}

\newtheorem{example}[theorem]{Example}

\newcommand{\D}{\mathcal D}
\newcommand{\E}{\mathcal E}
\newcommand{\F}{\mathcal F}
\newcommand{\G}{\mathcal G}
\renewcommand{\H}{\mathcal H}
\newcommand{\I}{\mathcal I}

\newcommand{\M}{\mathcal M}

\renewcommand{\P}{\mathcal P}
\newcommand{\Q}{\mathcal Q}

\renewcommand{\S}{\mathcal S}

\newcommand{\V}{\mathcal V}
\newcommand{\W}{\mathcal W}
\newcommand{\X}{\mathcal X}
\newcommand{\Y}{\mathcal Y}
\newcommand{\Z}{\mathcal Z}

\newcommand{\CC}{\mathbb C}

\newcommand{\NN}{\mathbb N}

\newcommand{\RR}{\mathbb R}

\newcommand{\CCC}{\mathbf C}

\newcommand{\aA}{\mathfrak A}
\newcommand{\bB}{\mathfrak B}

\newcommand{\fF}{\mathfrak F}
\newcommand{\gG}{\mathfrak G}

\renewcommand{\:}{\colon}
\newcommand{\To}{\rightarrow}
\newcommand{\inv}{^{-1}}
\newcommand{\subsetof}{\subseteq}

\newcommand{\suchthat}{\,|\,}

\newcommand{\tensor}{\otimes}
\newcommand{\iso}{\cong}

\newcommand{\Union}{\bigcup}
\newcommand{\union}{\cup}

\newcommand{\up}{\mathop{\uparrow}}

\newcommand{\ran}{\mathop{\mathrm{ran}}}

\renewcommand{\above}{\sqsupseteq}
\newcommand{\below}{\sqsubseteq}
\newcommand{\cat}{\mathbf}

\newcommand{\atomof}{\text{\raisebox{1.2pt}{\smaller\,$\propto$\,}}}

\newcommand{\Tr}{\mathrm{Tr}}
\newcommand{\At}{\mathrm{At}}

\newcommand{\Eval}{\mathrm{Eval}}
\newcommand{\proj}{\mathrm{proj}}
\newcommand{\op}{\mathrm{op}}

\newcommand{\Set}{\mathbf{Set}}
\newcommand{\POS}{\mathbf{POS}}
\newcommand{\CPO}{\mathbf{CPO}}
\newcommand{\qRel}{\mathbf{qRel}}
\newcommand{\qSet}{\mathbf{qSet}}

\newcommand{\qPOS}{\mathbf{qPOS}}

\newcommand{\qCPO}{\mathbf{qCPO}}

\allowdisplaybreaks

\makeatletter
\newcommand{\circlearrow}{}
\DeclareRobustCommand{\circlearrow}{%
  \mathrel{\vphantom{\rightarrow}\mathpalette\circle@arrow\relax}%
}
\newcommand{\circle@arrow}[2]{%
  \m@th
  \ooalign{%
    \hidewidth$#1\circ\mkern1mu$\hidewidth\cr
    $#1\longrightarrow$\cr}%
}
\makeatother

\begin{document}

\title{Categories of quantum cpos}
\author{Andre Kornell, Bert Lindenhovius, Michael Mislove}

\begin{abstract}
    This paper unites two research lines. The first involves finding categorical models of quantum programming languages with recursion and their type systems. The second line concerns the program of \emph{quantization} of mathematical structures, which amounts to finding noncommutative generalizations (also called \emph{quantum} generalizations) of these structures. Using a quantization method called \emph{discrete quantization}, which essentially amounts to the internalization of structures in a category of von Neumann algebras and quantum relations, we find a noncommutative generalization of $\omega$-complete partial orders (cpos), called \emph{quantum cpos}. Cpos are central in domain theory, and are widely used to construct categorical models of programming languages with recursion. We show that quantum cpos have similar categorical properties to cpos and are therefore suitable for the construction of categorical models for quantum programming languages, which is illustrated with some examples. Because of their noncommutative character, quantum cpos may form the backbone of a future quantum domain theory that provides structural methods for the denotational semantics of recursive quantum programming languages. 
\end{abstract}

\maketitle

\section{Introduction}
\subsection{Motivation}
\emph{Quantum computing} is an approach to computation that utilizes quantum-mechanical principles such as superposition and entanglement to enhance certain calculations compared to classical computers. Because understanding all computational implications of these principles is challenging, designing quantum programming languages and their type systems forms a very active research field, in which new language paradigms and features are regularly proposed. Here, a place to begin is a \emph{type system} - a formal system that acts as a blueprint for programming languages, similar to how the lambda calculus serves as a blueprint for functional programming languages. The formation and derivation rules of the type system provide an \emph{operational model} that can be used to reason about the evolution of a quantum program, albeit at an abstract level.

\emph{Recursive} methods are crucial for various aspects of quantum computing as seen in algorithms like Grover's quantum search algorithm \cite{grover}, where a unitary gate is applied multiple times. Also the construction of quantum circuits sometimes depends on recursive methods. For instance, the Quantum Fourier Transform in Shor's algorithm for prime factorization \cite{shor} involves iteratively building on lower qubit counts. The design of quantum programming languages and their type systems is challenging, in particular when recursion is added. 

In particular, an open problem is how to abstract away from the quantum circuit paradigm, just as classical computing advanced by shifting from assembly languages to high-level programming. A key  component of this paradigm shift would be the implementation of a \emph{quantum control flow}. Currently, in most existing quantum programming languages the control flow is classical: conditions in \textbf{if} statements and for the termination of recursive loops are deterministically true or false. Even when a program branches based on the state of a qubit, that state must first be measured, collapsing it into a classical value before the control decision is made. In contrast, quantum control flow allows branching and looping decisions to depend on quantum states \emph{without} prior measurement. This enables control structures that operate coherently on superpositions, allowing quantum operations to proceed along multiple computational paths in parallel, without collapsing the underlying quantum information. As an example, the \emph{quantum switch} \cite{quantum-switch} is a higher-order process that applies two operations $A$ and $B$ to a qubit, but the order of application depends on the state of a control qubit. For instance, if the state is $|0\rangle$, we perform $A;B$ (i.e., $A$ followed by $B$), and if the state is $|1\rangle$, we perform $B;A$. We explicitly do not require a measurement of the control qubit, and as a consequence, the state of the control qubit can be in a superposition of $|0\rangle$ and $|1\rangle$, whence so can be the order of application of $A$ and $B$. For this reason, we say that the quantum switch is an \emph{indefinite causal} process.

Designing type systems that can express such inherently quantum control mechanisms remains a central challenge, particularly when it comes to achieving clean, higher-order and widely adopted solutions. A deep understanding of how higher-order functions behave under quantum interactions—such as superposition and entanglement—is crucial for designing type systems that can reliably express these control mechanisms.

 One approach to gaining such an understanding is through \emph{denotational semantics}, which provides a rigorous mathematical framework for interpreting programming languages and type systems.
 By constructing a \emph{denotational} or \emph{categorical model} of a type system, one can precisely describe the behavior of types and terms under complex computational phenomena—including quantum effects. Formally, a model consists of a category $\mathbf C$, along with a \emph{interpretation function} that assigns objects and morphisms of $\mathbf C$ to types and terms of the system, respectively, and type constructors are interpreted by functors on $\mathbf C$. 
 A crucial requirement is that the interpretation function is \emph{compositional}, meaning that the interpretation of a compound term is determined solely by the interpretations of its immediate subterms. In this way, the semantics can be built ``bottom up": the meaning of complex terms can be expressed as the composition of the meanings of their constituent components.

An example of how denotational semantics influences the design of programming languages  is provided by Moggi’s seminal work modeling computational effects using monads \cites{moggi,moggi1991}, which inspired new computational paradigms such as Levy’s \emph{Call-by-push-value} paradigm~\cite{levy-cbpv} that combines \emph{Call-by-value} and \emph{Call-by-name}. Another example of a new paradigm is provided by coherence spaces, from which Girard derived linear logic \cite{girard-linear}, eventually leading to substructural type systems. These type systems are actually used for quantum programming languages, because they reflect that quantum information is  subject to principles such as the No-Cloning Theorem. Hence, we expect that new, more refined categorical models for higher-order quantum programming languages may provide insights into the possibilities for implementing recursion in quantum computing, and in particular may help clarify whether recursion with quantum control is feasible.

Currently, most models of quantum programming languages consist of a category $\mathbf M$ whose objects represent finite collections of qubits, and whose morphisms represent first-order quantum processes, such as quantum gates in the case of pure quantum computations, or quantum channels in the case of impure computations. To support higher-order functions, one often applies categorical constructions such as taking presheaves over $\mathbf M$. However, the free higher-order structure of such a model does not capture the quantum interactions between higher-order functions such as the quantum switch, which are essential for modeling quantum control.

We propose constructing denotational models of quantum computing using the principles of  \emph{noncommutative geometry}  \cite{connes:ncg}, which is a mathematical framework in which quantum phenomena are described by noncommutative structures, typically in terms of bounded (=continuous linear) operators on a Hilbert space. Examples of these structures are \emph{operator algebras} \cite{Blackadar}, i.e., algebras of bounded operators on a Hilbert space, and \emph{operator spaces} \cite{EffrosRuan}, which are subspaces of the space of all bounded operators on a Hilbert space. A core ingredient of noncommutative geometry is \emph{quantization}, the process of extending classical mathematical structures into the noncommutative realm. A canonical example is provided by \emph{C*-algebras}, which form a class of operator algebras that can be regarded as noncommutative generalizations of locally compact Hausdorff spaces. This connection is formalized by \emph{Gelfand duality}, which states that any commutative $C^*$-algebra $A$ is isomorphic to the algebra $C_0(X)$ of continuous complex-valued functions vanishing at infinity on some locally compact Hausdorff space $X$.

The guiding insight of noncommutative geometry is that many quantum phenomena have classical counterparts and can be understood as noncommutative generalizations of familiar structures. In classical physics, for instance, the observables of a system are continuous real-valued functions on its phase space $X$, which is usually a locally compact Hausdorff spaces. In the language of operator algebras, these observables correspond to the selfadjoint elements of the C*-algebra $C_0(X)$. Quantum systems follow the same pattern: they are described by (typically noncommutative) C*-algebras, with observables identified as selfadjoint elements. In a similar way, we expect that quantizing the mathematical structures underlying the denotational semantics of ordinary higher-order programming languages will yield models of higher-order quantum programming languages that provide the clearest and most faithful representation of higher-order quantum functions and their interactions.

In the denotational semantics of  programming languages and their type systems, the support of recursion requires the existence of canonical fixpoints of morphisms in a categorical model. As a consequence, categorical models for programming languages with recursion often consist of structures for which there exists a fixpoint theorem. For example, $\omega$-complete partial ordered sets (cpos) are widely used in the denotational semantics of ordinary higher-order programming languages with recursion, because Kleene's fixpoint theorem ensures the existence of fixpoints on pointed cpos, and because the category $\mathbf{CPO}$ of cpos and Scott continuous maps is cartesian closed, hence supports higher-order functions. \emph{Domain theory} is the branch of mathematics that studies cpos and their more refined variants, such as directed-complete partial orders (dcpos), and develops systematic methods for constructing denotational models that support higher-order functions and recursion on the basis of these structures.

We expect that the quantization of domain-theoretic structures will yield a \emph{quantum domain theory} with the following properties: 
\begin{itemize}
	\item[(1)] Quantum domain theory should provide \emph{systematic methods} for the construction of categorical models of higher-order quantum programming languages with recursion, just as ordinary domain theory does for ordinary programming languages.
     \item[(2)] Because of their noncommutative character, the resulting models describe higher-order quantum processes and their interactions  \emph{more accurately}.
	\item[(3)] Since its structures are noncommutative analogues of classical domain-theoretic ones, many of the \underline{techniques and intuitions} of ordinary domain theory should carry over.
    \item[(4)] As a generalization of domain theory, its techniques should also apply to classical programming languages, thereby unifying classical and quantum computation in the \underline{same denotational framework}.
	\item[(5)] Alternative semantic structures for ordinary languages—metric, topological, or coherence spaces—are deeply connected to domains and may be quantized in the same way. The resulting quantum structures would inherit \underline{similar connections} to quantum domain theory, enabling a richer analysis of quantum programming models.
\end{itemize}


\subsection{Contributions, methodology, and outline of the article}
In this article, as a first step toward developing a quantum domain theory, we introduce \emph{quantum cpos}, which are  noncommutative generalizations of ordinary cpos, obtained via a quantization process that we call \emph{discrete quantization}. This process, which is described in greater detail in Section \ref{sec:preliminaries}, is based on noncommutative generalizations of sets and binary relations between sets, called \emph{quantum sets} and \emph{binary relations between quantum sets}\footnote{To avoid an overload of the modifier `quantum', we follow \cite{Kornell18} and only use the  modifier to noncommutative generalizations of structures such as quantum sets, quantum graphs or quantum posets. We do not use the modifier for noncommutative generalizations of notions on structures that coincide with the original notion when restricting to the classical case. For this reason, we speak about `binary relations between quantum sets' instead of `quantum binary relations between quantum sets'.}, respectively. The idea behind discrete quantization is that most mathematical structures can be defined in terms of sets and binary relations subject to constraints, hence one can quantize such a structure by replacing each instance of a set in its definition by a quantum set, and any instance of a binary relation between ordinary sets by a binary relation between quantum sets, while requiring that the same constraints hold. Categorically, discrete quantization amounts to the \emph{internalization} of structures in the category $\mathbf{qRel}$ of quantum sets and binary relations. Using discrete quantization, one can easily obtain the notions of \emph{functions between quantum sets} and \emph{partial orders on quantum sets} as noncommutative generalizations of functions between ordinary sets and partial order relations on ordinary sets, respectively. 

Quantum sets and functions between quantum sets form a category $\qSet$ that can be regarded as a noncommutative generalization of the category $\Set$. This category $\qSet$ shares many categorical properties with $\Set$, for instance it is complete and cocomplete, but instead of cartesian closed, it is symmetric monoidal closed, reflecting the linear character of quantum data. This pattern seems to be typical, also the cartesian closed category $\POS$ of posets and monotone maps has a noncommutative counterpart $\qPOS$ that is symmetric monoidal closed. The objects of $\qPOS$ are called \emph{quantum posets}, and consist of a quantum set equipped with a partial order relation. 

In Section \ref{sec:modeling physical systems}, we explore what computations can be modelled with $\qSet$ in the absence of recursion. Here, models based on quantum sets follow the same pattern as set-theoretic models of ordinary programming languages without recursion. $\qSet$ models pure quantum computations just as $\Set$ models pure classical computations. As in the classical case, where impure computations (i.e., side effects) are described by monads on $\Set$, impure quantum computations, typically represented by quantum channels, can be modelled by a monad $\D$ on $\qSet$ that is the quantum analogue of the (countable) distributions monad.  The \emph{CP-Löwner order}\footnote{The CP-Löwner order is an adaptation of the usual Löwner order to the setting of completely positive maps, and is used in quantum information theory as the appropriate order of quantum channels. On states, the CP-Löwner order coincides with the Löwner order.} emerges as a the counterpart of the pointwise order of subdistributions. Since the subdistributions monad is defined on thee category $\Set$,  subdistributions do not suffice for the support of recursion. In the same spirit, the CP-Löwner order does not support full recursion in quantum computing. Another perspective on why support for recursion fails is that classically-controlled recursion typically takes place during the phase of circuit generation, whereas the CP-Löwner order describes the stage of completion of a computation during the phase of circuit execution. 

In Section \ref{Quantum cpos}, we introduce quantum cpos, defined as quantum posets whose partial order relation generalizes $\omega$-complete partial orders in the noncommutative setting. We also introduce noncommutative generalizations of basic domain-theoretic concepts, such as Scott continuity, and we show that the category $\qCPO$ of quantum cpos and Scott continuous maps is enriched over $\CPO$, complete, and has coproducts. The existence of arbitrary colimits is more complicated, and proven in Section \ref{sec:qCPO is cocomplete}.

In Section \ref{sec:monoidal}, we show that $\qCPO$ is symmetric monoidal closed, hence it supports higher-order functions. Similar to $\qSet$ and $\qPOS$, the monoidal product of $\qCPO$ is not the categorical product, even though it generalizes the cartesian product of ordinary cpos.

In Section \ref{sec:cpos are qcpos}, we investigate the relation between cpos and quantum cpos: We show every cpo is a quantum cpo, and that there is a fully faithful strong monoidal functor $`(-):\CPO\to\qCPO$, which has a right adjoint, hence forming a linear/nonlinear model (see also Definition \ref{def:lnlmodel}).  

The effect of nontermination, necessary for support of recursion, is supported by the construction of the lift monad on $\qCPO$ in Section \ref{sec:lifting}, which we show to be commutative. We introduce pointed quantum cpos and strict Scott continuous maps, and show that they form a category $\qCPO_{\perp!}$ that is equivalent to the Kleisli category of the lift monad on $\qCPO$. Moreover, we show that $\qCPO_{\perp!}$ it is \emph{$\CPO$-algebraically compact}, a strong categorical property characteristic of models of type systems with \emph{recursive types}, i.e., type systems in which types can be defined recursively. This form of recursion is particularly powerful and typically implies recursion at the term level as well. Moreover, we show that there is a linear/nonlinear adjunction between $\CPO$ and $\qCPO_{\perp!}$ and we conclude the article by proving that this adjunction forms model that is both sounds for \emph{ECLNL}, a quantum circuit description language with recursive terms, and sound and computationally adequate for \emph{LNL-FPC}, a type system that combines linear and recursive types, see also the Related work section. It follows that quantum cpos form an appropriate structure to model recursion at the stage of circuit generation. 

Finally, in the Conclusions section, we discuss the possible existence of a probabilistic power domain monad on $\qCPO$, the challenges in finding such a monad, and its connection with support for quantum-controlled recursion.

\subsection{Related work}\label{sec:related work}

Our work is based on two lines of research. The first concerns type systems with recursion for quantum programming languages and their denotational semantics. The starting point for this work is \emph{Proto-Quipper-M} \cite{pqm}, a quantum circuit description language. In collaboration with Vladimir Zamdzhiev, the last two authors introduced \emph{ECLNL} \cite{stringdiagramssemantics}, which is essentially an extension of Proto-Quipper-M with recursive terms, and \emph{LNL-FPC} \cite{lnl-fpc-lmcs}, which can be viewed both as an extension of Plotkin's Fixpoint Calculus (a blueprint for a recursive type system) with linear types, and as an extension of the circuit-free fragment of Proto-Quipper-M with recursive types. In this article, we construct a model in terms of quantum cpos that is sound for ECLNL and both sound and computationally adequate for LNL-FPC.  

As mentioned above, models for quantum computing are often obtained by applying categorical constructions to a category $\mathbf M$ that models first-order quantum processes. For instance, in \cite{stringdiagramssemantics}, where ECLNL is defined, a presheaf model on $\mathbf M$ was proposed, where  $\mathbf M$ was taken to be the category of finite-dimensional matrix algebras and completely positive maps. Moreover, in \cite{quant-semantics}, a model for Selinger and Valiron's \emph{quantum lambda calculus} \cite{selingervaliron:quantumlambda} was constructed, where $\mathbf M$ was taken to be the category of finite-dimensional matrix algebras and completely positive maps. The eventual model is obtained by first taking a suitable full subcategory of the Karoubi envelope of $\mathbf M$, then enriching it over the category $\mathbf{DCPO}$ of dcpos via a change-of-basis operation, and finally taking the free finite biproduct completion. The resulting categorical model might be difficult to analyse because of this use of several categorical constructions on top of each other. The quantum lambda calculus and FPC were combined in the language \emph{Quantum FPC}, which was also modelled by a presheaf category \cite{QuantumFPC}. We expect that the existence of a probabilistic power domain monad on $\qCPO$ will allow us to construct an alternative, more transparent models for the quantum lambda calculus and Quantum FPC in terms of quantum cpos.

The other line of research, mathematical quantization, is a rich subject with many applications in quantum information theory. Our preferred quantization method, discrete quantization, is based on the notion of a quantum relation between \emph{von Neumann algebras}\footnote{A class of operator algebras that can be regarded as noncommutative generalizations of measure spaces}, which was distilled by Weaver \cite{Weaver10} from his work with Kuperberg on the quantization of metric spaces \cite{kuperbergweaver:quantummetrics}. 
Their \emph{quantum metric spaces} can be regarded as an internal version of metric spaces in the category $\mathbf{WRel}$ of von Neumann algebras and quantum relations. For example, Kuperberg and Weaver showed that the quantum Hamming distance used in quantum error correction can be understood as a quantum metric \cite{kuperbergweaver:quantummetrics}. 

We will refer to the internalization of structures in $\mathbf{WRel}$ as \emph{W*-quantization}. This process is closely related to discrete quantization, as follows from the work of the first author, who  showed that the category $\qSet$ is dually equivalent to the category $\mathbf{WStar}_\mathrm{HA}$ of \emph{hereditarily atomic von Neumann algebras}\footnote{Von Neumann algebras that are isomorphic to $\ell^\infty$-sums of matrix algebras.} and normal $*$-homomorphisms \cite{Kornell18}, and that there is an inclusion of the category $\qRel$ of quantum sets and binary relations into the category $\mathbf{WRel}$ of von Neumann algebras and quantum relations with hereditarily atomic von Neumann algebras as essential image \cite[Section A.2]{Kornell-Discrete-I}. 
In this work, we employ discrete quantization, because $\qRel$ has better categorical properties than $\mathbf{WRel}$\footnote{For instance , as shown by the first author $\qRel$ is compact closed \cite{Kornell18}, which is a fundamental property in the program of Categorical Quantum Mechanics \cites{AbramskyCoecke08,coeckekissinger,heunenvicary}, whereas $\mathbf{WRel}$, is not.}, and because most quantum computations only involve qubits, for which $\qRel$ suffices.  

\emph{Quantum posets} and \emph{quantum graphs} are other examples of quantum structures that can be obtained via both quantization processes \cite{Weaver10}. Quantum graphs were introduced as noncommutative confusability graphs for quantum error correction \cite{duanseveriniwinter}, and Weaver generalized this notion to arbitrary von Neumann algebras \cites{Weaver10, Weaver21}. The first author showed that these structures, as well as discrete quantum groups can be characterized in terms of the internal logic of this category \cites{Kornell-Discrete-I,Kornell-Discrete-II}. The categorical properties of quantum posets were explored by the authors in \cite{KLM20}. Building on this work and in collaboration with Gejza Jen\v{c}a, the second author introduced a noncommutative generalization of complete lattices called \emph{quantum suplattices} \cite{qSup}. This work heavily relies on the compact-closed structure of $\qRel$.

Von Neumann algebras provide prototypical noncommutative models for quantum computation. A first step in this direction was given by the first author, who established that the category of von Neumann algebras and normal $*$-homomorphisms is monoidal closed, thereby showing that higher-order pure computations can already be captured in this operator-algebraic setting \cite{Kornell17}. In \cite{cho:semantics}, it was shown that the category $\mathbf{WCPSU}$ of von Neumann algebras and completely positive subunital maps is $\mathbf{DCPO}$-enriched with respect to the CP-Löwner order, can be used to model the language \emph{QPL}. It was also shown that the category models an extension of QPL with inductive types \cite{PPRZ19}. In \cite{ChoWesterbaan16}, it was shown that the inclusion of $\mathbf{WStar}$ into $\mathbf{WCPSU}$ has an adjoint, and that the resulting adjunction can be used to model the recursion-free fragment of the quantum lambda calculus. In Section \ref{sec:modeling physical systems}, we will see that when restricted to hereditarily atomic von Neumann algebras, this adjunction yields a monad on $\qSet$ that is the quantum version of the subdistributions monad.

\section{Preliminaries}\label{sec:preliminaries}
\subsection{Quantum sets and binary relations}
In this section, we introduce quantum sets, binary relations and functions between quantum sets, and order relations on quantum sets. We recall the convention that we only use the modifier `quantum' to noncommutative generalizations of structures such as quantum sets, quantum graphs or quantum posets. We do not use the modifier for noncommutative generalizations of notions on structures that coincide with the original notion when restricting to the classical case, such as binary relations and orders.  We refer to \cite{Kornell18} and \cite{KLM20} for a detailed discussion including proofs of the statements in these subsections. 

The most transparent formulation of the definition of $\qRel$, and of quantum sets in particular, is  in terms of a known categorical construction. For this, we must first define the category $\mathbf{FdOS}$

\begin{definition}[$\mathbf{FdOS}$]
The category $\mathbf{FdOS}$ has objects that are nonzero finite-dimen\-sion\-al Hilbert spaces. A morphism $A:X\to Y$ in $\mathbf{FdOS}$ is a \emph{concrete operator space}, i.e., a subspace of the vector space $L(X,Y)$ of linear maps $X\to Y$. The composition of $A$ with a morphism $B:Y\to Z$ is the operator space $B\cdot A:=\mathrm{span}\{ba:a\in A,b\in B\}$. The identity morphism on $X$ is the operator space $\CC 1_X$. 
\end{definition}

The space $L(X,Y)$ is actually a finite-dimensional Hilbert space  via the inner product $(a,b)\mapsto\mathrm{Tr}(a^\dag b)$, where $a^\dag:Y\to X$ denotes the hermitian adjoint of $a\in L(X,Y)$. Hence, the homset $\mathbf{FdOS}(X,Y)$ is a complete modular ortholattice, where the order on the homset is explicitly given by $A\leq B$ if $A$ is a subspace of $B$. Since composition in $\mathbf{FdOS}(X,Y)$ preserves suprema, $\mathbf{FdOS}$ is enriched over the category $\mathbf{Sup}$ of complete lattices and suprema-preserving functions;\footnote{The base of an enriched category should be a monoidal category; for a discussion of the monoidal structure of $\mathbf{Sup}$, see \cite{eklund2018semigroups}} any such category is also called a \emph{quantaloid}. Products and coproducts in a quantaloid $\mathbf Q$ coincide and are also called \emph{sums}. The free sum-completion of $\mathbf Q$ can be described by the quantaloid $\mathrm{Matr}(\mathbf Q)$, whose objects are $\mathbf{Set}$-indexed families of objects $(X_i)_{i\in I}$ of $\mathbf Q$, and whose morphisms $R:(X_i)_{i\in I}\to(Y_j)_{j\in J}$ are `matrices' whose $(i,j)$-component $R(i,j)$ is a $\mathbf Q$-morphism $X_i\to Y_j$.\footnote{The matrix-construction as the free sum-completion of quantaloids was introduced in \cite{Heymans-Stubbe} and is a special case of a matrix-construction for more general bicategories as described in \cite{BCSW}.} 
Composition is defined by `matrix multiplication': we define $S\circ R$ for $S:(Y_j)_{j\in J}\to (Z_k)_{k\in K}$ by $(S\circ R)(i,k)=\bigvee_{j\in J}S(j,k)\cdot R(i,j)$ for each $i\in I$ and $k\in K$, where $\cdot$ denotes the composition of morphisms in $\mathbf Q$. The $(i,i')$-component of the identity morphism on an object $(X_i)_{i\in I}$ is the identity morphism of $X_i$ if $i=i'$ and is $0$ otherwise. The order on the homsets of $\mathrm{Matr}(\mathbf Q)$ is defined componentwise.

\begin{definition}[$\qRel$]\label{def:qRel}
The category $\mathbf{qRel}$ is the quantaloid $\mathrm{Matr}(\mathbf{FdOS})$. An object $\X=(X_i)_{i\in I}$ of $\qRel$ is called a \emph{quantum set}. For convenience and without loss of generality, we will assume that $X_i\neq X_j$ for each $i\neq j$ (but we allow that possibly $X_i\cong X_j)$, so there is a 1-1 correspondence between $I$ and the set $\At(\X):=\{X_i:i\in I\}$, which allows us to write $\X=(X)_{X\in\At(\X)}$. We will refer to the elements of $\At(\X)$ as \emph{atoms}, and write $X\atomof\X$ if $X\in\At(\X)$. 
It is important to maintain a formal distinction between $\X$ and $\At(\X)$: $\At(\X)$ is a set, while $\X$ should be regarded as a quantum generalization of a set; we do not interpret a $d$-dimensional atom of $\X$ as an element, but as an indecomposable subset of $\X$ consisting of $d^2$ elements that are inextricably clumped together. As a consequence, two quantum sets can have equinumerous sets of atoms, but they are isomorphic as quantum sets only if there is a bijection between their sets of atoms that respects the dimensions of the atoms. 
If a quantum set consists of a single atom, it is said to be \emph{atomic}. We often denote atomic quantum sets by $\H$. If the dimension $d$ of the atom of the atomic quantum set is relevant, we write $\H_d$ instead of $\H$.
\end{definition}

The atomic quantum set $\H_d$ can be used to represent a \emph{qudit}, i.e., a $d$-level quantum system. In particular, a qubit can be represented by $\H_2$. Given an arbitrary quantum set $\X$, its atoms  intuitively represent the superselection sectors\footnote{The superselection sectors of a quantum system represented by a C*-algebra $A$ correspond to the (equivalence classes of the) irreducible representations of $A$ \cite[Section III.C]{haag-kastler}} of the discrete quantum system represented by $\X$. For any quantum set $\X$, we can consider the C*-algebra $A=\bigoplus_{X\in\At(\X)}L(X)$, in which case these representations are precisely the representations $A\to L(X)$ for $X\in\At(\X)$.
Formally, $\At(\X)$ is an index set for the quantum set $\X$,
For this reason, we use the notation $X\atomof\X$ to express that $X$ is an atom of $\X$.  

To summarize, to any quantum set $\X$, we associate an ordinary set $\At(\X)$, whose elements are the atoms of $\X$, so $X\atomof\X$ if and only if $X\in\At(\X)$. Conversely, to any ordinary set $M$ consisting of finite-dimensional Hilbert spaces, we associate the unique quantum set $\Q M$ whose set of atoms $\At(\Q M) = \{H\in M\mid \dim(H) > 0\}$.

Several set-theoretic concepts can also be generalized to the quantum setting. We call a quantum set $\X$ \emph{empty} if $\At(\X)=\emptyset$; we call $\X$ \emph{finite} if $\At(\X)$ is finite; we write $\X\subseteq\Y$ if $\At(\X)\subseteq\At(\Y)$, in which case $\X$ is called a \emph{subset} of $\Y$. The union $\X\cup\Y$ is defined to be the quantum set $\Q(\At(\X)\cup\At(\Y))$, and the \emph{Cartesian product} $\X\times\Y$ of quantum sets is defined to be the quantum set $\Q\{X\otimes Y:X\atomof\X,Y\atomof\Y\}$, where $\otimes$ denotes the Hilbert space tensor product.

A morphism $R:\X\to\Y$ in $\qRel$ is called a \emph{binary relation} or more simply a \emph{relation}; as we have seen, it assigns to each $X\atomof \X$ and each $Y\atomof\Y$  an operator space $R(X,Y):X\to Y$. Sometimes we will write $R_X^Y$ instead of $R(X,Y)$, a notation that improves the readability of large calculations. Given our conventions, composition with $S:\Y\to\Z$ is then given by $(S\circ R)(X,Z)=\bigvee_{Y\atomof\Y}S(Y,Z)\cdot R(X,Y)$. We denote the identity morphism of $\X$ by $I_\X$, whose only nonvanishing components are $I_\X(X,X)=\CC 1_X$ for $X\atomof\X$.  The order on the homset $\qRel(\X,\Y)$ is given by $R\leq S$ if and only if $R(X,Y)\leq S(X,Y)$ for all $X\atomof\X$ and $Y\atomof\Y$.

The following theorem succinctly summarizes several relevant properties of this category:
\begin{theorem*}\cite{Kornell18}*{Theorem 3.6}
$\qRel$ is dagger compact.
\end{theorem*}
To make the dagger compact structure of $\qRel$ more explicit, we first extend the Cartesian product $\times$ of quantum sets to a monoidal product on $\qRel$ as follows. Given relations $R:\X_1\to\Y_1$ and $S:\X_2\to\Y_2$, we define $R\times S$ by $(R\times S)(X_1\otimes X_2,Y_1\otimes Y_2)=R(X_1,Y_1)\otimes S(Y_1,Y_2)$. The monoidal unit is the quantum set $\mathbf 1:=\Q\{\CC\}$. The dagger $R^\dag:\Y\to\X$ of a relation $R:\X\to\Y$ is given by $R^\dag(Y,X)=\{r^\dag:r\in R(X,Y)\}$, where $r^\dag$ denotes the hermitian adjoint of the linear map $r:X\to Y$. The dual $\X^*$ of a quantum set $\X$ is the quantum set $\Q\{X^*:X\atomof\X\}$, where $X^*$ is the Banach space dual of the Hilbert space $X$.

We emphasize that a relation $R$ in $\qRel$ is not a binary relation in the traditional sense; it is not a set of ordered pairs of atoms. Binary relations between quantum sets can be regarded as a generalization of ordinary binary relations because we have a fully faithful functor $`(-):\mathbf{Rel}\to\mathbf{qRel}$ (cf.~Section \ref{sec:quantum sets generalize ordinary sets}). We also stress the difference in terminology between quantum relations and binary relations. The term `quantum relation' was introduced by Weaver and Kuperberg, who do not follow our convention on the modifier `quantum'. The term refers to a morphism between von Neumann algebras, whereas the terms `relation' and `binary relation' refer to morphisms between sets or quantum sets. These classes of morphisms form categories that are equivalent but not equal: indeed, each quantum set $\X$ corresponds to the hereditarily atomic von Neumann algebra $\ell^\infty(\X):=\bigoplus_{X\atomof\X}L(X)$, where $\bigoplus$ denotes the $\ell^\infty$-sum, and this object-level correspondence extends to an equivalence between $\mathbf{qRel}$, whose morphisms are relations, and $\cat{WRel}_{\mathrm{HA}}$, whose morphisms are normal unital $\ast$-homomorphisms~\cite[Propositions A.2.1 \& A.2.2]{Kornell-Discrete-I}. Since dagger categories are self-dual, we also obtain an equivalence between $\qRel$ and $\cat{WRel}_{\textrm{HA}}^\op$, whose morphisms are quantum relations.

\subsection{Conventions} We will generally use calligraphic letters, e.g., $\X$, to name quantum sets and the corresponding capital letters, e.g., $X$, to name their atoms. In particular, if $\X$ is an atomic quantum set, then we will use $X$ to refer to its unique atom. We adopt the convention that $\H$ always denotes an atomic quantum set. As usual, the Dirac delta symbol $\delta_{a,b}$ names the complex number $1$ if $a = b$ and otherwise names $0$. Similarly, the symbol $\Delta_{a,b}$ names the maximum binary relation on $\mathbf 1$ if $a=b$ and otherwise names the minimum binary relation on $\mathbf 1$. Thus, $\delta_{a,b} \in \Delta_{a,b}(\CC,\CC)$. The binary relation $\Delta_{a,b}$ is a scalar in the symmetric monoidal category $\cat{qRel}$ \cite{AbramskyCoecke08}. Hence, for each binary relation $R$ from a quantum set $\X$ to a quantum set $\Y$, we write $\Delta_{a,b}R$ for the binary relation from $\X$ to $\Y$ obtained by composing $\Delta_{a,b} \times R$ with unitors in the obvious way. Finally, we will reserve the notation $\leq$ for the partial order on the homsets of $\qRel$. For any other partial order, we use the notation $\sqsubseteq$. We use the adjective `ordinary' to emphasize that we are using a term in its standard sense. Thus, an ordinary poset is just a poset.

\subsection{Quantum sets as a generalization of ordinary sets}\label{sec:quantum sets generalize ordinary sets}
To each ordinary set $S$, we associate a quantum set $`S$ as follows. For each $s\in S$, we introduce a one-dimensional Hilbert space $\CC_s$ such that $\CC_s\neq\CC_t$ for $s\neq t$. For example, we might define $\CC_s:=\ell^2(\{s\})$. This allows us to define $`S$ formally as the quantum set $\Q\{\CC_s:s\in S\}$. 

We can now extend the assignment $S\mapsto `S$ to a functor $`(-):\cat{Rel}\to\qRel$ as follows. Given a binary relation $r:S\to T$ between ordinary sets $S$ and $T$, we define $`r:`S\to`T$ to be the relation given by 
\[ `r(\CC_s,\CC_t)=\begin{cases} L(\CC_s,\CC_t), & (s,t)\in r,
\\
0, & \text{otherwise}.
\end{cases}\]
Since $L(\CC_s,\CC_t)$ is one dimensional, it follows easily that $`(-):\cat{Rel}\to\qRel$ is a fully faithful functor, which justifies our view that binary relations between quantum sets generalize the ordinary binary relations. Furthermore, $`(-)$ preserves the other relevant structure: it is a strong monoidal functor that preserves the dagger, and the order structure of the homsets. 

The quote notation may look a bit odd, but there is a reason for this choice of notation: we want to treat the objects and morphisms of $\cat{Rel}$ as special cases of objects and morphisms of $\qRel$, while still keeping a formal distinction. Hence, the notation for the embedding $\cat{Rel}\to\qRel$ should not be too prominent visually, and this is exactly what the quote notation achieves.

\subsection{Discrete quantization via quantum sets}\label{sec:discquant}
The following summarizes results from \cite{Kornell18}, where proofs can be found. Since $\qRel$ is a dagger category that is equivalent to $\cat{WRel}_\mathrm{HA}$, we could equivalently define discrete quantization to be the process of internalizing structures in $\qRel$. The basic principle is the following: given a category $\mathbf C$ whose objects and morphisms can be described in terms of sets, binary relations, and constraints expressed in the language of dagger compact quantaloids, we define a category $\mathbf{qC}$ by replacing any set in the definition of $\mathbf C$ by a quantum set, and any binary relation between sets by a binary relation between quantum sets such that the original constraints still hold. In many cases, the functor $`(-):\cat{Rel}\to\qRel$ induces a fully faithful functor $\mathbf C\to\mathbf{qC}$, which we also notate $`(-)$. If this fully faithful functor exists, then the objects in $\mathbf{C}$ can also be regarded as objects in $\mathbf{qC}$.
Often, $\mathbf{qC}$ has the same categorical properties as $\mathbf C$, except that if $\mathbf C$ has a categorical product, this product translates to a symmetric monoidal product on $\mathbf{qC}$, due to the quantum character of the objects of the latter category.

As an example, a function $f:S\to T$ between ordinary sets can be described in terms of the dagger compact quantaloid $\cat{Rel}$ as a morphism such that $f^\dag\circ f\geq 1_S$ and $f\circ f^\dag\leq 1_T$, where the dagger of a binary relation in $\cat{Rel}$ is just the converse relation. In the same way, we can define a function $F:\X\to\Y$ between quantum sets as a binary relation that satisfies $F^\dag\circ F\geq I_\X$ and $F\circ F^\dag\leq I_\Y$. We call the category of quantum sets and functions $\qSet$. We call a function $F:\X\to\Y$ \emph{injective} if $F^\dag\circ F=I_\X$, \emph{surjective} if $F\circ F^\dag=I_\Y$, and \emph{bijective} if it is both injective and surjective, in which case $F^\dag$ is the \emph{inverse} of $F$. 
The injective functions coincide with the monos in $\qSet$; the surjective functions coincide with the epis in $\qSet$, and the bijections are precisely the isomorphisms of $\qSet$. 

An important example of an injective function occurs in case that a quantum set $\X$ is a subset of a quantum set $\Y$. Then we define the canonical injection $J_\X^\Y:\X\to\Y$ by $J_\X^\Y(X,Y)=\CC1_X$ if $X=Y$, and $J_\X^\Y(X,Y)=0$ otherwise, for $X\atomof \X$ and $Y\atomof \Y$. When the ambient quantum set $\Y$ is clear, we often write $J_\X$ instead of $J_\X^\Y$. Given a relation $R:\X\to\Y$ between quantum sets, we define the \emph{restriction} $R|_\W:=\W\to\X$ of $R$ to a subset $\W\subseteq\X$ as the quantum relation $R\circ J^\X_\W$. Likewise, we define the \emph{corestriction} $R|^\Z:\X\to\Z$ of $R$ to a subset $\Z\subseteq\Y$ to be the quantum relation $(J_\Z^\Y)^\dag\circ R$. We write $R|_\W^\Z$ for relation $(J_\Z^\Y)^\dag\circ R\circ J_\W^\X$.

The symmetric monoidal product $\times$ on $\qRel$ restricts to a symmetric monoidal product on $\qSet$, which we also notate $\times$, even though this is not the categorical product on $\qSet$ either. We use this notation because this monoidal product on $\qSet$ is the correct quantum generalization of the categorical product $\times$ on $\Set$. Moreover, $\qSet$ is closed, and $\Y^\X$ denotes the internal hom between quantum sets $\X$ and $\Y$.

The following theorem, whose content is taken from \cite{Kornell18}, summarizes the categorical properties of $\qSet$:
\begin{theorem}\label{thm:qSet} \cite{Kornell18}
     $\qSet$ is complete, cocomplete, and symmetric monoidal closed. Moreover, $`(-):\cat{Rel}\to\qRel$ restricts and corestricts to a fully faithful functor $`(-):\Set\to\qSet$. Finally, the assignment $\X\mapsto \ell^\infty(\X)$ extends to an equivalence $\ell^\infty:\qSet\to\cat{WStar}_\mathrm{HA}^\op$. Given a quantum function $F:\X\to\Y$, we denote the corresponding normal unital $*$-homomorphism $\ell^\infty(\Y)\to\ell^\infty(\X)$ by $F^\star$.
\end{theorem}
We regard ordinary sets as a special case of quantum sets via the functor $`(-)$.

\subsection{Modeling physical systems and quantum computations via quantum sets}\label{sec:modeling physical systems}
Finitary physical types are naturally modeled by quantum sets. We can use the atomic quantum set $\H_d$, whose only atom is a $d$-dimensional Hilbert space (cf. Definition \ref{def:qRel}) to model a qudit. In particular, the qubit can be modeled by the quantum set $\H_2$. By contrast, the classical bit is modeled by the quantum set $`\{0,1\}$, which has two one-dimensional atoms. The `memory' of an idealized quantum computer might consist of finitely many qubits, and finitely many bits, so we can model such a quantum computer as a composite physical system.

The symmetric monoidal product $\times$ on $\mathbf{qRel}$ and $\mathbf{qSet}$ generalizes the Cartesian product of ordinary sets. Composite systems consisting of two fully quantum systems, that is, of quantum systems modeled by single Hilbert spaces, are obtained by forming the tensor product of the two Hilbert spaces. Similarly, composite systems consisting of two fully classical systems, that is, of quantum systems modeled by ordinary sets, are obtained by forming the ordinary Cartesian product of the two sets. Thus, our generalized product models composite systems consisting entirely of fully quantum systems, or entirely of fully classical systems. In fact, it is appropriate for modeling mixed quantum/classical systems as well. In particular, an idealized quantum computer possessing $n$ qubits and $m$ bits can be modeled by the quantum set
$$ \underbrace{\H_2 \times \cdots \times \H_2}_{n} \times \underbrace{`\{0,1\} \times \cdots \times `\{0,1\}}_{m}.$$
This quantum set has $2^m$ atoms, each of dimension $2^n$.

Other natural type constructors are also easily modeled in quantum sets. Sum types are modeled by disjoint unions of quantum sets, portraying a kind of classical disjunction of potentially quantum systems. The resulting physical system may be in a state of the first type or in a state of the second type, but no superpositions may occur. For example, an idealized quantum computer possessing a single qubit and a single bit is modeled by a quantum set with two atoms, because the computer can be in a configuration where the qubit has value $0$, and it can be in a configuration where the qubit has value $1$, but it cannot be in a superposition between two such states. The exponential modality $!$ that permits duplication, is modeled by an operation on quantum sets that extracts just their one-dimensional atoms, so $!\X=\X_1$, where \[\X_1:=\Q\{X\in\At(\X):\dim(X)=1\}.\] This is essentially an ordinary set, and it is the largest subset that admits a duplication map. Indeed, classical sets all admit duplication maps, but the no-cloning theorem forbids the duplication of higher dimensional atoms. A useful result that characterizes the atoms of $\X_1$ is the following:

\begin{lemma}\label{lem:one-dimensional atoms}
Let $\X$ be a quantum set. Then we have the following canonical bijections:
\begin{itemize}
\item $b_\X\: \At(\X_1) \to \cat{qSet}(\mathbf 1, \X)$, which maps each one-dimensional atom $X \atomof \X$ to the function $b_\X(X)\: \mathbf 1 \to \X$ defined by $b_\X(X)(\CC, X) = L(\CC, X)$, with the other components vanishing;
\item $B_\X:\X_1\to`\qSet(\mathbf 1,\X)$ defined by $B_\X(X,\CC_{b_\X(X)}) = L(X,\CC_{b_\X(X)}) $, for $X \atomof \X_1$, with the other components vanishing. 
\end{itemize}
\end{lemma}
\begin{proof}It is easy to see that $b_\X(X):\mathbf 1\to\X$ is a function for each atom $X$ of $\X$. Using the duality between $\qSet$ and $\mathbf{WStar}_\mathrm{HA}$ (cf. Theorem \ref{thm:qSet}), it easily follows from \cite{Kornell18}*{Proposition 7.5} that all functions $\mathbf 1\to\X$ are of this form. The binary relation $B_\X$ is a bijection by \cite{Kornell18}*{Proposition 4.4}
\end{proof}

We recall that the type systems of many quantum programming languages are based on intuitionistic linear logic. We further note that this logic can be modeled by \emph{linear/nonlinear} models~\cite{benton}:
\begin{definition}[Linear/nonlinear model]\label{def:lnlmodel}
    A \emph{linear/nonlinear model} consists of a cartesian closed category $\mathbf V$, a symmetric monoidal closed category $\mathbf C$ and a strong monoidal functor $\mathbf V\to\mathbf C$ that has a right adjoint. The category $\mathbf V$ is called the \emph{non-linear} category, whereas $\mathbf C$ is called the \emph{linear} category.
\end{definition}
Given such a model, the exponential modality $!\colon \mathbf C\to\mathbf C$ is the comonad induced by the linear/nonlinear adjunction. For instance, the functor $`(-):\Set\to\qSet$ described in Theorem \ref{thm:qSet} has a right adjoint given by $\X\mapsto\qSet(\mathbf 1,-)$, hence $\Set$ and $\qSet$ form a linear/nonlinear model.\footnote{We omit the proof, which can be distilled from the proof of Theorem \ref{thm:CPO-qCPO adjunction} below. In fact, all complicated details in the proof of \ref{thm:CPO-qCPO adjunction} are due to the order relation, which is irrelevant in case of $\qSet$.} By Lemma \ref{lem:one-dimensional atoms}, the comonad induced by the linear/nonlinear adjunction between $\Set$ and $\qSet$ is precisely the operation described above of extracting the one-dimensional atoms from a quantum set. 

Higher types are modeled by quantum function sets, which are more challenging to describe, mainly due to the large automorphism group of the physical qubit. Indeed, the ordinary set $\mathbf{qSet}(\H_2, \H_2)$ is in canonical bijection with the rotation group $\mathrm{SO}(3)$. The quantum function set $\H_2^{\H_2}$ also has atoms of dimension higher than $1$. In general, the quantum function set from a quantum set $\X$ to a quantum set $\Y$ has an atom of dimension $d$ for each unital normal $*$-homomorphism $\rho\: \ell^\infty(\Y) \To \ell^\infty(\X) \overline \otimes L(H_d)$ that is irreducible in the sense that $\rho(\ell^\infty(\Y))' \cap (\CC \overline \otimes L(H_d)) = \CC$, where $\overline \tensor$ denotes the spatial tensor product of von Neumann algebras, and $(\,\cdot \,)'$ denotes the commutant.

\begin{example}
      \emph{Proto-Quipper-M}, introduced by Rios and Selinger \cite{pqm}, is a circuit description language that describes the generation of circuits, represented by a small symmetric monoidal category $\mathbf M$. In case of quantum circuits, one can take $\mathbf M$ to be the opposite of the category whose objects are finite sums of matrix algebras, and whose morphisms are $*$-homomorphisms. Here, a qudit is represented by $M_d(\mathbb C)$. The tensor product of vector spaces is used to represent combined systems of several qudits.

      An abstract axiomatization of models of Proto-Quipper-M was given in \cite{stringdiagramssemantics}: such a model consists of a symmetric monoidal closed category $\mathbf C$ with coproducts enriched over a cartesian closed category $\mathbf V$ with coproducts such that $\mathbf V$ and $\mathbf C$ form a linear/nonlinear model, 
      satisfying $\mathbf M$ embeds monoidally into $\mathbf C$. Since $\mathbf{Set}$ and $\mathbf{qSet}$ form a linear/nonlinear model with coproducts, and $\mathbf M$ can be a regarded as a subcategory of $\mathbf{WStar}_\mathrm{HA}$, $\qSet$ clearly forms a model of Proto-Quipper-M.
\end{example}

\begin{example}\label{ex:pure}
\emph{Pure quantum computation} refers to the model of quantum computation in which the evolution of a system is described entirely by quantum gates, typically formalized by unitary operators \cite{huotstaton:qptheory}. In this pure setting, a quantum circuit is just the sequential composition of such unitaries. There is no measurement, randomness, or environmental interaction: computation proceeds coherently and reversibly, so quantum gates are automorphisms of fully quantum systems, which consist of qubits. In the category of quantum sets and functions,  such an automorphism is formalized by a closely related function between quantum sets. Indeed, each function $F_1\: \H_d \To \H_d$, for any positive integer $d$, is defined by $F_1(H_d, H_d) = \CC\cdot u$, for some unitary operator $u$. For example, the Hadamard gate is formalized by the function $F_1\: \H_2 \To \H_2$ defined by $$F_1(H_2, H_2) = \CC \cdot \left[ \begin{matrix} 1 & 1 \\ 1 & -1 \end{matrix} \right].$$
This function $F_1$ is an automorphism of $\H_2$ in $\qSet$. 
The unital normal $*$-homomorphism $M_2(\CC) \iso \ell^\infty(\H_2) \To \ell^\infty(\H_2) \iso M_2(\CC)$ that corresponds to this function in the sense of Theorem \ref{thm:qSet} is simply conjugation by the Hadamard matrix above, appropriately normalized.

\end{example}

\begin{example}\label{ex:measurement}
Measurement is a channel from a fully quantum system to a fully classical system, and cannot be described by a unitary operator. As a consequence, it is an \emph{impure} quantum operation, which in $\qSet$ is formalized by a function from an atomic quantum set to a classical quantum set. 
It follows from the definition of a function between quantum sets that functions from an atomic quantum set $\H_d$ to a classical quantum set $`S$ are in bijective correspondence with projection-valued measurements on $\H_d$. Explicitly, $F_2(H_d, \CC_s) = L(H_d, \CC_s) \cdot  p_s$ for each $s \in S$, where $(p_s: s \in S)$ is a projective POVM.\footnote{POVM stands for \emph{Positive Operator-Valued Measurement}.} In particular, the standard measurement of a qubit is formalized by a function from $\H_2$ to $`\{1, -1\}$ defined by
$$F_2(H_2, \CC_s) = \begin{cases} \CC \cdot [\begin{matrix} 1 & 0 \end{matrix} ] & {s = 1}; \\ 
\CC \cdot [\begin{matrix} 0 & 1 \end{matrix}]  & s = {-1}.    \end{cases}$$
This function $F_2$ is an epimorphism in $\qSet$. Equivalently, it is surjective in the sense that $F_2 \circ F_2^\dagger \geq I$.
The unital normal $*$-homomorphism $\CC^2 \iso \ell^\infty(\{1, -1\}) \To \ell^\infty(\H_2) \iso M_2(\CC)$ corresponding to this function in the sense of Theorem \ref{thm:qSet} includes $\CC^2$ into $M_2(\CC)$ diagonally.
\end{example}

\begin{example}\label{ex:distributions-monad}
Generally, quantum computations are described by \emph{quantum channels}, i.e., completely positive unital maps\footnote{Here, we assume the Heisenberg picture of quantum physics. In the Schrödinger picture, which is dual to the Heisenberg picture, quantum channels are described by completely positive trace-preserving maps.}. Classically, impure computations correspond to side effects, which are modeled by monads. In a similar way, quantum channels can be described by Kleisli morphisms of a monad $\D$ on $\qSet$. This monad is induced by the adjunction between the category $\mathbf{WStar}_\mathrm{HA}^\op$, which is equivalent to $\qSet$, and  the category $\mathbf{WCPU}_\mathrm{HA}^\op$, the opposite of the category of hereditarily atomic von Neumann algebras and normal completely positive unital maps. Here, the left adjoint is given by the canonical  embedding of the former category into the latter. Thus, for any two quantum sets $\X$ and $\Y$, we have a natural bijection $\qSet(\X,\D(\Y))\cong\mathbf{WCPU}_\mathrm{HA}(\ell^\infty(\Y),\ell^\infty(\X))$, 
arbitrary quantum channels can be modeled in $\qSet$.

If we identify the states on $\X$ with normal completely positive unital functionals on $\ell^\infty(\X)$, then it follows from Lemma \ref{lem:one-dimensional atoms} that \[\mathrm{States}(\X)=\mathbf{WCPU}_\mathrm{HA}(\ell^\infty(\X),\mathbb C)=\qSet(\mathbf 1,\D(\X))\cong\At(\D(\X)_1),\] hence by the same lemma, we have an bijection 
\begin{equation}\label{iso:states}
`\mathrm{States}(\X)\cong \D(\X)_1.
\end{equation}
The monad $\D$ on $\qSet$ is a noncommutative version of the (countable) distributions monad $D$ on $\Set$ in the sense that if $S$ is an ordinary set, then there is a bijection $`D(S)\cong \D(`S)_1$, so essentially, $D(S)$ is the classical part of $\D(`S)$. This follows by taking $\X=`S$ in Equation (\ref{iso:states}), and from a standard result in the theory of von Neumann algebras that the predual of the (hereditarily atomic) von Neumann algebra $\ell^\infty(S)$ is isometrically isomorphic to $\ell^1(S)$, the space of functions $f:S\to\CC$ such that $\sum_{x\in S}|f(x)|<\infty$. This isomorphism restricts and corestricts to a bijection $\mathbf{WCPU}(\ell^\infty(S),\CC)\cong D(S)$, where we recognize the left-hand side as $\mathrm{States}(`S)$. 
\end{example}

\begin{example}
State preparation cannot be formalized directly as a function, in the manner of the Examples \ref{ex:pure} and \ref{ex:measurement} Indeed, from the perspective of quantum mechanics, state preparation is a nondeterministic binary channel, but there are no functions at all from $`\{1, -1\}$ to $\H_2$. However, it is possible to formalize state preparation in $\qSet$ using the monad $\D$ of Example \ref{ex:distributions-monad}. In this example, we saw that $\D$ has the property that for each quantum set $\X$, the one-dimensional atoms of $\D(\X)$ correspond bijectively to the states on $\X$, which we may define to be the normal states on $\ell^\infty(\X)$ in the sense of operator algebras. We thus have a monomorphism $`\mathrm{States}(\X) \rightarrowtail \D(\X)$ in $\qSet$.

A simple experimental procedure, which consists of preparing a qubit, applying the Hadamard gate, and then measuring its value is modeled as the following composition of functions in $\qSet$:

\[
\begin{tikzcd}
`\mathrm{States}(\H_2)
\arrow[tail]{r}
&
\D(\H_2)
\arrow{r}{\D(F_1)}
&
\D(\H_2)
\arrow{r}{\D(F_2)}
&
\D(`\{1,-1\}),
\end{tikzcd}
\]
where the injection is obtained from the bijection $`\mathrm{States}(\H_2)\cong\D(\H_2)_1$ in Example \ref{ex:distributions-monad}, and where $F_1$ and $F_2$ are the functions from Examples \ref{ex:pure} and \ref{ex:measurement}, respectively. This yields a function from $`\mathrm{States}(\H_2)$ to $\D(`\{1,-1\})$, or equivalently a function\\
\noindent $\mathrm{States}(\H_2) \To \mathrm{States}(`\{1,-1\})$ that maps each prepared state to the corresponding probability distribution on outcomes. This equivalence is due to the fact that a function between quantum sets necessarily maps one-dimensional atoms to one-dimensional atoms. Type theoretically, it is related to the fact that any term of type $A$ that has only nonlinear parameters lifts to a term of type $!A$. 
\end{example}

\begin{example}\label{ex:subdistributions-monad}
   Incomplete quantum computations are described by normal completely positive subunital maps, which generalize subdistributions, and which are in a 1-1 correspondence with the Kleisli maps of a monad $\S$ on $\qSet$ that is the quantum version of the subdistributions monad $S$ on $\Set$. This monad $\S$ can be obtained in a similar way as the quantum distributions monad on $\qSet$ in Example \ref{ex:distributions-monad}, namely via the adjunction between $\mathbf{WStar}^\op_\mathrm{HA}$ and the category $\mathbf{WCPSU}^\op_\mathrm{HA}$ of hereditarily atomic von Neumann algebras and normal completely positive subunital maps. The existence of this adjunction is proven in ~\cites{ChoWesterbaan16,Westerbaanthesis}. 
   As a consequence, the Kleisli morphisms $\X\to\S(\Y)$ correspond to normal completely positive subunital maps $\ell^\infty(\Y)\to\ell^\infty(\X)$
   
   The Kleisli morphisms of the ordinary subdistributions monad $S$ are the substochastic maps, which can be ordered pointwise. The analog of this order in the quantum setting is a variation of the Löwner order, the so-called \emph{CP-Löwner order} on normal completely positive subunital maps\footnote{Given two parallel completely positive maps $\varphi$ and $\psi$, we define $\varphi\leq\psi$ if and only if $\psi-\varphi$ is completely positive. For states, the CP-Löwner order coincides with the Löwner order.}  It was shown in \cites{cho:semantics} that the category of von Neumann algebras with completely positive subunital maps is $\mathbf{DCPO}$-enriched with respect to the CP-Löwner order\footnote{A similar result with respect to the usual Löwner order was proven in \cite{rennela:domains}.}, and in \cite{JKLMZ} that the Kleisli category $\qSet_\S$ is likewise enriched over continuous domains with respect to this order. Nevertheless, $\qSet_\S$ does not support recursion, for the same reason that $\Set_S$ (which is also enriched over continuous domains) does not support for recursions: in both cases, the base category ($\Set$ and $\qSet)$ is not $\CPO$-enriched. Another explanation is that in classically-controlled quantum computing, we have two runtimes: circuit generation time and circuit execution time. Typically, recursion in  classically-controlled quantum computing unfolds during the circuit generation time, whereas the CP-Löwner describes the stage of completion of a computation during the circuit execution time. The order of quantum cpos, on the other hand, describe the stage of completion of a computation during circuit generation time, and offer better support for recursion in quantum computing. 
\end{example}

\subsection{Quantum posets}\label{sub:quantum posets}
Quantum cpo are in particular \emph{quantum posets}, which are essentially due to Weaver \cite[Definition 2.4]{Weaver10}, and consist of a pair $(\X,R)$ consisting of quantum set $\X$ and a binary relation $R$ on $\X$ that generalizes orders to the noncommutative setting. We start by reviewing the exact notion of a quantum poset. We note that the notions of partial order, monotone function, etc.~introduced here are all direct generalizations of the standard notions of set theory.  

\begin{definition}[Quantum poset]\label{def:quantpos}
We call a binary relation $R$ on a quantum set $\X$
\begin{itemize}
    \item \emph{reflexive} if $I_\X\leq R$;
    \item \emph{transitive} if $R\circ R\leq R$;
    \item \emph{antisymmetric} if $R\wedge R^\dag\leq I_\X$.
\end{itemize}
A relation satisfying these three conditions is an \emph{order}. 
We define a \emph{quantum poset} to be a pair $(\X,R)$ consisting of a quantum set $\X$ equipped with an order $R$. 
\end{definition}

\begin{example}
    Let $\X$ be a quantum set. Then $I_\X$ is easily seen to be an order on $\X$, the \emph{trivial} or \emph{flat} order. 
\end{example}
\begin{example}\cite[Example 1.4]{KLM20}\label{ex:quantum order on H2}
    Let $\H$ be an atomic quantum set whose single atom is named $H$. Then a binary relation $V$ on $\H$ is an order if and only if its single component $V(H,H)$ is a unital algebra $A$ of operators on $H$ that is \emph{anti-symmetric} in the sense that $A\cap A^\dag=\mathbb C1$. We give a more specific example for  $H=\mathbb C^2$. Let $V(H,H)$ be the span of the nilpotent matrix $a = \left[\begin{smallmatrix} 0 & 1 \\ 0 & 0 \end{smallmatrix}\right]$ and the identity matrix $1 = \left[\begin{smallmatrix} 1 & 0 \\ 0 & 1 \end{smallmatrix}\right]$. This is an antisymmetric unital algebra of operators on $\CC^2$, so $(\H,V)$ is a quantum poset.
\end{example}

A function $F:(\X,R)\to(\Y,S)$ between quantum posets is said to be \emph{monotone} if $F\circ R\leq S\circ F$. If $R=F^\dag\circ S\circ F$, we say that $F$ an \emph{order embedding}. A surjective order embedding is called an \emph{order isomorphism}, and is precisely a bijective monotone function whose inverse is also monotone. The category of quantum posets and monotone maps is denoted by $\qPOS$. 
\noindent The elementary properties of this category are established in \cite{KLM20}:

\begin{theorem*}
The category $\cat{qPOS}$ of quantum posets and monotone functions is complete, cocomplete, symmetric monoidal closed, and $\mathbf{POS}$-enriched, where $\POS$ is the category of ordinary posets and monotone functions.  Moreover, the functor $`(-):\POS\to\qPOS$, given by $(S,{\sqsubseteq})\mapsto (`S,`{\sqsubseteq})$ on posets and by $f\mapsto `f$ on monotone maps, is a fully faithful functor.
\end{theorem*}
Here, the monoidal product $(\X,R)\times(\Y,S)$ of two quantum posets $(\X,R)$ and $(\Y,S)$ is the quantum poset $(\X\times \Y,R\times S)$. The monoidal product of two monotone functions is the monoidal product of the two functions when regarded as morphisms of $\qSet$.
\begin{example}
    Let $\overline \NN$ be the set of natural numbers with infinity, ordered in the obvious way by $\sqsubseteq$. By the above theorem, $`{\sqsubseteq}$ is an order on $`\overline\NN=\Q\{\CC_1,\CC_2,\ldots,\CC_\infty\}$ whose non-vanishing components are given by $`{\sqsubseteq}(\CC_r,\CC_s)=L(\CC_r,\CC_s)$ if $r\sqsubseteq s$ in $\overline\NN$. Since every atom of $`\overline{\NN}$ is one-dimensional, and $L(\CC_r,\CC_s)$ is also one-dimensional, we can represent $`{\sqsubseteq}$ by the following matrix:
    \[\left(
\begin{matrix}
`{\sqsubseteq}(\CC_\infty, \CC_\infty) & \cdots & `{\sqsubseteq}(\CC_2, \CC_\infty) & `{\sqsubseteq}(\CC_1, \CC_\infty) \\
\vdots & \ddots &  \vdots &\vdots \\
`{\sqsubseteq}(\CC_\infty, \CC_2) & \cdots & `{\sqsubseteq}(\CC_2, \CC_2) &`{\sqsubseteq}(\CC_1, \CC_2) \\
`{\sqsubseteq}(\CC_\infty, \CC_1) & \cdots & `{\sqsubseteq}(\CC_2, \CC_1) & `{\sqsubseteq}(\CC_1, \CC_1)
\end{matrix}
\right)
: = \left(
\begin{matrix}
\CC & \cdots & \CC & \CC \\
\vdots & \ddots &  \vdots &\vdots \\
0 & \cdots & \CC &  \CC \\
0 & \cdots & 0 &\CC
\end{matrix}
\right).\]
The triangular shape of the matrix of $`{\sqsubseteq}$ reflects that $\sqsubseteq$ is an ordinary order.
\end{example}
\begin{example}\label{ex:orderNtimesH}
    Let $\H$ be the atomic quantum set with single atom $H=\CC^2$, and let $V$ be the quantum order from Example \ref{ex:quantum order on H2}, i.e., $V(H,H)=\CC a+\CC1$, where $a = \left[\begin{smallmatrix} 0 & 1 \\ 0 & 0 \end{smallmatrix}\right]$ and $1$ is the identity on $\CC^2$. Let $\X:=`\overline\NN\times\H$. Then all atoms of $\X$ are of the form $X_i : = \CC_i \tensor \CC^2$ for some $i\in\overline\NN$. If $\sqsubseteq$ denotes the usual order on $\overline\NN$ as in the previous example, then the above theorem implies that $S:=`{\sqsubseteq}\times V$ is an order on $`\overline\NN\times\H$. Since each atom $X_i$ of $\X$ is canonically isomorphic to $\CC^2$, we may depict $S$ as a matrix of matrix subspaces as follows:
    $$
\left(
\begin{matrix}
S(X_\infty, X_\infty) & \cdots & S(X_2, X_\infty) & S(X_1, X_\infty) \\
\vdots & \ddots &  \vdots &\vdots \\
S(X_\infty, X_2) & \cdots & S(X_2, X_2) &S(X_1, X_2) \\
S(X_\infty, X_1) & \cdots & S(X_2, X_1) & S(X_1, X_1)
\end{matrix}
\right)
: =
\left(
\begin{matrix}
\CC a+\CC 1  & \cdots & \CC a + \CC 1 & \CC a + \CC 1 \\
\vdots & \ddots & \vdots & \vdots \\
0 & \cdots & \CC a + \CC 1  & \CC a + \CC 1 \\
0 & \cdots & 0  & \CC a + \CC 1 \\
\end{matrix}
\right).
$$
\end{example}

The following result, from \cite{KLM20} shows the category $\mathbf{qPOS}$ is enriched over $\mathbf{POS}$, using a generalization of the pointwise order on functions from an ordinary set into an ordinary poset.

\begin{lemma}[Lemma~4.1~\cite{KLM20}]\label{lem:qPos}
For two functions $F,G:\X\to\Y$ from a quantum set $\X$ to a quantum poset $(\Y,S)$, the following conditions are equivalent:
\begin{itemize}
    \item $G\leq S\circ F$;
    \item $F\leq S^\dag\circ G$;
    \item $S\circ G\leq S\circ F$;
    \item $S^\dag\circ F\leq S^\dag\circ G$;
    \item $G\circ F^\dag\leq S$.
\end{itemize}
We define $F\sqsubseteq G$ if the above conditions hold. Then $\qSet(\X,\Y)$ is an ordinary poset in the order $\sqsubseteq$.
\end{lemma}
We regard ordinary posets as a special case of quantum posets via the functor $`(-):\mathbf{POS}\to\mathbf{qPOS}$.

\section{Quantum cpos}\label{Quantum cpos}

We next define quantum cpos, which generalize complete partial orders to the quantum setting. To begin, recall that a \emph{complete partial order (cpo)} is a poset in which every monotone increasing sequence has a supremum. A poset $(X,\sqsubseteq)$ is a cpo if and only if the homset $\Set(1,X)$ equipped with the pointwise order is a cpo. This suggests that we could define a quantum cpo to be a quantum poset $(\X,R)$ such that $\qSet(\mathbf 1,(\X,R))$ is a cpo when equipped with the pointwise order of Lemma \ref{lem:qPos}. Unfortunately, this definition is too weak, for several reasons. Firstly, $\Set$ has a single generator, namely $1$, but $\qSet$ is generated by the atomic quantum sets of the form $\H_d$ for $d\in\mathbb N$. The functions $\mathbf 1\to\X$ correspond only to the one-dimensional atoms (see Lemma \ref{lem:one-dimensional atoms}), so to the classical part of $\X$. This is related to the fact that higher-dimensional atoms of $\X$ can be interpreted as subsets of $\X$ consisting of elements that cannot be accessed individually. However,  also requiring that $\qSet(\H_d,\X)$ is a cpo for each $d\in\mathbb N$ yields a definition that is too weak, because this amounts to taking `external' suprema of monotonically ascending sequences, i.e., purely defined in terms of homsets, whereas a proper quantization of cpos requires an `internal' notion of suprema of such sequences.

To find the correct definition of a quantum cpos, we examine the classical case more closely. Let $(X,{\below})$ be an ordinary poset, let $W$ be an ordinary set, and let $k_1\sqsubseteq k_2\sqsubseteq \cdots$ be a monotonically increasing sequence of functions from $W$ to $X$ in the pointwise order. For each  $w \in W$, the monotonically increasing sequence $k_1(w) \below k_2(w) \below \cdots$ has supremum $x_\infty \in X$ if and only if the intersection of the upward-closed sets $\up k_n(w):=\{x\in X:k_n(w)\sqsubseteq x\}$ is the upward-closed set $\up x_\infty$. Thus, if each sequence $k_1(w) \below k_2(w) \below \cdots$ has a supremum in $X$, then we obtain a function $k_\infty\: W \To X$ with the property that $\up k_\infty(w) = \bigwedge_{n \in \NN} \up k_n(w)$ for all $w \in W$. In fact, this is an equivalence of conditions, and if we view both the order ${\below}$ and the functions $k_n$ as morphisms in the category $\cat{Rel}$, then we can render the latter condition as $(\sqsubseteq)\circ k_\infty=\bigwedge_{n\in\NN}(\sqsubseteq)\circ k_n$. In short, a poset $(X,{\below})$ is a cpo if and only if, for each set $W$ and each monotonically increasing sequence of functions $k_1 \below k_2 \below \cdots\: W \to X$, there exists a function $k_\infty\:W \to X$ such that $(\sqsubseteq)\circ k_\infty=\bigwedge_{n\in\NN}(\sqsubseteq)\circ k_n$. By analogy, a quantum poset $(\X,R)$ should be a quantum cpo if and only if, for each quantum set $\W$ and each monotonically increasing sequence of functions $K_1 \below K_2 \below \cdots\: \W \to \X$, there exists a function $K_\infty\:\W \to \X$ such that $R\circ K_\infty=\bigwedge_{n\in\NN}R \circ K_n$. Establishing this as the correct condition is our next goal.

\subsection{Convergence} For a monotonically increasing sequence of functions $K_1 \below K_2 \below \cdots \below K_\infty$ from a quantum set $\W$ to a quantum poset $\X$, we next define the notation $K_n \nearrow K_\infty$, whose intuitive meaning is that $K_\infty$ is the  \emph{internal} pointwise supremum of the sequence $K_1, K_2, \ldots$. Example~\ref{ex:sup} then shows $K_n \nearrow K_\infty$ is a strictly stronger relation than $K_\infty$ being the external supremum of the sequence $K_1, K_2, \ldots$ among all the functions from $\W$ to $\X$.

\begin{definition}[Limit of increasing sequence of $\qSet$-morphisms]\label{def:limit}
Let $(\X,R)$ be a quantum poset, and let $\W$ be a quantum set. Equip $\qSet(\W,\X)$ with the pointwise order $\sqsubseteq$ defined in Lemma~\ref{lem:qPos}. 
\begin{enumerate}
\item If the supremum of a  monotonically increasing sequence $K_1\sqsubseteq K_2\sqsubseteq\cdots$ in $\qSet(\W,\X)$ equipped with the pointwise order exists, we denote it by $\bigsqcup_{n\in\NN}K_n$. 
\item If there exists some $K_\infty:\W\to\X$ such that $R\circ K_\infty =\bigwedge_{n\in\NN}R\circ K_n,$
then we write $K_n\nearrow K_\infty$.
\end{enumerate}
Thus, $K_n\nearrow K_\infty$ expresses that $K_\infty$ is the internal supremum of the sequence. In this case, we will call $K_\infty$ the \emph{limit} of the sequence. On the other hand, $K_\infty=\bigsqcup_{n\in\NN}K_n$ expresses that $K_\infty$ is the external supremum of the sequence. In this case, we omit the word `external', and just call $K_\infty$ the \emph{supremum} of the sequence.\end{definition}
In order to show that property (2), $K_n\nearrow K_\infty$ is the appropriate one for the quantum setting, we first prove that the limit of a monotonically ascending sequence is also its supremum, as one would expect from the classical case.

\begin{lemma}\label{lem:lim is sup}
Let $(\X,R)$ be a quantum poset, and let $\W$ be a quantum set. If a monotonically ascending sequence of functions $K_1\sqsubseteq K_2\sqsubseteq\cdots\sqsubseteq K_\infty :\W\to\X$ satisfies $K_n\nearrow K_\infty$, then $K_\infty=\bigsqcup_{n\in\NN}K_n$. In particular, $K_\infty$ is the unique limit of the $K_n$. 
\end{lemma}
\begin{proof}
Since $R\circ K_\infty=\bigwedge_{n\in\NN}R\circ K_n$, we have $R\circ K_\infty\leq R\circ K_n$ for each $n\in\NN$, i.e., $K_n\sqsubseteq K_\infty$ for each $n\in\NN$; hence, $K_\infty$ is an upper bound of the $K_n$ in $\qSet(\W,\X)$. Let $G$ be any upper bound of the $K_n$. For each $n \in \NN$, $K_n\sqsubseteq G$, i.e., $R\circ G\leq R\circ K_n$, and so, $R\circ G\leq \bigwedge_{n\in\NN}R\circ K_n=R\circ K_\infty$. Thus, $K_\infty\sqsubseteq G$; we conclude that $K_\infty$ is indeed the supremum of the functions $K_n$ in $\qSet(\W,\X)$.
If $G:\W\to\X$ is any function satisfying $K_n\nearrow G$, then the same argument shows $G=\sup_{n\in\NN}K_n=K_\infty$; hence, limits are unique.
\end{proof}

The converse of Lemma \ref{lem:lim is sup} is not true, as the next example shows. That is, there exist a quantum set $\W$, a quantum poset $(\X,R)$ and functions $K_1 \below K_2 \below \cdots \below K_\infty\: W \To \X$ such that $K_\infty=\bigsqcup_{n\in\NN}K_n$, but not $K_n \nearrow K_\infty$. Nevertheless, Lemma \ref{lem:lim is sup} implies that $K_\infty=\bigsqcup_{n\in\NN}K_n$ is equivalent to $K_n \nearrow K_\infty$ when $(\X,R)$ is a quantum cpo.
 
\begin{example}\label{ex:sup}
Let $\X=`\overline\NN\times\H$ ordered by $S$ as in Example \ref{ex:orderNtimesH}. Later we will see that $(\X,S)$ is a quantum cpo according to Proposition~\ref{prop:finite qposet is qcpo}, Corollary~\ref{classical.L}, and Corollary~\ref{cor:monoidal product of qCPOs}. Here, we modify $S$ slightly in order to obtain a quantum poset for which not every monotonically ascending sequence has a limit. So let $R$ be the quantum order that equals $S$ on all atoms, except we equip the atom $X_\infty=\CC_\infty \tensor \CC^2$ the flat order, i.e., $R(X_\infty,X_\infty)=\CC 1_{X_\infty}$. 
Since each atom $X_i=\CC_i\otimes\CC^2$ is canonically isomorphic to $\CC^2$, we can depict $R$ as a matrix of matrix subspaces as follows:
$$
\left(
\begin{matrix}
R(X_\infty, X_\infty) & \cdots & R(X_2, X_\infty) & R(X_1, X_\infty) \\
\vdots & \ddots &  \vdots &\vdots \\
R(X_\infty, X_2) & \cdots & R(X_2, X_2) &R(X_1, X_2) \\
R(X_\infty, X_1) & \cdots & R(X_2, X_1) & R(X_1, X_1)
\end{matrix}
\right)
: =
\left(
\begin{matrix}
\CC 1  & \cdots & \CC a + \CC 1 & \CC a + \CC 1 \\
\vdots & \ddots & \vdots & \vdots \\
0 & \cdots & \CC a + \CC 1  & \CC a + \CC 1 \\
0 & \cdots & 0  & \CC a + \CC 1 \\
\end{matrix}
\right)
$$ So, the only difference between $R$ and $S$ is the first row, first column entry.

We show $R$ is an order and that there is a monotonically increasing sequence of functions $K_1 \below K_2 \below \cdots\colon \H\to (\X,R)$ with a supremum, $K_\infty$, that does not satisfy $R \circ K_\infty  = \bigwedge_{n \in \NN} R \circ K_n$. 
Appealing to the fact that the algebra $\CC1 + \CC a$ is antisymmetric, it is easy to verify that $R$ is an order on $`\NN \times \H$. The ``origin'' of the matrices is in the lower right corner, to reflect the intuition that higher-indexed atoms $X_i$ correspond to atomic subsets that are higher in the order on $\X$.

For each $n \in \NN$ we define the function $K_n\: \H \to \X$ and compute $R \circ K_n$ as follows:
$$
\left(
\begin{matrix}
K_n (H, X_\infty) \\
\vdots \\
K_n (H, X_{n+1}) \\
K_n (H, X_n) \\
K_n (H, X_{n-1}) \\
\vdots \\
K_n(H, X_2) \\
K_n(H, X_1) 
\end{matrix}
\right)
: =
\left(
\begin{matrix}
0  \\
\vdots \\
0 \\
\CC 1 \\
0\\
\vdots \\
0 \\
0
\end{matrix}
\right)
\qquad
\left(
\begin{matrix}
(R \circ K_n) (H, X_\infty) \\
\vdots \\
(R \circ K_n) (H, X_{n+1}) \\
(R \circ K_n) (H, X_n) \\
(R \circ K_n) (H, X_{n-1}) \\
\vdots \\
(R \circ K_n)(H, X_2) \\
(R \circ K_n)(H, X_1) 
\end{matrix}
\right)
=
\left(
\begin{matrix}
\CC a + \CC 1 \\
\vdots \\
\CC a + \CC 1 \\
\CC a + \CC 1 \\
0\\
\vdots \\
0 \\
0
\end{matrix}
\right)
$$
We also define $K_\infty\: \H \To \X$ below (it's straightforward to show $K_\infty = \sup_{n\in \NN} K_n$), and we compute $R \circ K_\infty$ as follows:
$$
\left(
\begin{matrix}
K_\infty (H, X_\infty) \\
\vdots \\
K_\infty(H, X_2) \\
K_\infty(H, X_1) 
\end{matrix}
\right)
: =
\left(
\begin{matrix}
\CC 1 \\
\vdots \\
0 \\
0
\end{matrix}
\right)
\qquad
\left(
\begin{matrix}
(R \circ K_\infty) (H, X_\infty) \\
\vdots \\
(R \circ K_\infty)(H, X_2) \\
(R \circ K_\infty)(H, X_1) 
\end{matrix}
\right)
=
\left(
\begin{matrix}
\CC 1 \\
\vdots \\
0 \\
0
\end{matrix}
\right)
$$
Thus, we find that
$$
\left(
\begin{matrix}
(R \circ K_\infty) (H, X_\infty) \\
\vdots \\
(R \circ K_\infty)(H, X_2) \\
(R \circ K_\infty)(H, X_1) 
\end{matrix}
\right)
=
\left(
\begin{matrix}
\CC 1 \\
\vdots \\
0 \\
0
\end{matrix}
\right)
\neq
\left(
\begin{matrix}
\CC a + \CC 1\\
\vdots \\
0 \\
0
\end{matrix}
\right)
= \bigwedge_{n \in \NN}
\left(
\begin{matrix}
(R \circ K_n) (H, X_\infty) \\
\vdots \\
(R \circ K_n)(H, X_2) \\
(R \circ K_n)(H, X_1) 
\end{matrix}
\right).
$$
In short, $R \circ K_\infty \neq \bigwedge_{n \in \NN} R \circ K_n$, so we do not have that $K_n \nearrow K_\infty$ in $\qPOS(\H,(\X,R))$.

On the other hand, we claim that $K_\infty=\bigsqcup_{n\in\NN}K_n$, i.e., that $K_\infty$ is the supremum of the sequence $K_1 \sqsubseteq K_2 \sqsubseteq \cdots$ in  $\qPOS(\H,(\X,S))$ in the order defined in Lemma~\ref{lem:qPos}. Indeed, let $F$ be a function such that $K_n \below F$ for each $n \in \NN$. It follows immediately that $F \leq \bigwedge_{n \in \NN} R \circ K_n$:
$$
\left(
\begin{matrix}
F (H, X_\infty) \\
\vdots \\
F(H, X_2) \\
F(H, X_1) 
\end{matrix}
\right)
\leq
\left(
\begin{matrix}
\CC a +\CC 1 \\
\vdots \\
0 \\
0
\end{matrix}
\right).
$$
By the definition of a function between quantum sets, $F \circ F^\dagger \leq I_\X$. Thus, $F(H, X_\infty)$ is a subspace of $2\times 2$ complex matrices such that $F(H, X_\infty)\cdot F(H, X_\infty)^\dagger \leq \CC 1$. It is straightforward to show that this implies that either $F(H, X_\infty)= 0$ or $F(H, X_\infty) = \CC u$ for some unitary matrix $u$. By the definition of a function between quantum sets, $F^\dagger \circ F \geq I_\H$, so certainly $F(H, X_\infty) \neq 0$. Hence, $F(H, X_\infty) = \CC u$ for some unitary matrix $u$. Applying \cite{Weaver19}*{Proposition 3.4}, we find that $u$ is a scalar and, therefore, that $F = K_\infty$. Thus, $K_\infty$ is not only the least upper bound of the sequence $K_1 \below K_2 \below \cdots$, it is the \emph{only} upper bound of this sequence.
\end{example}

\subsection{Quantum cpos} We now define quantum cpos and the Scott continuous functions between them. 

\begin{definition}[Quantum cpo]\label{def:quantum cpo}
A quantum poset $(\X,R)$ is called a \emph{quantum cpo} if, for each atomic quantum set $\H$, every monotonically ascending sequence $K_1\sqsubseteq K_2\sqsubseteq\cdots:\H\to\X$ has a limit $K_\infty:\H\to\X$, i.e., $K_n\nearrow K_\infty$. 
\end{definition}

\begin{example}\label{ex:trivially ordered set is qcpo}
Let $\X$ be a quantum set equipped with the trivial order $I_\X$. Then $(\X,I_\X)$ is a quantum cpo. This can be seen as follows: Let $K_1\sqsubseteq K_2\sqsubseteq\cdots:\H\to\X$ be a monotonically ascending sequence of functions for some atomic quantum set $\H$. For $n\leq m$, we have $K_n\sqsubseteq K_m$, hence $K_m=I_\X\circ K_m\leq I_\X\circ K_n=K_n.$ It now follows from \cite[Lemma A.7]{KLM20} that $K_n=K_m$, so the sequence is constant. Taking $K_\infty=K_1$ gives $\bigwedge_{n\in\NN}I_\X\circ K_n=\bigwedge_{n\in\NN}K_n=K_1=K_\infty=I_\X\circ K_\infty,$
which shows that indeed $K_n\nearrow K_\infty$.
\end{example}

The next result shows that we can replace $\H$ by any quantum set $\W$ in Definition \ref{def:quantum cpo}, as one would expect, since the atomic quantum sets generate $\qSet$.

\begin{proposition}\label{prop:suprema of sequences of functions}
Let $(\X,R)$ be a quantum cpo, and let $\W$ be a quantum set. For every monotonically ascending sequence of functions $F_1\sqsubseteq F_2\sqsubseteq\cdots:\W\to\X$ there is a function $F_\infty:\W\to\X$ such that $F_n\nearrow F_\infty$.
\end{proposition}
\begin{proof}
For each $W\atomof\W$, let $J_W:\Q\{W\}\to\W$ be the canonical inclusion. Let $n,m \in \NN$ such that $n \leq m$. Then $F_n\sqsubseteq F_m$, so $F_n\circ J_W\sqsubseteq F_n\circ J_W$.
Hence $F_1\circ J_W\sqsubseteq F_2\circ J_W\sqsubseteq\cdots:\Q\{W\}\to\X$
is a monotonically ascending sequence of functions, and since $\Q\{W\}$ is atomic, it follows that there is a function $F_W:\Q\{W\}\to\X$ such that $F_n\circ J_W\nearrow F_W$, i.e.,
$R\circ F_W=\bigwedge_{n\in\NN}R\circ F_n\circ J_W.$

We obtain a function $F_\infty=[F_W:W\atomof\W]: \W = \biguplus_{W\atomof\W}\Q\{W\}\to\X$, which satisfies $F_\infty(W,X)=F_W(W,X)$ for each $W\atomof\W$ and each $X\atomof\X$. We then find that
\begin{align*} (R\circ F_\infty)(W,X) &  = \bigvee_{X'\atomof\X}R(X',X)\cdot F_\infty(W,X')=\bigvee_{X'\atomof\X} R(X',X)\cdot F_W(W,X')\\
& =(R\circ F_W)(W,X)=\bigwedge_{n\in\NN}(R\circ F_n\circ J_W)(W,X)=\bigwedge_{n\in\NN} (R\circ F_n)(W,X).
\end{align*}
Therefore, $R\circ F_\infty=\bigwedge_{n\in\NN}R\circ F_n$, that is, $F_n\nearrow F_\infty$.
\end{proof}

\begin{definition}[Scott-continuous function]
Let $(\X,R)$ and $(\Y,S)$ be quantum posets. A monotone function $F:(\X,S)\to(\Y,S)$ is called \emph{Scott continuous} if for each atomic quantum set $\H$ and each monotone sequence of functions
$K_1\sqsubseteq K_2\sqsubseteq\cdots\sqsubseteq K_\infty:\H\to \X$
satisfying $K_n\nearrow K_\infty$, we have $F\circ K_n\nearrow F\circ K_\infty$.
\end{definition}
Since left multiplication by the monotone function $F$ is monotone \cite{KLM20}*{Lemma 4.4}, it follows that the functions $F\circ K_n$ form a monotonically increasing sequence.

\begin{example}\label{ex:function with trivially ordered domain is Scott continuous}
Let $\X$ be a quantum set, equipped with the trivial order $I_\X$, and let $(\Y,S)$ be a quantum poset. Then any function $F:\X\to\Y$ is Scott continuous: Indeed, by Example \ref{ex:trivially ordered set is qcpo}, $(\X,R)$ is a quantum cpo, since any sequence $K_1\sqsubseteq K_2\sqsubseteq\cdots:\H\to\X$ is constant with limit is $K_1$. Hence $\bigwedge_{n\in\NN}S\circ F\circ K_n=S\circ F\circ K_1,$
which expresses that $F\circ K_n\nearrow F\circ K_1$.
\end{example}

\begin{lemma}\label{lem:order iso is Scott continuous}
Let $(\X,R)$ and $(\Y,S)$ be quantum posets, and let $F:\X\to\Y$ be an order isomorphism. Then $F$ is Scott continuous.
\end{lemma}
\begin{proof}
Let $\H$ be an atomic quantum set, and let $K_1\sqsubseteq K_2\sqsubseteq\cdots\subseteq K_\infty:\H\to\X$ satisfy $K_n\nearrow K_\infty$, i.e.,
$\bigwedge_{n\in\NN}R\circ K_n=R\circ K_\infty.$
Since $F$ is a surjective order embedding by \cite[Proposition~2.7]{KLM20}, we have $R=F^\dag\circ S\circ F$. Using the surjectivity of $F$ and \cite[Proposition~A.6]{KLM20}, we then find that
\begin{align*}  \bigwedge_{n\in\NN}S\circ F\circ K_n & =F\circ F^\dag\circ\bigwedge_{n\in\NN}S\circ F\circ K_n=F\circ\bigwedge_{n\in\NN}F^\dag\circ S\circ F\circ K_n\\
& =F\circ \bigwedge_{n\in\NN}R\circ K_n=F\circ R\circ K_\infty=F\circ F^\dag\circ S\circ F\circ K_\infty=S\circ F\circ K_\infty.
\end{align*}
Therefore, $F\circ K_n\nearrow F\circ K_\infty$.
\end{proof}

\begin{proposition}\label{prop:finite qposet is qcpo}
Let $\X$ be a finite quantum set, and let $R$ be an order on $\X$. Then, $(\X,R)$ is a quantum cpo. Moreover, if $(\Y,S)$ is a quantum poset and $F:\X\to\Y$ is monotone, then $F$ is Scott continuous.  
\end{proposition}
\begin{proof}
Let $\H$ be an atomic quantum set and let $K_1\sqsubseteq K_2\sqsubseteq K_3\sqsubseteq\cdots$ be a monotonically ascending sequence in $\qSet(\H,\X)$. By definition of the order on $\qSet(\H,\X)$, 
$R\circ K_1\geq R\circ K_2\geq R\circ K_3\geq\cdots$ is a non-ascending sequence in $\qRel(\H,\X)$. Fix $X\in\X$, and let $H$ be the unique atom of $\H$. Then, 
$ (R\circ K_1)(H,X)\geq (R\circ K_2)(H,X)\geq\cdots$ is a non-ascending sequence in $L(H,X)$, which is finite-dimensional. Thus, there is some $n_X\in\NN$ such that $(R\circ K_n)(H,X)=(R\circ K_{n_X})(H,X)$ for all $n\geq n_X$. Since $\X$ is finite, we know that $\{n_X:X\atomof\X\}$ has a maximum element, which we call $m$. Let $K_\infty=K_m$. Then, for each $X\atomof\X$ and each $n\geq m$, we have $(R\circ K_n)(H,X)=(R\circ K_\infty)(H,X)$, and therefore, $R\circ K_n=R\circ K_\infty$. We conclude that $ \bigwedge_{n\in\NN}R\circ K_n=R\circ K_\infty$, i.e., that $K_n\nearrow K_\infty$. Therefore, $(\X,R)$ is a quantum cpo.

Now let $(\Y,S)$ be a quantum poset, and let $F:\X\to\Y$ be monotone. Let $K_1\sqsubseteq K_2\sqsubseteq\cdots\subseteq K_\infty:\H\to\X$ satisfy $K_n\nearrow K_\infty$. We have already shown that there is some $m\in\NN$ such that $K_n=K_m$ for each $n\geq m$ and that $K_\infty=K_m$. Since $F$ is monotone, left multiplication with $F$ is monotone \cite{KLM20}*{Lemma 4.4}, hence $F\circ K_1\sqsubseteq F\circ K_2\sqsubseteq \cdots\sqsubseteq F\circ K_m=F\circ K_{m+1}=\cdots = F\circ K_\infty:\H\to\X$ is a monotonically ascending sequence of functions that stabilizes at $m$, hence
 \[S\circ F\circ K_1\geq S\circ F\circ K_2\geq\cdots\geq  S\circ F\circ K_m=S\circ F\circ K_{m+1}=\cdots=S\circ F\circ K_\infty.\]
Thus $\bigwedge_{n\in\NN} S\circ F\circ K_n=S\circ F\circ K_m=S\circ F\circ K_\infty,$
so $F\circ K_n\nearrow F\circ K_\infty$. Therefore, $F$ is Scott continuous.
\end{proof}

\begin{lemma}\label{lem:composition is Scott continuous}
Let $(\X,R)$, $(\Y,S)$ and $(\Z,T)$ be quantum posets, and let $F:\X\to\Y$ and $G:\Y\to\Z$ be Scott continuous. Then $G\circ F$ is Scott continuous.
\end{lemma}
\begin{proof}
As the composition of monotone functions, $G\circ F$ is monotone. Let $\H$ be an atomic quantum set and let $K_1\sqsubseteq K_2\sqsubseteq\cdots\subseteq K_\infty:\H\to\X$ satisfy $K_n\nearrow K_\infty$. Because $F$ is Scott continuous, $F \circ K_n \nearrow F \circ K_\infty$, and because $G$ is Scott continuous, $G \circ F \circ K_n \nearrow G \circ F \circ K_\infty$. Therefore, $G \circ F$ is Scott continuous.
\end{proof}

Let $\X$ be a quantum poset, let $\H$ be an atomic quantum set, and let $K_1\sqsubseteq K_2\sqsubseteq\cdots:\H\to\X$ be a monotonically ascending sequence with limit $K_\infty$. Then
\[ \bigwedge_{n\in\NN}R\circ I_\X\circ K_n=\bigwedge_{n\in\NN}R\circ K_n=R\circ K_\infty=R\circ I_\X\circ K_\infty,\]
and hence $I_\X:\X\to\X$ is Scott continuous. Together with the previous lemma, this shows that quantum posets and Scott continuous maps form a subcategory $\qPOS_{\mathrm{SC}}$ of $\qPOS$. By restricting the class of objects to quantum cpos we obtain the category $\cat{qCPO}$.

\begin{definition}[$\qCPO$]
We define $\qCPO$ to be the category of quantum cpos and Scott continuous functions.
\end{definition}

\subsection{Enrichment over $\CPO$}

We show that for  quantum cpos $(\X,R)$ and $(\Y,S)$, the subset $\cat{qCPO}((\X,R), (\Y,S)) \subsetof \cat{qPOS}((\X,R), (\Y,S))$ is a cpo in the induced order. Thus, $\cat{qCPO}$ is enriched over $\cat{CPO}$.

\begin{proposition}\label{prop:limit of functions is Scott continuous}
Let $(\X,R)$ be a quantum poset, and let $(\Y,S)$ be a quantum cpo. Let $F_1\sqsubseteq F_2\sqsubseteq\cdots\sqsubseteq F_\infty:\X\to\Y$ satisfy $F_n\nearrow F_\infty$. If each $F_n$ is monotone, then so is $F_\infty$, and if each $F_n$ is Scott continuous, so is $F_\infty$.
\end{proposition}

\begin{proof}
By the definition of limit (Definition \ref{def:limit}),
\begin{equation}\label{eq:limit of functions}
S\circ F_\infty=\bigwedge_{n\in\NN} S\circ F_n.
\end{equation}
Assume that each $F_n$ is monotone, i.e., that $S\circ F_n\geq F_n\circ R$. Then
\begin{align*}
    S\circ F_\infty & =\bigwedge_{n\in\NN}S\circ S\circ F_n\geq \bigwedge_{n\in\NN}S\circ F_n\circ R\geq \left(\bigwedge_{n\in\NN}S\circ F_n\right)\circ R=S\circ F_\infty\circ R\geq F_\infty\circ R,
\end{align*}
where the first equality follows by combining Equation (\ref{eq:limit of functions}) with the definition of an order; that
definition is also used in the last inequality. We use \cite[Lemma~A.1]{KLM20} in the second inequality.

Now, assume that each $F_n$ is Scott continuous. Let $\H$ be an atomic quantum set, and let $K_1\sqsubseteq K_2\sqsubseteq\cdots:\H\to\X$ be a monotonically ascending sequence of functions with limit $K_\infty$. Then, for each $n \in\NN$, we have that $F_n\circ K_m\nearrow F_n\circ K_\infty$, i.e., 
\begin{equation}\label{eq:limit of functions 2}
 \bigwedge_{m\in\NN}S\circ F_n\circ K_m=S\circ F_n\circ K_\infty.
 \end{equation}
Using \cite[Proposition~A.6]{KLM20} implicitly, we find that
$S\circ F_\infty\circ K_\infty  = \bigwedge_{n\in\NN}S\circ F_n\circ K_\infty = \bigwedge_{n,m\in\NN} S\circ F_n\circ K_m =\bigwedge_{m\in\NN}S\circ F_\infty\circ K_m,$
where we use Equation (\ref{eq:limit of functions}) in the first and last equalities and Equation (\ref{eq:limit of functions 2}) in the middle equality.
\end{proof}

\begin{corollary}\label{cor:hom sets are cpos}
Let $(\X,R)$ and $(\Y,S)$ be quantum cpos. Then, $\qCPO(\X,\Y)$ is a cpo. 
\end{corollary}
\begin{proof}
Let $F_1\sqsubseteq F_2\sqsubseteq\cdots$ be a monotone sequence in $\qCPO(\X,\Y)$. Then $F_n\nearrow F_\infty$ for some $F_\infty:\X\to\Y$ by Proposition \ref{prop:suprema of sequences of functions}. By Proposition \ref{prop:limit of functions is Scott continuous}, we find that $F_n\in \qCPO(\X,\Y)$. By Lemma \ref{lem:lim is sup}, it follows that $F_\infty=\bigsqcup_{n\in\NN}F_n$ in $\qCPO(\X,\Y)$.
\end{proof}

\begin{lemma}\label{lem:right multiplication is Scott continuous}
Let $\X$ and $\Y$ be quantum sets, and let $(\Z,T)$ be a quantum cpo. Let $F\: \X \To \Y$ be a function. For any monotonically increasing sequence of functions $G_1 \below G_2 \below \cdots\sqsubseteq G_\infty\: \Y \to \Z$ satisfying $G_n\nearrow G_\infty$, we have $G_n \circ F \nearrow G_\infty \circ F$.
\end{lemma}
\begin{proof}
Since right multiplication is monotone \cite{KLM20}*{Lemma 4.2}, it follows that 
$G_1\circ F\sqsubseteq G_2\circ F\sqsubseteq\cdots:\X\to\Z$ is a monotonically ascending sequence of functions. Then,
\[  \bigwedge_{n\in\NN}T\circ G_n\circ F=\left(\bigwedge_{n\in\NN}T\circ G_n\right)\circ F=T\circ G_\infty\circ F,\]
where the first equality is by \cite[Proposition~A.6]{KLM20}, and the second equality since $G_n\nearrow G_\infty$. We conclude that $G_n\circ F\nearrow G_\infty\circ F$. 
\end{proof}

\begin{lemma}\label{lem:left multiplication by Scott continuous function is Scott continuous}
Let $\X$ be a quantum set, and let $(\Y,S)$ and $(\Z,T)$ be quantum cpos. Let $G:\Y\to\Z$ be Scott continuous. If $F_1\sqsubseteq F_2\sqsubseteq\cdots\sqsubseteq F_\infty:\X\to\Y$ satisfies $F_n\nearrow F_\infty$, then $G\circ F_n\nearrow G\circ F_\infty$. 
\end{lemma}
\begin{proof}
Fix $X \atomof \X$, and let $J_X:\Q\{X\}\to\X$ be the inclusion function. By Lemma \ref{lem:right multiplication is Scott continuous}, $F_n \circ J_X \nearrow F_\infty \circ J_X$. Since $\Q\{X\}$ is atomic, it follows from the Scott continuity of $G$ that 
$G\circ F_n\circ J_X\nearrow G\circ F_\infty\circ J_X$. We now reason that
\begin{align*}&
\left(\bigwedge_{n\in\NN}T\circ G\circ F_n\right) \circ J_X
=
\bigwedge_{n\in\NN}T\circ G\circ F_n \circ J_X
=
T \circ G \circ F_\infty \circ J_X,
\end{align*}
where the second equality follows by \cite[Proposition~A.6]{KLM20}. We vary $X \atomof \X$ to conclude that $\bigwedge_{n\in\NN}T\circ G\circ F_n = T \circ G \circ F_\infty$, i.e., that $G \circ F_n \nearrow G \circ F_\infty$, as desired.
\end{proof}

\begin{theorem}\label{thm:qCPO enriched over CPO}
The category $\qCPO$ is enriched over $\CPO$.
\end{theorem}
\begin{proof}
Let $(\X,R)$, $(\Y,S)$ and $(\Z,T)$ be quantum cpos. The homsets $\cat{qCPO}(\X,\Y)$ and $\cat{qCPO}(\Y,\Z)$ are cpos by Corollary \ref{cor:hom sets are cpos}. It follows that for functions $F_1 \below F_2 \below \cdots\below F_\infty\in \cat{qCPO}(\X,\Y)$, the condition $F_\infty = \bigsqcup_{n \in \NN} F_n$ is equivalent to the condition $F_n \nearrow F_\infty$, and likewise for functions $G_1 \below G_2 \below \cdots \below G_\infty\in \cat{qCPO}(\Y,\Z)$. We apply Lemmas \ref{lem:right multiplication is Scott continuous} and \ref{lem:left multiplication by Scott continuous function is Scott continuous} to infer that composition $\cat{qCPO}(\X,\Y) \times \cat{qCPO}(\Y,\Z) \To \cat{qCPO}(\X,\Z)$ is Scott continuous in each of its two variables separately. It follows that the iterated limits are equal and they equal the double limit, so we conclude composition is Scott continuous. Therefore $\cat{qCPO}$ is enriched over $\cat{CPO}$. 
\end{proof}

\subsection{Completeness} We next show that the category $\cat{qCPO}$ is complete. For denotational semantics, completeness on its own has no special meaning. Later, we use completeness of $\cat{qCPO}$ to show that $\cat{qCPO}$ is also cocomplete, which in turn will be used to construct a category of quantum cpos that is $\CPO$-algebraically compact, a property that is fundamental for modeling recursive types. Eventually, we will show that $\qCPO$ is symmetric monoidal closed. In combination with completeness and cocompleteness, this implies that $\qCPO$ is a \emph{cosmos}, which is a category that has suitable properties for enriched category theory.

\begin{theorem}\label{thm:qCPO has all limits}
The category $\qCPO$ is complete. Any limit in $\qCPO$ can be calculated in the ambient category $\qPOS$, and the embedding of $\qCPO$ in $\qPOS_{\mathrm{SC}}$ preserves all limits, where $\qPOS_{\mathrm{SC}}$ denotes the category of quantum posets and Scott continuous functions.
\end{theorem}

\begin{proof}
Consider a diagram of shape $A$, consisting of objects $(\X_\alpha,R_\alpha)\in\qCPO$ and with Scott continuous maps $F_f:\X_\alpha\to \X_\beta$ for each morphism $f:\alpha\to\beta$ in $A$. We can regard this as a diagram in $\mathbf{qPOS}$. The latter category is complete by Theorem 5.5~\cite{KLM20}, which also states that the limit of the diagram in $\qPOS$ is given by $(\X,R)$, where $\X$ is the limit of the diagram $(\X_\alpha)_{\alpha\in A}$ in $\mathbf{qSet}$, and where  $R=\bigwedge_{\alpha\in A}F_\alpha^\dag\circ R_\alpha\circ F_\alpha$, where $F_\alpha:\X\to\X_\alpha$ are the limiting functions of the diagram in $\mathbf{qSet}$. Theorem 5.5~\cite{KLM20} also assures that these limiting functions $F_\alpha$ are monotone, so they are also limiting functions in $\mathbf{qPOS}$. 

We first show that $(\X,R)$ is a quantum cpo. For a fixed atomic quantum set $\H$ let 
$K_1\sqsubseteq K_2\sqsubseteq\ldots:\H\to\X$
be a monotonically ascending sequence of functions. Since left multiplication is monotone \cite[Lemma 4.4]{KLM20}, we find for each $\alpha\in A$ that $F_\alpha\circ K_1\sqsubseteq F_\alpha\circ K_2\sqsubseteq\ldots:\H\to\X_\alpha$
is a monotonically ascending sequence of functions, too. 

We first claim that if there exists a function $K_\infty:\H\to\X$ such that 
\begin{equation}\label{eq:lim in lim}
F_\alpha\circ K_n\nearrow F_\alpha\circ K_\infty
\end{equation}
for each $\alpha\in A$, then $K_n\nearrow K_\infty$. Equation (\ref{eq:lim in lim}) means that 
$ R_\alpha\circ F_\alpha\circ K_\infty=\bigwedge_{n\in\NN}R_\alpha\circ F_\alpha\circ K_n$; hence we find that
\begin{align*}
    R\circ K_\infty & = \left(\bigwedge_{\alpha\in A}F_\alpha^\dag\circ R_\alpha\circ F_\alpha\right)\circ K_\infty=\bigwedge_{\alpha\in A}F_\alpha^\dag\circ R_\alpha\circ F_\alpha\circ K_\infty \\
    & = \bigwedge_{\alpha\in A}F_\alpha^\dag\circ \bigwedge_{n\in\NN}R_\alpha\circ F_\alpha\circ K_n=\bigwedge_{\alpha\in A} \bigwedge_{n\in\NN}F_\alpha^\dag\circ R_\alpha\circ F_\alpha\circ K_n\\
    & = \bigwedge_{n\in\NN}\bigwedge_{\alpha\in A} F_\alpha^\dag\circ R_\alpha\circ F_\alpha\circ K_n= \bigwedge_{n\in\NN}\left(\bigwedge_{\alpha\in A} F_\alpha^\dag\circ R_\alpha\circ F_\alpha\right)\circ K_n = \bigwedge_{n\in\NN}R\circ K_n,
\end{align*}
where we use \cite[Proposition~A.6]{KLM20} in the second, the third and the penultimate equalities. Therefore, $K_n\nearrow K_\infty$, as claimed.

Next we construct a function $K_\infty$ that satisfies (\ref{eq:lim in lim}). Since $(\X_\alpha,R_\alpha)$ is by assumption a quantum cpo for each $\alpha \in A$, there exists a function $G_\alpha:\H\to\X_\alpha$ such that $F_\alpha\circ K_n\nearrow G_\alpha$. 
Consider a morphism $f:\alpha\to\beta$ in $A$. By the definition of $F_\alpha$ and $F_\beta$ being limiting maps, we have 
$ F_f\circ F_\alpha=F_\beta.$ By the Scott continuity of $F_f$, we have 
$ F_\beta\circ K_n=F_f\circ F_\alpha\circ K_n\nearrow F_f\circ G_\alpha.$
It now follows from Lemma \ref{lem:lim is sup} that $G_\beta=F_{f}\circ G_\alpha$. Thus, the functions $G_\alpha:\H\to\X_\alpha$ form a cone, and since the functions $F_\alpha:\X\to\X_\alpha$ form the limiting cone, there must be a monotone function $K_\infty:\H\to\X$ such that $G_\alpha=F_\alpha\circ K_\infty$. Since $F_\alpha\circ K_n\nearrow G_\alpha$, it follows that Equation (\ref{eq:lim in lim}) holds for each $\alpha\in A$, hence we conclude that $K_n\nearrow K_\infty$. Therefore, $(\X,R)$ is a quantum cpo. 

We have to show that the limiting functions $F_\alpha$ are Scott continuous. Let $\H$ be an atomic quantum set, and let 
$K_1\sqsubseteq K_2\sqsubseteq\cdots:\H\to\X$
be a monotonically ascending sequence. Then $K_n\nearrow K_\infty'$ for some $K'_\infty:\H\to\X$  since $\X$ is a quantum cpo. Now, as before, construct some $K_\infty$ such that (\ref{eq:lim in lim}) holds for each $\alpha\in A$, ,and note that $K_n\nearrow K_\infty$. Then Lemma \ref{lem:lim is sup} implies $K_\infty=K_\infty'$. Thus, Equation (\ref{eq:lim in lim}) expresses that $F_\alpha$ is Scott continuous. We have shown that the vertex $(\X,R)$ of the limiting cone is a quantum cpo and that the limiting functions $F_\alpha$ are all Scott continuous; hence, this cone is entirely in $\cat{qCPO}$. It remains to show that it is also limiting in $\cat{qCPO}$.

Let $(\Y,S)$ be a quantum poset, and let $C_\alpha:\Y\to\X_\alpha$ be Scott continuous maps that form a cone. Since $(\X,R)$ is the limit of the $(\X_\alpha,R_\alpha)$ in $\qPOS$, there exists a unique monotone map $M:\Y\to\X$ such that $F_\alpha\circ M=C_\alpha$ for each $\alpha\in A$.
We claim that $M$ is Scott continuous: 
Indeed, let $\H$ be an atomic quantum set, let $E_1\sqsubseteq E_2\sqsubseteq\cdots:\H\to\Y$ be a monotonically ascending sequence of functions, and assume that there exists a function $E_\infty:\H\to\Y$ such that $E_n\nearrow E_\infty$. Note that $E_\infty$ exists if we assume that $(\Y,S)$ is a quantum cpo. Let $K_n=M\circ E_n$ for each $n\in\{1,2,\ldots,\infty\}$. Since $M$ is monotone, left multiplication with $M$ is monotone \cite{KLM20}*{Lemma 4.4}, hence the sequence $K_1\sqsubseteq K_2\sqsubseteq\cdots$ is also monotonically ascending. Since the $C_\alpha$ are Scott continuous, we have $C_\alpha\circ E_n\nearrow C_\alpha\circ E_\infty$. So, for each $\alpha\in A$ we obtain
\[ F_\alpha\circ K_n=F_\alpha\circ M\circ E_n=C_\alpha\circ E_n\nearrow C_\alpha\circ E_\infty=F_\alpha\circ M\circ E_\infty=F_\alpha\circ K_\infty,\]
which shows the $K_n$ satisfy condition (\ref{eq:lim in lim}) for each $\alpha\in A$; i.~e.,  we obtain $K_n\nearrow K_\infty$. But this exactly expresses that $M\circ  E_n\nearrow M\circ E_\infty$, so $M$ is Scott continuous. If $\Y$ is a quantum cpo, it follows that the Scott continuous functions $F_\alpha:\X\to\X_\alpha$ form a limiting cone in $\qCPO$, so $(\X,R)$ is the limit of the diagram.

Choosing $(\Y,S)$ to be a quantum cpo shows that $\qCPO$ is complete. Since we allowed $(\Y,S)$ to be an arbitrary quantum poset, it follows that the embedding of $\qCPO$ into $\qPOS_{\mathrm{SC}}$ preserves all limits.
\end{proof}

\section{Coproducts and Cocompleteness}\label{sec:qCPO is cocomplete}

In this section, we show that the category $\cat{qCPO}$ has coproducts and that it is cocomplete. Coproducts are used in semantics to support sum types and conditional branching, while cocompleteness is required to construct a $\CPO$-algebraic compact category of quantum cpos, which is key for modeling recursive types.  We begin with coproducts.

\subsection{Coproducts} The main technical challenge of showing the existence of coproducts in $\qCPO$ is the case that a monotone ascending sequence $K_1 \below K_2 \below \cdots\colon \H\to \X_i\uplus \X_2$ of functions from an atomic quantum set to a coproduct of quantum posets - intuitively a family of monotone ascending sequences in $\X_1 \uplus \X_2$ indexed by $\H$ - have non-zero range in both summands.

We briefly recall the construction of coproducts of quantum sets and of quantum posets. The coproduct $\biguplus_{\alpha\in A}\X_\alpha$ of a family $(\X_\alpha)_{\alpha\in A}$ of quantum sets coincides in $\mathbf{qSet}$ and $\mathbf{qRel}$, and is defined as follows. If all quantum sets in the family are pairwise disjoint, their coproduct is simply the union $\bigcup_{\alpha\in A}\X_\alpha:=\Q\bigcup_{\alpha\in A}\At(\X_\alpha)$, otherwise the coproduct is given by the disjoint union $\bigcup_{\alpha\in A}\X_\alpha\times `\{\alpha\}$. 
So, just as for ordinary sets, for most results involving coproducts we can assume that the family over which the coproduct is formed consists of pairwise disjoint quantum sets. For each $\beta\in A$, we denote the canonical injection $\X_\beta\to\biguplus_{\alpha\in A}\X_\alpha$ by $J_\beta$.  

If each quantum set in the family $(\X_\alpha)_{\alpha\in A}$ is equipped with an order $R_\alpha$, then the coproduct is $\biguplus_{\alpha\in A}(\X_\alpha,R_\alpha):=(\X,R)$, where $\X=\biguplus_{\alpha\in A}\X_\alpha$ in $\mathbf{qSet}$, and $R=\biguplus_{\alpha\in A}R_\alpha$ in $\mathbf{qRel}$, i.e., the unique relation on $\X$ such that $R\circ J_\alpha=R_\alpha\circ J_\alpha$ for each $\alpha\in A$ \cite[Proposition 6.2]{KLM20}.

\begin{definition}[Decomposition of $\H$]\label{def:decomposition}
Let $\H$ be an atomic quantum set, and let $\X$ be any quantum set. Let $H$ be the unique atom of $\H$. A function $D\: \H \to \X$ is a \emph{decomposition} of $\H$ if there exists a family $\{X_\alpha\}_{\alpha \in A}$ of pairwise orthogonal subspaces of $H$ that sum to $H$ such that
\begin{itemize}
    \item[(1)] $\X =\Q\{X_\alpha:\alpha\in A\}$, and
    \item[(2)] $D(H, X_\alpha) = \CC \cdot \mathrm{proj}_{X_\alpha}$ for each $\alpha \in A$,
\end{itemize}
where $\mathrm{proj}_{X_\alpha}\in L(H, X_\alpha)$ is the orthogonal projection onto $X_\alpha$.
\end{definition}
We have a canonical one-to-one correspondence between the decompositions of $\H$ and orthogonal direct sum decompositions $H = \bigoplus_{\alpha \in A} X_\alpha$. The non-zero subspaces $X_\alpha$ are pair-wise distinct, so $\bigcup_{\alpha \in A} \Q\{X_\alpha\} = \biguplus_{\alpha \in A} \Q\{X_\alpha\}$.

\begin{lemma}\label{lem:decomposition is surjective}
Let $\H$ be an atomic quantum set, and let $D\: \H \to \biguplus_{\alpha \in A} \Q\{X_\alpha\}$ be any decomposition of $\H$. Then, $D$ is surjective.
\end{lemma}
\begin{proof}
For each $\alpha\in A$, let $\proj_{X_\alpha}:H\to H$ be the projection operator whose range is the subspace $X_\alpha$ of $H$. Clearly $\proj_{X_\alpha}\proj_{X_\beta}^\dag = 0$ for distinct $\alpha, \beta \in A$. We now calculate that for all $\alpha, \beta \in A$, 
\begin{align*}
   (D\circ D^\dag)(X_\beta,{X_\alpha}) & = \bigvee_{H\atomof\H}D(H,X_\alpha)D^\dag({X_\beta},H)=D(H, X_\alpha)D^\dag(X_\beta, H)
 =D(H, X_\alpha)D(H,X_\beta)^\dag\\
 &  = \CC \proj_{X_\alpha} \proj_{X_\beta}^\dag  =\delta_{\alpha,\beta}\CC \proj_{X_\alpha}\proj_{X_\alpha}^\dag 
 = \delta_{\alpha,\beta}\CC 1_{X_\alpha}=I_{\X}(X_\beta,X_\alpha).\qedhere
 \end{align*}
\end{proof}

In general, a decomposition $D:\H\to\X$ of $\H$ is not injective. In fact, $D$ is injective precisely when it is bijective, in which case $\X$ also is atomic. In this case, it necessarily follows that the unique atoms of $\X$ and $\H$ have the same dimension.

\begin{lemma}\label{lem:decomposition}
Let $F$ be a function from an atomic quantum set $\H=\Q\{H\}$, to the disjoint union of an indexed family of quantum sets $\{\Y_\alpha\}_{\alpha \in A}$. Then, there exists a decomposition $D\: \H \To \biguplus_{\alpha \in A} \Q\{X_\alpha\}$ and a family of functions $\{F_\alpha \: \Q\{X_\alpha\} \To \Y_\alpha\}_ {\alpha \in A}$ such that $(\biguplus_{\alpha \in A} F_\alpha) \circ D = F$.

\end{lemma}

\begin{proof}
Fix $\alpha \in A$, and let $J_\alpha$ be the inclusion of $\Y_\alpha$ into the disjoint union $\biguplus_{\alpha'} \Y_{\alpha'}$. We have that $J_\alpha^\dagger \circ F \circ F^\dagger \circ J_\alpha \leq J_\alpha^\dagger \circ J_\alpha = I_{\Y_\alpha}$, and thus, $J_\alpha^\dagger \circ F$ is a partial function. Applying \cite[Proposition C.6]{Kornell18}, we obtain a subspace $X_\alpha \leq H$ and a function $F_\alpha\: \Q\{X_\alpha\} \to \Y_\alpha$ such that $F_\alpha\circ K_\alpha=J_\alpha^\dag\circ F$, where $K_\alpha$ is the partial function from $\H$ to $\X_\alpha = \Q\{X_\alpha\}$ defined by $K_\alpha(H, X_\alpha) = \CC \cdot \proj_{X_\alpha}^\dagger$.

If we vary $\alpha \in A$, we obtain families $\{X_\alpha\atomof\X_\alpha\}_{\alpha \in A}$, $\{K_\alpha\}_{\alpha \in A}$ and $\{F_\alpha\}_{\alpha \in A}$. Now, let $D = [K_\alpha^\dagger : \alpha \in A]^\dagger$, where the square-bracket notation refers to the coproduct in $\qRel$. We compute that 
\begin{align*}
\left(\biguplus _\alpha F_\alpha \right) & \circ D
=
\left(\biguplus _\alpha F_\alpha \right) \circ [K_\alpha^\dagger: \alpha \in A]^\dagger
= 
\left( [K_\alpha^\dagger: \alpha \in A] \circ \left(\biguplus _\alpha F_\alpha^\dagger\right)\right)^\dagger
\\ & =
[K_\alpha^\dagger \circ F_\alpha^\dagger : \alpha \in A]^\dagger
=
[(F_\alpha\circ K_\alpha)^\dagger : \alpha \in A]^\dagger
=
[(J_\alpha^\dagger \circ F)^\dagger : \alpha \in A]^\dagger
\\ & =
[F^\dagger \circ J_\alpha : \alpha \in A]^\dagger
=
(F^\dagger)^\dagger
=
F.
\end{align*}
Thus, $(\biguplus_\alpha F_\alpha) \circ D = F$. 

To show that $D$ is a decomposition, we show that it is a function. For distinct $\alpha, \beta \in A$, we apply \cite[Lemma~6.1(a)]{KLM20} to find that 
\begin{align*}
K_\alpha \circ K_\beta^\dagger
& \leq
F_\alpha^\dagger \circ F_\alpha \circ K_\alpha \circ K_\beta^\dagger \circ F_\beta^\dagger \circ F_\beta
=
F_\alpha^\dagger \circ J_\alpha^\dagger \circ F \circ F^\dagger \circ J_\beta \circ F_\beta
\\ & \leq 
F_\alpha^\dagger\circ J_\alpha^\dagger \circ J_\beta \circ F_\beta
=
F_\alpha^\dagger \circ \bot \circ F_\beta
=
\bot,
\end{align*}
so $K_\alpha \circ K_\beta^\dagger = \bot$. It follows that
\begin{align*}
D \circ D^\dagger
& =
\bigvee_{\alpha, \beta\in A} J_\alpha \circ J_\alpha^\dagger \circ D \circ D^\dagger \circ J_\beta \circ J_\beta^\dagger
=
\bigvee_{\alpha, \beta\in A} J_\alpha \circ (D^\dagger \circ J_\beta)^\dagger \circ ( D^\dagger \circ J_\beta) \circ J_\beta^\dagger
\\ & =
\bigvee_{\alpha, \beta\in A} J_\alpha \circ K_\alpha \circ K_\beta^\dagger \circ J_\beta^\dagger
=
\bigvee_{\alpha\in A} J_\alpha \circ K_\alpha \circ K_\alpha^\dagger \circ J_\alpha^\dagger
\leq
\bigvee_{\alpha\in A} J_\alpha \circ J_\alpha^\dagger
=
I_{(\biguplus_\alpha \Y_\alpha)}.
\end{align*}
Thus, $D$ is a partial function.

The equation $(\biguplus_\alpha F_\alpha) \circ D = F$ in the category $\mathbf{qPar}$ of quantum sets and partial functions is equivalent to the equation $D^\star \circ (\biguplus_\alpha F_\alpha)^\star = F^\star$ in the category $\mathbf{WStar}_\mathrm{HA}$ of hereditarily atomic von Neumann algebras and normal $*$-homomorphisms (cf. Theorem \ref{thm:qSet}). The normal $*$-homomorphism $F^\star$ is unital, because $F$ is a function. Then $D^\star \circ (\biguplus_\alpha F_\alpha)^\star = F^\star$ implies the normal $*$-homomorphism $D^\star$ is also unital, so $D$ is also a function.

It follows that $D$ is a decomposition: By a short direct computation, we have that $D(H, X_\alpha) = K_\alpha(H, X_\alpha) = \CC \cdot \mathrm{proj}_{X_\alpha}^\dagger$ for all $\alpha \in A$. For all distinct $\alpha, \beta \in A$, if $X_\alpha$ and $X_\beta$ are both nonzero, then they must be orthogonal anyway, because
$$\CC \cdot \mathrm{proj}_{X_\alpha} \cdot \mathrm{proj}_{X_\beta}^\dagger = D(H,X_\alpha)^\dagger \cdot D(H, X_\beta) \leq I_{(\biguplus_{\alpha\in A} \X_\alpha)}(X_\alpha, X_\beta) = 0.$$
Similarly,
$1_H \in I_\H(H, H) = \bigvee_{\alpha\in A} D(H, X_\alpha) \cdot D(H, X_\alpha)  = \bigvee_{\alpha\in A} \CC \cdot \mathrm{proj}_{X_\alpha}^\dagger \cdot \mathrm{proj}_{X_\alpha}$,
so $1_H = \sum_{\alpha\in A} \mathrm{proj}_{X_\alpha}^\dagger \cdot \mathrm{proj}_{X_\alpha}$, implying that $\bigvee_{\alpha\in A} X_\alpha = H$.
\end{proof}

\begin{lemma}\label{lem:coproduct of qposets decomposition}
Let $\{(\Y_\alpha,S_\alpha)\}_{\alpha\in A}$ be an indexed family of quantum posets, and let $\Y=\biguplus_{\alpha\in A}\Y_\alpha$ be the coproduct of $\{\Y_\alpha\}_{\alpha \in A}$ in $\qSet$. Let $\H$ be an atomic quantum set, let $S$ be an order on $\Y$, and let $F^1,F^2:\H\to\Y$ be functions such that $F^1\sqsubseteq F^2$. For each $i \in \{1,2\}$, let $D_i:\H\to\biguplus_{\alpha\in A}\Q\{X^i_\alpha\}$ be a decomposition of $\H$ and let $\{F^i_\alpha:\Q\{X^i_\alpha\}\to\Y_\alpha\}_{\alpha \in A}$ be an indexed family of functions such that 
\[ F^i=\left(\biguplus_{\alpha\in A}F^i_\alpha\right)\circ D_i,\]
as in Lemma \ref{lem:decomposition}.

Let $\beta \in A$. If $J_\beta\circ S_\beta=S\circ J_\beta,$ where $J_\beta\: \Y_\beta \to\Y$ is the canonical injection, then
\begin{itemize}
    \item[(1)] $X_\beta^1\perp X_\gamma^2$ for each $\gamma\in A$ distinct from $\beta$ and
\item[(2)]  $X^1_\beta = X^2_\beta$ implies $F_\beta^1\sqsubseteq F_{\beta}^2$.
\end{itemize}
\end{lemma}
\begin{proof}
The condition $F^1\sqsubseteq F^2$ is equivalent to $F^2\circ (F^1)^\dag\leq S$. Hence, we obtain
\[ \left(\biguplus_{\alpha\in A} F^2_\alpha\right)\circ D_2\circ D_1^\dag\circ \left(\biguplus_{\alpha\in A}F^1_\alpha\right)^\dag\leq S,  \]
and hence
\begin{align*} D_2\circ D_1^\dag & \leq \left(\biguplus_{\alpha\in A} F_\alpha^2\right)^\dag\circ
\left(\biguplus_{\alpha\in A} F_\alpha^2\right)\circ D_2\circ D_1^\dag\circ \left(\biguplus_{\alpha\in A}F^1_\alpha\right)^\dag\circ \left(\biguplus_{\alpha\in A}F^1_\alpha\right)\\
& \leq \left(\biguplus_{\alpha\in A} F_\alpha^2\right)^\dag\circ  S \circ \left(\biguplus_{\alpha\in A}F^1_\alpha\right).
\end{align*}
For $i \in \{1,2\}$ and $\beta \in A$, let $J_\beta^i:\Q\{X_\beta^i\}\hookrightarrow\biguplus_{\alpha\in A}\Q\{X^i_\alpha\}$ be inclusion. Let $\beta,\gamma\in A$ be distinct. Then,
\begin{align*}
 (J_\gamma^2)^\dag\circ D_2\circ D_1^\dag\circ J_\beta^1 & \leq (J_\gamma^2)^\dag\circ\left(\biguplus_{\alpha\in A} F_\alpha^2\right)^\dag\circ  S \circ \left(\biguplus_{\alpha\in A}F^1_\alpha\right)\circ J^1_\beta
 =(F_\gamma^2)^\dag\circ J_\gamma^\dag\circ S\circ J_\beta\circ F_\beta^1\\
 &= (F_\gamma^2)^\dag\circ J_\gamma^\dag\circ J_\beta \circ S_\beta\circ F_\beta^1= \bot, 
 \end{align*}
since $J_\gamma^\dag\circ J_\beta=\bot$ by \cite[Lemma~6.1(a)]{KLM20}.
Applying \cite[Lemma~A.5]{KLM20}, we now infer that  $((J^1_\beta)^\dag\circ D_1)(H,X_\beta)=D_1(H,X_\beta)=\CC \cdot \mathrm{proj}_{X^1_{\beta}}$ 
and similarly that $((J_\gamma^2)^\dag\circ D_2)(H,X^2_\gamma)=\CC\cdot \mathrm{proj}_{X^2_\gamma}.$ Thus, for distinct $\beta,\gamma\in A$, if $X_\beta^1$ and $X_\gamma^2$ are non-zero, we find that 
\begin{align*}
0 & = \bot(X^1_\beta, X^2_\gamma) =  [(J_\gamma^2)^\dag\circ D_2\circ D_1^\dag\circ J^1_\beta](X_\beta^1,X^2_\gamma)=[((J_\gamma^2)^\dag\circ D_2)\circ ( (J^1_\beta)^\dag\circ D_1)^\dag](X^1_\beta,X^2_\gamma)\\
& =
\bigvee_{H\atomof\H}((J_\gamma^2)^\dag\circ D_2)(H,X_\gamma^2)\cdot((J_\beta^1)^\dag\circ D_1)^\dag(X^1_\beta,H)=((J_\gamma^2)^\dag\circ D_2)(H,X_\gamma^2)\cdot ((J^1_\beta)^\dag\circ D_1)(H,X^1_\beta)^\dag\\
& = (\CC \cdot \mathrm{proj}_{X_\gamma^2}) \cdot(\CC \cdot \mathrm{proj}_{X_\beta^1})^\dag
= \CC \cdot \mathrm{proj}_{X_\gamma^2} \cdot \mathrm{proj}_{X_\beta^1}^\dag,
\end{align*}
which implies that $X_\gamma^2$ and $X_\beta^1$ are orthogonal subspaces of $H$.

For (2), assume $X_1^\beta=X_2^\beta$.  For each $i \in \{1,2\}$, the decomposition $D_i$ is surjective by Lemma \ref{lem:decomposition is surjective}, and thus, we have that
\begin{align*}
    S_\beta\circ F_\beta^i & =  J_\beta^\dag\circ J_\beta\circ S_\beta\circ F_\beta^i=J_\beta^\dag\circ S\circ J_\beta\circ F^i_\beta= J_\beta^\dag\circ S\circ\left(\biguplus_{\alpha\in A}F_\alpha^i\right)\circ J_\beta^i\\
    & =J_\beta^\dag\circ S\circ\left(\biguplus_{\alpha\in A}F_\alpha^i\right)\circ D_i \circ D_i^\dag \circ J_\beta^i =J_\beta^\dag\circ S\circ F^i\circ D_i^\dag\circ J_\beta^i. 
\end{align*}
As above, the partial functions $(J^1_\beta)^\dag\circ D_1$ and $(J^2_\beta)^\dag\circ D_1$ are completely determined by $X^1_\beta$ and $X^2_\beta$, respectively, and hence we have that $(J^1_\beta)^\dag\circ D_1  = (J^2_\beta)^\dag\circ D_2$. Recalling that $F^1\sqsubseteq F^2$ is equivalent to $S\circ F^2\leq S\circ F^1$, we calculate
\begin{align*}
    S_\beta\circ F^2_\beta & = 
    J_\beta\circ S\circ F^2\circ D_2^\dag\circ J^2_\beta=
    J_\beta\circ S\circ F^2\circ D_1^\dag\circ J^1_\beta \\ & \leq J_\beta\circ S\circ F^1\circ D_1^\dag\circ J^1_\beta
    =S_\beta\circ F_\beta^1.
\end{align*}
Therefore, $F^1_\beta \below F^2_\beta$, as claimed.
\end{proof}

\begin{theorem}\label{thm:qCPO has all coproducts}
The category $\qCPO$ has all coproducts. The coproduct of an indexed family of quantum cpos in $\cat{qCPO}$ is also their coproduct in $\cat{qPOS}$.
\end{theorem}
\begin{proof}
Let $\{(\Y_\alpha,S_\alpha)\}_{\alpha\in A}$ be an indexed family of quantum cpos, and let $(\Y,S)$ be its coproduct $(\Y,S)$ in $\qPOS$ (cf. \cite[Proposition 6.2]{KLM20}). Hence,
\begin{equation}\label{eq:condition for applying lemma}
 J_\alpha\circ S_\alpha=S\circ J_\alpha
 \end{equation}
for each $\alpha\in A$.
Let $\H = \Q\{H\}$ be an atomic quantum set, and let $K^1\sqsubseteq K^2\sqsubseteq\cdots:\H\to\Y$ be a monotonically ascending sequence of functions.
By Lemma \ref{lem:decomposition}, for each $n\in\NN$, there exist a decomposition $D_n:\H\to\biguplus_{\alpha\in A}\Q\{X^n_\alpha\}$ and an indexed family $\{K^n_\alpha: \Q\{X^n_\alpha\} \to \Y_\alpha\}_{\alpha \in A}$ of functions such that $K^n=\left(\biguplus_{\alpha\in A}K_\alpha^n\right)\circ D_n$. 

We show that the decompositions $D_n$ are all the same, as follows: Since Equation (\ref{eq:condition for applying lemma}) holds for each $\alpha\in A$, we can apply Lemma \ref{lem:coproduct of qposets decomposition}(1) to conclude that for each $n\in\NN$, we have $X^1_\alpha\perp X^{n}_\beta$ for each $\alpha\neq\beta$ in $A$. 
Since $H=\bigoplus_{\alpha\in A}X_\alpha^n$ by the definition of a decomposition, it follows that $X_\alpha^1=X_{\alpha}^{n}$ for each $\alpha\in A$. Hence $D_1=D_{n}$ for each $n\in\NN$. Let $D=D_1 = D_2 = \ldots$, and for each $\alpha \in A$, let $X_\alpha = X_\alpha^1 = X_\alpha^2 = \cdots$.

We now apply Lemma \ref{lem:coproduct of qposets decomposition}(2) to conclude that $K_\alpha^1\sqsubseteq K_\alpha ^2\sqsubseteq\cdots$ for each $\alpha\in A$.
If $X_\alpha$ is zero-dimensional, then $\Q\{X_\alpha\} = `\emptyset$ is initial in $\qSet$; hence $K^1_\alpha=K^n_\alpha$ for each $n\in\NN$. Taking $K^\infty_\alpha=K^1_\alpha$, we automatically have $K_\alpha^n\nearrow K_\alpha^\infty$.
If $X_\alpha$ is not zero-dimensional, then the fact that $\Y_\alpha$ is a quantum cpo ensures the existence of a function $K^\infty_\alpha:\X_\alpha\to\Y_\alpha$ such that $K^n_\alpha\nearrow K^\infty_\alpha$. In either case, 
\begin{equation}\label{eq:intersection coproduct} S_\alpha\circ K_\alpha^\infty=\bigwedge_{n\in\NN}S_\alpha\circ K_\alpha^n.\end{equation}

Let $K^\infty=\left(\biguplus_{\alpha\in A}K_\alpha^\infty\right)\circ D$. We show that $K^n\nearrow K^\infty$: 
\begin{align*}
    S\circ K^\infty & = \left(\biguplus_{\alpha\in A}S_\alpha^\infty\right)\circ\left(\biguplus_{\alpha\in A}K_\alpha^\infty\right)\circ D =\left(\biguplus_{\alpha\in A}(S_\alpha\circ K_\alpha^\infty)\right)\circ D=\left(\biguplus_{\alpha\in A}\left(\bigwedge_{n\in\NN}(S_\alpha\circ K_\alpha^n)\right)\right)\circ D\\
    & =\left(\bigwedge_{n\in\NN}\left(\biguplus_{\alpha\in A}(S_\alpha\circ K_\alpha^n)\right)\right)\circ D=\bigwedge_{n\in\NN}\left(\biguplus_{\alpha\in A}(S_\alpha\circ K_\alpha^n)\circ D\right)\\
    & =\bigwedge_{n\in\NN}\left(\left(\biguplus_{\alpha\in A}S_\alpha\right)\circ\left(\biguplus_{\alpha\in A}K_\alpha^n\right)\circ D\right)=\bigwedge_{n\in\NN}S\circ K^n, 
\end{align*}
where the fourth equality follows from \cite[Lemma~6.1]{KLM20}, and the fifth equality follows from \cite[Proposition~A.6]{KLM20}. Therefore, $\Y$ is a quantum cpo.

Next, we show that the canonical injections $J_\beta:\Y_\beta\to\Y$ are Scott continuous for each $\beta\in A$. So, let $\H$ be an atomic quantum set, and let $E_1\sqsubseteq E_2\sqsubseteq\cdots:\H\to\Y_\beta$ be a monotonically ascending sequence of functions with limit $E_\infty$. Let $K^i=J_\beta\circ E_i$ for each $i\in\NN\cup\{\infty\}$. Since the $J_\alpha$ are monotone by \cite[Proposition~6.2]{KLM20}, left multiplication with $J_\alpha$ is monotone \cite{KLM20}*{Lemma 4.4}, and it follows  that $K^1\sqsubseteq K^2\sqsubseteq\cdots\sqsubseteq$ is a monotonically ascending sequence. We need to show that $K^n\nearrow K^\infty$. For each $\alpha\in A$, we calculate that
\begin{align*} J_\alpha^\dag\circ\bigwedge_{n\in\NN} S\circ K^n & =\bigwedge_{n\in\NN}J_\alpha^\dag\circ S\circ K^n=\bigwedge_{n\in\NN}J_\alpha^\dag\circ S\circ J_\beta\circ E_n=\bigwedge_{n\in\NN}J_\alpha^\dag\circ J_\beta\circ S_\beta\circ E_n\\
& =\Delta_{\alpha,\beta} \cdot \bigwedge_{n\in\NN}S_\beta\circ E_n=\Delta_{\alpha,\beta}\cdot (S_\beta\circ E_\infty) =J_\alpha^\dag\circ J_\beta\circ S_\beta\circ E_\infty\\
&= J_\alpha^\dag\circ S\circ J_\beta\circ E_\infty =J_\alpha^\dag\circ S\circ K_\infty,
\end{align*}
where we have used \cite[Proposition~A.6]{KLM20} for the first and the penultimate equality, \cite[Proposition~6.2]{KLM20} for the third equality, and \cite[Lemma~6.1]{KLM20} for the fourth and sixth equalities. We use $E_n\nearrow E_\infty$ for the fifth equality. It now follows by \cite[Lemma~6.1]{KLM20} that $\bigwedge_{n\in\NN}S\circ K^n=S\circ K^\infty$, i.e., that $K^n\nearrow K^\infty$. Therefore, each of the canonical injections $J_\beta$, for $\beta \in A$, is Scott continuous.

Finally, let $(\Z,T)$ be a quantum cpo, and for each $\alpha \in A$, let $F_\alpha:\Y_\alpha\to \Z$ be a Scott continuous function. We need to show that $[F_\alpha:\alpha\in A]:\Y\to\Z$ is Scott continuous, where the notation refers to the universal property of the coproduct. Hence, let $\H$ be an atomic quantum set and let 
$K^1\sqsubseteq K^2\sqsubseteq\cdots:\H\to\Y$
be a monotonically ascending sequence of functions with limit $K^\infty$.
As before, we have a decomposition $D:\H\to\biguplus_{\alpha\in A}\Q\{X_\alpha\}$ and an indexed family of functions $\{K^i_\alpha:\{X_\alpha\}\to\Y_\alpha\}_{\alpha \in A}$ such that
$K^i=\left(\biguplus_{\alpha\in A}K^i_\alpha\right)\circ D$, for each index $i \in \NN \union \{\infty\}$, and $K^n_\alpha\nearrow K^\infty_\alpha$. Since each function $F_\alpha$ is Scott continuous, we also have that $F_\alpha\circ K_\alpha^n\nearrow F_\alpha\circ K_\alpha^\infty$. We show that $\left[F_\alpha:\alpha\in A\right]\circ K^n\nearrow \left[F_\alpha:\alpha\in A\right]\circ K^\infty$: 
\begin{align*}
\bigwedge_{n\in\NN}T\circ [F_\alpha:\ &\alpha\in A]\circ K^n  =  \bigwedge_{n\in\NN}T\circ [F_\alpha:\alpha\in A]\circ \left(\biguplus_{\alpha\in A}K^n_\alpha\right)\circ D  \\
& = \bigwedge_{n\in\NN}[T\circ F_\alpha\circ K_\alpha^n:\alpha\in A]\circ D
= \left(\bigwedge_{n\in\NN} [T\circ F_\alpha\circ K_\alpha^n:\alpha\in A]\right)\circ D\\
& = \left[\bigwedge_{n\in\NN}T\circ F_\alpha\circ K_\alpha^n:\alpha\in A\right]\circ D
= \left[T\circ F_\alpha\circ K_\alpha^\infty:\alpha\in A\right]\circ D\\
& = T\circ \left[F_\alpha\in A\right]\circ\left(\biguplus_{\alpha\in A}K_\alpha^\infty\right)\circ D 
= T\circ[F_\alpha:\alpha\in A]\circ K^\infty,
\end{align*}
where the third equality follows from \cite[Proposition~A.6]{KLM20}, the fourth from \cite[Lemma~6.1]{KLM20}, and the fifth equality is $F_\alpha\circ K_\alpha^n\nearrow F_\alpha\circ K_\alpha^\infty$. Therefore, the monotone function $[F_\alpha:\alpha\in A]:\Y\to\Z$, whose existence and uniqueness is guaranteed by the universal property of the coproduct in $\cat{qPOS}$, is Scott continuous. All together, we find that $(\Y,S)$ is also its coproduct of the family $\{(\Y_\alpha, S_\alpha)\}_{\alpha \in A}$ in $\cat{qCPO}$.
\end{proof}

\subsection{Factorizations in $\mathbf{qSet}$}\label{subsec:EM-factorization}

Let $F:\X\to\Y$ be a function. Then the restriction of $F|_\Z=F\circ J_\Z$ to any subset of $\Z\subseteq\X$ is always a function, but the corestriction $F|^\W=J_\W^\dag\circ F$ of $F$ to a subset $\W\subseteq \Y$ is not always a function. However, we prove below that $F|^\W$ is a function if $\W$ contains the \emph{range} of $F$, i.e., the subset 
\[\ran F:=\Q\{Y\atomof\Y:F(X,Y)\neq 0\text{ for some }X\atomof\X\}.\]
We denote the embedding $J_{\ran F}^\Y:\ran F\to\Y$ by $J_F$.

\begin{lemma}\label{lem:surjective function and range}
Let $F:\X\to\Y$ and $G:\Y\to\Z$ be functions between quantum sets. Then:
\begin{itemize}
    \item[(a)] $F$ is surjective if and only if $\ran F=\Y$;
    \item[(b)] $\ran (G\circ F)\subseteq\ran G$;
    \item[(c)] $\ran (G\circ F)=\ran G$ if $F$ is surjective.
\end{itemize}
\end{lemma}
\begin{proof}
Assume $F$ is surjective, i.e., $F\circ F^\dag=I_\Y$. Let $Y\atomof\Y$. Then $\CC 1_Y=(I_\Y)_Y^Y=(F\circ F^\dag)_Y^Y=\bigvee_{X\atomof\X}F_X^Y\cdot (F^\dag)_Y^X,$ hence there must be some $X\atomof\X$ such that $F(X,Y)=F_X^Y\neq 0$, i.e., $Y\atomof\ran F$. Conversely, assume that $\ran F=\Y$. Let $Y\atomof\Y$. Then there is some $X\atomof\X$ such that $F(X,Y)\neq 0$, hence there is some nonzero linear map $a:X\to Y$ in $F(X,Y)$, then also $0\neq aa^\dag$, which is an element in $F(X,Y)\cdot F(X,Y)^\dag=F(X,Y)\cdot F^\dag(Y,X)$, hence $0<F(X,Y)\cdot F^\dag(Y,X)\leq (F\circ F)^\dag(Y,Y)\leq I_\Y(Y,Y)=\CC 1_Y$, and since $\CC 1_Y$ is one-dimensional, we must have $(F\circ F)^\dag(Y,Y)=I_\Y(Y,Y)$. If $Y,Y'\atomof\Y$ are distinct, then $(F\circ F)^\dag(Y,Y')\leq I_\Y(Y,Y')=0$ forces $(F\circ F)^\dag=I_\Y(Y,Y')$, whence $F\circ F^\dag=I_\Y$, i.e., $F$ is surjective, which proves (a).

Let $Z\atomof\ran(G\circ F)$. Then $(G\circ F)(X,Z)\neq 0$ for some $X\atomof\X$, hence there must be some $Y\atomof\Y$ such that $0<G(Y,Z)\cdot F(X,Y)\leq(G\circ F)(X,Z)$, which can only happen if $G(Y,Z)\neq 0$, i.e., $Z\atomof\ran G$, which proves (b). For (c), assume that $Z\atomof\ran G$, so $G(Y,Z)\neq 0$ for some $Y\atomof\Y$. Since $F$ is surjective, we have $F\circ F^\dag=I_\Y$, whence $(G\circ F\circ F^\dag)(Y,Z)\neq 0$. It follows that there must be some $X\atomof\X$ such that $0<(G\circ F)(X,Z)\cdot F^\dag(Y,X)\leq (G\circ F\circ F^\dag)(Y,Z)$, forcing $(G\circ F)(X,Z)\neq 0$. Thus $Z\atomof\ran(G\circ F)$.
\end{proof}

\begin{lemma}\label{lem:bar F is surjective function}
    Let $F:\X\to\Y$ be a function. Then the restriction $\overline F:=F|^{\ran F}$ to its range is a surjective function such that $F=J_F\circ \overline F$.
\end{lemma}
\begin{proof}We have 
\[\overline F\circ \overline F^\dag=J_F\circ F\circ F^\dag\circ J_F^\dag\leq J_F\circ J_F^\dag=I_{\ran F},\]
whereas for each $X\atomof\X$, we have
\begin{align*}
   (\overline F^\dag\circ \overline F)_X^X  &  = \bigvee_{W\atomof\ran F}(F^\dag\circ J_F )_W^X\cdot (J_F\circ F)_X^W=\bigvee_{W\atomof\W}F^X_W\cdot F_X^W\\
   & =\bigvee_{Y\atomof\Y}F^X_Y\cdot F_X^Y=(F^\dag\circ F)_X^X=(I_\X)_X^X.
\end{align*}
For atoms $X\neq X'$ of $\X$, we have
$(\overline F^\dag\circ \overline F)_X^{X'}\geq 0=(I_\X)_X^X$, so $\overline F^\dag\circ \overline F\geq I_\X$, which shows that $\overline F$ is indeed a function. Since $\overline F$ is a function $\X\to\ran F$, we have $\ran \overline F\subseteq\ran F$. Let $Y\atomof\ran F$. Then $F(X,Y)\neq 0$ for some $X\atomof\Y$, hence $\overline F(X,Y)=(J_F^\dag\circ F)(X,Y)=F(X,Y)\neq 0$, so $Y\in\ran \overline F$. We conclude that $\ran \overline F=\ran F$, hence it follows from Lemma \ref{lem:surjective function and range} that $\overline F$ is surjective. Finally, we have $J_F\circ\overline F=J_F\circ J_F^\dag\circ F\leq F$ for $J_F$ is a function. Since parallel functions are only comparable if they are equal \cite[Lemma A.7]{KLM20}, it follows that $F=J_F\circ\overline F$.
\end{proof}

The next lemma now characterizes precisely for which subsets $\W$ of $\Y$ the corestriction $F|^\W$ is a function.

\begin{lemma}\label{lem:factors through a subset}
Let $\X$ and $\Y$ be quantum sets, and let $F \colon \X \to \Y$ be a function. Let $\W$ be a subset of $\Y$. Then the following are equivalent:
\begin{itemize}
    \item[(a)] $\ran F \subsetof \W$;
    \item[(b)] there is a function $G:\X\to\W$ such that $F=J_\W\circ G$;
    \item[(c)] $F|^\W$ is a function.
\end{itemize}
In case these equivalent conditions hold, we have:
\begin{itemize}
    \item[(1)] $F|^\W$ is the unique function $G:\X\to\W$ such that $F=J_\W\circ G$;
\item[(2)] $\ran F=\ran F|^\W$;
\item[(3)] $F|^\W$ is injective if and only if $F$ is injective;
\item[(4)] $F|^\W$ is surjective if and only if $\W=\ran F$.
\end{itemize} 
\end{lemma}

\begin{proof}
Assume that $\ran F \subsetof \W$. By Lemma \ref{lem:bar F is surjective function}, we have $F = J_F \circ \overline F$. Then, \[F = J_F \circ \overline F = J_{\W} \circ J_F^\W \circ \overline F.\] We conclude that $G = J_F^\W \circ \overline F$ is a function $\X \to \W$ such that $F = J_\W \circ G$, so (a) implies (b). Now, (c) follows from (b) because \[F|^\W=J_\W^\dag\circ F=J_\W^\dag\circ J_\W\circ G=G,\] which also shows that $F|^\W$ is the unique function such that $F=J_\W\circ G$. Conversely, assume that $F|^\W$ is a function. Then $J_\W\circ F|^\W$ is also a function, and $J_\W\circ F|^\W=J_\W\circ J_\W^\dag\circ F\leq F$, since $J_\W$ is a function. Now, parallel functions are only comparable if they are equal \cite[Lemma A.7]{KLM20}, so $J_\W\circ F|^\W=F$. It follows that  $\ran F=\ran (J_\W\circ F|^\W).$ 
Let $Y\atomof\ran F$. Then there is some $X\atomof\X$ such that $0\neq F(X,Y)=(J_\W\circ F|^\W)(X,Y)$, hence there must be some $W\atomof\W$ such that  $F|^\W(X,W)\neq 0$ and $J_\W(W,Y)\neq 0$. The former expression implies that $W\atomof\ran F|^\W\subseteq\W$. The latter expression only holds if $W=Y$, so $Y\in\ran F|^\W\subseteq \W$, which shows that $\ran F\subseteq \ran F|^\W$ and that (c) implies (a).

We already showed that $G$ in (b) must be equal to $F|^\W$, and that $\ran F\subseteq \ran F|^\W$. For the remainder of the proof, assume that conditions (a)-(c) holds. Let $W\atomof\ran F|^\W$. Then $F|^\W(X,W)\neq 0$ for some $X\atomof\X$, and $J_\W(W,W)=\mathbb C 1_W$, hence \[F(X,W)=(J_\W\circ F|^\W)(X,W)\geq J_\W(W,W)\cdot F|^\W(X,W)\neq 0,\] implying that $W\atomof\ran F$. We conclude that $\ran F=\ran F|^\W$.
For (3), let $F$ be injective, so $F^\dag\circ F=I_\X$. Then
\[ I_\X\leq(F|^\W)^\dag\circ F|^\W=F^\dag\circ J_\W\circ J_\W^\dag\circ F\leq F^\dag\circ F=I_\X,\]
where the first inequality follows since $F|^\W$ is a function and the second because $J_\W$ is a function. Hence, $(F|^\W)^\dag\circ F|^\W=I_\X$, i.e., $F|^\W$ is injective. Conversely, if $F|^\W$ is injective, then $F=J_\W\circ F|^\W$ is injective as the composition of injective functions.
Finally, by (2) and by Lemma \ref{lem:surjective function and range}, $F|^\W$ is surjective if and only if $\ran F=\ran \overline F=\W$, which shows (4). 
\end{proof}

In particular, $\overline F$ is the unique surjection $G$ such that $F=J_F\circ G$.  More generally, the injections and surjections of $\cat{qSet}$ form a factorization structure, and the factorization $F = J_F \circ \overline F$ is the canonical factorization of $F$ for this factorization system.

\begin{proposition}\label{prop:factorization structure}
The surjections and injections of $\cat{qSet}$ form a factorization structure on $\cat{qSet}$ in the sense of \cite{adameketall:joyofcats}*{Definition 14.1}, i.e., if $\E$ denotes the class of surjections and $\M$ the class of injections in $\mathbf{qSet}$, then $\E$ and $\M$ are closed under composition with isomorphisms (which in $\mathbf{qSet}$ are bijections), any morphism $F$ in $\mathbf{qSet}$ can be written as $F=M\circ E$ for $M\in\M$ and $E\in\E$, and for each commutative square as below, with $E'\in\E$ and $M'\in\M$, there is a unique diagonal map $D:\W\to\Y$ such that the diagram commutes:
\[ \begin{tikzcd}
\X\ar{r}{E'}\ar{d}[swap]{F} & \W\ar{d}{G}\arrow[dl,dashed,swap,"D"]\\
\Y\ar{r}[swap]{M'} & \Z.
\end{tikzcd} \]
\end{proposition}

\begin{proof}
The surjections and injections in the category $\mathbf{WStar}$ of von Neumann algebras and unital normal $*$-homomorphisms form a factorization structure. Indeed, it is well known that the range of a unital normal $*$-homomorphism from one von Neumann algebra to another is itself a von Neumann algebra, and this readily implies that each unital normal $*$-homomorphism factors into an surjective unital normal $*$-homomorphism followed by a injective one.
It is routine to verify the unique diagonalization property. The surjections and injections in the full subcategory $\mathbf{WStar}_\mathrm{HA}$ of hereditarily atomic von Neumann algebra therefore also form a factorization structure. Appealing to Theorem \ref{thm:qSet}, and to \cite[Propositions 8.1 \& 8.4]{Kornell18}, we conclude that the surjections and injections in $\cat{qSet}$ form a factorization structure.
\end{proof}

\subsection{Sub-cpo's}
Recall that a sub-cpo of a cpo $(S, \below)$ is  a subset of $S$ that is closed under suprema of $\omega$-chains. We generalize this notion to the quantum setting.

\begin{definition}[Quantum sub-cpo]
Let $(\Y,S)$ be a quantum cpo. We call $\X\subseteq\Y$ a \emph{sub-cpo of $\Y$} if given any atomic quantum set $\H$ and a monotonically increasing sequence
$K_1\sqsubseteq K_2\sqsubseteq\cdots:\H\to\Y$ with  $\ran K_i\subseteq\X$, the limit $K_i\nearrow K_\infty$ satisfies $\ran K_\infty\subseteq\X$.
\end{definition}

\begin{proposition}\label{prop:subqcpo}
Let $(\Y,R)$ be a quantum cpo and let $\X\subseteq\Y$ be a subposet with the relative order $R|_\X^\X=J_\X^\dag\circ R\circ J_\X$. Then $\X$ is a sub-cpo of $\Y$ if and only if $(\X,R|_\X^\X)$ is a quantum cpo and $J_\X$ is Scott continuous.
\end{proposition}
\begin{proof}
First assume that $\X$ is a sub-cpo of $\Y$ and let $K_1\sqsubseteq K_2\sqsubseteq\cdots:\H\to\X$ be a monotonically ascending sequence of functions. By definition of $R$, it follows that $J_\X$ is an order embedding, hence $J_\X$ is injective and monotone by \cite[Lemma~2.4]{KLM20}. Let $G_n=J_\X\circ K_n$. Since left multiplication with the monotone function $J_\X$ is monotone \cite{KLM20}*{Lemma 4.4}, it follows that $G_1\sqsubseteq G_2\sqsubseteq\cdots:\H\to\Y$ 
is a monotonically ascending sequence of functions, which has a limit $G_\infty$, because $\Y$ is a quantum cpo.
Since $G_n$ factors via $J_\X$, it follows that $\ran G_n\subseteq\X$. By definition of a sub-cpo, we have that $\ran G_\infty\subseteq\X$, and by Lemma \ref{lem:factors through a subset}, it follows that $G_\infty= J_\X\circ K_\infty$ for some function $K_\infty:\H\to\X$. 
Then
\begin{align*}
    R|_\X^\X\circ K_\infty & = J_\X^\dag\circ R\circ J_\X\circ K_\infty=J_\X^\dag\circ R\circ G_\infty = J_\X^\dag\circ\bigwedge_{n\in\NN}R\circ G_n\\
    & = \bigwedge_{n\in\NN}J_\X^\dag\circ R\circ G_n =\bigwedge_{n\in\NN}J_\X^\dag\circ R\circ J_\X\circ K_n = \bigwedge_{n\in\NN}R|_\X^\X\circ K_n,
\end{align*}
where the first and the last equalities are by definition of $R$, the second and the penultimate equalities are by definition of the $G_i$, the third equality is by $G_n\nearrow G_\infty$, and the fourth equality is by \cite[Proposition~A.6]{KLM20}. We conclude that $K_n\nearrow K_\infty$, so $(\X,R|_\X^\X)$ is a quantum cpo. 

Next we show that $J_\X$ is Scott continuous. Let $K_1\sqsubseteq K_2\sqsubseteq\cdots\sqsubseteq K_\infty\:\H\to\X$ with $K_n\nearrow K_\infty$. Then $\ran K_\infty\subseteq \X$ since $\X$ is a sub-cpo of $\Y$. By the first part of the proof, the sequence $J_\X\circ K_n$ has a limit $G_\infty$ and there is some $K'_\infty$ satisfying $\ran K'_\infty\subseteq\X$, $J_\X\circ K'_\infty = G_\infty$, and $K_n\nearrow K'_\infty$. But, limits are unique by Lemma \ref{lem:lim is sup}, so $K_\infty = K'_\infty$, and then $J_\X\circ K_n\nearrow J_\X \circ K_\infty$, so $J_\X$ is Scott continuous.

Now assume that $(\X,R|_\X^\X)$ is a quantum cpo and that $J_\X$ is Scott continuous. Let $G_1\sqsubseteq G_2\sqsubseteq \cdots\sqsubseteq G_\infty:\H\to\Y$ satisfy $G_n\nearrow G_\infty$. Assume that $\ran G_n\subseteq\X$ for each $n\in\NN$. We need to show that $\ran G_\infty\subseteq\X$, too. For each $n\in\NN$, since $\ran G_n\subseteq\X$, it follows from Lemma \ref{lem:factors through a subset} that $G_n=J_\X\circ K_n$ for some function $K_n:\H\to\X$.
Then for each $n,m\in\NN$ such that $n\leq m$, we have
$R|_\X^\X\circ K_m =J_\X^\dag\circ R\circ J_\X\circ K_m=J_\X^\dag\circ R\circ G_m\leq J_\X^\dag\circ R\circ G_n=J_\X^\dag\circ R\circ J_\X\circ K_n=R|_\X^\X\circ K_n,$
where the inequality follows from $G_n\sqsubseteq G_m$, hence also $K_n\sqsubseteq K_m$.
We conclude that $K_1\sqsubseteq K_2\sqsubseteq\cdots:\H\to\X$
is a monotonically ascending sequence of functions, which has a limit $K_\infty$, since $\X$ is a quantum cpo. By Scott continuity of $J_\X$ we obtain $G_n=J_\X\circ K_n\nearrow J_\X\circ K_\infty,$
and since also $G_n\nearrow G_\infty$, it follows from  Lemma \ref{lem:lim is sup} that $G_\infty=J_\X\circ K_\infty$. Lemma \ref{lem:factors through a subset} now implies that $\ran G_\infty\subseteq\X$, so $\X$ is a sub-cpo of $\Y$.
\end{proof}

\begin{proposition}\label{prop:finite subset is subqcpo}
Let $(\Y,S)$ be a quantum cpo and let $\X\subseteq\Y$ be a finite subset. Then $\X$ is a sub-cpo of $\Y$.
\end{proposition}
\begin{proof}
This follows directly from Propositions \ref{prop:finite qposet is qcpo} and \ref{prop:subqcpo}.
\end{proof}

Classically, if $Y$ is a cpo, then the sub-cpos of $Y$ form the closed sets of the \emph{$d$-topology on $Y$}~\cite{Zhao-Fan}. This topology also has an intrinsic description: a subset $X\subseteq Y$ is $d$-closed iff $X$ is closed under suprema of ascending sequences, and the $d$-closure of a subset is the intersection of the sub-cpos containing the subset. The fact that quantum sets are not well-pointed makes it challenging to describe  quantum topologies, but the next results allow use to define the `quantum' $D$-closed sets of a quantum cpo, along with the $D$-closure of a subset.

\begin{lemma}\label{lem:sub-cpo generated by subset}
Let $\Y$ be a quantum cpo, and let $\{\X_\alpha\}_{\alpha\in A}$ be a collection of sub-cpos of $\Y$. Then $\X=\bigcap_{\alpha\in A}\X_\alpha$ is a sub-cpo of $\Y$. In particular, if $\Z\subset\Y$ is non-empty, and $\{\X_\alpha\mid \alpha \in A\}$ is the family of sub-cpos of $\Y$ containing $\Z$, then $\X = \bigcap_\alpha \X_\alpha$ is a non-empty sub-cpo of $\Y$ that contains $\Z$, so it is the smallest such sub-cpo.
\end{lemma}
\begin{proof}
Let $K_1\sqsubseteq K_2\sqsubseteq\cdots\sqsubseteq K_\infty:\H\to\Y$ satisfy $K_n\nearrow K_\infty$ such that $\ran K_n\subseteq\X$ for each $n\in\NN$. Fix $\alpha\in A$. Then for each $n\in\NN$, we have $\ran K_n\subseteq\X\subseteq\X_\alpha$, and since $\X_\alpha$ is a sub-cpo of $\Y$, it follows that $\ran K_\infty\subseteq\X_\alpha$. Since $\alpha\in A$ is arbitrary, we conclude that $\ran K_\infty\subseteq\bigcap_{\alpha\in A}\X_\alpha=\X$, whence $\X$ is indeed a sub-cpo of $\Y$.

The proof of the second assertion follows because $\Y$ is a sub-cpo containing $\Z$, so the family of such sub-cpos of $\Y$ is non-empty. Since the intersection of this family  clearly contains $\Z$, it follows that the intersection is the smallest sub-cpo of $\Y$ containing $\Z$.
\end{proof}

Thus the following concept is well defined:
\begin{definition}[Sub-cpo $\overline\X$ generated by $\X$]\label{def:sub-cpo generated by subset}
Let $\Y$ be a quantum cpo and let $\X\subseteq\Y$. Then we define
\[\overline\X=\bigcap\{\Z\subseteq\Y:\Z\text{ is a sub-cpo of }\Y\text{ containing }\X\},\]
which we call the \emph{sub-cpo of} $\Y$ \emph{generated by} $\X$ or the \emph{$D$-closure} of $\X$ in $\Y$.
\end{definition}

\subsection{The $\D$-completion of quantum posets and cocompleteness of $\qCPO$}
Let $\POS_{\mathrm{SC}}$ be the category of ordinary posets and Scott continuous maps. Then the inclusion $\CPO\to\POS_{\mathrm{SC}}$ has a left adjoint, which assigns to each poset $X$ its \emph{$D$-completion} \cite{Zhao-Fan}. Instrumental for this completion is the \emph{$D$-topology} on a cpo $Y$; the closed subsets of this topology are precisely the sub-cpos of $Y$. As a consequence, the $D$-closure of a subset $S$ of $Y$ is the intersection of all sub-cpos of $Y$ containing $S$. 

We will generalize 
this construction to the quantum realm, i.e., if $\qPOS_\mathrm{SC}$ denotes the category of quantum posets and Scott continuous maps, we show that the inclusion $\I:\qCPO\to\qPOS_\mathrm{SC}$ has a left adjoint $\D$. It turns out that the results required to show that $\D$ exists are almost identical to the results required for the cocompleteness of $\qCPO$. 

We first need to show that $\qCPO$ is \emph{wellpowered}. This means that for each quantum cpo $\Y$, the class of all equivalence classes of monomorphisms in $\qCPO$ with codomain $\Y$ under a suitable equivalence relation $\sim$ forms a set. Specifically, for monomorphisms $F:\X\to\Y$ and $F':\X'\to \Y$, we define $F\sim F'$ if there exists an isomorphism $K:\X\to\X'$ in $\qCPO$ such that $F'=F\circ K$. In this case, $F$ and $F'$ and are said to be \emph{isomorphic}.

\begin{proposition}\label{prop:monos in qCPO}Monomorphisms in $\mathbf{qPOS}_\mathrm{SC}$ and in $\mathbf{qCPO}$ are precisely the injective Scott continuous maps.
\end{proposition}
\begin{proof}
The proof of this statement is almost identical to the proof of \cite[Lemma~3.1]{KLM20}, which states that monomorphisms in $\mathbf{qPOS}$ are precisely the injective monotone maps. In particular, the proof that any injective monotone morphism is a monomorphism is identical. For the other direction, let $(\X,R)$ be a quantum poset and let $(\Y,S)$ be a quantum cpo. Let $M:(\X,R)\to(\Y,S)$ be a Scott continuous map that is not injective. Just as in \cite[Lemma~3.1]{KLM20}, it follows that $M:\X\to\Y$ is not a monomorphism in $\mathbf{qSet}$, hence there is a quantum set $\W$ and distinct functions $F,G:\W\to\X$ such that $M\circ F=M\circ G$. By Example \ref{ex:trivially ordered set is qcpo}, $(\W,I_\W)$ is a quantum cpo, and it follows from Example \ref{ex:function with trivially ordered domain is Scott continuous} that the functions $F,G:(\W,I_\W)\to(\X,R)$ are Scott continuous, which implies that $M$ is not an monomorphism in $\mathbf{qPOS}_{\mathrm{SC}}$. If $(\X,R)$ is a quantum cpos, it follows from the same argumentation that $M$ is not a monomorphism in $\qCPO$.  
\end{proof}

\begin{proposition}\label{prop:qCPO is wellpowered}
$\qCPO$ is wellpowered.
\end{proposition}
\begin{proof}
Let $(\Y,S)$ be a quantum cpo, and let $F:(\X,R)\to(\Y,S)$ be a monomorphism in $\qCPO$. By \cite[Lemma~3.1]{KLM20} and Proposition \ref{prop:monos in qCPO} it follows that any monomorphism in $\qCPO$ is a monomorphism in the ambient category $\qPOS$. Moreover, by Lemma \ref{lem:order iso is Scott continuous}, any isomorphism in $\qCPO$ is an isomorphism in $\qPOS$, and any isomorphism in $\qPOS$ between objects in $\qCPO$ is an isomorphism in $\qCPO$. As a consequence, the class $\aA$ of monomorphisms in $\qCPO$ with codomain $\Y$ is a subclass of the class $\bB$ of monomorphisms in $\qPOS$ with codomain $\Y$. Since isomorphisms in $\qCPO$ coincide with isomorphisms in $\qPOS$, it follows that $\aA/\sim$ is a subclass of $\bB/\sim$, which is a set by \cite[Theorem~6.4]{KLM20}. We conclude that $\aA/\sim$ must also be a set, whence $\qCPO$ is wellpowered.
\end{proof}

\begin{proposition}\label{prop:range is finite}
Let $\X$ and $\Y$ be quantum sets, and let $F\colon \X \to \Y$ be a function. If $\At(\X)$ is finite, then $\At(\ran F)$ is finite. If $\At(\X)$ is infinite, then $\mathrm{card}(\At(\ran F))\leq \mathrm{card}(\At(\X))$.
\end{proposition}

\begin{proof}
We factor $F$ through its range $F = J_F \circ \overline F$ \cite[Definition~3.2]{KLM20}. The function $\overline F\: \X \to \ran F$ is surjective, so $\overline F^\star\: \ell^\infty(\ran F) \to \ell^\infty(\X) = \bigoplus_{X\atomof\X}L(X)$ is an injective unital normal $*$-homomorphism as follows from Theorem \ref{thm:qSet} and \cite{Kornell18}*{Proposition 8.1}.

Now, $\At(\X)$ is finite if and only if $\ell^\infty(\X)=\bigoplus_{X\in\At(\X)}L(X)$ is finite-dimensional, hence if $\At(\X)$ is finite, then the injectivity of $\overline F^\star$ forces $\ell^\infty(\ran F)$ to be finite-dimensional, so $\At(\ran F)$ must be finite, too. For the second statement, assume that $\At(\X)$ is infinite. We note that the dimension of each $X\in\At(\X)$ is equal to the maximum cardinality of any set of pairwise orthogonal projections in $L(X)$. Then the maximum cardinality of any set of pairwise orthogonal projections in $\ell^\infty(\X)=\bigoplus_{X\in\At(\X)}L(X)$ is $\sum_{X\in\At(\X)}\dim(X)$, which equals the cardinality of $\At(\X)$ because every atom is finite-dimensional. Now, since injective $*$-homomorphisms map sets of pairwise orthogonal projections to sets or pairwise orthogonal projections, the second statement follows from the injectivity of $\overline{F}^\star$. 
\end{proof}

We next generalize the $D$-closure of a subset to the quantum setting.
Since we lack topological tools, we first show we can prove statements about the $D$-closure by transfinite induction. We say two functions $F:\W\to\Y$ and $G:\V\to\Y$ with the same codomain are \emph{equivalent} if $F=\V\circ K$ for some bijection $K:\W\to\V$; it is easy to see that equivalent functions have the same range. The fact that up to isomorphism there are only countably many atomic quantum sets then assures that $\X_{\alpha+1}$ in the following lemma is well defined. 

\begin{lemma}\label{lem:generated subqcpo}
Let $\Y$ be a quantum cpo, and let $\X\subseteq\Y$.
Define $\X_0=\X$, and for an ordinal $\alpha$, define $\X_{\alpha+1} = \bigcup_{K_\infty\in {\mathcal K}} \ran K_\infty$ where $\mathcal K$ is the family of all  $K_\infty:\H\to\Y$ satisfying there is an atomic quantum set $\H$ and a sequence $K_n\nearrow K_\infty\:\H\to\Y$ with $\ran K_n\subseteq\X_{\alpha}$ for each $n\in\NN$.
For each limit ordinal $\lambda$ define
\[\X_\lambda=\bigcup_{\alpha<\lambda}\X_\alpha.\]
Then 
\begin{itemize}
    \item[(a)] $\X_\alpha\subseteq\X_\beta$ for ordinals $\alpha$ and $\beta$ such that $\alpha<\beta$;
    \item[(b)] $\overline\X=\X_{\omega_1}$.
    \end{itemize}
\end{lemma}

\begin{proof}
Firstly, for any successor ordinal $\alpha+1$, we have $\X_\alpha\subseteq\X_{\alpha+1}$: let $X\atomof\X_\alpha$. Then $\H=\Q\{X\}$ is atomic. Consider the constant sequence $K_n:\H\to\Y$ given by $K_n=J_\X$ for $n\in\NN$. Then $\ran K_n=\Q\{X\}\subseteq\X_\alpha$, and $K_n\nearrow J_X$. Hence $\Q\{X\}=\ran J_X\subseteq\X_{\alpha+1}$. Given a limit ordinal $\lambda$, we clearly have $\X_\alpha\subseteq\X_\lambda$ for each $\alpha<\lambda$, whence
for any two ordinals $\alpha$ and $\beta$, we have $\alpha<\beta$ implies $\X_\alpha\subseteq\X_\beta$.

Since $\X_0=\X$, it follows that $\X\subseteq\X_{\omega_1}$. We show that $\X_{\omega_1}$ is a sub-cpo of $\Y$. So let $\H$ be an atomic quantum set and let $K_1\sqsubseteq K_2\sqsubseteq\cdots:\H\to\X_{\omega_1}$ be a monotonically ascending sequence of functions, and note that $\ran K_n\subseteq\X_{\omega_1}=\bigcup_{\alpha<\omega_1}\X_\alpha$. Since $\Y$ is a cpo, the sequence $K_n$ has a limit $K_\infty$. We show $\ran K_\infty \subseteq \X_{\omega_1}$.

Fix $n\in\NN$ and let $X\atomof\ran K_n$. Then $X\atomof\X_{\beta_X}$ for some $\beta_X<\omega_1$. Because $\H$ consists of a single atom, the range of $K_n$ consists of finitely many atoms (Proposition \ref{prop:range is finite}).
It follows that $\alpha_n=\max\{\beta_X:X\atomof\ran K_n\}$ exists and is smaller than $\omega_1$. Thus $\ran K_n\subseteq\X_{\alpha_n}$. Since $\alpha_n<\omega_1$, $\alpha_n$ is countable, and then $\alpha:=\sup_{n\in\NN}\alpha_n=\bigcup_{n\in\NN}\alpha_n$ is countable, so $\alpha<\omega_1$. Hence, $\ran K_n\subseteq\X_\alpha$ for each $n\in\NN$, so $\ran K_\infty\subseteq\X_{\alpha+1}$, and since also $\alpha+1<\omega_1$, it follows that $\ran K_\infty\subseteq\X_{\omega_1}$. Hence $\X_{\omega_1}$ is indeed a sub-cpo of $\Y$, and since it contains $\X$, it follows that $\overline\X\subseteq\X_{\omega_1}$.

For the reverse inclusion, $\X_0 = \X\subseteq\overline\X$ by Lemma~\ref{lem:sub-cpo generated by subset}. We claim that for an ordinal $\alpha$, if $\X_\alpha\subseteq\overline\X$ then $\X_{\alpha+1}\subseteq\overline\X$. Indeed, assume $\X_\alpha\subseteq\overline\X$. Let $\H$ be an atomic quantum set, and let $K_1\sqsubseteq K_2\sqsubseteq\cdots:\H\to\Y$ be a sequence with limit $K_\infty$ and with $\ran K_n\subseteq\X_\alpha$ for each $n\in\NN$. Then $\ran K_n\subseteq\overline\X$ for each $n\in\NN$, and since $\bar X$ is a sub-cpo of $\Y$, it follows that $\ran K_\infty\subseteq\overline\X$, whence $\X_{\alpha+1}\subseteq\overline\X$. 

If $\lambda$ is a limit ordinal such that $\X_\alpha\subseteq\overline\X$ for each $\alpha<\lambda$, it follows by definition of $\X_\lambda$ that $\X_\lambda\subseteq\overline\X$. Thus $\X_{\omega_1}\subseteq\overline\X$ by ordinal induction. 
\end{proof}

Our next goal is to show that any Scott continuous function $F$ between quantum cpos can be written as $F=M\circ E$ for some monomorphism $M$ and some epimorphism $E$. We already characterized monomorphisms, hence we proceed with investigating epimorphisms.

\begin{lemma}\label{lem:dense range is epi}
Let $\X$ be a quantum poset, let $\Y$ be a quantum cpo, and let $E:\X\to\Y$ be Scott continuous. If  $\overline{\ran E}=\Y$, then $E$ is an epimorphism in $\qPOS_{\mathrm{SC}}$. If $\X$ is a quantum cpo, then $E$ is also an epimorphism in $\qCPO$.
\end{lemma}
\begin{proof}
First assume that $\overline{\mathrm{ran}(E)}=\Y$. Let $(\Z,T)$ be a quantum poset and let $F,G:\Y\to\Z$ be Scott continuous functions such that $F\circ E=G\circ E$.

Write $\W=\mathrm{ran}(E)$. We claim that $F$ and $G$ coincide on $\W$, i.e., $F\circ J_\W=G\circ J_\W$, where $J_\W:\W\to\Y$ is the inclusion. Indeed, by Lemma \ref{lem:factors through a subset}, there is a surjective function $K:\X\to\W$ such that $E=J_\W\circ K$.
Then 
\[F\circ J_\W=F\circ J_\W\circ K\circ K^\dag=F\circ E\circ K^\dag=G\circ E\circ K^\dag=G\circ J_\W\circ K\circ K^\dag=G\circ J_\W,\]
where $K\circ K^\dag=I_\W$ by surjectivity of $K$.

We recall Lemma \ref{lem:generated subqcpo} and follow the same notation. Hence let $\W_0=\W$ and  define $\W_\lambda=\bigcup_{\alpha<\lambda}\W_\alpha$ for each limit ordinal $\lambda$. Finally, 
for any ordinal $\alpha$, define
$\W_{\alpha+1}= \bigcup_{K_\infty\in {\mathcal K}} \ran K_\infty$ where $\mathcal K$ is the family of all $K_\infty:\H\to \Y$, with $\H$ an atomic quantum set, for which there is a sequence $K_1\sqsubseteq K_2\sqsubseteq\cdots:\H\to\Y$ with $K_n\nearrow K_\infty$, and $\ran K_n\subseteq\W_{\alpha}$ for each $n\in\NN$. 

Let $\beta$ be an ordinal, and assume that $F\circ J_{\W_\alpha}=G\circ J_{\W_\alpha}$ for each $\alpha<\beta$. We aim to show that also $F\circ J_{\W_\beta}=G\circ J_{\W_\beta}$. First assume that $\beta$ is a successor ordinal, so $\beta=\alpha+1$. 
Let $Y\atomof\W_\beta$ and $Z\atomof\Z$. Then $Y\atomof\ran K_\infty$ for some atomic quantum set $\H$ and some function $K_\infty:\H\to\Y$ for which there is a sequence $K_1\sqsubseteq K_2\sqsubseteq\cdots:\H\to\Y$ with $K_n\nearrow K_\infty$ and $\ran K_n\subseteq\W_\alpha$.
The last inclusion implies there is a function $\tilde K_n:\H\to\W_\alpha$ such that $K_n=J_{W_\alpha}\circ \tilde K_n$ (Lemma \ref{lem:factors through a subset}). Let $\V=\ran K_\infty$. Then there is a surjective function $\overline K_\infty:\H\to\V$ such that $K_\infty=J_{\V}\circ \overline K_\infty$. 

Since $K_n\nearrow K_\infty$, and $F$ and $G$ are Scott continuous continuous, we have
\begin{align*}
F\circ K_n\nearrow F\circ K_\infty & =F\circ J_{\V}\circ \overline K_\infty,\quad \mathrm{and}\\
G\circ K_n\nearrow G\circ K_\infty & =G\circ J_{\V}\circ \overline K_\infty.
\end{align*}
Since $F\circ J_{\W_\alpha}=G\circ J_{\W_\alpha}$ by assumption, we obtain
$F\circ K_n=F\circ J_{\W_\alpha}\circ \tilde K_n=G\circ J_{\W_\alpha}\circ \tilde K_n=G\circ K_n,$
and since limits are unique, we obtain
$F\circ J_{\V}\circ\overline K_\infty =G\circ J_{\V}\circ\overline K_\infty,$
and the surjectivity of $\overline K_\infty$ then implies
\[F\circ J_\V=F\circ J_\V\circ \overline K_\infty\circ\overline K_\infty^\dag=G\circ J_\V\circ \overline K_\infty\circ\overline K_\infty^\dag=G\circ J_\V.\]
Since $\V=\ran K_\infty\subseteq\W_\beta$, and $Y\atomof\V$, we obtain
\[ (F\circ J_{\W_\beta})(Y,Z)=F(Y,Z)=(F\circ J_\V)(Y,Z)=(G\circ J_\V)(Y,Z)=G(Y,Z)=(G\circ J_{\W_\beta})(Y,Z),\]
and since $Y\atomof\W_\beta$ and $Z\atomof\Z$ are arbitrary, it follows that $F\circ J_{\W_\beta}=G\circ J_{\W_\beta}$.

Now assume that $\beta$ is a limit ordinal. Let $Y\atomof\W_\beta$ and $Z\atomof\Z$. Then $Y\atomof\W_\alpha$ for some $\alpha<\beta$, whence
\[(F\circ J_{\W_\beta})(Y,Z)=F(Y,Z)=(F\circ J_{\W_\alpha})(Y,Z)=(G\circ J_{\W_\alpha})(Y,Z)=G(Y,Z)=(G\circ J_{\W_\beta})(Y,Z),\]
and since $Y\atomof\W_\beta$ and $Z\atomof\Z$ are arbitrary, it follows that $F\circ J_{\W_\beta}=G\circ J_{\W_\beta}$.

Now by Lemma \ref{lem:generated subqcpo}, we have $\bar\W=\W_{\omega_1}$, so we have $F\circ J_{\W_{\omega_1}}=G\circ J_{\W_{\omega_1}}$. By assumption $\Y=\bar\W$, so $F\circ J_\Y=G\circ J_\Y$, but since $J_\Y:\Y\to\Y$ is the identity $I_\Y$ on $\Y$, we obtain $F=G$. So $E$ is indeed an epimorphism.
\end{proof}

\begin{proposition}\label{prop:epi-mono in qCPO} 
Let $(\X,R)$ be a quantum poset, let $(\Y,S)$ be a quantum cpo, and let $F:\X\to\Y$ be Scott continuous. Let $\Z\subseteq\Y$ be a subset such that $\ran F\subseteq\Z$ and equip $\Z$ with the relative order $J_\Z^\dag\circ S\circ J_\Z$ (cf. \cite[Definition~2.2]{KLM20}). Then:
\begin{itemize}
    \item[(a)] The corestriction $F|^\Z:\X\to\Z$ of $F$ is a Scott continuous function;
    \item[(b)] If $\Z=\overline{\ran F}$, then $\overline{\ran F|^\Z}=\Z$.
\end{itemize}
\end{proposition}
Note that $\Z$ does not need to be a sub-cpo of $\Y$, even if $\Z=\ran F$.
\begin{proof}
 Since $\ran F\subseteq\Z$ it follows from Lemma \ref{lem:factors through a subset} that $F|^\Z$ is a function such that $J_\Z\circ F|^\W=F$, so we proceed with showing that it is Scott continuous. So for an atomic quantum set $\H$, let  $K_1\sqsubseteq K_2\sqsubseteq\cdots\sqsubseteq K_\infty\:\H\to\X$ be functions satisfying $K_n\nearrow K_\infty$. Then
\begin{align*}
    \bigwedge_{n\in\NN}J_\Z^\dag\circ S\circ J_\Z\circ F|^\Z\circ K_n & =\bigwedge_{n\in\NN}J_\Z^\dag\circ S\circ F\circ K_n=J_\Z^\dag\circ \bigwedge_{n\in\NN}S\circ F\circ K_n\\
    & =J_\Z^\dag\circ S\circ F\circ K_\infty=J_\Z^\dag\circ S\circ J_\Z\circ F|^\Z\circ K_\infty, 
\end{align*}
where we used \cite[Proposition~A.6]{KLM20} in the second equality, and Scott continuity of $F$ in the penultimate equality. We conclude that $F|^\Z$ is indeed Scott continuous. 

For (b) assume that $\Z=\overline{\ran F}$. Then $\ran F\subseteq\Z$, so it follows from Lemma \ref{lem:factors through a subset} that $\ran F=\ran F|^\Z$. Hence, $\overline{\ran F^\Z}=\overline{\ran F}=\Z$.
\end{proof}

\begin{corollary}(Epi-Mono Factorization in $\qCPO$)\label{cor:epi-mono factorization in qCPO}
Let $F:\X\to\Y$ be a Scott continuous function from a quantum poset $\X$ to a quantum cpo $\Y$. Let $\Z=\overline{\ran F}$. Then $\Z$ is a sub-cpo of $\Y$, $J_\Z$ is a monomorphism in $\qPOS_\mathrm{SC}$, and the corestriction $E:=F|^\Z$ is an epimorphism in $\qPOS_\mathrm{SC}$ such that $\overline{\ran E}=\Z$. If $\X$ is a quantum cpo, then $J_\Z$ is a monomorphism in $\qCPO$ and $E$ is an epimorphism in $\qCPO$.
\end{corollary}
\begin{proof}
By construction $\Z$ is a sub-cpo of $\Y$, whence $J_\Z$ is Scott continuous by Proposition \ref{prop:subqcpo}. Since $J_\Z$ is injective, it is a monomorphism in $\qPOS_\mathrm{SC}$ by Proposition \ref{prop:monos in qCPO}. Note that if $\X$ is a quantum cpo, the same proposition assures that $J_\Z$ is a monomorphism in $\qCPO$. Now, Proposition \ref{prop:epi-mono in qCPO} implies that $E$ is Scott continuous function. It now follows from Lemma \ref{lem:dense range is epi} that $E$ is an epimorphism in $\qPOS_\mathrm{SC}$. The same lemma assures that $E$ is an epimorphism in $\qCPO$ if $\X$ is a quantum cpo. 
\end{proof}

\begin{lemma}\label{lem:card lemma 1}
Let $\X$ and $\Y$ be finite quantum sets. Then $\mathrm{card}(\qSet(\X,\Y))\leq 2^{\aleph_0}$.
\end{lemma}
\begin{proof}
We have a bijection between $\qSet(\X,\Y)$ and $\mathbf{WStar}_\mathrm{HA}(\ell^\infty(\Y),\ell^\infty(\X))$ (cf. Theorem \ref{thm:qSet}), where the latter is a subspace of $L(\ell^\infty(\Y),\ell^\infty(\X))$. Since both $\X$ and $\Y$ are finite, it follows that $\ell^\infty(\X)$ and $\ell^\infty(\Y)$, and hence also $L(\ell^\infty(\Y),\ell^\infty(\X))$ are finite-dimensional. Any finite-dimensional space over $\CC$ is isomorphic as a vector space to the product of a finite number of copies of $\CC$, whose cardinality is $2^{\aleph_0}$. Hence the cardinality $\qSet(\X,\Y)$ can be at most $2^{\aleph_0}$, too.
\end{proof}

\begin{lemma}\label{lem:card lemma 2}
Let $\X$ be a finite quantum set, and let $\Y$ be an arbitrary quantum set. Then $\mathrm{card}(\qSet(\X,\Y))\leq 2^{\aleph_0}\cdot\mathrm{card}(\At(\Y)).$
\end{lemma}
\begin{proof}
Because $\X$ is finite, every function $F:\X\to\Y$ has a finite range (Proposition \ref{prop:range is finite}). The number of finite subsets of $\Y$ is $\aleph_0\cdot\mathrm{card}(\At(\Y))$. Multiplying with $2^{\aleph_0}$ obtained from Lemma \ref{lem:card lemma 1} gives us the cardinality of $\qSet(\X,\Y)$. Finally, we use that $\aleph_0\cdot 2^{\aleph_0}=2^{\aleph_0}$. 
\end{proof}

\begin{lemma}\label{lem:upper bound cardinality of sub-cpo generated by subset}
Let $\Y$ be a quantum cpo, and let $\X\subseteq\Y$ be a quantum subset. Then
\[\mathrm{card}(\At(\overline\X))\leq(\mathrm{card}(\At(\X)))^{\aleph_0}.\]
\end{lemma}
\begin{proof}
If $\X$ is finite, it is already a sub-cpo by Proposition \ref{prop:finite subset is subqcpo}, in which case $\overline\X=\X$, for which the statement clearly holds. 

Assume that $\X$ is infinite.
By Lemma \ref{lem:generated subqcpo} it is sufficient to show that
\begin{equation}\label{ineq:IH generated subqcpo}\mathrm{card}(\X_\alpha)\leq (\mathrm{card}(\At(\X)))^{\aleph_0}
\end{equation}
holds for each ordinal $\alpha\leq\omega_1$. Note that the same lemma assures that each $\X_\alpha$ is infinite. Clearly (\ref{ineq:IH generated subqcpo}) holds for $\alpha=0$.
Given an ordinal $\beta\leq \omega_1$ assume as an induction hypothesis that (\ref{ineq:IH generated subqcpo}) holds for each ordinal $\alpha<\beta$.
Assume first that $\beta$ is a succesor ordinal, so $\beta=\alpha+1$ for some ordinal $\alpha$. 
For each $d\in\NN$ fix an atomic quantum set $\H_d$ with single atom $H_d$ of dimension $d$ and let $S_d$ be the set of all sequences $K_1\sqsubseteq K_2\sqsubseteq\cdots:\H_d\to\Y$ such that $\ran K_n\subseteq X_\alpha$ for each $n\in\NN$. Given a function $K:\H_d\to\Y$ with $\ran K\subseteq X_\alpha$, we know that $K(H_d,Y)=0$ if $Y\atomof\Y$ is not an atom of $\X_\alpha$, whereas for each $X\atomof\X_\alpha\subseteq\Y$, $K(H_d,X)$ is a subspace of the finite-dimensional space $L(H_d,X)$, hence there are only finitely many possible choices for $K(H_d,X)$. Since $K$ is completely determined by  $K(H_d,X)$ for each $X\atomof\X_\alpha$, there at most $\mathrm{card}(\At(\X_\alpha))$ many functions $K:\H_d\to\X_\alpha$, so there are at most $2^{\aleph_0}\cdot \mathrm{card}(\At(\X_\alpha))$ elements in $S_d$. This means that 
\[\mathrm{card}\left(\bigcup_{d\in\NN}S_d\right)\leq \aleph_0\cdot 2^{\aleph_0}\mathrm{card}(\At(\X_\alpha))\leq \aleph_0\cdot 2^{\aleph_0}\cdot\mathrm{card}(\At(\X)^{\aleph_0}=\mathrm{card}(\At(\X))^{\aleph_0},\]
where we used the induction hypothesis in the last inequality. Let $X$ be an atom of $\X_\beta$. Then there is some atomic quantum set $\H$ and a monotonically ascending sequence of functions $K_1\sqsubseteq K_2\sqsubseteq\cdots:\H\to\Y$ with with limit $K_\infty$ such that $\ran K_n\subseteq\X_\alpha$ for each $n\in\NN$ and $X\atomof\ran K_\infty$. 
Let $d$ be the dimension of the atom $H$ of $\H$. Then there is a bijection $F:\H_d\to\H$, and it follows from Lemma \ref{lem:right multiplication is Scott continuous} that $(K_n\circ F)_{n\in\NN}$ is a monotonically ascending sequence in $\qSet(\H_d,\X_\alpha)$ with limit $K_\infty\circ F$. Thus $(K_n\circ F)_{n\in\NN}$ is an element of $S_d$, and clearly $\ran K_\infty=\ran K_\infty\circ F$, which is finite by Proposition \ref{prop:range is finite}.
Hence $\mathrm{card}(\X_\beta)\leq \aleph_0\cdot\mathrm{card}\left(\bigcup_{d\in\NN} S_d\right)$, from which it follows that (\ref{ineq:IH generated subqcpo}) also holds for $\beta$.

Now assume $\beta$ is a limit ordinal. Then (\ref{ineq:IH generated subqcpo}) hold for each $\alpha<\beta$, hence
\[\mathrm{card}(\At(\X_\beta))=\mathrm{card}\left(\bigcup_{\alpha<\beta}\X_\alpha\right)\leq2^{\aleph_0}\cdot\mathrm{card}(\At(\X_\alpha))\leq 2^{ \aleph_0}\cdot\mathrm{card}(\At(\X))^{\aleph_0}=\mathrm{card}(\At(\X))^{\aleph_0},\]
where we used that $\mathrm{card}(\omega_1)=\aleph_1\leq 2^{\aleph_0}$ in the first inequality (equality if the Continuum Hypothesis is true); the second inequality follows from the induction hypothesis. Hence (\ref{ineq:IH generated subqcpo}) holds for any ordinal $\alpha\leq\omega_1$.
\end{proof}

Let $G:\mathbf D\to\mathbf C$ be a functor. Given an object $X$ in $\mathbf C$, a pair $(f,Y)$ consisting of an object $Y\in\mathbf D$ and a morphism $f:X\to GY$ in $\mathbf C$ is called a \emph{G-structure with domain} $X$ \cite[Definition 8.30]{adameketall:joyofcats}. A $G$ structure $(f,Y)$ is \emph{generating} if for any parallel morphisms $r,s:Y\to Y'$ in $\mathbf D$, the equality $Gr\circ f=Gs\circ f$ implies $r=s$. A generating $G$-structure $(f,Y)$ with domain $X$ is  \emph{extremally generating} if every monomorphism $m:Y'\to Y$ in $\mathbf D$ and $G$-structure $(g,Y')$ with $f=Gm\circ g$ satisfy $m$ is an isomorphism in $\mathbf D$. Two $G$-structures $(f,Y)$ and $(f',Y')$ with the same domain $X$ are called \emph{isomorphic} if there is an isomorphism $k:Y\to Y'$ in $\mathbf D$ such that $Gk\circ f=f'$ \cite[Definition 8.34]{adameketall:joyofcats}. In this case, we write $(f,Y)\sim (f',Y')$. The relation $\sim$ on the class of $G$-structures with domain $X$ is an equivalence relation. We call $G$ \emph{(extremally) co-wellpowered} if for each object $X$ of $\mathbf C$, the class of all equivalence classes under $\sim$ of (extremally) generating $G$-structures with domain $X$ is a set \cite[Definition 8.37]{adameketall:joyofcats}.

Let $\mathbf D$ be a category. We recall that two epimorphisms $e:X\to Y$ and $e':X\to Y'$ in $\mathbf D$ are called \emph{isomophic} if there is an isomorphism $k:Y\to Y'$ such that $e'=k\circ e$, in which case we write $e\sim e'$. The relation $\sim$ is an equivalence relation on the class of all epimorphisms with domain $X$. Moreover, we recall that an epimorphism $e:X\to Y$ is called \emph{extremal} if for each monomorphism $m:Y'\to Y$ such that $e=m\circ g$ for some $g:X\to Y'$, it follows that $m$ is an isomorphism. We call $\mathbf D$ \emph{extremally co-wellpowered} if for each object $X\in\mathbf D$ the class of all equivalence classes under $\sim$ of epimorphisms with domain $X$ is a set.

Let $G:\mathbf D\to\mathbf C$ be a faithful extremally co-wellpowered functor. Then $\mathbf D$ is not necessarily extremally co-wellpowered \cite[Remark 8.39]{adameketall:joyofcats}. However, if $\mathbf D$ is a subcategory of $\mathbf C$, and $G$ is the inclusion, then for each object $X$ of $\mathbf D$, we clearly have that a morphism $f:X\to Y$ is an extremal epimorphism if and only if $(f,Y)$ is an extremally generating $G$-structure, and $f\sim f'$ if and only if $(f,Y)\sim(f',Y)$ for any other epimorphism $f':X\to Y'$.  Hence, we obtain:
\begin{lemma}\label{lem:ext-co-wellpowered-inclusion-implies-ext-co-wellpowered-domain}
    Let $\mathbf C$ be a category, and let $\mathbf D$ be a subcategory of $\mathbf C$. If the inclusion $G:\mathbf D\to\mathbf C$ is extremally co-wellpowered, so is $\mathbf D$.
\end{lemma}
The importance of these concepts lie in the following two theorems that we will use:
\begin{theorem}\cite[Proposition 12.5 \& Theorem 18.14]{adameketall:joyofcats}\label{thm:adjoint functor theorem}
Let $\mathbf D$ be a complete, wellpowered and extremally co-wellpowered category. Then a functor $G:\mathbf D\to\mathbf C$ has a left adjoint if and only if it is extremally co-wellpowered.
\end{theorem}

\begin{theorem} \cite[Theorem 5.11]{NAKAGAWA1989563}\label{thm:coequalizers}
A complete, wellpowered and extremally co-wellpowered category  has all coequalizers.
\end{theorem}

\begin{proposition}\label{prop:extremal epis have d-dense image}
Let $\I:\mathbf{qCPO}\to\mathbf{qPOS}_\mathrm{SC}$ be the inclusion. Let $E:(\X,R)\to\I(\Y,S)$ be an extremally generating $\I$-structure. Then $\overline{\ran E}=\Y$.
\end{proposition}
\begin{proof}
By definition of an $\I$-structure, $(\Y,S)$ is a quantum cpo, and $E:(\X,R)\to (\Y,S)$ is Scott continuous. Let $\Z:=\overline{\ran E}$. By definition, $\Z$ is a sub-cpo of $\Y$, hence it follows from Proposition \ref{prop:subqcpo} that $J_\Z$ is Scott continuous. Since $J_\Z$ is injective, it follows from Proposition \ref{prop:monos in qCPO} that it is a monomorphism in $\qCPO$. Since $\ran E\subseteq\Z$, Proposition \ref{prop:epi-mono in qCPO} assures the existence of a Scott continuous function $G:\X\to\Z$ such that $E=J_\Z\circ G$. Since $\Z$ is a quantum cpo, we have that $G$ is actually a morphism $\X\to \I\Z$, so $G$ is an $\I$-structure. Thus, $E=J_\Z\circ G=\I(J_\Z)\circ G$, and since $E$ is an extremally generating $\I$-structure, it follows that $J_\Z$ is an isomorphism in $\mathbf{qCPO}$. In particular, $J_\Z$ must be a bijection, which forces $\Z=\Y$. 
\end{proof}

\begin{proposition}
\label{prop:ext-cowellpowered}
Let $(\X,R)$ be a quantum poset, and let $\fF$ be the class of all Scott continuous functions $F:\X\to\Y_F$ where $\Y_F$ is a quantum cpo such that $\overline{\ran F}=\Y_F$. Define an equivalence relation $\sim$ on $\fF$ by $F\sim F'$ if and only if there is an isomorphism $G:\Y_F\to\Y_{F'}$ in $\qCPO$ such that $F'=G\circ F$. Then $\fF/\sim$ is a set.
\end{proposition}
\begin{proof}
Let $F\sim F'$ in $\fF$. Since order isomorphisms are the isomorphisms in $\qCPO$ (see also Lemma \ref{lem:order iso is Scott continuous}), there is an order isomorphism $G:\Y_F\to \Y_{F'}$, which in particular implies the existence of a bijection $f:\At(\Y_F)\to\At(\Y_{F'})$ such that $\dim f(Y)=\dim Y$ for each $Y\atomof\Y_F$ \cite[Proposition 4.4]{Kornell18}. It follows that the equivalence class of $F\in\fF$ is determined by the cardinality of $\At(\Y_F)$, the dimensions of the atoms of $\Y_F$, and the possible orders on $\Y_F$ relative to which it is a quantum cpo with $\overline{\ran F}=\Y_F$. The statement follows if we can show that the  cardinality of $\At(\Y_F)$ is bounded, because we have only a countable choice of possible dimensions for each atom of $\Y_F$, and any suitable order on $\Y_F$ is an element of $\qRel(\Y_F,\Y_F)$, which is a set since $\qRel$ is locally small.

By Lemma \ref{lem:upper bound cardinality of sub-cpo generated by subset} we have
\[\mathrm{card}(\At(\Y_F))=\mathrm{card}(\At(\overline{\ran F}))\leq \mathrm{card}(\At(\ran F))^{\aleph_0}.\]
Given a surjective function $G:\X\to\W$, we claim that 
$\mathrm{card}(\At(\W))\leq\aleph_0\cdot\mathrm{card}(\At(\X))$. Indeed, since $\X\cong\coprod_{X\atomof\X}\Q\{X\},$ with canonical injections $J_X:\Q\{X\}\to\X$, we have functions $G_X:\Q\{X\}\to\W$ such that $[G_X:X\atomof\X]=G$, i.e., $G\circ J_X=G_X$ for each $X\atomof\X$. By Proposition \ref{prop:range is finite}, the range of each $G_X$ is finite, hence $\mathrm{card}(\ran G_X)\leq\aleph_0$. Then
\begin{align*} \ran G & =\ran[G_X:X\atomof\X]=\{W\atomof\W:[G_X](X',W)\neq 0\text{ for some }X'\atomof\X\}\\
& = \{W\atomof\W:G_X(X,W)\neq 0\text{ for some }X\atomof\X\}=\bigcup_{X\atomof\X}\ran G_\X,
\end{align*}
hence $\mathrm{card}(\W)=\mathrm{card}(\ran G)=\aleph_0\cdot\mathrm{card}(\At(\X)).$
We factor $F$ through its range, $F = J_F \circ \overline F$, and take $\W = \ran F$. Hence, $\mathrm{card}(\At(\ran F))\leq\aleph_0\cdot\mathrm{card}(\At(\X))$, whence
\[\mathrm{card}(\At(\Y_F))\leq\aleph_0\cdot \mathrm{card}(\At(\X))^{\aleph_0}.\]
 We conclude that $\fF/\sim$ indeed is a set.
 \end{proof}

\begin{corollary}
\label{cor:ext-cowellpowered}\label{cor:qCPO is extremal co-wellpowered}
Both $\mathbf{qCPO}$ and the inclusion $\I:\mathbf{qCPO}\to\mathbf{qPOS}_\mathrm{SC}$ are extremally co-wellpowered.
\end{corollary}
\begin{proof}
    Let $\X$ be a quantum poset, and let $\gG$ be the class of all extremally generating $\I$-structures with domain $\X$. If $(F,\Y)$ is an element of $\gG$, then $\Y$ is a quantum cpo, and $F:\X\to \Y$ is Scott continuous. By Proposition \ref{prop:extremal epis have d-dense image}, we have $\overline{\ran F}=\Y$. It follows that $\gG$ is a subclass of the class $\fF$ of  Proposition \ref{prop:ext-cowellpowered}. Since this proposition asserts that $\fF/\sim$ is a set, so must be $\gG/\sim$ from which it follows that $\I$ is extremally co-wellpowered. It now follows from Lemma \ref{lem:ext-co-wellpowered-inclusion-implies-ext-co-wellpowered-domain} that $\qCPO$ is extremally co-wellpowered, too.
\end{proof}

In \cite[Definition 1]{Zhao-Fan}, the $D$-completion of a poset $X$ is defined as a dcpo $D(X)$ together with a Scott continuous map $\eta:X\to D(X)$ such that for any other dcpo $Y$ and Scott continuous map $f:X\to Y$ there is a unique Scott continuous map $k:D(X)\to Y$ such that $f=\hat f\circ\eta$. The following theorem asserts we have a similar completion for quantum posets.

\begin{theorem}\label{thm:d-completion}
The inclusion $\I:\qCPO\to\mathbf{qPOS}_\mathrm{SC}$ has a left adjoint $\D$. In particular, if $\X$ is a quantum poset, then there exists a quantum cpo $\D(\X)$ and a Scott continuous function $E_\X:\X\to\D(\X)$ satistying the following universal property: for each quantum cpo $\Z$ and each Scott continuous map $G:\X\to \Z$, there is a unique Scott continuous map $K:\D(\X)\to\Z$ such that the following diagram commutes:
\[ \begin{tikzcd}
\X\ar{r}{E_\X}\ar{dr}[swap]{G} & \D(\X)\ar{d}{K}\\
& \Z
\end{tikzcd} \]
\end{theorem}
\begin{proof}
The category $\qCPO$ is complete by Theorem \ref{thm:qCPO has all limits}, wellpowered by Proposition \ref{prop:qCPO is wellpowered}, and extremally co-wellpowered by Corollary \ref{cor:qCPO is extremal co-wellpowered}. Since the same corollary assures that the inclusion $\I:\qCPO\to\mathbf{qPOS}_\mathrm{SC}$ is extremally co-wellpowered, it follows from  Theorem  \ref{thm:adjoint functor theorem} that $\I$ has a left adjoint $\D$. The universal property of the function $E_\X$ is expresses the universal property of the unit $E$ of the adjunction $\D\dashv \I$. 

Let $J:\overline{\ran E_\X}\to\D(\X)$ be the embedding. Since ${\mathcal D}(\X)$ is a quantum cpo, Proposition \ref{prop:epi-mono in qCPO} assures that the corestriction $G:\X\to\overline{\ran E_\X}$ of $E_\X:\X\to\D(\X)$, i.e., the unique function such that $E_\X=J\circ G$, exists and is Scott continuous. Then from the universal property of $E_\X$, it follows that there exists a function $K:\D(\X)\to\overline{\ran E_\X}$ such that $K\circ E_\X=G$. Then $E_\X=J\circ G=J\circ K\circ E_\X$, and again using the universal property of $E_\X$ yields $J\circ K=I_{\D(\X)}$, i.e., $J$ is a split epimorphism in $\qSet$. By \cite[Proposition 8.1]{Kornell18}, epimorphisms in $\qSet$ are precisely the surjections, which forces $\overline{\mathrm{ran}~E_\X}=\D(\X)$.
\end{proof}

\begin{theorem}\label{thm:qCPO is cocomplete}
$\qCPO$ is cocomplete. 
\end{theorem}
\begin{proof}
By Theorem \ref{thm:qCPO has all limits}, Proposition \ref{prop:qCPO is wellpowered} and Corollary \ref{cor:qCPO is extremal co-wellpowered}, $\qCPO$ is complete, wellpowered, and extremally co-wellpowered, hence it has all coequalizers by Theorem \ref{thm:coequalizers}. By Theorem \ref{thm:qCPO has all coproducts}, $\qCPO$ has also all coproducts. It is a standard result that any category with all products and all equalizers is complete, see for instance \cite[Theorem 2.8.1]{borceux:handbook1}, hence by duality it follows that $\qCPO$ has all colimits.
\end{proof}

\section{Monoidal structure on $\cat{qCPO}$}\label{sec:monoidal}
In this section, we prove that $\qCPO$ is symmetric monoidal closed. Monoidal products are relevant to denotational semantics as they offer support for pair types, whereas the internal hom arising from the monoidal closed structure can be used to model function types.

The monoidal product of two quantum cpos $(\X_1, R_1)$ and $(\X_2, R_2)$ is simply their monoidal product as posets: $(\X_1 \times \X_2, R_1 \times R_2)$. The monoidal structure on $\cat{qSet}$ is closed, yielding quantum functions sets $\Y^\X$ for all quantum sets $\X$ and $\Y$; see also Theorem \ref{thm:qSet}. The monoidal structure on $\cat{qCPO}$ is similarly closed, with each inner hom object $[\X,\Y]_\uparrow$ equal to a subset of $\Y^\X$, appropriately ordered. Specifically, $[\X,\Y]_\uparrow$ is the largest subset of $\Y^\X$ on which the evaluation function $\Y^\X \times \X \to \Y$ is Scott continuous in the second variable, in a sense that is still to be defined. This subset $[\X,\Y]_\uparrow$ is also a subset of the quantum poset $[\X,\Y]_\below$, and it is ordered accordingly.

\subsection{Monoidal product.} We show that the monoidal product of any two quantum cpos $(\X,R)$ and $(\Y,S)$ is itself a quantum cpo. Let $P\: \X \times \Y \to \X$ and $Q\: \X \times \Y \to \Y$ be the canonical projection functions \cite{Kornell18}*{Section 10}. Given a monotonically ascending sequence $K_1 \below K_2 \below \cdots$ of functions from an atomic quantum set $\H$ to a monoidal product of quantum cpos $\X \times \Y$, it is not so difficult to show that $P \circ K_n \nearrow F_\infty$, for some function $F_\infty\: \H \to \X$, and symmetrically, that $Q \circ K_n \nearrow G_\infty$, for some function $G_\infty\: \H \to \Y$. The main technical challenge is to show that $F_\infty$ and $G_\infty$ are compatible in the sense of \cite{Kornell18}*{Definition 10.3}, i.e., to show that $F_\infty$ and $G_\infty$ are the components of a function $K_\infty\: \H \to \X \times \Y$.

\begin{definition}\label{def:compatible-functions}
Two functions $F\: \H \to \X$ and $G\: \H \to \Y$ are said to be \emph{compatible} if there exists a function $(F,G)\: \H \to \X \times \Y$ that makes the following diagram commute:
$$
\begin{tikzcd}
&&
\H \arrow{lld}[swap]{F} \arrow{rrd}{G}
\arrow[dotted]{d}[swap]{(F,G)}
&&
\\
\X
&&
\X \times \Y \arrow{ll}{P} \arrow{rr}[swap]{Q}
&&
\Y
\end{tikzcd}
$$
The duality between quantum sets and hereditarily atomic von Neumann algebras makes it clear that if such a function $(F, G)$ exists, then it is unique (cf. Theorem \ref{thm:qSet}). Furthermore, if $F'$ is a function from $\X$ to a quantum set $\X'$ and $G$ is a function from $\Y$ to a quantum set $\Y'$, then $(F' \times G') \circ (F, G) = (F' \circ F, G' \circ G)$.
\end{definition}

\begin{lemma}\label{lem:tensor product is Scott continuous bifunctor}
Let $\V$ and $\W$ be quantum sets, and let $(\X,R)$ and $(\Y,S)$ be quantum posets. Let $F_1\sqsubseteq F_2\sqsubseteq F_3\sqsubseteq\cdots:\V\to\X$ and $G_1\sqsubseteq G_2\sqsubseteq G_3\sqsubseteq\cdots:\W\to\Y$
be monotonically ascending sequences of functions with limits $F_\infty$ and $G_\infty$, respectively. 
Then $F_n\times G_n\nearrow F_\infty\times G_\infty$. 
\end{lemma}
\begin{proof}
By the definition of the order of functions, we have $R \circ F_1 \geq R \circ F_2 \geq \cdots$ and $S \circ G_1 \geq S \circ G_2 \cdots$. Thus, the binary relation $(R \circ F_n) \times (S \circ G_m)$ is monotonically decreasing in both variables $n$ and $m$. Hence, we have $\bigwedge_{n\in\NN}(R\circ F_n)\times( S\circ G_n) = \bigwedge_{n,m\in\NN}(R\circ F_n)\times( S\circ G_m)$. We calculate that
\begin{align*}
    \bigwedge_{n\in\NN}(R\times S)\circ (F_n\times G_n)  & =     \bigwedge_{n\in\NN}(R\circ F_n)\times( S\circ G_n) = \bigwedge_{n,m\in\NN}(R\circ F_n)\times( S\circ G_m)\\
    & = \bigwedge_{n\in\NN}(R\circ F_n)\times\bigwedge_{m\in\NN}( S\circ G_m) = (R\circ F_\infty)\times(S\circ G_\infty)\\
    & = (R\times S)\circ (F_\infty\times G_\infty),
\end{align*}
where we use \cite[Lemma~A.3]{KLM20} for the third equality.
\end{proof}

\begin{lemma}\label{lem:monoidal product qCPO}
Let $(\X,R)$ be a quantum poset, and let $\H$ be an atomic quantum set. Let $\X_1,\X_2\subseteq\X$, and let 
$J_1=J_{\X_1}$ and $J_2=J_{\X_2}$. Let $\bar F_1:\H\to\X_1$ and $\bar F_2:\H\to\X_2$ be surjective functions, and let $M$ be the binary relation from $\X_1$ to $\X_2$ defined by 
$M=\bar F_2\circ \bar F_1^\dag.$
If both
\begin{align}
\label{assumption 1} J_1\circ \bar F_1 & \sqsubseteq J_2\circ\bar F_2\qquad
\mathrm{and}\\
\label{assumption 2} J_2^\dag\circ R\circ J_2\circ \bar F_2& =J_2^\dag\circ R\circ J_1\circ\bar F_1,
\end{align}
hold, then
\begin{itemize}
    \item[(a)] $M$ is a function $\X_1\to\X_2$,
    \item[(b)] $M$ is monotone for the induced orders on $\X_1$ and $\X_2$, and
    \item[(c)] $J_2^\dag\circ R\circ J_2\circ M=J_2^\dag\circ R\circ J_1$.
    \end{itemize}
\end{lemma}

\begin{proof}
We first show (c) by a direct calculation:
\[ J_2^\dag\circ R\circ J_2\circ M=J_2^\dag\circ R\circ J_2\circ \bar F_2\circ \bar F_1^\dag=J_2^\dag\circ R\circ J_1\circ\bar F_1\circ\bar F_1^\dag= J_2^\dag\circ R\circ J_1,\]
where the second equality follows from Equation (\ref{assumption 2}) and the last equality holds because $\bar F_1$ is surjective.
For (a), we calculate that
\[M^\dag\circ M=(\bar F_2\circ\bar F_1^\dag)^\dag\circ (\bar F_2\circ\bar F_1^\dag)=\bar F_1\circ\bar F_2^\dag\circ \bar F_2\circ \bar F_1^\dag\geq \bar F_1\circ I\circ \bar F_1^\dag=\bar F_1\circ\bar F_1^\dag =I,\]
where we use the surjectivity of $\bar F_1$ for the last equality. Similarly, we calculate that
\begin{align*} M\circ M^\dag & = I_{\X_2}\circ M\circ M^\dag\leq J_2^\dag\circ R\circ J_2\circ M\circ M^\dag=J_2^\dag\circ R\circ J_1\circ M^\dag \\
& = J_2^\dag\circ R\circ J_1\circ\bar F_1\circ\bar F_2^\dag= J_2^\dag\circ R\circ J_2\circ \bar F_2\circ \bar F_2^\dag=J_2^\dag\circ R\circ J_2,
\end{align*}
where the inequality follows because $J_2^\dag\circ R\circ J_2$ is an order on $\X_2$ \cite[Lemma~2.1]{KLM20}, the next equality follows from (c), the penultimate equality follows from Equation (\ref{assumption 2}), and the last equality follows from the  surjectivity of $\bar F_2$.
Since $(M\circ M^\dag)^\dag=M\circ M^\dag$, we have
\begin{align*} M\circ M^\dag& =(M\circ M^\dag)\wedge (M\circ M^\dag)^\dag\leq (J_2^\dag\circ R\circ J_2)\wedge (J_2^\dag\circ R\circ J_2)^\dag\\
& = (J_2^\dag\circ R\circ J_2)\wedge (J_2^\dag\circ R^\dag\circ J_2)=J_2^\dag\circ (R\wedge R^\dag)\circ J_2=J_2^\dag\circ J_2=I_{\X_2},
\end{align*}
where we use \cite[Proposition~A.6]{KLM20} in the second equality of the second line and the antisymmetry axiom of an order in the penultimate equality. The last equality follows by the injectivity of $J_2$. 
Finally, we show (b). 
\begin{align*}
    M\circ J_1^\dag\circ R\circ J_1 &  = \bar F_2\circ \bar F_1^\dag\circ J_1^\dag\circ R\circ J_1 = \bar F_2\circ (R^\dag \circ J_1\circ \bar F_1)^\dag\circ J_1 \leq \bar F_2\circ (R^\dag\circ J_2\circ\bar F_2)\circ J_1\\
    & =\bar F_2\circ\bar F_2^\dag\circ J_2^\dag\circ R\circ J_1=J_2^\dag\circ R\circ J_1 = J_2^\dag\circ R\circ J_2\circ M
    \end{align*}
where the inequality follows from Equation (\ref{assumption 1}) and \cite[Lemma~4.1]{KLM20}, the penultimate equality follows from the surjectivity of $\bar F_2$, and the last equality is (c).
\end{proof}

\begin{theorem}\label{thm:monoidal product of qCPOs}
Let $(\X,R)$ and $(\Y,S)$ be quantum posets, let $\H = \Q\{H\}$ be an atomic quantum set, and let $E_1\sqsubseteq E_2\sqsubseteq\cdots:\H\to\X\times\Y$ be a monotonically ascending sequence of functions.
For each $n\in\NN$, let $F_n=P\circ E_n:\H\to\X$, and $G_n=Q\circ E_n:\H\to\Y$. Assume that there are functions  $F_\infty:\H\to\X$ and $G_\infty:\H\to\Y$ such that $F_n\nearrow F_\infty$ and $G_n\nearrow G_\infty$.
Then, there exists a natural number $\ell$ such that the binary relation $E_\infty:\H\to\X\times\Y$, defined by
$E_\infty= (F_\infty\times G_\infty)\circ (F_\ell^\dag\times G_\ell^\dag)\circ (F_\ell,G_\ell),$
is a function that satisfies $E_n\nearrow E_\infty$, $P\circ E_\infty=F_\infty$, and $Q\circ E_\infty=G_\infty.$  
\end{theorem}

Thus, we find that $F_\infty$ and $G_\infty$ are compatible, with $(F_n, G_n) \nearrow (F_\infty, G_\infty)$.

\begin{proof}[Proof of Theorem \ref{thm:monoidal product of qCPOs}]
First, we show that $F_\infty$ and $G_\infty$ are compatible. We begin by factoring each function $F_i$ through its range, and likewise each function $G_i$ \cite[Definition~3.2]{KLM20}. Thus, writing $J_i$ for the inclusion of $\X_i:= \ran F_i$ into $\X$ and $K_i$ for the inclusion of $\Y_i:= \ran G_i$ into $\Y$, we have $F_i = J_i \circ \overline F_i$ and $G_i = K_i \circ \overline G_i$, for all $i \in \{0,1, \ldots, \infty\}$. Each function $\overline F_i \: \H \to \X_i$ is surjective, and this implies that its codomain $\X_i$ is finite in the sense that it has only finitely many atoms. Indeed, by \cite{Kornell18}*{Proposition 8.1} the unital $*$-homomorphism $F_i^\star\: \ell^\infty(\X_i) \To \ell^\infty(\H) = L(H)$ is injective. Hence, each $\X_i$ is finite, for $i = \{1, 2, \ldots, \infty\}$, and similarly, each $\Y_i$ is finite.

For each atom $X\atomof\X_\infty$, the sequence of operator subspaces $(R\circ F_1)(H,X)\geq (R\circ F_2)(H,X)\geq \cdots$ is monotonically decreasing, simply because the sequence of functions $F_1 \below F_2 \below$ is monotonically ascending. These subspaces are all finite-dimension; hence, we conclude that this monotonically decreasing sequence of subspaces must stablize. Similarly, for each atom $Y\atomof\Y_\infty$, the monotonically decreasing sequence $(S\circ G_1)(H,Y)\geq (S\circ G_2)(H,Y)\geq\cdots$ must also stabilize.
Since both $\X_\infty$ and $\Y_\infty$ are finite, this means that we can find an index $\ell\in\NN$ such that for each index $i \in \{\ell, \ell+1, \ldots,\}$, each atom $X\atomof\X_\infty$ and each atom $Y\atomof\Y_\infty$, we have
\begin{align*}
    (R\circ F_i)(H,X)  =(R\circ F_\infty)(H,X),\quad\text{and}\quad
    (S\circ G_i)(H,Y)  =(S\circ G_\infty)(H,Y).
\end{align*}
We apply \cite[Lemma~A.5]{KLM20} to find that 
\begin{align*}
    J_\infty^\dag\circ R\circ F_i  =J_\infty^\dag\circ R\circ F_\infty,\quad\text{and}\quad
    K_\infty^\dag\circ S\circ G_i  = K_\infty^\dag\circ S\circ G_\infty.
\end{align*}
In other words, $
    J_\infty^\dag\circ R\circ J_i\circ\bar F_i =J_\infty^\dag\circ R\circ J_\infty\circ \bar F_\infty$, and $
    K_\infty^\dag\circ S\circ K_i\circ \bar G_i  = K_\infty^\dag\circ S\circ K_\infty\circ \bar G_\infty. 
$
We certainly also have that 
$
    J_i\circ\bar F_i \sqsubseteq J_\infty\circ\bar F_\infty,$ and $
    K_i\circ\bar G_i  \sqsubseteq  K_\infty\circ\bar G_\infty,
$
so we can apply Lemma \ref{lem:monoidal product qCPO} at each index $i \in \{\ell, \ell+1, \ldots\}$.

For each index $i \in \{\ell, \ell+1, \ldots\}$, let $M_i:=\bar F_\infty\circ \bar F_i^\dag$ and $N_i:= \bar G_\infty\circ \bar G_i^\dag$. We have that
\begin{itemize}
    \item[(a)] $M_i:\X_i\to\X_\infty$ and $N_i:\Y_i\to\Y_\infty$ are functions,
    \item[(b)] $M_i$ and $N_i$ are monotone,
    \item[(c)] $
        J_\infty^\dag\circ R\circ J_\infty\circ M_i = J_\infty^\dag\circ R\circ J_i$ and $K_\infty^\dag\circ S\circ K_\infty\circ N_i  = K_\infty^\dag\circ S\circ K_i$.
\end{itemize}
Note that $F_i = J_i\circ \bar F_i$ and $G_i = K_i\circ\bar G_i$ are compatible by definition. Since $J_i$ and $K_i$ are injective, it follows by \cite[Corollary~B.4]{KLM20} that $\overline F_i$ and $\overline G_i$ are also compatible. Let $\bar E_i = (\bar F_i, \bar G_i)$. We reason that
\[ (J_i\times K_i)\circ\bar E_i=(J_i\times K_i)\circ (\bar F_i,\bar G_i)=(J_i\circ \bar F_i,K_i\circ\bar G_i)=(F_i,G_i)=E_i.\]
Now, define $
    \bar E_\infty  =(M_\ell\times N_\ell)\circ \bar E_\ell$ and 
    $E_\infty =(J_\infty\times K_\infty)\circ \bar E_\infty.$
We calculate that
\begin{align*}
    E_\infty & =(J_\infty\times K_\infty)\circ (M_\ell\times N_\ell)\circ \bar E_\ell\\
    & = (J_\infty\times K_\infty)\circ ((\bar F_\infty\circ \bar F_\ell^\dag)\times(\bar G_\infty\circ \bar G_\ell^\dag ))\circ (\bar F_\ell,\bar G_\ell)\\
    & = (J_\infty\times K_\infty)\circ (\bar F_\infty\times \bar G_\infty)\circ (\bar F^\dag_\ell\times \bar G_\ell^\dag)\circ (\bar F_\ell,\bar G_\ell)\\
    & = (J_\infty\times K_\infty)\circ (\bar F_\infty\times \bar G_\infty)\circ (\bar F^\dag_\ell\times \bar G_\ell^\dag)\circ(J_\ell^\dag\times K_\ell^\dag)\circ (J_\ell\times K_\ell)\circ  (\bar F_\ell,\bar G_\ell)\\
    & =((J_\infty\circ \bar F_\infty)\times (K_\infty\circ \bar G_\infty))\circ ((\bar F^\dag_\ell\circ J_\ell^\dag)\times (\bar G_\ell^\dag\circ K_\ell^\dag))\circ  (J_\ell\circ\bar F_\ell,K_\ell\circ \bar G_\ell)\\
    & = (F_\infty\times G_\infty)\circ (F_\ell^\dag\times G_\ell^\dag)\circ (F_\ell,G_\ell),
\end{align*}
where we use \cite[Corollary~B.4]{KLM20} in the penultimate equality.
Hence, for each finite $i \geq \ell$,
\begin{align*}
    P\circ (J_\infty\times K_\infty)\circ (M_i\times N_i)\circ\bar E_i & =     P\circ (J_\infty\times K_\infty)\circ (M_i\times N_i)\circ(\bar F_i,\bar G_i)\\
    & =     P\circ (J_\infty\circ M_i\circ\bar F_i,K_\infty\circ N_i\circ\bar G_i)\\
    & = J_\infty\circ M_i\circ\bar F_i=J_\infty\circ\bar F_\infty\circ \bar F_i^\dag\circ\bar F_i\\
    & \geq J_\infty\circ\bar F_\infty\circ I=J_\infty\circ\bar F_\infty= F_\infty.
\end{align*}
By \cite[Lemma~A.7]{KLM20}
we obtain 
$P\circ (J_\infty\times K_\infty)\circ (M_i\times N_i)\circ\bar E_i =F_\infty,$
and in a similar way, we find
$ Q \circ (J_\infty\times K_\infty)\circ (M_i\times N_i)\circ\bar E_i=G_\infty.$ Therefore, $F_\infty$ and $G_\infty$ are compatible and
$(F_\infty,G_\infty)= (J_\infty\times K_\infty)\circ (M_i\times N_i)\circ\bar E_i$ for each $i \in \{\ell ,\ell+1, \cdots\}$.
In particular, we conclude that 
$ (F_\infty,G_\infty)= (J_\infty\times K_\infty)\circ (M_\ell\times N_\ell)\circ\bar E_\ell=(J_\infty\times K_\infty)\circ \bar E_\infty=E_\infty.$

Next we show that $E_n\nearrow E_\infty$. 
Since $F_n\sqsubseteq F_\infty$ and $G_n\sqsubseteq G_\infty$ for each $n\in\NN$, it follows from \cite[Proposition~7.6]{KLM20} that $E_n\sqsubseteq E_\infty$ for each $n\in\NN$. Thus, $(R\times S)\circ E_\infty\leq (R\times S)\circ E_n$ for each $n\in\NN$, and therefore,
\begin{equation}
    \label{ineq:product supremum}(R\times S)\circ E_\infty\leq \bigwedge_{n\in\NN}(R\times S)\circ E_n.
\end{equation} 

Fix an atom $X\otimes Y\atomof\X\times\Y$. Let $J_X$ be the inclusion function $J_{\Q\{X\}}:\Q\{X\}\to\X$. 
Then we have 
$J_{X\otimes Y}=J_{\Q\{X\otimes Y\}}=J_{\Q\{X\}\times\Q\{Y\}}=J_{\Q\{X\}}\times J_{\Q\{Y\}}=J_{X}\times J_Y$
where we used \cite[Proposition~B.3]{KLM20} in the penultimate equality.
Since $L(H,X)$ and $L(H,Y)$ are finite-dimensional, the monotonically decreasing sequences
\begin{align*}
    (R\circ F_1)(H,X)\geq (R\circ F_2)(H,X)\geq\cdots,\quad \text{and} \quad
    (S\circ G_1)(H,Y)\geq (S\circ G_2)(H,Y)\geq\cdots
\end{align*}
stabilize, and since $\bigwedge_{n\in\NN}R\circ F_n=R\circ F_\infty$ and $\bigwedge_{n\in\NN}S\circ G_n=S\circ G_\infty$, this means that we can find a natural number $m\geq \ell$ such that
\begin{align*}
(R\circ F_n)(H,X) = (R\circ F_\infty)(H,X),\quad \text{and} \quad
(S\circ G_n)(H,Y)  = (S\circ G_\infty)(H,Y)
\end{align*}
for all $n\geq m$. Applying \cite[Lemma~A.5]{KLM20}, we infer that
\begin{align*}
(J_X^\dag\circ R\circ F_n)(H,X) & = (J_X^\dag\circ R\circ F_\infty)(H,X), \; \text{and} \\
(J_Y^\dag\circ S\circ G_n)(H,Y) & = (J_Y^\dag\circ S\circ G_\infty)(H,Y).
\end{align*}
Hence,
\begin{align}
\label{eq: tensor identity 1} J_X^\dag\circ R\circ F_n & = J_X^\dag\circ R\circ F_\infty\qquad \mathrm{and}\\
\label{eq: tensor identity 2} J_Y^\dag\circ S\circ G_n & = J_Y^\dag\circ S\circ G_\infty
\end{align}
for each $n\geq m$.
Now, we calculate that
\begin{align*}
    J_{X\otimes Y}^\dag\circ&\bigwedge_{n\in\NN}(R\times S)\circ E_n  =     \bigwedge_{n\in\NN}J_{X\otimes Y}^\dag\circ(R\times S)\circ E_n\\
           & =      \bigwedge_{n\in\NN}(J_{X }^\dag\times J_{Y}^\dag)\circ(R\times S)\circ (F_n,G_n)\\
               & =      \bigwedge_{n\in\NN}(J_{X }^\dag\times J_{Y}^\dag)\circ(R\times S)\circ (J_n\times K_n)\circ (\bar F_n,\bar G_n)\\
        & \leq      \bigwedge_{n\geq m}(J_{X }^\dag\times J_{Y}^\dag)\circ(R\times S)\circ (J_n\times K_n)\circ \big((\bar F_n\circ\bar F_\ell^\dag)\times(\bar G_n\circ\bar G_\ell^\dag)\big)\circ (\bar F_\ell,\bar G_\ell)\\
               & =     \bigwedge_{n\geq m}\big((J_{X }^\dag\circ R\circ  F_n)\times (J_{Y}^\dag\circ S\circ G_n)\big) \circ (\bar F_\ell^\dag\times \bar G_\ell^\dag)\circ (\bar F_\ell,\bar G_\ell)\\
        & =     \big((J_{X }^\dag\circ R\circ  F_\infty)\times (J_{Y}^\dag\circ S\circ G_\infty)\big) \circ (\bar F_\ell^\dag\times \bar G_\ell^\dag)\circ (\bar F_\ell,\bar G_\ell)\\
               & =   \big((J_{X }^\dag\circ R\circ  J_\infty)\times (J_{Y}^\dag\circ S\circ K_\infty)\big) \circ (M_\ell\times N_\ell)\circ (\bar F_\ell,\bar G_\ell)\\
        & =J_{X\otimes Y }^\dag\circ (R\times S)\circ  E_\infty,
\end{align*}
where the first equality follows by \cite[Proposition~A.6]{KLM20} and the inequality follows by \cite[Proposition~B.7]{KLM20}. The third-to-last equality follows from Equations (\ref{eq: tensor identity 1}) and (\ref{eq: tensor identity 2}).
We now apply \cite[Lemma~A.5]{KLM20} to calculate that
\begin{align*}& \left(\bigwedge_{n\in\NN}(R\times S)\circ E_n\right)(H,X\otimes Y)  =\left(J_{X\otimes Y}^{\dag}\circ \bigwedge_{n\in\NN}(R\times S)\circ E_n\right)(H,X\otimes Y)\\
& \leq\left(J_{X\otimes Y}^{\dag}\circ(R\times S)\circ E_\infty\right)(H,X\otimes Y)
 = ((R\times S)\circ E_\infty)(H,X\otimes Y).
\end{align*}
We now vary the atom $X\otimes Y\atomof\X\times\Y$ to conclude that $\bigwedge_{n\in\NN}(R\times S)\circ E_n\leq (R\times S)\circ E_\infty.$
With Equation (\ref{ineq:product supremum}), this implies that $E_n\nearrow E_\infty$, as claimed.
\end{proof}

\begin{corollary}\label{cor:monoidal product of qCPOs}
Let $(\X,R)$ and $(\Y,S)$ be quantum cpos. Then $(\X\times \Y,R\times S)$ is a quantum cpo, and the projection maps $P:\X\times \Y\to\X$ and $Q:\X\times \Y\to\Y$ are Scott continuous.
\end{corollary}
\begin{proof}
Let $\H$ be an atomic quantum set, and let $E_1\sqsubseteq E_2\sqsubseteq\ldots$ be a monotonically ascending sequence of functions $\H\to \X\times\Y$. Let $F_n=P\circ E_n$, and let $G_n=Q\circ E_n$. By \cite[Lemma~7.3]{KLM20}, the projection functions $P$ and $Q$ are monotone; hence $F_1\sqsubseteq F_2\sqsubseteq\ldots$ and $G_1\sqsubseteq G_2\sqsubseteq\ldots$ are monotonically ascending sequences. Since both $\X$ and $\Y$ are quantum cpos, there exist functions $F_\infty:\H\to\X$ and $G_\infty:\H\to\Y$ such that $F_n\nearrow F_\infty$ and $G_n\nearrow G_\infty$. Now, we apply Theorem \ref{thm:monoidal product of qCPOs} to infer that there exists a function $E_\infty\: \H \to \X \times \Y$ such that $E_n\nearrow E_\infty$, $P\circ E_\infty=F_\infty$, and $Q\circ E_\infty=G_\infty.$  
Since 
    $P\circ E_n  =F_n\nearrow F_\infty=P\circ E_\infty$ and
    $Q\circ E_n = G_n\nearrow G_\infty = Q\circ E_\infty$, we also infer that $P$ and $Q$ are Scott continuous.
\end{proof}

\begin{corollary}
Let $(\X,R)$ and $(\Y,S)$  be quantum cpos, and let $(\Z,T)$ be a quantum poset. Let $M:\Z\to\X\times\Y$ be a function.  Then $M$ is Scott continuous if and only if $P\circ M:\Z\to\X$ and $Q\circ M:\Z\to\Y$ are Scott continuous.
\end{corollary}
\begin{proof}
Assume that $M$ is Scott continuous. Since $P$ and $Q$ are Scott continuous by Corollary \ref{cor:monoidal product of qCPOs} and the composition of Scott continuous maps is Scott continuous (cf. Lemma \ref{lem:composition is Scott continuous}), it follows that $P\circ M$ and $Q\circ M$ are Scott continuous. For the converse, assume that $P\circ M$ and $Q\circ M$ are Scott continuous. Then, $P\circ M$ and $Q\circ M$ are monotone, so $M$ is also monotone by \cite[Proposition~7.4]{KLM20}.

Let 
$K_1\sqsubseteq K_2\sqsubseteq\cdots$
be a monotonically ascending sequence of functions $\H\to\Z$ with limit $K_\infty$. For each $i\in \{1,2,\ldots,\infty\}$, let $E_i=M\circ K_i$, let 
$F_i=P\circ E_i$ and $G_i=Q\circ E_i$.
By the Scott continuity of $P\circ M$ and $Q\circ M$, we have
\begin{align*}
    F_n=P\circ M\circ K_n \nearrow P\circ M\circ K_\infty=F_\infty,\quad \text{and} \quad
    G_n=Q\circ M\circ K_n \nearrow Q\circ M\circ K_\infty =G_\infty.
\end{align*}
By the monotonicity of $M$, 
$E_1\sqsubseteq E_2\sqsubseteq\cdots$
is a monotonically ascending sequence of functions $\H\to\X\times\Y$. We now apply Theorem \ref{thm:monoidal product of qCPOs} to we conclude that $E_n \nearrow E_\infty$, because the function $E_\infty$ is uniquely determined by the equations $P \circ E_\infty = F_\infty$ and $Q \circ E_\infty = G_\infty$. In other words, $M \circ K_n \nearrow M \circ K_\infty$. Therefore, $M$ is Scott continuous.
\end{proof}

\begin{lemma}\label{lem:tensor product of Scott continuous maps is Scott continuous}
Let $(\X_1,R_1)$, $(\X_2,R_2)$, $(\Y_1,S_1)$ and $(\Y_2,S_2)$ be quantum cpos, and let $F:(\X_1,R_1)\to(\X_2,R_2)$ and $G:(\Y_1,S_1)\to(\Y_2,S_2)$ be Scott continuous. Then \[F\times G:(\X_1\times \Y_1,R_1\times S_1)\to(\X_2\times \Y_2,R_2\times S_2)\] is Scott continuous.
\end{lemma}
\begin{proof}
Let $K_1\sqsubseteq K_2\sqsubseteq\cdots:\H\to\X_1\times\Y_1$ be a monotonically ascending sequence with limit $K_\infty$, and for each $i\in\{1,2,\ldots,\infty\}$, let $E_i=(F\times G)\circ K_i$.
By \cite[Lemma~7.2]{KLM20}, $F\times G$ is monotone; hence
$E_1\sqsubseteq E_2\sqsubseteq\cdots:\H\to\X_2\times\Y_2$
is a monotonically ascending sequence of functions, too. It follows from Theorem \ref{thm:monoidal product of qCPOs} that $E_n\nearrow E$ for some $E:\H\to\X_2\times\Y_2$ such that $P_2\circ E_n\nearrow P_2\circ E$ and $Q_2\circ E_n\nearrow Q_2\circ E$. Here, $P_i\: \X_i \times \Y_i \To \X_i$ and $Q_i\: \X_i \times \Y_i \To \Y_i$ are the canonical projections, for both $i \in \{1,2\}$.

It also follows from Corollary \ref{cor:monoidal product of qCPOs} that $P_1\circ K_n\nearrow P_1\circ K_\infty$ and $Q_1\circ K_n\nearrow Q_1\circ K_\infty$. Since $F$ and $G$ are Scott continuous, it follows that $F\circ P_1\circ K_n\nearrow F\circ P_1\circ K_\infty$ and $G\circ Q_1\circ K_n\nearrow G\circ Q_1\circ K_\infty$. Using \cite[Proposition~B.2]{KLM20}, we obtain
\begin{align*}
P_2\circ E_n & = P_2\circ (F\times G)\circ K_n=F\circ P_1\circ K_n\nearrow F\circ P_1\circ K_\infty = P_2\circ (F\times G)\circ K_\infty= P_2\circ E_\infty,\\
Q_2\circ E_n & = Q_2\circ (F\times G)\circ K_n=G\circ Q_1\circ K_n\nearrow G\circ Q_1\circ K_\infty = Q_2\circ (F\times G)\circ K_\infty= Q_2\circ E_\infty.
\end{align*}
By the uniqueness of limits (Lemma \ref{lem:lim is sup}), we have that $P\circ E=P\circ E_\infty$ and $Q\circ E=Q\circ E_\infty$. It then follows that $E=E_\infty$ \cite[Appendix~B]{KLM20}; hence $ (F\times G)\circ K_n=E_n\nearrow E=E_\infty = (F\times G)\circ K_\infty$. Therefore, $F\times G$ is indeed Scott continuous.
\end{proof}

\begin{theorem}\label{thm:qCPO is monoidal}
The symmetric monoidal structure on $\qPOS$ restricts to $\qCPO$.
\end{theorem}
\begin{proof}
By Theorem \ref{thm:monoidal product of qCPOs} and Lemma \ref{lem:tensor product of Scott continuous maps is Scott continuous}, the monoidal product on $\qPOS$ restricts to a bifunctor on $\qCPO$. The tensor unit $(\mathbf 1,I_{\mathbf 1}$) of $\qPOS$ is a quantum cpo by Proposition  \ref{prop:finite qposet is qcpo}. Finally, the components of the associator, the left and right unitors and the symmetry in $\qPOS$ are order isomorphisms (cf. \cite[Theorem~7.5]{KLM20}); hence, they are Scott continuous by Lemma \ref{lem:order iso is Scott continuous}, i.e., they are morphisms of $\qCPO$. 
\end{proof}

\subsection{Scott continuity in one of two variables} Let $(\X,R)$, $(\Y,S)$ and $(\Z,T)$ be quantum cpos, and let $D\: \X \times \Y \to \Z$ be a function. We define the Scott continuity of $D$ in the first variable and, symmetrically, the Scott continuity of $D$ in the second variable. We then show that $D$ is Scott continuous if and only if it is Scott continuous in both variables.

\begin{definition}\label{lem:monotone in the first variable iff monotone when second factor is trivially ordered}
Let $\X$ and $\Y$ be quantum sets, and let $(\Z,T)$ be a quantum poset. Let $D:\X\times\Y\to\Z$ be a function. If $R$ is an order on $\X$, we say that $D$ is \emph{monotone in the first variable} if 
$D\circ (R\times I_\Y)\leq T\circ D$, or equivalently, if $D$ is a monotone function from $(\X \times \Y, R \times I_\Y)$ to $(\Z, T)$. Similarly, if $S$ is an order on $\Y$, we say that $D$ is \emph{monotone in the second variable} if 
$ D\circ (I_\X\times S)\leq T\circ D$, or equivalently, if $D$ is a monotone function from $(\X \times \Y, I_\X \times S)$ to $(\Z, T)$.
\end{definition}

\begin{lemma}\label{lem:monotone in both variables iff monotone}
Let $(\X,R)$, $(\Y,S)$ and $(\Z,T)$ be quantum posets, and let $D:\X\times\Y\to\Z$ be a function. Then, $D$ is monotone if and only if it is monotone in both variables.
\end{lemma}
\begin{proof}
Assume that $D$ is monotone, i.e.,
$ D\circ (R\times S)\leq T\circ D.$
Since $I_\X\leq R$ and $I_\Y\leq S$, we have that
$D\circ (I_\X\times S)\leq D\circ (R\times S)\leq T\circ D$
and $D\circ (R\times I_\Y)\leq D \circ (R \times S) \leq  T\circ D$ \cite[Lemma~A.3(b)]{KLM20}. Hence, $D$ is monotone in both variables. 

For the converse, assume that $D$ is monotone in both variables. Then,
\[ D\circ (S\times R)=D\circ (R\times I_\Y)\circ (I_\X\times S)\leq T\circ D\circ (I_\X\times R)\leq T\circ T\circ D\leq T\circ D,\]
where the first inequality holds because $D$ is monotone in the second variable and the second inequality holds because $D$ is monotone in the first variable. 
\end{proof}

\begin{definition}
Let $\X$ and $\Y$ be quantum sets, and let $(\Z,T)$ be a quantum poset. If $R$ is an order on $\X$, we say that a function $D:\X\times\Y\to \Z$ is \emph{Scott continuous in the first variable} if it is monotone in the first variable and, for each atomic quantum set $\H$ and each monotonically ascending sequence 
$F_1\sqsubseteq F_2\sqsubseteq\cdots$ of functions $\H \to \X$, if $F_n\nearrow F_\infty$ for some function $F_\infty:\H\to\X$, then $D\circ (F_n\times I_\Y)\nearrow D\circ (F_\infty\times I_\Y)$.

Similarly, if $S$ is an order on $\Y$, we say that a function $D\: \X \times \Y \to \Z$ is \emph{Scott continuous in the second variable} if it is monotone in the second variable and, for each atomic quantum set $\H$ and each monotonically ascending sequence $G_1\sqsubseteq G_2\sqsubseteq\cdot$ of functions $\H \to \Y$, if $G_n\nearrow G_\infty$ for some function  $G_\infty:\H\to\Y$, then $D\circ (I_\X\times G_n)\nearrow D\circ (I_\X\times G_\infty)$.
\end{definition}


\begin{lemma}\label{lem:D Scott continuous in first variable and F Scott continuous then DFtimes G Scott continuous in first variable}
Let $\V$, $\W$, $\X$, $\Y$ and $\Z$ be quantum sets, and let $Q$, $S$ and $T$ be orders on $\W$, $\Y$ and $\Z$, respectively. Let $F\: \V \to \X$, $G\: \W \to \Y$ and $D\: \X \times \Y \to \Z$ be functions. If $G$ is Scott continuous and $D$ is Scott continuous in the second variable, then $D \circ (F \times G)$ is also Scott continuous in the second variable.
\end{lemma}

\begin{proof}
The function $D$ is Scott continuous in the second variable, and in particular, it is monotone in the second variable. Equivalently, it is monotone with respect to the order $I_\X \times S$ on $\X \times \Y$ (Definition \ref{lem:monotone in the first variable iff monotone when second factor is trivially ordered}). The function $G:\V\to\Y$ is trivially monotone if we order $\Y$ and $\V$ trivially \cite[Example~1.11]{KLM20}. Hence, $F\times G:\W\times\V\to\X\times\Y$ is also monotone, by \cite[Theorem~7.5]{KLM20}. Thus, $D\circ(F\times G)$ is monotone with respect to the order $I_\V\times Q$ on $\V\times\W$, i.e., it is monotone in the second variable.

Now, let $\H$ be an atomic quantum set, and let $K_1\sqsubseteq K_2\sqsubseteq \cdots$ be a monotonically ascending sequence of functions $\H\to\W$ with limit $K_\infty$. By the Scott continuity of $G$, it follows that $G\circ K_1\sqsubseteq G\circ K_2\sqsubseteq \cdots$ is a monotonically ascending sequence of functions $\H\to\Y$ with limit $G\circ K_\infty$. Since $D$ is Scott continuous in the second variable, it follows that $D\circ ( I_\X \times (G\circ K_n))\nearrow D\circ ( I_\X \times (G\circ K_\infty))$. Furthermore, by Lemma \ref{lem:right multiplication is Scott continuous}, 
\[ D\circ (I_\X \times (G\circ K_n))\circ (F\times I_\Y)\nearrow D\circ (I_\X \times (G\circ K_\infty))\circ (F\times I_\Y),  \]  which is nothing more than
\[D\circ (F\times G)\circ (I_\V \times K_n)\nearrow  D\circ (F\times G)\circ (I_\V \times K_\infty).\]
Therefore, $D\circ (F\times G)$ is indeed Scott continuous in the second variable.
\end{proof}

\begin{proposition}\label{prop:Scott continuous iff Scott continuous in both variables}
Let $(\X,R)$, $(\Y,S)$ and $(\Z,T)$ be quantum cpos. Then, a function $D:\X\times\Y\to\Z$ is Scott continuous if and only if it is Scott continuous in both variables.
\end{proposition}
\begin{proof}
First, assume that $D$ is Scott continuous. It is monotone in both variables by Lemma \ref{lem:monotone in both variables iff monotone}. Hence, let $\H$ be an atomic quantum set, and let $F_1 \below F_2 \below \cdots$ be a monotonically ascending sequence of functions $\H \to \X$ with limit $F_\infty$. By Lemma \ref{lem:tensor product is Scott continuous bifunctor}, $F_n \times I_\Y \nearrow F_\infty \times I_\Y$. By Lemma \ref{lem:left multiplication by Scott continuous function is Scott continuous}, $D \circ (F_n \times I_\Y) \nearrow D \circ (F_\infty \times I_\Y)$. Therefore, $D$ is Scott continuous in the first variable, and symmetrically, it is Scott continuous in the second variable.

Conversely, assume that $D$ is Scott continuous in both variables. In particular, it is monotone in both variables; hence it is monotone by Lemma \ref{lem:monotone in both variables iff monotone}. Let $E_1\sqsubseteq E_2\sqsubseteq\cdots$ be a monotonically ascending sequence of functions $\H\to\X\times\Y$ with limit $E_\infty$. 
Let $F_i=P\circ E_i$ and $G_i=Q\circ E_i$, for $i\in\{1,2,\ldots,\infty\}$. 
Since $P$ and $Q$ are monotone \cite[Lemma~7.3]{KLM20}, it follows that $F_1\sqsubseteq F_2\sqsubseteq\cdots\sqsubseteq F_\infty$ and $G_1\sqsubseteq G_2\sqsubseteq\cdots\sqsubseteq G_\infty.$ By Corollary \ref{cor:monoidal product of qCPOs}, we have $F_n\nearrow F_\infty$ and $G_n\nearrow G_\infty$. We apply Theorem \ref{thm:monoidal product of qCPOs} to obtain a natural number $\ell$ such that
$E_\infty=(F_\infty\times G_\infty)\circ (F_\ell^\dag\times G_\ell^\dag)\circ (F_\ell,G_\ell).$

Fix $Z \atomof \Z$, and let $J_Z\: \Q\{Z\} \hookrightarrow \Z$ be inclusion.
For each index $i \in \{0, 1 ,\ldots, \infty\}$, we factor the function $F_i$ through its range: $F_i = J_{\X_i} \circ \bar F_i$, as in \cite[Definition~3.2]{KLM20}. Similarly, for each index $j \in \{0, 1, \ldots, \infty\}$, we factor the function $G_j$ through its range: $G_j = J_{\Y_j} \circ \bar G_j$. For each $n \in \NN$, the function $\bar F_n\: \H \to \X_n$ is surjective, so the unital normal $*$-homomorphism $\bar F_n^\star\: \ell^\infty(\X_n) \to \ell^\infty(\H) = L(H)$ is injective by \cite{Kornell18}*{Proposition 8.1}; thus, $\X_n$ has only finitely many atoms. We now reason that, because $D\: \X \times \Y \to \Z$ is Scott continuous in the second variable, the sequence $$T \circ D \circ (I_\X \times G_1) \geq T \circ D \circ (I_\X \times G_2) \geq \cdots$$ is monotonically decreasing in $\cat{qRel}(\X \times \H, \Z)$, with infimum $T \circ D \circ (I_\X \times G_\infty)$, and therefore, the sequence $$J_\Z^\dagger \circ T \circ D \circ (I_\X \times G_1) \circ (J_{\X_n} \times I_\H) \geq J_\Z^\dagger \circ  T \circ D \circ (I_\X \times G_2) \circ (J_{\X_n} \times I_\H) \geq \cdots$$ is monotonically decreasing in $\cat{qRel}(\X_n \times \H, \Q\{Z\})$, with infimum $J_\Z^\dagger \circ  T \circ D \circ (I_\X \times G_\infty) \circ (J_{\X_n} \times I_\H)$. The poset $\cat{qRel}(\X_n \times \H, \Q\{Z\})$ clearly has finite height because the quantum sets $\X_n$, $\Z$ and $\Q\{Z\}$ all have only finitely many atoms. Therefore, for each $n \in \NN$, there exists a number $k_n \in \NN$ such that $J_\Z^\dagger \circ T \circ D \circ (I_\X \times G_m) \circ (J_{\X_n} \times I_\H) = J_\Z^\dagger \circ T \circ D \circ (I_\X \times G_\infty) \circ (J_{\X_n} \times I_\H)$ for all $m \geq k_n$. Similarly, there exists a number $k_\infty$ such that $J_Z^\dagger \circ T \circ D \circ (F_n \times I_\Y) \circ (I_\X \times J_{\Y_\infty}) = J_Z^\dagger \circ T \circ D \circ (F_\infty \times I_\Y) \circ (I_\X \times J_{\Y_\infty})$ for all $n \geq k_\infty$, because $D$ is Scott continuous in the first variable.

We now calculate that
\begin{align*}
 J_Z^\dagger &\circ \left( \bigwedge_{n\in\NN}T\circ D\circ E_n \right) \mathop{=}^{(a)}  \bigwedge_{n\in\NN}J_Z^\dagger \circ T\circ D\circ E_n
=
\bigwedge_{n\in\NN} J_Z^\dagger \circ T\circ D\circ (F_n,G_n)
\\& \mathop{\leq}^{(b)} 
\bigwedge_{n\in\NN} J_Z^\dagger \circ T\circ D\circ (F_n\times G_n)\circ (F_\ell^\dag\times G^\dag_\ell)\circ (F_\ell,G_\ell)
\\ & \mathop{=}^{(c)}
\bigwedge_{n\in\NN}\bigwedge_{m\in\NN} J_Z^\dagger \circ T\circ D\circ (F_n\times G_m)\circ (F_\ell^\dag\times G^\dag_\ell)\circ (F_\ell,G_\ell)
\\ & 
\leq \bigwedge_{n\geq k_\infty}\bigwedge_{m\geq k_n} J_Z^\dagger \circ T\circ D\circ (F_n\times G_m)\circ (F_\ell^\dag\times G^\dag_\ell)\circ (F_\ell,G_\ell)
\\ & =
\bigwedge_{n\geq k_\infty}\bigwedge_{m\geq k_n} J_Z^\dagger \circ T\circ D\circ (I_\X\times G_m)\circ(J_{\X_n}\times I_{\H})\circ (( \bar F_n\circ F_\ell^\dag)\times G^\dag_\ell)\circ (F_\ell,G_\ell)
\\ & =
\bigwedge_{n\geq k_\infty}\bigwedge_{m\geq k_n} J_Z^\dagger \circ T\circ D\circ (I_\X\times G_\infty)\circ(J_{\X_n}\times I_{\H})\circ (( \bar F_n\circ F_\ell^\dag)\times G^\dag_\ell)\circ (F_\ell,G_\ell)
\\ & =
\bigwedge_{n\geq k_\infty} J_Z^\dagger \circ T\circ D\circ (I_\X\times G_\infty)\circ(J_{\X_n}\times I_{\H})\circ (( \bar F_n\circ F_\ell^\dag)\times G^\dag_\ell)\circ (F_\ell,G_\ell)
\\ &=
\bigwedge_{n\geq k_\infty} J_Z^\dagger \circ T\circ D\circ (F_n\times G_\infty)\circ (F_\ell^\dag\times G^\dag_\ell)\circ (F_\ell,G_\ell) 
\\ & = \bigwedge_{n\geq k_\infty}J_Z^\dagger \circ T\circ D\circ (F_n\times I_\Y)\circ (I_{\H}\times J_{\Y_\infty})\circ (F_\ell^\dag\times(\bar G_\infty\circ G^\dag_\ell))\circ (F_\ell,G_\ell)
\\ & = \bigwedge_{n\geq k_\infty} J_Z^\dagger \circ  T\circ D\circ (F_\infty\times I_\Y)\circ (I_{\H}\times J_{\Y_\infty})\circ (F_\ell^\dag\times (\bar G_\infty\circ G^\dag_\ell))\circ (F_\ell,G_\ell)
\\ & = J_Z^\dagger \circ  T\circ D\circ (F_\infty\times I_\Y)\circ (I_{\H}\times J_{\Y_\infty})\circ (F_\ell^\dag\times (\bar G_\infty\circ G^\dag_\ell))\circ (F_\ell,G_\ell)
\\ & = J_Z^\dagger \circ  T\circ D\circ (F_\infty\times G_\infty)\circ (F_\ell^\dag\times G^\dag_\ell)\circ (F_\ell,G_\ell)
= J_Z^\dagger \circ T \circ D \circ E_\infty.
\end{align*}
The equality marked $(a)$ follows from \cite[Proposition~A.6]{KLM20}. The inequality marked $(b)$ follows from \cite[Proposition~B.7]{KLM20}. The equality marked $(c)$ follows from the fact that the binary relation $ T\circ D\circ (F_n\times G_m) \in \cat{qRel}(\H \times \H, \Z)$ can be written as both $T \circ D\circ (F_n\times I_{\Y}) \circ (I_\H \times G_m)$ and $T \circ D\circ (I_\X \times G_m) \circ (F_n\times I_{\H})$, and thus, $ T\circ D\circ (F_n\times G_m)$ is monotonically decreasing in both $n \in \NN$ and $m \in \NN$, because $D$ is monotone in both variables.

We vary $Z \atomof \Z$ to conclude that $\bigwedge_{n\in\NN}T\circ D\circ E_n \leq T \circ D \circ E_\infty$. The converse inequality follows easily from the fact that $E_n \below E_\infty$ for each $n \in \NN$. Therefore, $\bigwedge_{n\in\NN}T\circ D\circ E_n = T \circ D \circ E_\infty$; in other words, $D \circ E_n \nearrow D \circ E_\infty$. More generally, we conclude that $D$ is Scott continuous.
\end{proof}

\subsection{Quantum families of Scott continuous functions} We show that the monoidal structure on $\cat{qCPO}$ is closed, by distinguishing a subset of the quantum function set $\Y^\X$ that intuitively consists of Scott continuous functions, for all quantum cpos $(\X,R)$ and $(\Y,S)$. Together with the evaluation function $\Eval\: \Y^\X \times \X \To \Y$, the quantum functions set $\Y^\X$ may be regarded as the universal quantum family of functions from $\X$ to $\Y$, and it is essentially constructed as such \cite{Kornell18}*{Definition 9.2}. Similarly, to construct the $\cat{qCPO}$ inner hom $[(\X,R),(\Y,S)]_\uparrow$, we begin by defining a quantum family of Scott continuous functions from $(\X,R)$ to $(\Y,S)$ indexed by a quantum set $\W$ to be a function $F\: \W \times \X \to \Y$ that is Scott continuous in the second variable.

\begin{definition}
Let $\W$, $\X$ and $\Y$ be quantum sets. A \emph{quantum family of functions} $\X \to \Y$ indexed by $\W$ is just a function $\W \times \X \To \Y$. Furthermore, if $(\X,R)$ and $(\Y,S)$ are quantum cpos, then a \emph{quantum family of Scott continuous functions} from $(\X,R)$ to $(\Y,S)$ is a function $\W \times \X \to \Y$ that is Scott continuous in the second variable, or equivalently, a Scott continuous function $(\W, I_\W) \times (\X, R) \to (\Y, S)$. 
\end{definition}

\begin{lemma}\label{lem:rabarber}
Let $(\X,R)$ and $(\Y,S)$ be quantum cpos, and let $\V$ and $\W$ be quantum sets. Let $F:\W\times\X\to\Y$ be a quantum family of functions, and let $G:\V\to\W$ be a surjective function. If $F\circ (G\times I_\X)$ is Scott continuous in the second variable, then so is $F$.
\end{lemma}

\begin{proof}
First, we verify that $F$ is monotone in the second variable. Since $F\circ (G\times I_\X)$ is monotone in the second variable, we have that
$ F\circ (G\times I_\X)\circ (I_\V\times R)\leq S\circ F\circ (G\times I_\X)$. Hence,
$
    F\circ (I_\W\times R)=F\circ (I_\W\times R)\circ(G\times I_\X)\circ  (G^\dag\times I_\X)
=  F\circ (G\times I_\X)\circ (I_\V\times R)\circ (G^\dag\times I_\X)  
    \leq S\circ F\circ (G\times I_\X)\circ (G^\dag\times I_\X)=S\circ F.
$ Therefore, $F$ is monotone in the second variable.

Now, let $\H$ be an atomic quantum set, and let $K_1\sqsubseteq K_2\sqsubseteq\cdots$ be a monotonically ascending sequence of functions $\H\to\X$ with limit $K_\infty$. We calculate that
\begin{align*} 
\bigwedge_{n \in \NN} S \circ F \circ (I_\W \times K_n) &
= 
\left(\bigwedge_{n \in \NN} S \circ F \circ (I_\W \times K_n) \right) \circ (G \times I_\H) \circ (G^\dagger \times I_\H)
\\ & =
\left(\bigwedge_{n \in \NN} S \circ F \circ (I_\W \times K_n) \circ (G \times I_\H) \right) \circ (G^\dagger \times I_\H)
\\ & =
\left(\bigwedge_{n \in \NN} S \circ F \circ (G \times I_\X)\circ (I_\V \times K_n)  \right) \circ (G^\dagger \times I_\H)
\\ & =
S \circ F \circ (G \times I_\X)\circ (I_\V \times K_\infty) \circ (G^\dagger \times I_\H)
 =
S \circ F \circ (I_\V \times K_\infty),
\end{align*}
where we apply \cite[Proposition~A.6]{KLM20} in the second equality and the Scott continuity in the second variable of $F\circ (G\times I_\X)$ in the fourth equality. We conclude that $F$ is Scott continuous in the second variable.
\end{proof}

\begin{proposition}\label{prop:maximum Scott continuous subset}
Let $(\X,R)$ and $(\Y,S)$ be quantum cpos. Let  
$\mathfrak A$ be the set of all subsets $\V\subseteq[\X,\Y]_\sqsubseteq$ \cite[Theorem~8.3]{KLM20} such that $\Eval_{\sqsubseteq}\circ (J_\V\times I_\X):\V\times\X\to\Y$ is Scott continuous in the second variable. Then $\mathfrak A$ has a greatest element.
\end{proposition}
\begin{proof}
Let $\W=\bigcup\mathfrak A$. We will show that $\W\in\mathfrak A$, i.e., that $\Eval_\sqsubseteq\circ (J_{\W}\times I_\X)$ is Scott continuous in the second variable. First, we show that $\Eval_\sqsubseteq\circ (J_\W\times I_\X)$ is monotone in the second variable. Fix $W \atomof \W$, and let $\V \in \mathfrak A$ be such that $W \atomof \V$. Let $J_W^\V\: \Q\{W\} \hookrightarrow \V$ and $J_W^\W\: \Q\{W\} \hookrightarrow \W$ be the inclusion functions of $\Q\{\W\}$ into $\V$ and $\W$, respectively. Similarly, let $J_\V\: \V \hookrightarrow [\X,\Y]_\below$ and $J_\W\: \W \hookrightarrow [\X,\Y]_\below$ be the inclusion functions of $\V$ and $\W$, respectively, into $[\X,\Y]_\below$. We certainly have that $J_\V \circ J_W^\V = J_\W \circ J_W^\W$, since both sides of the equality are just the inclusion function $\Q\{W\} \hookrightarrow [\X, \Y]_\below$. We now compute that
\begin{align*}&
\Eval_\below \circ (J_\W \times I_\X) \circ (I_\W \times R) \circ (J_W^\W \times I_\X)
=
\Eval_\below \circ (I \times R) \circ (J_\W \times I_\X) \circ (J_W^\W \times I_\X)
\\ & =
\Eval_\below \circ (I \times R) \circ (J_\V \times I_\X) \circ (J_W^\V \times I_\X)
 =
\Eval_\below \circ (J_\V \times I_\X) \circ (I_\V \times R) \circ (J_W^\V \times I_\X)
\\ & \leq
S \circ \Eval_\below \circ (J_\V \times I_\X) \circ (J_W^\V \times I_\X)
 =
S \circ \Eval_\below \circ (J_\W \times I_\X) \circ (J_W^\W \times I_\X).
\end{align*}
Thus, $\Eval_\below \circ (J_\W \times I_\X) \circ (I_\W \times R) \circ (J_W^\W \times I_\X) \leq S \circ \Eval_\below \circ (J_\W \times I_\X) \circ (J_W^\W \times I_\X)$. We vary $W \atomof \W$ to conclude that $\Eval_\below \circ (J_\W \times I_\X) \circ (I_\W \times R) \leq S \circ \Eval_\below \circ (J_\W \times I_\X)$. Therefore, $\Eval_\below \circ (J_\W \times I_\X)$ is monotone in the second variable.

Let $\H$ be an atomic quantum set, and let $K_1 \below K_2 \below \cdots$ be a monotonically ascending sequence of functions $\H \to \X$, with limit $K_\infty$. Fix $W \atomof \W$, and let $\V \in \mathfrak A$ be such that $W \atomof \V$. We now compute that
\begin{align*}
\bigg( \bigwedge_{n \in \NN} S \circ \Eval_\below \circ (J_\W \times I_\X)& \circ (I_\W \times K_n) \bigg) \circ (J_W^\W \times I_\H)
\\ & =
\bigwedge_{n \in \NN} S \circ \Eval_\below \circ (J_\W \times I_\X) \circ (I_\W \times K_n) \circ (J_W^\W \times I_\H)
\\ & =
\bigwedge_{n \in \NN} S \circ \Eval_\below \circ (J_\V \times I_\X) \circ (I_\V \times K_n) \circ (J_W^\V \times I_\H)
\\ & = 
\bigg( \bigwedge_{n \in \NN} S \circ \Eval_\below \circ (J_\V \times I_\X) \circ (I_\V \times K_n) \bigg) \circ (J_W^\V \times I_\H)
\\ & =
S \circ \Eval_\below \circ (J_\V \times I_\X) \circ (I_\V \times K_\infty) \circ (J_W^\V \times I_\H)
\\ & =
S \circ \Eval_\below \circ (J_\W \times I_\X) \circ (I_\W \times K_\infty) \circ (J_W^\W \times I_\H),
\end{align*}
where we appeal to \cite[Proposition~A.6]{KLM20} for the first and third equalities. We vary $W \atomof \W$ to conclude that $\bigwedge_{n \in \NN} S \circ \Eval_\below \circ (J_\W \times I_\X) \circ (I_\W \times K_n) = S \circ \Eval_\below \circ (J_\W \times I_\X) \circ (I_\W \times K_\infty)$. In other words, $\Eval_\below \circ (J_\W \times I_\X) \circ (I_\W \times K_n) \nearrow \Eval_\below \circ (J_\W \times I_\X) \circ (I_\W \times K_\infty)$. Thus, $\Eval_\below \circ (J_\W \times I_\X)$ is Scott continuous in the second variable, and therefore, $\W \in \mathfrak A$ as claimed.
\end{proof}

Proposition \ref{prop:maximum Scott continuous subset} ensures that the subset $[\X, \Y]_\uparrow \subsetof [\X,\Y]_\below$ is well defined:

\begin{definition}\label{def:inner hom qCPO}
Let $\X$ and $\Y$ be quantum cpos. We define $[\X,\Y]_\uparrow$ to be the largest subset $\W$ of $[\X,\Y]_\sqsubseteq$ such that $\Eval_{\sqsubseteq}\circ(J_\W\times I_\X)$ is Scott continuous in the second variable.
\end{definition}

In this context, we let $J$ be the inclusion function $[\X,\Y]_\uparrow\hookrightarrow [\X,\Y]_\sqsubseteq$, and we equip $[\X,\Y]_\uparrow$ with the induced order $J^\dag \circ Q \circ J$ \cite[Definition~2.2, Theorem~8.3]{KLM20}. Furthermore, we define $\Eval_\uparrow:[\X,\Y]_\uparrow\times\X\to\Y$ to be the composition $\Eval_{\sqsubseteq}\circ(J\times I_\X)$. 

\begin{lemma}\label{lem:step 1 monoidal closure qCPO}
Let $(\X,R)$, $(\Y,S)$ and $(\Z,T)$ be quantum cpos. Let $F:\Z\times \X\to\Y$ be a Scott continuous function. Then there is a unique monotone function $G:\Z\to[\X,\Y]_\uparrow$ such that the following diagram commutes:
\[\begin{tikzcd}
\Z \times \X
\arrow{rrrd}{F}
\arrow[dotted]{d}[swap]{G \times I_\X}
&
&
&
\\
\phantom{}[\X,\Y]_{\uparrow} \times \X
\arrow{rrr}[swap]{\mathrm{Eval_{\uparrow}}}
&
&
&
\Y
\end{tikzcd}\]
\end{lemma}
\begin{proof}
By \cite[Theorem~8.3]{KLM20}, which expresses that $\qPOS$ is monoidal closed, there exists a unique monotone $\tilde G:\Z\to[\X,\Y]_\sqsubseteq$ such that $F=\Eval_\sqsubseteq\circ ( \tilde G\times I_\X)$. We factor $\tilde G$ through its range $\V: = \ran \tilde G$, writing $\tilde G = J_\V \circ \bar G$ \cite[Definition~3.2]{KLM20}. The function $F = \Eval_\sqsubseteq\circ (\tilde G\times I_\X)=\Eval_\sqsubseteq\circ ( J_\V\times I_\X)\circ (\bar G\times I_\X)$ is given to be Scott continuous, and in particular, it is Scott continuous in the second variable. Since $\bar G$ is surjective, we apply Lemma \ref{lem:rabarber} to conclude that $\Eval_\sqsubseteq\circ ( J_\V\times I_\X)$ is Scott continuous in the second variable. Therefore, by the definition $[\X,\Y]_\uparrow$, we have that $\V \subsetof [\X,\Y]_\uparrow$.

Writing $J$ for the inclusion $[\X, \Y]_\uparrow \hookrightarrow [\X, \Y]_\below$, as before, we clearly have that $J_\V = J \circ J_\V^{[\X, \Y]_\uparrow}$. Hence, defining $G:= J_\V^{[\X, \Y]_\uparrow} \circ \bar G$, we have that $J \circ G = J \circ J_\V^{[\X, \Y]_\uparrow} \circ \bar G = J_\V \circ \bar G = \tilde G$. We argue that $G$ is monotone. The function $\tilde G$ is certainly monotone, so $J\circ G\circ T=\tilde G\circ T\leq Q\circ \tilde G=Q\circ J\circ G$, where $Q$ is the canonical order on $[\X, \Y]_\below$. The injectivity of $J$ now implies that $G\circ T\leq J^\dag\circ Q\circ J\circ G,$
which expresses that $G$ is monotone, too.

By construction, $\Eval_\uparrow\circ (G\times I_\X ) = \Eval_\sqsubseteq \circ (J\times I_\X )\circ (G\times I_\X ) = \Eval_\sqsubseteq\circ (\tilde G\times I_\X) = F$. The function $G$ is also the unique function with this property, simply because $J$ is injective. Indeed, if $G_1$ and $G_2$ are two functions $\Z \to [\X,\Y]_\uparrow$ making the diagram commute, then we have that $F = \Eval_\uparrow \circ (G_i \times I_\X) = \Eval_\below \circ ((J \circ G_i) \times I_\X)$ for each $i \in \{1,2\}$. The universal property of $\Eval_\below$ implies that $J \circ G_1 = J \circ G_2$, and the injectivity of $J$ is then implies that $G_1 = G_2$. Therefore, $G$ is the unique monotone function making the diagram commute, as claimed.
\end{proof}

\begin{proposition}\label{prop:relating Q and W}
Let $(\X,R)$ and $(\Y,S)$ be quantum posets. Let $Q$ be the order on $[\X,\Y]_\sqsubseteq$ as in \cite[Theorem 8.3]{KLM20}. Let $W$ be a relation on $[\X,\Y]_\sqsubseteq$ such that one of the following (equivalent) conditions hold:
\begin{itemize}
    \item[(a)] $W\times I_\X\leq\Eval_\sqsubseteq^\dag\circ S\circ\Eval_\sqsubseteq$;
    \item[(b)] $\Eval_\sqsubseteq\circ (W\times I_\X)\leq S\circ \Eval_\sqsubseteq$.
\end{itemize}
Then $W\leq Q$.
\end{proposition}
\begin{proof}
Since $I_{[\X,\Y]_\sqsubseteq}\leq \Eval_\sqsubseteq^\dag\circ  \Eval_\sqsubseteq$, the conditions (a) and (b) are equivalent.

In order to prove that any of these conditions imply $Q'\leq Q$, we first note that by \cite[Theorem 8.3]{KLM20},  $Q$ is the largest pre-order on $[\X,\Y]_\sqsubseteq$ such that 
\begin{equation}\label{ineq:Q def2}
\Eval_{\sqsubseteq}\circ (Q\times R)\leq S\circ \Eval_{\sqsubseteq},
\end{equation} holds. Assume that condition (b) holds. Then
\begin{align*}
    \Eval_\sqsubseteq\circ (W\times R) & = \Eval_\sqsubseteq\circ (W\times I_\X)\circ(I_{[\X,\Y]_\sqsubseteq}\times R)\leq S\circ \Eval_\sqsubseteq\circ(I_{[\X,\Y]_\sqsubseteq}\times R)\\
    & \leq S\circ\Eval_\sqsubseteq\circ(Q\times R)\leq S\circ S\circ \Eval_\sqsubseteq\leq S\circ\Eval_\sqsubseteq,
\end{align*}
where the first inequality is (b), the second inequality follows from \cite[Lemma A.3]{KLM20}, and the penultimate inequality follows since $\Eval_\sqsubseteq$ is monotone (see also the identity (\ref{ineq:Q def2})).
Since $Q$ is the largest pre-order on $[\X,\Y]_\sqsubseteq$ such that (\ref{ineq:Q def2}) holds, it follows that $W\leq Q$.
\end{proof}

\begin{proposition}\label{prop:step 2 monoidal closure qCPO}
Let $(\X,R)$ and $(\Y,S)$ be quantum cpos. Then $([\X,\Y]_{\uparrow},J^\dag\circ Q\circ J)$ is a quantum cpo, and $\Eval_{\uparrow}\:[\X,\Y]_\uparrow\times\X\to\Y$ is Scott continuous. 
\end{proposition}
\begin{proof}
The identity $I_\X$ is trivially monotone, and the inclusion $J$ is monotone by \cite[Lemma~2.4]{KLM20}. Since $\qPOS$ is a monoidal category \cite[Theorem~7.5]{KLM20}, it follows that $J\times I_\X$ is also monotone. Furthermore, since $\Eval_\sqsubseteq$ is monotone \cite[Theorem~8.3]{KLM20}, and monotone functions compose, we conclude that $\Eval_{\uparrow}$ is monotone. 

Let $\H$ be an atomic quantum set and $K_1\sqsubseteq K_2\sqsubseteq\cdots$ be a monotonically ascending sequence of functions $\H\to[\X,\Y]_\uparrow$. We will show that this sequence has a limit $K_\infty$, to establish that $[\X,\Y]_\uparrow$ is a quantum cpo. We will also show that $\Eval_\uparrow\circ (K_n\times I_\X)\nearrow\Eval_\uparrow\circ (K_\infty\times I_\X)$, to establish that $\Eval_\uparrow$ is Scott continuous in the first variable. By definition of $[\X,\Y]_\uparrow$, we already know that $\Eval_\uparrow$ is Scott continuous in the second variable. Thus, we will conclude by Proposition \ref{prop:Scott continuous iff Scott continuous in both variables} that $\Eval_\uparrow$ is Scott continuous.

First, we define $K_\infty$ and show that $\Eval_\sqsubseteq\circ (J\times I_\X)\circ (K_n\times I_\X)\nearrow \Eval_\sqsubseteq\circ (J\times I_\X)\circ (K_\infty\times I_\X)$. 
For brevity, we write $F_n=\Eval_\uparrow \circ (K_n\times I_\X)$ for each $n\in\NN$.
By \cite[Corollary~7.7]{KLM20}, we have that
$ K_1\times I_\X\sqsubseteq K_2\times I_\X\sqsubseteq\cdots.$
By the monotonicity of $\Eval_\uparrow$, we have that
$F_1\sqsubseteq F_2\sqsubseteq \cdots,$
and since the codomain of the functions $F_n$ is the quantum cpo $\Y$, it follows from Proposition \ref{prop:suprema of sequences of functions} that $F_n\nearrow F_\infty$ for some function $F_\infty:\X\times\H\to\Y$.

We equip $\H$ with the trivial order $I_\H$. Then given another atomic quantum set $\H'$, any monotonically ascending sequence $G_1\sqsubseteq G_2\sqsubseteq\cdots$ of functions $\H' \to \H$ is constant, as in Example \ref{ex:trivially ordered set is qcpo}. If $T$ denotes the order on $[\X,\Y]_\uparrow$, then for each $n\in\NN$ we have  
$\bigwedge_{k\in\NN}T\circ F_n \circ (G_k\times I_\X)=T \circ F_n \circ(G_1\times I_\X)$. Hence, $F_n$ is automatically Scott continuous in the first variable.
Moreover, it is straightforward to see that the functions $K_n$ are Scott continuous. 
It follows from Lemma \ref{lem:D Scott continuous in first variable and F Scott continuous then DFtimes G Scott continuous in first variable} that each function $F_n=\Eval_\uparrow\circ (K_n\times I_\X)$ is also Scott continuous in the second variable. Thus, each function $F_n$ is Scott continuous (Proposition \ref{prop:Scott continuous iff Scott continuous in both variables}). We conclude that, by Proposition \ref{prop:limit of functions is Scott continuous}, its limit $F_\infty:\H\times\X\to\Y$ is Scott continuous, too. We now apply Lemma \ref{lem:step 1 monoidal closure qCPO} to conclude that there is a function $K_\infty:\H\to[\X,\Y]_\uparrow$ such that
$F_\infty =\Eval_\uparrow\circ (K_\infty\times I_\X)$.

We show that $K_n \nearrow K_\infty$. Fix $n\in\NN$. Since $F_n\sqsubseteq F_\infty$, we have that $F_\infty\leq S\circ F_n$, whence we obtain that
\begin{align*}
    K_\infty \times I_\X & \leq \Eval_\uparrow^\dag\circ\Eval_\uparrow \circ (K_\infty\times I_\X )= \Eval_\uparrow^\dag\circ F_\infty\\
    & \leq\Eval_\uparrow^\dag\circ S\circ F_n =\Eval_\uparrow^\dag\circ S\circ \Eval_{\uparrow}\circ (K_n\times I_\X )\\
    & =(\Eval_\sqsubseteq\circ (J\times I_\X ))^\dag\circ S\circ \Eval_\sqsubseteq\circ (J\times I_\X )\circ (K_n\times I_\X )\\
    & = (J^\dag\times I_\X )\circ\Eval_\sqsubseteq^\dag\circ S\circ \Eval_\sqsubseteq\circ ((J\circ K_n)\times I_\X ).
\end{align*}
Appealing to the definition of a function between quantum sets, we conclude that
\begin{align*}
(J\circ K_\infty\circ K_n^\dag\circ J^\dag)\times I_\X \leq \Eval_\sqsubseteq^\dag\circ S\circ \Eval_\sqsubseteq.
\end{align*}
It now follows from Proposition \ref{prop:relating Q and W} that
$J\circ K_\infty\circ K_n^\dag\circ J^\dag\leq Q$ for each $n\in\NN$.
Again appealing to the definition of a function between quantum sets, we infer that
$J\circ K_\infty\leq Q\circ J\circ K_n$;
hence,
\begin{equation}\label{ineq:qidentity}
J^\dag\circ Q\circ J\circ K_\infty\leq J^\dag\circ Q\circ Q\circ J\circ K_n\leq J^\dag\circ Q\circ J\circ K_n.    
\end{equation}
The binary relation $J^\dag\circ Q\circ J$ is of course the canonical order on $[\X,\Y]_\uparrow$, so $K_n\sqsubseteq K_\infty$. We vary $n \in \NN$ to conclude that $J^\dag\circ Q\circ J\circ K_\infty\leq \bigwedge_{n\in\NN}J^\dag\circ Q\circ J\circ K_n$, either from this inequality or directly from Inequality (\ref{ineq:qidentity}).

To establish that $K_n \nearrow K_\infty$, it remains to show that $J^\dag\circ Q\circ J\circ K_\infty\geq \bigwedge_{n\in\NN}J^\dag\circ Q\circ J\circ K_n$. The function $\Eval_\sqsubseteq$ is monotone, and in particular, it is monotone in the first variable (Lemma \ref{lem:monotone in both variables iff monotone}); hence $
Q\times I_\X\leq \Eval_{\sqsubseteq}^\dag \circ S\circ\Eval_\sqsubseteq$.
We now calculate that
\begin{align*}
    \left(\bigwedge_{n\in\NN} J^\dag\circ Q\circ J\circ K_n\right)\times I_\X & = \bigwedge_{n\in\NN}(J^\dag\times I_\X)\circ (Q\times I_\X)\circ (J\times I_\X)\circ (K_n\times I_\X)\\
&     \leq  \bigwedge_{n\in\NN}(J^\dag\times I_\X)\circ \Eval_\sqsubseteq^\dag\circ S\circ\Eval_\sqsubseteq\circ (J\times I_\X)\circ (K_n\times I_\X)\\
&     =
\bigwedge_{n\in\NN}(J\times I_\X)^\dag\circ \Eval_\sqsubseteq^\dag\circ S\circ F_n\\
& =
(J\times I_\X)^\dag\circ \Eval_\sqsubseteq^\dag\circ \bigwedge_{n\in\NN} S\circ F_n\\
& = 
(J\times I_\X)^\dag\circ \Eval_\sqsubseteq^\dag\circ S\circ F_\infty
\\
& = 
(J^\dag\times I_\X)\circ \Eval_\sqsubseteq^\dag\circ S\circ \Eval_\sqsubseteq\circ (J\times I_\X)\circ (K_\infty\times I_\X),
\end{align*}
where we use \cite[Lemma~A.3]{KLM20} in the first equality and \cite[Proposition~A.6]{KLM20} in the third equality. Appealing to the definition of a function between quantum sets, we infer that 
\[  \left(J\circ \left( \bigwedge_{n\in\NN} J^\dag\circ Q\circ J\circ K_n \right) \circ K_\infty^\dag\circ J^\dag \right)\times I_\X\leq \Eval_\sqsubseteq^\dag\circ S\circ \Eval_\sqsubseteq.\]
Applying Proposition \ref{prop:relating Q and W}, we obtain
$J\circ \left( \bigwedge_{n\in\NN} J^\dag\circ Q\circ J\circ K_n \right) \circ K_\infty^\dag\circ J^\dag\leq Q.$ Appealing to the injectivity of $J$, we reason that 
\[ \left( \bigwedge_{n\in\NN} J^\dag\circ Q\circ J\circ K_n \right) \circ K_\infty^\dag = J^\dag\circ J\circ  \left( \bigwedge_{n\in\NN} J^\dag\circ Q\circ J\circ K_n \right) \circ K_\infty^\dag\circ J^\dag\circ J\leq J^\dag\circ Q\circ J.\]
Therefore, $\bigwedge_{n\in\NN} J^\dag\circ Q\circ J\circ K_n \leq J^\dag\circ Q\circ J\circ K_\infty$, and more precisely, $\bigwedge_{n\in\NN} J^\dag\circ Q\circ J\circ K_n = J^\dag\circ Q\circ J\circ K_\infty$.

We conclude that $K_n \nearrow K_\infty$ and, more generally, that $[\X,\Y]_\uparrow$ is a quantum cpo. Furthermore, $\Eval_\uparrow\circ (K_n\times I_\X)=F_n\nearrow F_\infty=\Eval_\uparrow\circ (K_\infty\times I_\X)$, and thus, $\Eval_\sqsubseteq$ is Scott continuous in the first variable. Therefore, it is Scott continuous.
\end{proof}

\begin{theorem}\label{thm:qCPO is monoidal closed}
The category $\qCPO$ is monoidal closed with respect to the monoidal product $\times$, i.e., for each triple of quantum cpos $(\X,R)$,  $(\Y,S)$, and $(\Z,T)$ and for each Scott continuous function $F:\Z\times \X\to \Y$,
there is a Scott continuous function $G \: \Z \To  [\X,\Y]_{\uparrow}$ such that the following diagram commutes.
\[\begin{tikzcd}
\Z \times \X
\arrow{rrrd}{F}
\arrow[dotted]{d}[swap]{G \times I_\X}
&
&
&
\\
\phantom{}[\X,\Y]_{\uparrow} \times \X
\arrow{rrr}[swap]{\mathrm{Eval_{\uparrow}}}
&
&
&
\Y
\end{tikzcd}\]
\end{theorem}
\begin{proof}
Let $F:\Z\times\X\to\Y$ be Scott continuous. By Lemma \ref{lem:step 1 monoidal closure qCPO}, there is a unique monotone function $G:\Z\to[\X,\Y]_\uparrow$ such that the diagram in the statement of the theorem commutes. It remains only to show that $G$ is Scott continuous. 

Let $H$ be an atomic quantum set, and let $K_1\sqsubseteq K_2\sqsubseteq\cdots$ be a monotonically ascending sequence of functions $\H\to\Z$, with limit $K_\infty$. For each $n\in\NN$, let $\tilde K_n=G\circ K_n$. Since $G$ is monotone, 
$\tilde K_1\sqsubseteq \tilde K_2\sqsubseteq\cdots$ is a monotonically ascending sequence of functions $\H\to[\X,\Y]_\uparrow$, which must have some limit $\tilde K_\infty$. By \cite[Corollary~7.7]{KLM20},
$K_1\times I_\X\sqsubseteq K_2\times I_\X\sqsubseteq\cdots $
is a monotonically ascending sequence of functions $\H\times\X\to\Z\times\X$ whose limit is $K_\infty\times I_\X$ by Lemma \ref{lem:tensor product is Scott continuous bifunctor}, and similarly,
$ \tilde K_1\times I_\X\sqsubseteq \tilde K_2\times I_\X\sqsubseteq\cdots$
is a monotonically increasing sequence of functions $:\H\times\X\to[\X,\Y]_\uparrow\times\X$ whose limit $\tilde K_\infty\times I_\X$.  

The function $F$ is Scott continuous by assumption, and $\Eval_\uparrow$ is Scott continuous by Proposition \ref{prop:step 2 monoidal closure qCPO}. Hence, it follows from Lemma \ref{lem:left multiplication by Scott continuous function is Scott continuous}
that
\[F\circ (K_n\times I_\X)\nearrow F\circ (K_\infty\times I_\X), \quad \text{and} \quad\Eval_\uparrow\circ (\tilde K_n\times I_\X)\nearrow \Eval_\uparrow\circ (\tilde K_\infty\times I_\X).\]
Since
$F\circ (K_n\times I_\X)=\Eval_\uparrow\circ (G\times I_\X)\circ (K_n\times I_\X)=\Eval_\uparrow\circ (\tilde K_n\times I_\X)$ for each $n \in \NN$, it follows that 
$F\circ(K_\infty\times I_\X)=\Eval_\uparrow\circ(\tilde K_\infty\times I_\X)$ and, therefore, that
$ \Eval_\uparrow\circ((G\circ K_\infty)\times I_\X)=\Eval_\uparrow(G\times I_\X)\circ (K_\infty\times I_\X)=F\circ (K_\infty\times I_\X)=\Eval_\uparrow\circ (\tilde K_\infty\times I_\X).$ We thus conclude that $\Eval_\uparrow\circ((G\circ K_\infty)\times I_\X) = \Eval_\uparrow\circ (\tilde K_\infty\times I_\X)$.

Equip $\H$ with the trivial order $I_\X$. Writing $J$ for the inclusion function $[\X,\Y]_\uparrow \hookrightarrow [\X,\Y]_\below$, we have that $\Eval_\below\circ((J \circ G\circ K_\infty)\times I_\X) = \Eval_\below\circ ((J \circ \tilde K_\infty)\times I_\X)$. The functions $J \circ G\circ K_\infty$ and $J \circ \tilde K_\infty$ are both trivially monotone, so they are equal by the universal property of $\Eval_\below\: [\X,\Y]_\below \times \X \to \Y$. Because $J$ is an injection we conclude $G\circ K_\infty = \tilde K_\infty$. Thus, $G \circ K_n = \tilde K_n \nearrow \tilde K_\infty = G \circ K_\infty$. Therefore, $G$ is Scott continuous.
\end{proof}

\section{Inclusion of $\cat{CPO}$ into $\cat{qCPO}$}\label{sec:cpos are qcpos}

As we observed in \cite[Example~1.3]{KLM20}, the pair $(`S, `\below)$ is a quantum poset for every ordinary poset $(S, \below)$. We proceed to show that if $(S, \below)$ is a cpo, then $(`S, `\below)$ is a quantum cpo. The converse implication holds essentially by construction (Section \ref{Quantum cpos}), so we will have shown that the notion of a quantum cpo is indeed a quantum generalization of the ordinary notion.

We will also show that the ``inclusion'' functor $(S, \below) \mapsto (`S, `\below)$ from $\cat{CPO}$ to $\cat{qCPO}$ has a right adjoint, which takes each quantum cpo $(\X,R)$ to the ordinary cpo whose points are the one-dimensional atoms of $\X$, ordered according to $R$. For denotational semantics, the existence of this right adjoint assures that $\qCPO$ forms a linear/nonlinear model; the comonad induced by the adjunction can be used to model the exponential modality $!$ of linear logic. See also Definition \ref{def:lnlmodel}.

\subsection{Ordinary cpos as quantum cpos} Let $\H = \Q\{H\}$ be an atomic quantum set, and let $S$ be an ordinary set. A function $F\: \H \to `S$ is essentially an $S$-valued observable on $H$: it is uniquely determined by an indexed family $\{H_s\}_{s \in S}$ of pair-wise orthogonal subspaces of $H$ that together sum to $H$. This one-to-one correspondence may be defined by $H_s = H\cdot F^\star(\delta_s)$ for all $s \in S$, where $\delta_s \in \ell^\infty(`S)$ is  the minimal projection corresponding to $s$, which is the indicator function on the subset $\{s\}\subseteq S$ if we identify $\ell^\infty(`S)$ with $\ell^\infty(S)$, and where $F^\star$ is defined as in  Theorem \ref{thm:qSet}. Intuitively, the unit vectors in $H_s$ represent those states of the quantum system represented by $H$ that yield the outcome $s \in S$ with probability one, see also Example \ref{ex:measurement}.

Let $(S, \below)$ be a poset, and let $F$ and $F'$ be functions $\H \to `S$, intuitively, two $S$-valued observables. As we essentially show in Proposition \ref{classical.E}, the inequality $F \below F'$ is equivalent to the condition $H_{s} \perp H_{s'}'$ for all $s, s'\in S$ such that $s \not \below s'$. Thus, intuitively, $F \below F'$ if and only if the two observables are such that if we first measure $F$ and then we measure $F'$, then the outcome $s$ of the first measurement must be below the outcome $s'$ of the second measurement.

For illustration, we briefly examine the special case $S = \RR$, with $\below$ being the standard order. Functions $\H \to `\RR$ of course correspond to self-adjoint operators on $H$, and the order that we have defined is then just the spectral order on $H$ \cite{Olson}*{Definition 1}. Thus, our order on functions $\H \to `S$, for general $S$, is a generalization of Olson's spectral order.

\begin{definition}\label{classical.A}
Let $H$ be a Hilbert space. A \emph{partition} on $H$ is a set of nonzero projection operators on $H$ whose sum is the identity operator $1_H$. A projection is said to be a \emph{cell} of a partition $P$ if it is an element of $P$. Similarly, a projection is said to be a \emph{multicell} of a partition $P$ if it is a sum of elements of $P$, possibly an empty sum or an infinite sum.
\end{definition}

The partitions on a Hilbert space $H$ are obviously analogous to the partitions of an ordinary set $S$. Such a partition might be more naturally defined to be a decomposition of $H$ into a set of pairwise-orthogonal closed subspaces, but our equivalent definition in terms of projections is more convenient here. We remark that the projections in a partition are necessarily mutually orthogonal, otherwise their sum cannot equal $1_H$.

\begin{definition}\label{classical.B}
Let $H$ be a Hilbert space, and let $P_1$ and $P_2$ be partitions on $H$. We define $P_1 \leq P_2$ if each projection $p_1 \in P_1$ is contained in some projection $p_2 \in P_2$. The partitions on $H$ are thus partially ordered.
\end{definition}

Hence, the inequality $P_1 \leq P_2$ expresses that $P_1$ is finer than $P_2$. This partial order has a greatest element, the singleton partition, but it has no least element. In the pursuit of clarity, we specialize our terminology to the various partial orders that appear in this section. If $s_1$ and $s_2$ are the elements of a poset $(S, \below)$, then we say that $s_1$ is below $s_2$ to mean that $s_1 \below s_2$. If $p_1$ and $p_2$ are projection operators on a Hilbert space $H$, then we say that $p_1$ is contained in $p_2$ to mean that $p_1 \leq p_2$; by definition, this holds if and only if $p_1p_2=p_1$. If $P_1$ and $P_2$ are partitions on a Hilbert space $H$, then we say that $P_1$ refines $P_2$ to mean that $P_1 \leq P_2$. 
Next, we give a different characterization of the refinement relation. First, we say that  projection operators $p$ and $q$ on a Hilbert space $H$ are orthogonal if $pq=0$, or equivalently, if $p\leq (1_H-q)$, in which case we write $p\perp q$.
\begin{lemma}
Let $H$ be a finite-dimensional Hilbert space, and let $P_1$ and $P_2$ be partitions on $H$. Then $P_1$ refines $P_2$ if and only if every cell of $P_2$ is a multicell of $P_1$.    
\end{lemma}
\begin{proof}
    We first make the following observation. If $P$ is a partition on $H$, and $q$ an arbitrary projection operator on $H$, then $\sum P=1_H$ forces that $q\not\perp p$ for some $p\in P$. It follows that $q\leq \sum\{p\in P:p\not\perp q\}=1_H-\sum\{p\in P:p\perp q\}$.

    Now assume that $P_1\leq P_2$. For any $p\in P_1$, we have $p\leq q$ for some $q\in P_2$. If $r\in P_2\setminus\{q\}$, then $q\perp r$, hence $p\leq q\leq 1_H-r$, so also $p\perp r$. So for each $p\in P_1$ and each $q\in P_2$, we either have $p\leq q$ or $p\perp q$. Now, let $q\in P_2$. Then $q\leq\sum\{p\in P_1:p\not\perp q\}=\sum\{p\in P_1:p\leq q\}\leq q$, whence $q$ is a multicell of $P_1$. 

    For the converse, assume that every cell of $P_2$ is a multicell of $P_1$, and let $p\in P_1$. We have $p\not\perp q$ for some $q\in P_2$, and by assumption, $q=p_1+\cdots +p_n$ for some cells $p_1,\ldots,p_n\in P_1$. Then $p\not\perp p_i$ for some $i\in\{1,\ldots,n\}$, and since all elements of $P_1$ are mutually orthogonal, it follows that $p=p_i\leq p_1+\cdots+p_n=q$.
\end{proof}

\begin{definition}\label{classical.C}
Let $\H = \Q\{H\}$ be an atomic quantum set, and let $P$ be a partition on $H$. We define the binary relation $M_P\: \H \To `P$ by $M_P(H, \CC_p) = L(H,\CC_p)\cdot p$.
\end{definition}

It is immediate from the definition that $M_P$ is a surjective function. Indeed, for cells $p_1, p_2 \in P$, we have $(M_P \circ M_P^\dagger) (\CC_{p_1}, \CC_{p_2}) = L(H,\CC_{p_1}) \cdot p_1 \cdot p_2 \cdot L(\CC_{p_2},H) = \delta_{p_1,p_2} \cdot L (\CC_{p_1}, \CC_{p_2}) = I_{`P}(\CC_{p_1}, \CC_{p_2})$, where $\delta$ denotes the characteristic function of the equality relation on $P$. Thus, $M_P \circ M_P^\dagger = I_{`P}$. Similarly, $(M_P^\dagger \circ M_P)(H, H) = \sum_{p \in P} p \cdot L(\CC_{p},H)\cdot L(H,\CC_{p}) \cdot p = \sum_{p \in P} p \cdot L(H, H) \cdot p \geq \sum_{p \in P} p \cdot (\CC \cdot 1_{H})\cdot p = \CC \cdot \sum_{p \in P} p = \CC \cdot 1_{H} = I_{\H}(H, H)$, and thus, $M_P^\dagger \circ M_P \geq I_{\H}$.

\begin{proposition}\label{classical.D}
Let $\H = \Q\{H\}$ be an atomic quantum set, let $S$ be an ordinary set, and let $F$ be a function from $\H$ to $`S$. Let $P$ be the partition on $H$ defined by $$P = \{F^\star(\delta_s) \suchthat s \in S,\,F^\star(\delta_s) \neq 0 \}.$$ Then, there is a unique function $f$ from $P$ to $S$ such that $F = `f \circ M_P$.
\end{proposition}

\begin{proof}
Since $L(H, \CC_s)$ is isomorphic to the dual of $H$, each subspace $F(H, \CC_s) \leq L(H, \CC_s)$ is of the form $F(H, \CC_s) = L(H, \CC_s) \cdot p_s$ for some projection $p_s$ on $H$. Proposition 7.5 of \cite{Kornell18} implies that $p_s = F^\star(\delta_s)$. Define $f \: P \To S$ by $f(p_s) = s$, for all projections $p_s \in P$. Then, for each $s \in S$,
\begin{align*}
(`f \circ M_P)(H, \CC_s)
=
\sum_{p\in P} `f(\CC_p, \CC_s) \cdot M_P(H, \CC_p)
=
L(\CC_{p_s}, \CC_s) \cdot L(H, \CC_{p_s}) \cdot p_s
=
F(H, \CC_s).
\end{align*}
\noindent The second equality holds even if $p_s \not \in P$, because in this case, $p_s = 0$.
Thus, $`f \circ M_P = F$. The uniqueness of $f$ follows from the surjectivity of $M_P$.
\end{proof}

For us, the significance of this proposition is the existence of a factorization $F = `f \circ M_P$ for some partition $P$ on $H$. Many such factorizations in the sequel may be different from the canonical factorization given by Proposition \ref{classical.D}.

\begin{proposition}\label{classical.E}
Let $(S, \sqsubseteq)$ be a partially ordered set, and let $\H = \Q\{H\}$ be an atomic quantum set. Let $P_1$ and $P_2$ be partitions on $H$, and let $f_1\: P_1 \To S$ and $f_2\: P_2 \To S$ be ordinary functions. Then, $`f_1 \circ M_{P_1} \sqsubseteq `f_2 \circ M_{P_2}$ if and only if for all $p_1 \in P_1$ and $p_2 \in P_2$, the condition $p_1 \not \perp p_2$ implies $f_1(p_1) \sqsubseteq f_2(p_2)$.
\end{proposition}

\begin{proof}
The inequality $`f_1 \circ M_{P_1} \sqsubseteq `f_2 \circ M_{P_2}$ is equivalent to the inequality $`f_2 \circ M_{P_2} \leq (`{\sqsubseteq}) \circ `f_1 \circ M_{P_1}$, by our definition of the order on functions, and this inequality is in turn equivalent to  $`f_2 \circ M_{P_2} \circ M_{P_1}^\dagger \leq (`{\sqsubseteq}) \circ `f_1 $, by our definition of a function as a binary relation between quantum sets. The binary relation $M_{P_2} \circ M_{P_1}^\dagger$ is from $`P_1$ to $`P_2$, and it is evidently equal to $`r$, where $r$ is the ordinary binary relation that relates $p_1 \in P_1$ to $p_2 \in P_2$ if and only if $p_1 \not \perp p_2$. Hence, the inequality $`f_1 \circ M_{P_1} \sqsubseteq `f_2 \circ M_{P_2}$ is equivalent to $`f_2 \circ `r \leq (`{\sqsubseteq}) \circ `f_1$, and thus, also to $f_2 \circ r \leq ({\sqsubseteq}) \circ f_1$. This last inequality is obviously equivalent to the claimed implication.
\end{proof}

\begin{lemma}\label{classical.F}
Let $(S, \sqsubseteq)$ be a partially ordered set, let $\H = \Q\{H\}$ be an atomic quantum set, let $P$ be a partition on $H$, and let $f\: P \To S$ be a function. Then, for each $s_0 \in S$, we have $(`({\sqsubseteq}) \circ `f \circ M_P)(H, \CC_{s_0}) = L(H, \CC_{s_0}) \cdot \sum\{ p \in P \suchthat  f(p) \sqsubseteq s_0 \}$.
\end{lemma}

\begin{proof} \quad
\vspace{-4.1ex}
\begin{align*}
(`({\sqsubseteq}) \circ `f \circ M_P)(H, \CC_{s_0})
& =
\sum_{s \in S} \sum_{p \in P} `({\sqsubseteq})(\CC_s, \CC_{s_0}) \cdot `f(\CC_p, \CC_s) \cdot M_P(H, \CC_p)
\\ & =
\sum_{s \sqsubseteq s_0} \sum_{f(p) = s} L(\CC_s, \CC_{s_0}) \cdot L(\CC_p, \CC_s) \cdot L(H, \CC_p) \cdot p
\\ & =
\sum_{s \sqsubseteq s_0} \sum_{f(p) = s} L(H, \CC_{s_0}) \cdot p
=
L(H, \CC_{s_0}) \cdot \sum_{s \sqsubseteq s_0} \sum_{f(p) = s} p
\\ &=
L(H, \CC_{s_0}) \cdot \sum\{ p \in P \suchthat  f(p) \sqsubseteq s_0 \} \qedhere
\end{align*}
\end{proof}

\begin{proposition}\label{classical.G}
Let $H$ be a Hilbert space. Each nonempty collection $\P$ of partitions on $H$ has a least upper bound.
\end{proposition}

\begin{proof}
The union $\bigcup \P$ is a set of projections on $\H$. We define a graph structure on $\Union \P$ by defining two projections to be adjacent iff their product is nonzero. We define two projections in $\Union \P$ to be equivalent iff they are in the same component of the graph. Finally, we define $Q_0$ to be the set $\{  \bigvee [p] \suchthat [p] \in \Union \P / \sim \}$.

We claim that $Q_0$ is a partition on $H$. To show pairwise orthogonality, let $[p]$ and $[p']$ be distinct equivalence classes. By definition of the equivalence relation, every projection in $[p]$ is orthogonal to every projecion in $[p']$. Thus, $\bigvee [p]$ is orthogonal to $\bigvee [p']$. The join $\bigvee Q_0$ is equal to the join $\bigvee \bigcup \P$, which is clearly the identity operator $1_H$, simply because $\P$ is a nonempty collection of partitions.

By Definition \ref{classical.B} of the order on partitions, a partition $Q$ is an upper bound for $\P$ if and only if every projection $p$ in $\Union P$ is contained in a projection $q$ in $Q$. This immediately implies that $Q_0$ is an upper bound for $\P$, because $p \leq \bigvee[p]$. To show that $Q_0$ is the least upper bound, let $Q$ be an arbitrary upper bound for $\P$, and let $p$ be an arbitrary projection in $\Union P$. The projection $p$ is a cell of some partition $P \in \P$, which refines $Q$ by assumption, so there is a projection $q$ in $Q$ that contains $p$. Similarly, every projection $p' \in \Union \P$ that is adjacent to $p$ must be contained in some projection $q'$ in $Q$. If $q$ and $q'$ are distinct, then they are orthogonal, implying that $p$ and $p'$ are also orthogonal, and thus contradicting that $p$ and $p'$ are adjacent. Therefore, $p'$ also must be contained in $q$.
By induction on path length, every projection in $\Union \P$ that is equivalent to $p$ must be contained in $q$, so $\bigvee[p] \leq q$. We conclude that $Q_0 \leq Q$, and more generally, that $Q_0$ is the least upper bound of $\P$.
\end{proof}

\begin{lemma}\label{classical.H}
Let $H$ be a nonzero finite-dimensional Hilbert space. Let $q$ be any nonzero projection on $H$, and let $P_1, P_2, \ldots$ be any sequence of partitions on $H$. There exists a sequence of projections $p_1 \not \perp p_2 \not \perp \cdots$ such that $p_i \in P_i$ and $p_i \not \perp q$ for all $i \in \NN$. Moreover, such a sequence exists for any prescribed projection $p_n \in P_n$ that is not orthogonal to $q$, where $n$ is any positive integer.
\end{lemma}

\begin{proof}
The cells of the partition $P_1$ sum to the identity, so they cannot all be orthogonal to $q$. Let $p_1 \in P_1$ be any projection that is not orthogonal to $q$. Suppose by contradiction that every projection in $P_2$ is either orthogonal to $p_1$ or orthogonal to $q$. Then, we can divide $P_2$ into the set of projections orthogonal to $p_1$, and the set of projections orthogonal to $q$. Since $P_2$ is a partition, this implies that $p_1 \perp q$, contradicting our choice of $p_1$. Thus, there exists a projection $p_2 \in P_2$ that is neither orthogonal to $p_1$ nor orthogonal to $q$. Repeating this process recursively, we obtain a sequence of projections $p_1, p_2, \ldots$ with the desired properties.

We chose $p_1$ arbitrarily, but we can also construct such a sequence from any given projection $p_1 \in P_1$ that is not orthogonal to $q$. Furthermore, we can begin constructing such a sequence from any given projection $p_n \in P_n$, for any given positive integer $n$. If we do so, we may choose the projections $p_{n-1}, \ldots, p_1$ via the same process.
\end{proof}

\begin{lemma}\label{classical.I}
Let $H$ be a nonzero finite-dimensional Hilbert space, and let $P_1, P_2, \ldots$ be a sequence of partitions on $H$ with the property that the partition $Q=\sup_{i \geq n} P_i$ is independent of $n$. Let $r_1$ be a multicell of $P_1$, and for each $i>1$, recursively define $r_{i+1}$ to be the smallest multicell of $P_{i+1}$ that is at least as large as $r_i$. Then, the monotone increasing sequence $r_1 \leq r_2 \leq \cdots$ stabilizes at the smallest multicell $s$ of $Q$ that is at least as large as $r_1$.
\end{lemma}

\begin{proof}
The sequence $r_1 \leq r_2 \leq \cdots$ stablizes at some positive integer $n$, because it is a monotone increasing sequence of projections on a finite-dimensional Hilbert space. By induction, each multicell $r_i$ is contained in $s$, because each partition $P_i$ refines $Q$, and therefore every multicell of $Q$ is also a multicell of $P_i$. In particular, $r_n$ is contained in $s$.

By choice of $n$, the sequence $r_n$, $r_{n+1}, \cdots$ is constant. Thus, $r_n$ is a multicell of $P_i$ for each index $i \geq n$. This implies that the partition $\{r_n, 1-r_n\}$ is refined by each such partition $P_i$, so by definition of $Q$, the partition $\{r_n, 1-r_n\}$ is also refined by $Q$. It follows that $r_n$ is a multicell of $Q$. Thus, $r_n$ is a multicell of $Q$ that contains $r_1$ and is also contained in $s$, the smallest multicell of $Q$ that contains $r_1$. Therefore, $r_n = s$.
\end{proof}

\begin{lemma}\label{classical.II}
Let $(S, \sqsubseteq)$ be a cpo, let $\H = \Q\{H\}$ be an atomic quantum set, and let $F_1 \sqsubseteq F_2 \sqsubseteq \cdots$ be a sequence of functions $\H \To `S$. For each index $i \in \NN$, let $P_i$ be a partition on the Hilbert space $H$, and let $f_i\: P_i \To S$ be an ordinary function such that $F_i = `f_i \circ M_{P_i}$. Let $P_\infty = \limsup_{i \in \NN} P_i$. Let $p_\infty \in P_\infty$, and let $p_1 \not \perp p_2 \not \perp \cdots$ be a sequence of projections such that $p_i \in P_i$ and $p_i \not \perp p_\infty$ for all $i \in \NN$. For every positive integer $n$, there exists a positive integer $m \geq n$ such that $f_m(p) \above f_n(p_n)$ for every projection $p \in P_m$ satisfying $p \not \perp p_\infty$.
\end{lemma}

\begin{proof}
For each $n \in \NN$, define $\tilde P_n = \sup_{i \geq n} P_i$. The sequence of partitions $\tilde P_1, \tilde P_2, \ldots$ is decreasing, and it therefore stablizes. Thus, $P_\infty$ is well defined, and furthermore, for all $n$ greater than or equal to some positive integer $n_0$, we have $P_\infty = \sup_{i \geq n} P_i$. Fix $p_\infty \in P_\infty$. Let $p_1 \not \perp p_2 \not \perp \cdots$ be any sequence of projections with $p_i \in P_i$ and $p_i \not \perp p_\infty$ for all $i \in \NN$. By Proposition \ref{classical.E}, we have $f_1(p_1) \sqsubseteq f_2(p_2) \sqsubseteq \cdots$.

Let $n$ be any integer greater than or equal to $n_0$. The projection $p_n$ is a multicell of $P_n$, so by Lemma \ref{classical.I}, there is an increasing sequence of projections $$r_n \leq r_{n+1} \leq  \cdots \leq r_{m} = r_{m+1} = \cdots$$ with $r_n$ equal to $p_n$, with each subsequent $r_i$ equal to the smallest multicell of $P_i$ at least as large as $r_{i-1}$, and with $r_m$ also equal to the smallest multicell of $P_\infty$ at least as large as $r_n$. Since the partition $P_n$ refines $P_\infty$, the cell $p_n\in P_n$ must be contained in $p_\infty$. We conclude that $r_m = p_\infty$.

Each multicell $r_i$ for $i > n$ is clearly the sum of all projections in $P_i$ that are not orthogonal to $r_{i-1}$. Thus, for any projection $q_m \in P_m$ contained in $r_m$, by backward induction on $i$, we can find a sequence of projections $q_n \not \perp \cdots \not \perp q_m$, with $q_n$ equal to $p_n$, and with each $q_i$ in $P_i$. Now, applying Proposition \ref{classical.E}, we conclude that $f_n(p_n) = f_n(q_n) \leq \cdots \leq f_m (q_m)$. Therefore, $f_m(q_m) \above f_n(p_n)$ for every projection $q_m \in P_m$ such that $q_m \leq r_m$.

Let $p$ be any projection in $P_m$ that satisfies $p \not \perp p_\infty$. The partition $P_m$ refines $P_\infty$, so $p \leq p_\infty$. In other words, $p \in P_m$ and $p \leq r_m$. Applying the conclusion of the previous paragraph with $q_m = p$, we infer that $f_m(p) \above f_n(p_n)$. Therefore, for all $n \geq n_0$, there exists $m \geq n$ such that for all $p \in P_m$, if $p \not \perp p_\infty$, then $f_m(p) \above f_n(p_n)$.

Now, let $n$ be any positive integer less than $n_0$, if such a positive integer exists. We have already established that there is a positive integer $m \geq n_0 >n$ such that for all $p \in P_m$, if $p \not \perp p_\infty$, then $f_m(p) \above f_{n_0}(p_{n_0}) \above f_n(p_n)$. Thus, the lemma is proved.
\end{proof}

\begin{lemma}\label{classical.J}
Let $(S, \sqsubseteq)$ be a cpo, let $\H = \Q\{H\}$ be an atomic quantum set, and let $F_1 \sqsubseteq F_2 \sqsubseteq \cdots$ be a sequence of functions $\H \To `S$. For each index $i \in \NN$, let $P_i$ be a partition on the Hilbert space $H$, and let $f_i\: P_i \To S$ be an ordinary function such that $F_i = `f_i \circ M_{P_i}$. Let $P_\infty = \limsup_{i \in \NN} P_i$, and let $p_\infty \in P_\infty$. Then, for each sequence of projections $p_1 \not \perp p_2 \not \perp \cdots$ such that $p_i \in P_i$ and $p_i \not \perp p_\infty$ for all $i \in \NN$, the supremum $\sup_{i\in \NN} f_i(p_i)$ is well defined. Furthermore, it does not depend on the choice of sequence.
\end{lemma}

\begin{proof}
Let $p_1 \not \perp p_2 \not \perp \cdots$ and $p_1' \not \perp p_2' \not \perp \cdots$ be two such sequences of projections. By Proposition \ref{classical.E} the sequences $f_1(p_1), f_2(p_2), \ldots$ and $f_1(p_1'), f_2(p_2'), \ldots$ are  monotone increasing, and by our assumption that $(S, \sqsubseteq)$ is a cpo, $\sup_{i\in \NN} f_i(p_i)$ and $ \sup_{i \in \NN} f_i(p_i')$ are both well defined elements of $S$.

Applying Lemma \ref{classical.II}, we find that for each positive integer $n$, there is a positive integer $m$ such that $f_n(p_n) \below f_m(p_m') \below\sup_{i \in \NN} f_i(p_i')$. Thus, $\sup_{i\in \NN} f_i(p_i) \below \sup_{i\in \NN} f_i(p_i')$. Similarly, $\sup_{i\in \NN} f_i(p_i') \below \sup_{i\in \NN} f_i(p_i)$. Thus, we conclude that the two suprema are equal, and more generally, that the supremum does not depend on our choice of sequence.
\end{proof}

\begin{theorem}\label{classical.K}
Let $(S, \sqsubseteq)$ be a cpo, let $\H = \Q\{H\}$ be an atomic quantum set, and let $F_1 \sqsubseteq F_2 \sqsubseteq \cdots$ be a sequence of functions $\H \To `S$. For each index $i \in \NN$, let $P_i$ be a partition on the Hilbert space $H$, and let $f_i\: P_i \To S$ be an ordinary function such that $F_i = `f_i \circ M_{P_i}$. Let $P_\infty = \limsup_{i \in \NN} P_i$, and let $f_\infty\: P_\infty \To S$ be defined by $f_\infty(p_\infty) = \sup_{i \in \NN} f_i(p_i)$ for all $p_\infty \in P_\infty$, where $p_1 \not \perp p_2 \not \perp  \cdots$ is any sequence of projections such that $p_i \in P_i$ and $p_i \not \perp p_\infty$ for all $i \in \NN$. Let $F_\infty = `f_\infty \circ M_{P\infty}$. Then, $F_i \nearrow F_\infty$.
\end{theorem}

\begin{proof}
By the definition of the order on functions, we have that $`({\sqsubseteq}) \circ F_1 \geq `({\sqsubseteq}) \circ F_2 \geq \cdots$, and we are only to show that the infimum of this sequence is the binary relation $`({\sqsubseteq}) \circ F_\infty$.

First, we show that $`({\sqsubseteq}) \circ F_\infty$ is a lower bound, by proving that $F_n \sqsubseteq F_\infty$ for each $n \in \NN$. Fix $n \in \NN$, and let $p_n \in P_n$ and $p_\infty \in P_\infty$ be projections such that $p_n \not \perp p_\infty$. Applying Lemma \ref{classical.H}, we find a sequence $p_1 \not \perp p_2 \not \perp \cdots$ such that $p_i \in P_i$ and $p_i \not \perp p_\infty$ for all $i \in \NN$. By definition of $f_\infty$, we conclude that $f_n(p_n) \leq f_\infty(p_\infty)$. We now conclude by Proposition \ref{classical.E} that $F_n \sqsubseteq F_\infty$, which means that $`({\sqsubseteq}) \circ F_n \geq`({\sqsubseteq}) \circ F_\infty$. Therefore, $`({\sqsubseteq}) \circ F_\infty$ is indeed a lower bound for our sequence of binary relations.

To show that $`({\sqsubseteq}) \circ F_\infty$ is the greatest lower bound, we reason atomwise. Fix $s_0 \in S$. The decreasing sequence of subspaces $(`({\sqsubseteq}) \circ F_1)(H, \CC_{s_0}) \geq (`({\sqsubseteq}) \circ F_2)(H, \CC_{s_0}) \geq \cdots$ must stabilize, as must the decreasing sequence of partitions $\sup_{i \geq 1} P_i, \sup_{i \geq 2} P_i, \ldots$. Without loss of generality, we may assume that both sequences are constant.

For each index $i \in \NN \union \{\infty\}$, define $r_i = \sum \{p \in P_i \suchthat f_i(p) \sqsubseteq s_0\}$. By Lemma \ref{classical.F}, $(`({\sqsubseteq}) \circ F_i)(H, \CC_{s_0}) = L(H, \CC_{s_0}) \cdot r_i$, so the sequence $r_1, r_2, \ldots$ is also constant. Since we have already established that $`({\sqsubseteq}) \circ F_\infty$ is a lower bound, we already have $r_1 \geq r_\infty$, and it remains only to show that $r_1 \leq r_\infty$. If $r_1 = 0$, then there is nothing to prove.

Assume that $r_1 \neq 0$, and let $p_1 \in P_1$ be any projection contained in $r_1$. Since $r_1$ is also a multicell of $P_2$, we may choose a projection $p_2 \in P_2$ that is not orthogonal to $p_1$ and that is also contained in $r_1$. Continuing in this way, we obtain a sequence of projections $p_1 \not \perp p_2 \not \perp \cdots$ such that $p_i \in P_i$ and $p_i \leq r_1$ for each $i \in \NN$. By definition of $r_i$, we are assured that $f_i(p_i) \sqsubseteq s_0$ for each $i \in \NN$.

Each partition $P_i$ refines $P_\infty$, so each projection $p_i$ is contained in some projection in $P_\infty$. Thus, $p_1$ is contained in some projection $p_\infty \in P_\infty$. Furthermore, because $p_1 \not \perp p_2 \not \perp \ldots$, it follows by induction that $p_i \leq p_\infty$ for all $i \in \NN$. In particular, $p_i \not \perp p_\infty$ for all $i \in\NN$. We now apply the definition of $f_\infty$ to find that $f_\infty(p_\infty) = \sup_{i \in \NN} f_i(p_i) \sqsubseteq s_0$. Therefore, $p_1 \leq p_\infty \leq r_\infty$. Since $p_1$ was chosen arbitrarily, we conclude that $r_1 \leq r_\infty$.

Therefore, $(`({\sqsubseteq}) \circ F_1)(H, \CC_{s_0}) = L(H, \CC_{s_0}) \cdot r_1 =  L(H, \CC_{s_0}) \cdot r_\infty = (`({\sqsubseteq}) \circ F_\infty)(H, \CC_{s_0})$. We vary $s_0 \in S$ to conclude that $`({\sqsubseteq}) \circ F_1 = `({\sqsubseteq}) \circ F_\infty$, and thus, $`({\sqsubseteq}) \circ F_\infty$ is the infimum of $`({\sqsubseteq}) \circ F_1 \geq `({\sqsubseteq}) \circ F_2 \geq \cdots$.
\end{proof}

\begin{corollary}\label{classical.L}
Let $(S,\sqsubseteq)$ be a cpo. Then, $(`S, `{\sqsubseteq})$ is a quantum cpo.
\end{corollary}

\begin{proof}
Let $\H$ be an atomic quantum set, and let $F_1 \sqsubseteq F_2 \sqsubseteq \ldots$ be an ascending sequence of functions $\H \To `S$. By Proposition \ref{classical.D}, each function $F_i$ may be factored as $F_i = `f_i \circ M_{P_i}$, for some partition $P_i$ on $H$, and some ordinary function $f_i\: P_i \To S$. By Theorem \ref{classical.K}, there exists a function $F_\infty\: \H \To `S$ such that $F_i \nearrow F_\infty$.
\end{proof}

\begin{corollary}\label{classical.M}
Let $(S_1,\sqsubseteq_1)$ and $(S_2, \sqsubseteq_2)$ be cpos, and let $g\: S_1 \to S_2$ be a Scott continuous function. Then $`g\:`S_1 \To `S_2$ is also Scott continuous.
\end{corollary}

\begin{proof}
Let $\H$ be an atomic quantum set, and let $F_1, F_2, \ldots, F_\infty$ be functions $\H \To `S$ such that $F_i \nearrow F_\infty$. By Proposition \ref{classical.D}, each function $F_i$ may be factored as $F_i = `f_i \circ M_{P_i}$, for some partition $P_i$ on $H$, and some ordinary function $f_i\: P_i \To S$. By Theorem \ref{classical.K}, $F_\infty$ may be factored as $F_\infty = `f_\infty \circ M_{P_\infty}$, where $P_\infty = \limsup_{i \To \infty} P_i$ and $f_\infty\: P_\infty \To S$ is defined by $f_\infty(p_\infty) = \sup_{i \in \NN} f_i(p_i)$ for all $p_\infty \in P_\infty$, where $p_1 \not \perp p_2 \not \perp \cdots$ is any sequence of projections such that $p_i \in P_i$ and $p_i \not \perp p_\infty$ for all $i \in \NN$.

We now apply Theorem \ref{classical.K} again, to show that $`g \circ F_i \nearrow `g \circ F_\infty$. Indeed, each function $`g\circ F_i$ may be factored as $`g \circ F_i = `(g \circ f_i) \circ M_{P_i}$, and $`g \circ F_\infty$ may be similarly factored as $`g \circ F_\infty = `(g \circ f_\infty) \circ M_{P_\infty}$. Furthermore, for each projection $p_\infty \in P_\infty$, and each sequence $p_1 \not \perp p_2 \not \perp \cdots$ such that $p_i \in P_i$ and $p_i \not \perp p_\infty$ for all $i\in \NN$, we have $(g \circ f_\infty)(p_\infty) = g (f_\infty (p_\infty)) = g(\sup_{i \in \NN} f_i(n_i)) = \sup_{i \in \NN} g(f_i(n_i)) = \sup_{i \in \NN} (g \circ f_i)(n_i)$. The third equality in the calculation appeals to the given assumption that $g$ is Scott continuous, and to Proposition \ref{classical.E}. We conclude by Theorem \ref{classical.K} that $`g \circ F_i \nearrow `g \circ F_\infty$. Therefore, $`g$ is Scott continuous.
\end{proof}

\begin{proposition}\label{prop:quote strong monoidal}
    Let $(S,\sqsubseteq_S)$ and $(T,\sqsubseteq_T)$ be cpos. Denote by $\sqsubseteq_{S\times T}$ the product order on $S\times T$, i.e., $(s_1,t_1)\sqsubseteq_{S\times T}(s_2,t_2)$ if and only if $s_1\sqsubseteq_S s_2$ and $t_1\sqsubseteq_T t_2$. Then there is an order isomorphism $G:`(S,\sqsubseteq_S)\times`(T,\sqsubseteq_T)\to `(S\times T,\sqsubseteq_{S\times T})$ whose only nonzero components are given by $G(\CC_s\otimes\CC_t,\CC_{(s,t)})=L(\CC_s\otimes\CC_t,\CC_{(s,t)})$ for each $s\in S$ and each $t\in T$.
\end{proposition}
\begin{proof}
It was already stated in \cite[Section II]{Kornell18} that $`S\times `T\cong `(S\times T)$. Clearly, $G:`S\times `T\to`(S\times T)$ is a bijection. A  short calculation yields  \[\big(G\circ (`\sqsubseteq_S\times  `\sqsubseteq_T)\big)(\CC_{s_1}\otimes\CC_{t_1},\CC_{(s_2,t_2)})=\begin{cases}
        L(\CC_{s_1}\otimes\CC_{t_1},\CC_{(s_2,t_2)}), & s_1\sqsubseteq_S s_2\text{ and }t_1\sqsubseteq_Tt_2,
        \\
        0, & \text{else},
    \end{cases}\]
    and  \[\big(`\sqsubseteq_{S\times T}\circ G\big)(\CC_{s_1}\otimes\CC_{t_1},\CC_{(s_2,t_2)}) =\begin{cases}
        L(\CC_{s_1}\otimes\CC_{t_1},\CC_{(s_2,t_2)}), & s_1\sqsubseteq_S s_2\text{ and }t_1\sqsubseteq_Tt_2,
        \\
        0, & \text{else}.
    \end{cases}\]
    It follows that 
    $G\circ (`\sqsubseteq_S\times  `\sqsubseteq_T)=`\sqsubseteq_{S\times T}\circ G$, so $G$ is an order isomorphism by \cite[Proposition 2.7]{KLM20}.    
\end{proof}

\subsection{The right adjoint of inclusion}

It follows from Corollaries \ref{classical.L} and \ref{classical.M} that the functor $`(-):\POS\to\qPOS$ restricts and corestricts to a functor $`(-):\CPO\to\qCPO$. Now, we will show that the functor $\X\mapsto\qCPO(\mathbf 1,\X)$ is its right adjoint. Of course, $\qCPO(\mathbf 1,\X)$ is canonically a cpo, because $\cat{qCPO}$ is enriched over $\cat{CPO}$ (Theorem \ref{thm:qCPO enriched over CPO}). Any function from $\mathbf 1$ to a quantum cpo $\X$ is automatically Scott continuous, and hence, the points of $\qCPO(\mathbf 1,\X)$ are in canonical one-to-one correspondence with the one-dimensional atoms of $\X$. 

We begin by showing that the the one-dimensional atoms of a quantum cpo $(\X,R)$ form a sub-cpo. So, let $\X_1$ be the subset of $\X$ that consists of the one-dimensional atoms of $\X$, let $J_1\: \X_1 \hookrightarrow \X$ be its inclusion function, and let $R_1$ be the induced order on $\X_1$, i.e., let $R_1=J_1^\dag\circ R\circ J_1$ \cite[Definition~2.2]{KLM20}.

\begin{lemma}\label{classical.N}
Make the assumptions of Theorem \ref{classical.K}. Furthermore, let $(\X,R)$ be a quantum cpo, and assume that $S = \At(\X_1)$, with $X_1 \sqsubseteq X_1'$ iff $R(X_1, X_1') \neq 0$, for $X_1, X_1' \in S$. Let $E$ be the canonical isomorphism $`S \leftrightarrow \X_1$, and let $J_1$ be the inclusion function $\X_1 \hookrightarrow \X$. If $P_\infty=\{1\}$, then $J_1 \circ E \circ F_i \nearrow J_1 \circ E \circ F_\infty$.
\end{lemma}

\begin{proof} 
Let $p_\infty = 1$, and let $p_1 \not \perp p_2 \not \perp p_3 \not \perp \cdots$ be a sequence of projections such that $p_i \in P_i$ for all $i \in \NN$, as in the statement of Theorem \ref{classical.K}.  Assume that $P_\infty := \limsup_{i \in \NN} P_i$ equals $\{1\}$. For each $i \in \NN \union \{\infty\}$, let $h_i\: \{1\} \To S$ be the function defined by $h_i(1) = f_i(p_i)$. Thus, $`h_\infty \circ M_{\{1\}} = F_\infty$, but in general, $`h_i \circ M_{\{1\}}$ is distinct from $F_i$. However, we will eventually show that for each sufficiently large integer $n$, there is an integer $m \geq n$ such that $`h_n \circ M_{\{1\}} \sqsubseteq F_m$.

By definition of the functions $h_i$, for $i \in \NN \union \{\infty\}$, we certainly have that $h_i \nearrow h_\infty$. The functions $J_1$ and $E$ are both monotone, so we also know that $J_1 \circ E \circ `h_1 \sqsubseteq J_1 \circ E \circ `h_2 \sqsubseteq \cdots$ is a monotone increasing sequence. Because $(\X,R)$ is a quantum cpo, this monotone increasing sequence has a limit, also a function $\{1\} \To \X$. This limiting function must factor through $J_1$, and it must be the least upper bound of the sequence, so it must be equal to $J_1 \circ E \circ `h_\infty$. We conclude that $J_1 \circ E \circ `h_i \nearrow J_1 \circ E \circ `h_\infty$. As a consequence, we also have that $J_1 \circ E \circ `h_i \circ M_{\{1\}} \nearrow J_1 \circ E \circ `h_\infty \circ M_{\{1\}}$.

By Lemma \ref{classical.II}, for each positive integer $n$, there exists an integer $m \geq n$ such that $f_m(p) \sqsupseteq f_n(p_n)$ for every projection $p \in P_m$. It follows by Proposition \ref{classical.E} that for each positive integer $n$, there exists an integer $m(n) \geq n$ such that $F_{m(n)} \sqsupseteq `h_n \circ M_{\{1\}}$. Hence, $J_1 \circ E \circ `h_n \circ M_{\{1\}} \sqsubseteq J_1 \circ E \circ F_{m(n)}$ for all sufficiently large integers $n$.

We now calculate that
\begin{align*}
R \circ J_1 \circ E \circ F_\infty
& =
R \circ J_1 \circ E \circ `h_\infty \circ M_{\{1\}}
=
\bigwedge_{n \in \NN} R \circ J_1 \circ E \circ `h_n \circ M_{\{1\}}
\\ &\geq
\bigwedge_{n \in \NN} R \circ J_1 \circ E \circ F_{m(n)}
\geq
\bigwedge_{m \in \NN} R \circ J_1 \circ E \circ F_m.
\end{align*}
\noindent
It follows from the monotonicity of $J_1$ and $E$ that $J_1 \circ E \circ F_1 \sqsubseteq J_1 \circ E \circ F_2 \sqsubseteq \cdots$ is an increasing sequence, so the converse inequality is already established. Thus, the lemma is proved.
\end{proof}

\begin{theorem}\label{classical.O}
Let $\H = \Q\{H\}$ be an atomic quantum set, let $(\X, R)$ be a quantum cpo, and let $J_1\: \X_1 \hookrightarrow \X$ be the inclusion function. For each $i \in \NN \union\{\infty\}$, let $G_i$ be a function from $\H$ to $\X$. Assume that $G_i \nearrow G_\infty$, and that $G_i$ factors through $J_1$ for each $i \in \NN$. Then, $G_\infty$ also factors through $J_1$.
\end{theorem}

\begin{proof}
Let $S = \At(\X_1)$, and let $\sqsubseteq$ be the partial order on $S$ defined by $X_1 \sqsubseteq X_1'$ iff $R(X_1, X_1') \neq 0$ for $X_1, X_1' \in S$, so that $(`S, `{\sqsubseteq})$ is canonically isomorphic to $(\X_1, J_1^\dagger \circ R \circ J_1)$. Let $E$ be this isomorphism. By assumption, for each integer $i \in\NN$, the function $G_i$ factors through $J_1$. Since, $E$ is an isomorphism, it also factors through $J_1 \circ E$. Therefore, for each integer $i \in \NN$, let $F_i\: \H \To `S$ be such that $J_1 \circ E \circ F_i = G_i$.
By Proposition \ref{classical.D}, each such function $F_i$ may be further factored as $F_i = `f_i \circ M_{P_i}$ for some partition $P_i$ on $H$. Let $P_\infty = \limsup_{i \To \infty} P_i$. Without loss of generality, we may assume that $\sup_{i \in \NN} P_{i} = P_\infty$. Thus, $P_i \leq P_\infty$ for each index $i \in \NN$. We index the elements of $P_\infty$ by some set $A$, writing $P_\infty = \{ p_\infty^\alpha \suchthat \alpha \in A\}$.

Fix an index $\alpha \in A$. For each $i \in \NN\union\{\infty\}$, we define $P_i^\alpha := \{p_i \in P_i \suchthat p_i \leq p_\infty^\alpha\}$. Identifying each projection contained in $p^\alpha_\infty$ with a projection operator on $H^\alpha := p_\infty^\alpha H$ in the obvious way, we conclude that $P_i^\alpha$ is a partition on $H^\alpha$. The partition $P_\infty^\alpha = \{p_\infty^\alpha\}$ is indeed the supremum of the partitions $P^\alpha_i$, for $i \in \NN$, essentially by our definition of $P_\infty$. Thus, we plan to apply Lemma \ref{classical.N}. For each integer $i \in \NN$, we define $f_i^\alpha$ to be the restriction of $f_i$ to $P_i^\alpha$; in other words, $f_i^\alpha = f_i \circ j_i^\alpha$, where $j_i^\alpha\: P_i^\alpha \hookrightarrow P_i$ is the ordinary inclusion function. We now define $F_i^\alpha = `f_i^\alpha \circ M_{P_i^\alpha}$.

We claim that $F_1^\alpha \sqsubseteq F_2^\alpha \sqsubseteq \cdots$. Fix integers $i_1 < i_2 \in \NN$. We are given that $J_1 \circ E \circ F_{i_1} \sqsubseteq J_1 \circ E \circ F_{i_2}$, which means that $J_1 \circ E \circ F_{i_2} \leq R \circ  J_1 \circ E \circ F_{i_1}$. Appealing to the fact that $J_1$ and $E$ are injective, i.e., $J_1^\dagger \circ J_1 = I$ and $E^\dagger \circ E =I$, we infer that $F_{i_2} \leq E^\dagger \circ J_1^\dagger \circ R \circ J_1 \circ E \circ F_{i_1}$. The binary relation $ E^\dagger \circ J_1^\dagger \circ R \circ J_1 \circ E$ on $`S$ is equal to $`{\sqsubseteq}$, essentially by definition. Therefore, $F_{i_1} \sqsubseteq F_{i_2}$. Appealing to Proposition \ref{classical.E} twice, we conclude that $F_{i_1}^\alpha \sqsubseteq F_{i_2}^\alpha$. We vary $i_1$ and $i_2$ to conclude that $F_1^\alpha \sqsubseteq F_2^\alpha \sqsubseteq \cdots$.

By Lemma \ref{classical.N}, $J_1 \circ E \circ F_i^\alpha \nearrow J_1 \circ E \circ F_\infty^\alpha$ for some function $F_\infty^\alpha\: \H^\alpha \To `S$. We vary $\alpha \in A$ to define $F_\infty = [F_\infty^\alpha \suchthat \alpha \in A] \circ D_{P_\infty}$, where the bracket notation refers to the universal property of the coproduct $\biguplus_{\alpha \in A} \H^\alpha$, and $D_{P_\infty}\:\H \to \biguplus_{\alpha \in A} \H^\alpha$ is the decomposition function corresponding to $P_\infty$ (Definition \ref{def:decomposition}). We now compute as follows:
\begin{align*}
R \circ J_1&  \circ E \circ F_\infty
 =
R \circ J_1 \circ E \circ\left[ F_\infty^\alpha \, \middle| \, \alpha \in A \right] \circ D_{P_\infty}
=
\left[R \circ J_1 \circ E \circ F_\infty^\alpha \, \middle| \, \alpha \in A\right] \circ D_{P_\infty}
\\ &=
\left[\bigwedge_{i \in \NN} R \circ J_1 \circ E \circ F_i^\alpha \, \middle| \, \alpha \in A\right] \circ D_{P_\infty}
=
\left(\bigwedge_{i \in \NN}\left[ R \circ J_1 \circ E \circ F_i^\alpha \, \middle| \, \alpha \in A\right]\right) \circ D_{P_\infty}
\\ &=
\bigwedge_{i \in \NN}R \circ J_1 \circ E \circ\left[  F_i^\alpha \, \middle| \, \alpha \in A\right] \circ D_{P_\infty}
=
\bigwedge_{i \in \NN}R \circ J_1 \circ E \circ\left[  `f_i^\alpha \circ M_{P_i^\alpha}\, \middle| \, \alpha \in A\right] \circ D_{P_\infty}
\\ &=
\bigwedge_{i \in \NN}R \circ J_1 \circ E \circ\left[  `f_i \circ `j_i^\alpha \circ M_{P_i^\alpha}\, \middle| \, \alpha \in A\right] \circ D_{P_\infty}
\\ &=
\bigwedge_{i \in \NN}R \circ J_1 \circ E \circ`f_i \circ \left[  `j_i^\alpha \circ M_{P_i^\alpha}\, \middle| \, \alpha \in A\right] \circ D_{P_\infty}
\\ &=
\bigwedge_{i \in \NN}R \circ J_1 \circ E \circ`f_i \circ M_{P_i}
=
\bigwedge_{i \in \NN}R \circ J_1 \circ E \circ F_i
\end{align*}
\noindent
Thus, $J_1 \circ E \circ F_i \nearrow J_1 \circ E \circ F_\infty$. We conclude that $G_\infty = J_1 \circ E \circ F_\infty$, so $G_\infty$ does factor through $J_1$.
\end{proof}

\begin{corollary}\label{classical.P}
Let $(\X, R)$ be a quantum cpo. Then, the inclusion function $J_1\: \X_1 \hookrightarrow \X$ is Scott continuous.
\end{corollary}

\begin{proof}
Let $\H$ be an atomic quantum set, and let $F_1 \sqsubseteq F_2 \sqsubseteq \cdots \sqsubseteq F_\infty$ be a monotonically ascending sequence of functions $\H \To \X_1$ with $F_i \nearrow F_\infty$. Since $J_1$ is monotone, $J_1 \circ F_1 \sqsubseteq J_1 \circ F_2 \sqsubseteq \cdots \sqsubseteq J_1 \circ F_\infty$ is a monotonically ascending sequence of functions $\H \To \X$. By the definition of a quantum cpo, there is a function $G_\infty\: \H \To \X$ such that $J_1 \circ F_i \nearrow G_\infty$. In particular, $G_\infty \sqsubseteq J_1 \circ F_\infty$. By Theorem \ref{classical.O}, there is a function $F_\infty'\: \H \To \X_1$ such that $G_\infty = J_1 \circ F_\infty'$. Thus, $J_1 \circ F_1 \sqsubseteq J_1 \circ F_2 \sqsubseteq\ldots \sqsubseteq J_1 \circ F_\infty' \sqsubseteq J_1 \circ F_\infty$. Since $J_1$ is an order embedding, we infer that $F_1 \sqsubseteq  F_2 \sqsubseteq\ldots \sqsubseteq F_\infty' \sqsubseteq  F_\infty$. In other words, $F_\infty'$ is an upper bound of the sequence $F_1 \sqsubseteq  F_2 \sqsubseteq\ldots$ that is nevertheless below $F_\infty$. Since $F_\infty$ is the least upper bound, we conclude that $F_\infty' = F_\infty$. Therefore, $J_1 \circ F_i \nearrow J_1 \circ F_\infty$, as desired.
\end{proof}

Equip $\mathbf 1$ with its unique (trivial) order $I_\mathbf 1$. Let $\X$ be a quantum set and $R$ an order relation on $\X$. Then each function $F:\mathbf 1\to\X$ is automatically monotone, simply because $\mathbf 1$ is trivially ordered. In fact, if $(\X,R)$ is a quantum cpo, then $F$ is automatically Scott continuous, simply because $\mathbf 1$ is a finite quantum poset and monotone maps whose codomain is finite are Scott continuous  (Proposition \ref{prop:finite qposet is qcpo}). Hence, if $(\X,R)$ is a quantum cpo, then $\qSet(\mathbf 1,\X)=\qPOS(\mathbf 1,\X)=\qCPO(\mathbf 1,\X)$. 

\begin{lemma}\label{lem:b is order iso}
Let $(\X,R)$ be a quantum poset, and equip $\At(\X_1)$ with the order defined by $X \sqsubseteq X'$ if and only if $R(X,X') \neq 0$, for $X, X' \in \At(\X_1)$. Then, the canonical bijection $b_\X:\At(\X_1)\to\qSet(\mathbf 1,\X)$ of Lemma \ref{lem:one-dimensional atoms} is an order isomorphism. 
\end{lemma}

\begin{proof}
For visual brevity, write $b = b_\X$.
Let $X, X' \in \At(\X_1)$. We reason that
\begin{align*}
b(X) \below b(X')
& \; \Longleftrightarrow \;
b(X') \circ b(X)^\dagger \leq R
\; \Longleftrightarrow \;
b(X')(\CC, X') \cdot b(X) (\CC,X)^\dagger \leq R(X, X')
\\ & \; \Longleftrightarrow \;
L(X,X') = R(X,X')
\; \Longleftrightarrow \;
R(X,X') \neq 0
\end{align*}
The first equivalence follows by \cite[Lemma~4.1(4)]{KLM20}.
\end{proof}

\begin{lemma}\label{lem:B is order iso}
Let $(\X,R)$ be a quantum poset. Let $B_\X:\X_1\to`\qSet(\mathbf 1,\X)$ be the bijection from Lemma \ref{lem:one-dimensional atoms}, so $B_\X(X,\CC_{b_\X(X)}) = L(X,\CC_{b_\X(X)}) $, for $X \atomof \X_1$, with the other components vanishing. Then, $B_\X$ is an order isomorphism from $(\X_1, R_1)$ to $(`\qSet(\mathbf 1,\X), `{\below})$, where $R_1 = J_1^\dagger \circ R \circ J_1$ and $\below$ is the usual order on functions \cite[Lemma~4.1]{KLM20}.
\end{lemma}

\begin{proof}
For visual brevity, write $b = b_\X$ and $B = B_\X$.
Reasoning by cases, we find that the conclusion of Lemma \ref{lem:b is order iso} is that $b$ is a bijection and $`(\below)(\CC_{b(X_1)}, \CC_{b(X_2)}) \cdot L(X_1, \CC_{b(X_1)}) = L(X_2, \CC_{b(X_2)}) \cdot R_1(X_1, X_2)$ for all $X_1, X_2 \atomof \X_1$. We also calculate that for all $X_1, X_2 \atomof \X_1$,
\begin{align*}
(`(\below) \circ B)&( X_1, \CC_{b(X_2)})
 =
\bigvee_{P \in \cat{qSet}(\mathbf 1, \X)} `(\below)(\CC_{P}, \CC_{b(X_2)}) \cdot B(X_1, \CC_{P})
\\ & =
`(\below)(\CC_{b(X_1)}, \CC_{b(X_2)}) \cdot L(X_1, \CC_{b(X_1)})
=
L(X_2, \CC_{b(X_2)}) \cdot R_1(X_1, X_2)
\\ &=
\bigvee_{X \atomof \X_1} B(X,\CC_{b(X_2)}) \cdot R_1(X_1, X)
=
(B \circ R_1)(X_1, \CC_{b(X_2)}).
\end{align*}
Thus, $`(\below) \circ B = B \circ R_1$. It follows immediately that $B$ is an isomorphism of quantum posets in the sense of \cite[Definition~2.6]{KLM20}.
\end{proof}

For the next lemma, we introduce a minor notation that will allow us to distinguish functions $1 \to \X$ from the corresponding functions $1 \to \cat{qSet}(1, \X)$. Specifically, whenever $s$ is an element of some ordinary set $S$, we write $\ulcorner s \urcorner$ for the corresponding function $1 \to S$.

\begin{lemma}\label{lem:unnaming}
Let $\X$ be a quantum set, let $b_\X: \At(\X_1) \to \cat{qSet}(\mathbf 1, \X)$ be as in Lemma \ref{lem:b is order iso}, and let $B_\X: \X_1 \to `\cat{qSet}(\mathbf 1, \X)$ be as in Lemma \ref{lem:B is order iso}. Then, for all functions $P_0\:\mathbf 1 \to \X$, we have that $J_1 \circ B_\X^{\inv} \circ `\ulcorner P_0 \urcorner = P_0$.
\end{lemma}

\begin{proof}For visual brevity, write $b = b_\X$ and $B = B_\X$.
Note that $B^{-1} = B^\dagger$ \cite{Kornell18}*{Proposition 4.2}. Let $X_0 \atomof \X_1$ be such that $b_\X(X_0) = P_0$ (Lemma \ref {lem:b is order iso}).
We calculate that for all one-dimensional atoms $X \atomof \X$,
\begin{align*}
(J_1 \circ B^{-1} \circ  `\ulcorner P_0 \urcorner)& (\CC, X)
 = (B^{-1} \circ  `\ulcorner P_0 \urcorner)(\CC, X)
  = (B^\dagger \circ  `\ulcorner P_0 \urcorner)(\CC, X)
\\ & = 
\bigvee_{P \in \cat{qSet}(\mathbf 1, \X)} B^\dagger(\CC_P, X) \cdot `\ulcorner P_0 \urcorner(\CC, \CC_P)
 =
B^\dagger(\CC_{P_0}, X) \cdot L(\CC, \CC_{P_0})
\\ & =
B^\dagger(\CC_{b(X_0)}, X) \cdot L(\CC, \CC_{b(X_0)})
 =
\delta_{X_0, X} \cdot L(\CC_{b(X_0)}, X) \cdot L(\CC, \CC_{b(X_0)})
\\ & =
\delta_{X_0, X} \cdot L(\CC, X)
=
b(X_0)(\CC, X)
=
P_0(\CC, X).
\end{align*}
Of course, $(J_1 \circ B^{-1} \circ  `\ulcorner P_0 \urcorner)(\CC, X) = 0 = P(\CC, X)$ for all atoms $X \atomof \X$ of dimension larger than one. Therefore, $J_1 \circ B^{-1} \circ  `\ulcorner P_0 \urcorner = P_0$, as claimed.
\end{proof}

\begin{theorem}\label{thm:CPO-qCPO adjunction}
The functor $`(-):\CPO\to\qCPO$ is fully faithful, strong monoidal and left adjoint to $\qCPO(\mathbf 1,-):\qCPO\to\CPO$.
\end{theorem}

\begin{proof}
It easily follows from the order isomorphism in Proposition \ref{prop:quote strong monoidal} that $`(-)$ is strong monoidal. 
We show that the functors in the statement are adjoint to each other. Moreover, we show that every component of the unit is an isomorphism in $\CPO$, which implies that $`(-)$ is fully faithful by \cite[Theorem IV.3.1]{maclane}.

Fix a cpo $(S, \below)$. If we order the atoms of $`S$ by $\CC_{s_1} \below \CC_{s_2}$ if and only if $`(\below)(\CC_{s_1}, \CC_{s_2})\neq 0$, then $i_S: s \mapsto \CC_s$ becomes an order isomorphism $S \to \At(`S)$, simply by definition of $`(\below)$. The bijection $b_{`S}: \At(`S) \to \cat{qSet}(1,`S)$ is also an order isomorphism by Lemma \ref{lem:b is order iso}, and it is easy to see that $b_{`S}(\CC_s) = ` \ulcorner s \urcorner$ for all $s \in S$. Indeed, $b_{`S}(\CC_s)(\CC,\CC_{s'}) = \delta_{s,s'} \cdot L(\CC, \CC_{s'}) = ` \ulcorner s \urcorner (\CC, \CC_{s'})$ for all $s, s' \in S$. Thus, $\eta_{S} : = b_{S'} \circ i_S$ is an order isomorphism $S \to \cat{qSet}(\mathbf 1, `S) = \cat{qCPO}(\mathbf 1, `S)$ such that $\eta_S(s) = `\ulcorner s \urcorner$ for all $s \in S$.

We show that $\eta_S$ satisfies the universal property of the $S$-component of the unit the claimed adjunction. Let $(\X, R)$ be a quantum cpo, and let $f: S \to \cat{qCPO}(\mathbf 1, \X)$ be a Scott continuous function. We show that there exists a unique Scott continuous function $F\:`S \to X$ making the following diagram commute:
$$
\begin{tikzcd}
S
\arrow{r}{\eta_S}
\arrow{rd}[swap]{f}
&
\cat{qCPO}(\mathbf 1, `S)
\arrow{d}{F \circ(-)}
\\
&
\cat{qCPO}(\mathbf 1, \X).
\end{tikzcd}
$$
Let $F = J_1 \circ B_\X^{-1} \circ `f$. For each $s \in S$, we calculate that $$J_1 \circ B_\X^{-1} \circ `f \circ \eta_S(s) = J_1 \circ B_\X^{-1} \circ `f \circ `\ulcorner s \urcorner = J_1 \circ B_\X^{-1} \circ `\ulcorner f(s) \urcorner = f(s),$$ applying Lemma \ref{lem:unnaming} in the last step. Thus, $F\circ(-)$ makes the diagram commute. If $F'$ also makes the diagram commute, then we quickly find that $(F \circ(-)) = (F' \circ(-))\: \cat{qCPO}(\mathbf 1, `S) \to \cat{qCPO}(\mathbf 1, \X)$, because $\eta_S$ is an isomorphism. This implies that $F = F'$ because for all $s \in S$ and all $X \atomof \X$, we have that
\begin{align*}
F'(\CC_s, X) \cdot L(\CC, \CC_s) & = \bigvee_{s' \in S} F'(\CC_{s'}, X) \cdot `\ulcorner s \urcorner (\CC, \CC_{s'})
= (F' \circ ` \ulcorner s \urcorner) \\ & = (F \circ ` \ulcorner s \urcorner) = \bigvee_{s' \in S} F(\CC_{s'}, X) \cdot `\ulcorner s \urcorner (\CC, \CC_{s'}) = F(\CC_s, X) \cdot L(\CC, \CC_s).
\end{align*}

Therefore, the order isomorphisms $\eta_S$, for all cpos $S$, together form the unit of an adjuction between the functor $\qCPO(\mathbf 1,-):\qCPO\to\CPO$ and a functor $\CPO \to \qCPO$ that takes each cpo $S$ to the quantum cpo $`S$. To show that this left adjoint is $`(-):\CPO\to\qCPO$, as claimed, it is enough to show that the unit is a natural transformation from the identity functor on $\cat{CPO}$ to $\cat{qCPO}(\mathbf 1, `(\,-\,))$. Let $S_1$ and $S_2$ be cpos, and let $f$ be a Scott continuous function $S_1 \to S_2$. For all $s_1 \in S_1$, we calculate that $\cat{qCPO}(\mathbf 1, `f)(\eta_{S_1})(s_1) = \cat{qCPO}(\mathbf 1, `f)(`\ulcorner s_1 \urcorner) = `f \circ (`\ulcorner s_1 \urcorner) = ` \ulcorner f(s_1) \urcorner = \eta_{S_2} (f(s_1))$. We conclude that the unit is indeed a natural transformation from the identity functor on $\cat{CPO}$ to $\cat{qCPO}(\mathbf 1, `(\,-\,))$, completing the proof.
\end{proof}

\section{The Lift Monad}\label{sec:lifting}
Programs in simply-typed languages without a recursion operator always terminate. If a language has a recursion operator, there is the possibility of nonterminating programs, hence it is desirable that nontermination can be expressed within a denotational model for such a language. In the denotational semantics of ordinary programming languages, the effect of nontermination is often modeled by the ``lift monad" $(-)_\perp$ on $\mathbf{CPO}$, which adds a least element to each cpo. In this section, we prove that there is a lift monad on $\qCPO$, also denoted by $(-)_\perp$, which generalizes the lift monad on $\CPO$. Moreover, we show that $(-)_\perp$ is a \emph{symmetric monoidal monad}, which is a necessary and sufficient condition to conclude that the Kleisli category of the monad is symmetric monoidal. Showing that a monad $T$ is symmetric monoidal amounts to showing the existence of a \emph{double strength map} $TX \otimes TY \to T(X\otimes Y)$ and proving its related diagrams commute. Verifying the commutativity of these diagrams is tedious, but the verification of some of them can be circumvented by applying known results. 

We begin by sketching our approach in terms of the classical versions of our categories. If $X$ is an ordinary cpo, then the underlying set of $X_\perp$ is the coproduct $X\uplus 1$ of $X$ with the terminal object $1$ in $\Set$. In a category $\CCC$ with coproducts and a terminal object $1$, if the endofunctor $X\mapsto X\uplus 1$ forms a monad, it is called the \emph{maybe monad}. For example, it is well-known that this endofunctor is a monad on $\Set$ whose Kleisli category is equivalent to the category of sets and partial functions, which in turn is well-known to be symmetric monoidal. Hence we can conclude that the maybe monad on $\Set$ is symmetric monoidal without checking any diagram. Although the lift monad on $\CPO$ and the maybe monad on $\CPO$ differ, the action of both monads on the underlying set of a cpo coincides with the maybe monad on $\Set$, so the diagrams for a symmetric monoidal monad must commute in both cases. Therefore, the lift monad on $\CPO$ must be symmetric monoidal. Following this approach, we will first show that the maybe monad on $\qSet$ is well defined and symmetric monoidal, and then use the above reasoning to conclude the lift monad on $\qCPO$ is symmetric monoidal, too.

\subsection{The maybe monad on $\qSet$}

We fix some terminology. If $T$ is a monad on a category $\CCC$ with unit $\eta$ and multiplication $\mu$, then we denote its Kleisli category by $\CCC_T$. The objects of $\CCC_T$ coincide with the objects of $\CCC$, and for objects $X,Y\in\CCC_T$, we have $\CCC_T(X,Y)=\CCC(X,TY)$. For a morphism $f$ in this homset we write $f:X\circlearrow Y$ if we regard it as a morphism in $\CCC_T$ and $f:X\to TY$ if we regard it as a morphism in $\CCC$. Composition in $\CCC_T$ is denoted by $\bullet$, so $g\bullet f=\mu_Z\circ Tg\circ f$ for $f:X\circlearrow Y$ and $g:Y\circlearrow Z$ in $\CCC_T$, i.e., for $f:X\to TY$ and $g:Y\to TZ$ in $\CCC$.

If $\CCC$ is equipped with symmetric monoidal product $\otimes$, we say that $T$ is a \emph{symmetric monoidal} monad if for each $X,Y\in\CCC$ there exists a \emph{double strength}  $k_{X,Y}:TX\otimes TY\to T(X,Y)$ natural in $X$ and $Y$, satisfying conditions (1)-(5) in \cite[1.2]{Seal13}. Equivalently, $T$ is symmetric monoidal if and only if $\CCC_T$ can be equipped with a symmetric monoidal product such that the canonical functor $\CCC\to\CCC_T$ is strictly monoidal \cite[Proposition 1.2.2]{Seal13}. 


\begin{theorem}\label{thm:maybe monad on qSet}
The endofunctor $\M$ on $\qSet$ given by $\X\mapsto\X\uplus\mathbf 1$ is a symmetric monoidal monad. More specifically:
\begin{itemize}
    \item The $\X$-component of the unit $H$ of $\M$ is $J_\X:\X\to\X\uplus\mathbf 1$.
    \item Given $\X\in\qSet$, if we write $\M\X=\X\uplus\mathbf 1_1$ and $\M^2\X=\X\uplus\mathbf 1_1\uplus\mathbf 1_2$, where $\mathbf 1_i=\Q\{\CC_i\}$, then the $\X$-component $M_\X:\M^2\X\to\M\X$ of the multiplication of $\M$ is given by
\[M_\X(X,X')=\begin{cases}
L(\CC_2,\CC_1), & X=\CC_2, X'=\CC_1,\\
\CC 1_X, & X=X',\\
0, & \text{otherwise}
\end{cases}\]
for each $X\atomof\M^2\X$ and each $X'\atomof\M\X$.
\item For each $\X,\Y\in\qSet$ the double strength (also called the \emph{structural constraint}) $K_{\X,\Y}:\M\X\times\M\Y\to\M(\X\times\Y)$ is given for each $X\atomof\M\X$, $Y\atomof\M\Y$ and each $Z\atomof\M(\X\times\Y)$ by 
\begin{equation} 
K_{\X,\Y}(X\otimes Y,Z) = \begin{cases} \CC 1_{X\otimes Y}, & X\neq \CC_1\neq  Y\text{ and }Z=X\otimes Y,\\
L(X\otimes Y,Z), & Z=\CC_1 \text{ and }(X=\CC_1\text{ or }Y=\CC_1),\\
0, & \text{otherwise.}
\end{cases}
\end{equation}
\end{itemize} 
\end{theorem}

\begin{proof}
It follows from \cite[Example 5.1.4]{Riehl} that $\M$ is a monad on $\qSet$. This monad is obtained from an adjunction as follows. First, let $\qSet_*$ be the coslice category under $\mathbb 1$, whose objects are pairs $(\X,A)$ consisting of a quantum set $\X$ and a function $A:\mathbf 1\to\X$. Morphisms $F:(\X,A)\to(\Y,B)$ are functions $F:\X\to\Y$ such that $F\circ A=B$. Let $\G:\qSet_*\to\qSet$ be the forgetful functor $(\X,A)\mapsto \X$. Then $\G$ is right adjoint to the functor $\F:\qSet\to\qSet_*$ that acts on objects by $\X\mapsto (\X\uplus \mathbf 1,J_\mathbf 1)$ and on morphisms by $F\mapsto F\uplus I_\mathbf 1$. The $\X$-component of the unit $H$ of the adjunction $\F\dashv\G$ is $J_\X:\X\to\X\uplus\mathbf 1$, and the $(\X,A)$-component of the counit $E$ is the morphism $[I_\X,A]:(\X\uplus \mathbf 1,J_\mathbf 1^{\X\uplus\mathbf 1})\to(\X,A)$. It follows that the $\X$-component $M_\X=\G E_{\F\X}$ of the multiplication of the induced monad $\M$ is the function $[I_{\X\uplus\mathbf 1_1},J_{\mathbf 1_2}]:\X\uplus\mathbf 1_1\uplus\mathbf 1_2\to\X\uplus\mathbf 1_1$, which is precisely the expression for $M_\X$ in the statement. 

We thank the anonymous reviewer for the observation that the maybe monad on a semicartesian closed category with coproducts is automatically a symmetric monoidal monad, so that we can omit the direct verification of the axioms for a double strength in \cite[1.2]{Seal13}, which is straightforward but tedious.
We sketch how the semicartesian closed structure yields the double strength $K$.
Since $\qSet$ is symmetric monoidal closed, the monoidal product $\times$ preserves coproducts, whence for quantum sets $\X$ and $\Y$ there exists a bijection \[D_{\X,\Y}:(\X\uplus\mathbf 1)\times(\Y\uplus \mathbf 1)\to(\X\times\Y)\uplus(\X\times\mathbf 1)\uplus (\mathbf 1\times\Y)\uplus(\mathbf 1\times\mathbf 1).\] Clearly, the domain and codomain of $D_{\X,\Y}$ have the same atoms, namely $X\otimes Y$, $X\otimes\mathbb C$, $\mathbb C\otimes Y$ and $\mathbb C\otimes\mathbb C$ for $X\atomof\X$ and $Y\atomof\Y$. Hence, clearly $D_{\X,\Y}$ is the `identity'. Now, we obtain $K_{\X,\Y}$ as the composition $\big(I_{\X\times \Y}\uplus (L_\mathbf 1\circ [ !_\X\times I_\mathbf 1,I_\mathbf 1\times !_\Y,I_{\mathbf 1\times\mathbf 1}])\big)\circ D_{\X,\Y}.$ 
\end{proof}

\subsection{The lift monad on $\qCPO$}

Next we describe how to lift a given quantum poset $(\X,R)$ to pointed posets, which will be achieved by equipping $\M\X$ with an appropriate order. We will denote the added one-dimensional atom in $\M\X$ by $\CC_\perp$ instead of $\CC$ in order to emphasize that it will function as a least element (it  corresponds to an element since it is one-dimensional), hence we will write $\M\X=\X\uplus\Q\{\CC_\perp\}$. 
We note that if $\X$ already has an atom $\CC_\perp$ (for instance because it is the result of a previous lifting operation), we relabel that atom, and call the newly added atom $\CC_\perp$. This is also the common practice with ordinary posets, where the added bottom element is usually denoted by $\perp$, even when the given poset already had a bottom element, which then typically is renamed. 

\begin{proposition}\label{prop:lift functor qPOS}
Let $(\X,R)$ be a quantum poset. Define $(\X,R)_\perp=(\M\X,R_\perp)$, where 
\[R_\perp(X,Y)=\begin{cases}
R(X,Y), & X,Y\atomof\X;\\
L(X,Y), & X=\CC_\perp;\\
0, & X\atomof\X,Y=\CC_\perp. 
\end{cases}\]
Then $(\X,R)_\perp$ is a quantum poset, and $J_\X:(\X,R)\to(\X,R)_\perp$ is an order embedding such that
\begin{equation}\label{eq:Jorderembedding}
    R_\perp\circ J_\X=J_\X\circ R.
\end{equation}
\end{proposition}
\begin{proof}
We first show that $R_\perp$ is an order on $\M\X$. Let $X,Y\atomof\M\X$. In order to check that $I_{\M\X}(X,Y)\leq R_\perp(X,Y)$, we only have to check the case that $X=Y$, since $I_\X(X,Y)=0$ if $X\neq Y$. First assume $X\neq \CC_\perp$. Then  $I_{\M\X}(X,X)=\CC 1_X=I_\X(X,X)\leq R(X,X)=R_\perp(X,X).$ If $X=\CC_\perp$, then 
$I_{\M\X}(X,X)= \CC 1_X\leq L(\CC_\perp,\CC_\perp)=R_\perp(X,X).$ We conclude that $I_{\M\X}\leq R_\perp$.
Next, we check that $R_\perp\circ R_\perp\leq R_\perp$. For $X,Y\atomof\X$, we have $R_\perp(X,\CC_\perp)=0$, hence 
\begin{align*}
(R_\perp\circ R_\perp)(X,Y) & =R_\perp(\CC_\perp,Y)\cdot R_\perp(X,\CC_\perp)\vee \bigvee_{Z\atomof\X}R(Z,Y)\cdot R(X,Z)\\
& =\bigvee_{Z\atomof\X}R(Z,Y)\cdot R(X,Z)=(R\circ R)(X,Y)\leq R(X,Y).
\end{align*}
For $X=\CC_\perp$, we have
$(R_\perp\circ R_\perp)(X,Y)  \leq L(X,Y)=R_\perp(X,Y).$
For $Y=\CC_\perp$, and $X\atomof\X$, we recall that $R_\perp(Z,Y)=0$ for each $Z\in\X$, hence
\begin{align*}
    (R_\perp\circ R_\perp)(X,Y) & =R_\perp(X,\CC_\perp)=\bigvee_{Z\atomof\X_\perp}R_\perp(Z,\CC_\perp)\cdot R_\perp(X,Z)=R_\perp(\CC_\perp,\CC_\perp)\cdot R_\perp(X,\CC_\perp)\\
    & = L(\CC_\perp,\CC_\perp)\cdot L(X,\CC_\perp)=L(X,\CC_\perp)=R_\perp(X,Y).
\end{align*}
We conclude that $R_\perp\circ R_\perp\leq R_\perp$.
Next, we check that $R_\perp\wedge R_\perp^\dag=I_{\M\X}$. We note that for each $X,Y\atomof\M\X$, we have 
$(R_\perp\wedge R_\perp^\dag)(X,Y) = R_\perp(X,Y)\wedge R_\perp^\dag(X,Y) = R_\perp(X,Y)\wedge R_\perp(Y,X)^\dag,$
hence for $X,Y\atomof\X$, we obtain $(R_\perp\wedge R_\perp^\dag)(X,Y) = R(X,Y)\wedge R(Y,X)^\dag= R(X,Y)\wedge R^\dag(X,Y) = (R\wedge R^\dag)(X,Y) =I_\X(X,Y)=\delta_{X,Y}\CC 1_X=I_{\M\X}(X,Y).$ 
For $X=\CC_\perp$, we have $R_\perp(Y,X)=L(Y,X)$ if $Y=\CC_\perp$, so if $Y=X$, and $R_\perp(Y,X)=0$ otherwise, hence $(R_\perp\wedge R_\perp^\dag)(X,Y)= L(X,Y)\wedge R_\perp(Y,X)^\dag=R_\perp(Y,X)^\dag = \delta_{X,Y}L(Y,X)^\dag =\delta_{X,Y}L(X,Y)=\delta_{X,Y}\CC 1_X=I_{\M\X}(X,Y).$
If $Y=\CC_\perp$, and $X\atomof\X$, then
$(R_\perp\wedge R_\perp^\dag)(X,Y) = 0\wedge L(Y,X)^\dag=0=I_{\M\X}(X,Y).$
Thus $R_\perp\wedge R_\perp^\dag=I_{\M\X}$, and we conclude that $R_\perp$ is an order on $\M\X$. Finally, for each $X\atomof\X$ and $Y\atomof\M\X$, we have \begin{align*}
    (R_\perp\circ J_\X)(X,Y) & =R_\perp(X,Y)=\begin{cases}
        R(X,Y), & Y\atomof\M\X,\\
        0, & Y=\CC_\perp
    \end{cases}\\
    & = J_\X(X,Y)\cdot R(X,Y)=\bigvee_{Z\atomof\X}J_\X(Z,Y)\cdot R(X,Z)=(J_\X\circ R)(X,Y),
\end{align*}
so (\ref{eq:Jorderembedding}) holds, which yields $J_\X^\dag\circ R_\perp\circ J_\X=J^\dag_\X\circ J_\X\circ R=R$, using that $J_\X$ is injective. Thus, $J_\X$ is an order embedding.
\end{proof}

If $(\X,R)$ is a quantum poset, we will sometimes write $\X_\perp$ instead of $\M\X$.

\begin{proposition}\label{prop:lift functor qPOS2}
We obtain an endofunctor $(-)_\perp:\qPOS\to\qPOS$ if for each monotone map $F:(\X,R)\to(\Y,S)$ between quantum posets $(\X,R)$ and $(\Y,S)$ we define $F_\perp:(\X,R)_\perp\to(\Y,S)_\perp$ by $F_\perp=\M F$,  
which satisfies 
\begin{equation}\label{eq:lift is natural}
 F_\perp\circ J_\X=J_\Y\circ F.
 \end{equation}
\end{proposition}
\begin{proof}
 We check that $F_\perp$ is monotone, for which we use the monotonicity of $F$, i.e., $R\circ F\leq F\circ S$. Then for $X\atomof\X_\perp$ and $Y\atomof\Y_\perp$, we have: 
 \begin{align*}
    (F_\perp\circ R_\perp)(X,Y) & = \bigvee_{Z\atomof\X_\perp}F_\perp(Z,Y)\cdot R_{\perp}(X,Z)\\
    & = \begin{cases}
    \bigvee_{Z\atomof\X}F(Z,Y)\cdot R(X,Z), & X\atomof\X,Y\atomof\Y,\\
       \bigvee_{Z\atomof\X}F(Z,Y)\cdot L(\CC_\perp,Z), & X=\CC_\perp,Y\atomof\Y\\
  L(\CC_\perp,\CC_\perp)\cdot 0, & X\atomof\X,Y=\CC_\perp \\ 
    L(\CC_\perp,\CC_\perp)\cdot L(\CC_\perp,\CC_\perp), & X=\CC_\perp, Y=\CC_\perp  \end{cases}\\
        & = \begin{cases}
    (F\circ R)(X,Y), & X\atomof\X,Y\atomof\Y,\\
       \bigvee_{Z\atomof\X}F(Z,Y)\cdot L(\CC_\perp,Z), & X=\CC_\perp,Y\atomof\Y\\
  0, & X\atomof\X,Y=\CC_\perp \\ 
    L(\CC_\perp,\CC_\perp), & X=\CC_\perp, Y=\CC_\perp  \end{cases}\\
    &  \leq \begin{cases}
    (S\circ F)(X,Y), & X\atomof\X,Y\atomof\Y,\\
    L(\CC_\perp,Y) , & X=\CC_\perp,Y\atomof\Y\\
    0, & X\atomof\X,Y=\CC_\perp\\
     L(\CC_\perp,\CC_\perp), & X=\CC_\perp,Y=\CC_\perp
    \end{cases}\\
    & = \begin{cases}
    \bigvee_{Z\atomof\Y}S(Z,Y)\cdot F(X,Z), & X\atomof\X,Y\atomof\Y,\\
    L(\CC_\perp,Y)\cdot L(\CC_\perp,\CC_\perp) , & X=\CC_\perp,Y\atomof\Y\\
    L(\CC_\perp,\CC_\perp)\cdot 0, & X\atomof\X,Y=\CC_\perp\\
      L(\CC_\perp,\CC_\perp)\cdot L(\CC_\perp,\CC_\perp), & X=\CC_\perp,Y=\CC_\perp
    \end{cases}\\
    & = \bigvee_{Z\atomof\M\Y} S_\perp(Z,Y)\cdot F_\perp(X,Z)\\
    & = (S_\perp\circ F_\perp)(X,Y)
\end{align*}
Let $X\atomof\X$ and $Y\atomof\Y_\perp$. By \cite[Lemma~A.5]{KLM20}, we have
\[(F_\perp\circ J_\X)(X,Y)=F_\perp(X,Y)=\begin{cases}
F(X,Y), & Y\atomof\Y,\\
0, & Y=\CC_\perp,
\end{cases}\]
whereas
\[  (J_\Y\circ F)(X,Y) = \bigvee_{Z\atomof\Y}J_\Y(Z,Y)\cdot F(X,Z)=\begin{cases}
F(X,Y), & Y\atomof\Y,\\
0, & Y=\CC_\perp,\end{cases},\]
thus (\ref{eq:lift is natural}) holds.

Next, we show that $(-)_\perp$ is a functor.
Consider the identity $I_\X:(\X,R)\to(\X,R)$. Then
\begin{align*}(I_\X)_\perp(X,Y) & =\begin{cases}
I_\X(X,Y), &  X,Y\atomof\X,\\
L(X,Y), & X=Y=\CC_\perp,\\
0, & \textrm{otherwise},
\end{cases}\\
& = \begin{cases}
\delta_{X,Y}\CC 1_X, & X,Y\atomof\X,\\
\CC 1_X, &  X=Y=\CC_\perp,\\
0, & \mathrm{otherwise}
\end{cases}\\
& = \delta_{X,Y}\CC 1_X = I_{\X_\perp}(X,Y).
\end{align*}
If $(\Z,T)$ is another quantum poset, and $G:\Y\to\Z$ is monotone, then
\begin{align*}
    (G_\perp\circ F_\perp)(X,Z)& =\bigvee_{Y\atomof\M\Y}G_\perp(Y,Z)\cdot F_\perp(X,Y)\\
    & = \begin{cases}
    \bigvee_{Y\atomof\Y}G(Y,Z)\cdot F(X,Y), & X\atomof\X,Z\atomof\Z\\
    L(X,Y)\cdot L(Y,Z), & X=\CC_\perp,Z=\CC_\Z,\\
    0, & \textrm{otherwise}.
    \end{cases}\\
    & =
    \begin{cases}
    (G\circ F)(X,Z), & X\atomof\X,Z\atomof\Z\\
    L(X,Z), & X=\CC_\perp,Z=\CC_\perp,\\
    0, & \textrm{otherwise}.
    \end{cases}\\
    & =    (G\circ F)_\perp(X,Z).
\end{align*}
We conclude that $(-)_\perp$ is indeed a functor.
\end{proof}

\begin{lemma}\label{lem:lift lemma}
Let $(\X,R)$ be a quantum poset, and for an atomic quantum set $\H$, let
\begin{equation}\label{eq:lift sequence1}E_1\sqsubseteq E_{2}\sqsubseteq E_{3}\sqsubseteq\cdots:\H\to\X_\perp
\end{equation}
be a monotone sequence. Then there exists a $k\in\NN$ and a decomposition $D:\H\to\Y\uplus\Z$ (cf. Definition \ref{def:decomposition}) such that there are functions 
\begin{equation}\label{eq:lift sequence}F_k\sqsubseteq F_{k+1}\sqsubseteq F_{k+2}\sqsubseteq\cdots:\Y\to\X,
\end{equation}
and a function $G:\Z\to\Q(\CC_\perp)$ such that $E_n=(F_n\uplus G)\circ D$ for each $n\geq k$.

If $(\X,R)$ is a quantum cpo, then so is $(\X_\perp,R_\perp)$, the embedding $J_\X:\X\to\X_\perp$ is Scott continuous, and the limits $E_\infty$ and $F_\infty$ of the sequences in (\ref{eq:lift sequence1}) and (\ref{eq:lift sequence}), respectively, are related via $E_\infty=(F_\infty\uplus G)\circ D$.
\end{lemma}
\begin{proof}
Let $\H$ be an atomic quantum set, and let $E_1\sqsubseteq E_2\sqsubseteq\cdots:\H\to\X_\perp$ be a monotone sequence. If $\X$ is already a pointed quantum cpo, we can relabel its bottom element, hence independent of $\X$ being pointed, we can assume that $\CC_\perp$ is not an atom of $\X$. As a consequence, we have $\X_\perp=\X\uplus\Q\{\CC_\perp\}$ as a quantum set. 
Lemma \ref{lem:decomposition} gives us  decompositions $D_n:\H\to\Y_n\uplus\Z_n$ and functions $F_n:\Y_n\to\X$, $G_n:\Z_n\to\Q\{\CC_\perp\}$ such that $E_n=(F_n\uplus G_n)\circ D_n$.
Since (\ref{eq:lift sequence1}) holds, we can now apply Lemma \ref{lem:coproduct of qposets decomposition} to conclude that for each $n\in\NN$, we have 
$Y_n\perp Z_{n+1}$. By definition of a decomposition, we have 
\begin{equation}\label{eq:lift decomposition}
H=Y_{i}\oplus Z_{i}
\end{equation}
for each $i\in\NN$, so $H=Y_{n+1}\oplus Z_{n+1}$, whence $Y_n\leq Y_{n+1}$. Thus we obtain a monotone sequence
$Y_1\leq Y_2\leq\ldots$,
and since $Y_i\leq H$ for each $i\in\NN$ and $H$ is finite-dimensional, it follows that there is some $k\in\NN$ such that $Y_n=Y_{k}$, so $\Y_n=\Y_k$ for each $n\geq k$.
It follows from (\ref{eq:lift decomposition}) that also $Z_n=Z_k$, so $\Z_n=\Z_k$ for each $n\geq k$. Let $D=D_k$ and $G=G_k$, then $D_n=D$ for each $n\geq k$ by (3) of Definition \ref{def:decomposition}. Again applying Lemma \ref{lem:coproduct of qposets decomposition} yields
$F_k\sqsubseteq F_{k+1}\sqsubseteq F_{k+2}\sqsubseteq\ldots:\Y_k\to\X.$
Since $\Q\{\CC_\perp\}$ is terminal in $\qSet$, it follows from $\Z_n=\Z_k$ for each $n\geq k$ that we also have that $G_n=G$ for each $n\geq k$.

Since $\X$ is a quantum cpo, and $\Y$ is subatomic, there is some $F_\infty$ such that $F_n\nearrow F_\infty$ for $n\geq k$, i.e.,
$\bigwedge_{n\geq k}R\circ F_n=R\circ F_\infty.$
Let $E_\infty=(F_\infty\uplus G)\circ D_k$. Since $E_1\sqsubseteq E_2\sqsubseteq\ldots$, we have
$R_\perp\circ E_1\geq R_\perp\circ E_2\ldots,$
hence
$\bigwedge_{n\in\NN}R_\perp\circ E_n=\bigwedge_{n\geq k}R_{\perp}\circ E_n.$

For any $E:\H\to\X_\perp$, we have
$(R_\perp\circ E)_{H}^{\CC_\perp}=\bigvee_{X\atomof\X_\perp}(R_\perp)^X_{\CC_\perp}E_{H}^X=L(\CC_\perp,\CC_\perp)E_{H}^{\CC_\perp}=E_{H}^{\CC_\perp},$
hence
\begin{align*}
\left(\bigwedge_{n\geq k}R_\perp\circ E_n\right)_{H}^{\CC_\perp} & =\bigwedge_{n\geq k} (R_\perp\circ E_n)_{H}^{\CC_\perp}=\bigwedge_{n\geq k }(E_n)_{H}^{\CC_\perp} =\bigwedge_{n\geq k}\big((F_n\uplus G)\circ D\big)_{H}^{ \CC_\perp}\\
& =\bigwedge_{n\geq k}\bigvee_{W\atomof\Y\uplus\Z}(F_n\uplus G)_{W}^{\CC_\perp} D_{H}^W=\bigwedge_{n\geq k}\bigvee_{W\atomof\Z}G_{W}^{\CC_\perp} D_{H}^W \\
& =\bigvee_{W\atomof\Y\uplus\Z}(F_\infty\uplus G)_{W}^{\CC_\perp} D_{H}^W = \big((F_\infty\uplus G)\circ D\big)_{H}^{ \CC_\perp}=(E_\infty)_{H}^{ \CC_\perp}=(R_\perp\circ E_\infty)_{H}^{ \CC_\perp}.
\end{align*}
Now let $X\atomof\X$. Then for any $E:\H\to\X_\perp$, we have
$(R_\perp\circ E)_{H}^{X}=\bigvee_{Y\atomof\X_\perp}(R_\perp)^X_{Y}E_{H}^Y=\bigvee_{Y\atomof\X}R^X_{Y}E_{H}^Y=(R\circ E)_{H}^X,$
hence
\begin{align*}
\left(\bigwedge_{n\geq k}R_\perp\circ E_n\right)_{H}^{X} & =\bigwedge_{n\geq k} (R_\perp\circ E_n)_{H}^{X}=\bigwedge_{n\geq k }(R\circ E_n)_{H}^{X} =\bigwedge_{n\geq k}\big(R\circ (F_n\uplus G)\circ D\big)_{H}^{X}
\\
& =\bigwedge_{n\geq k}\bigvee_{V\atomof\X,W\atomof\Y\uplus\Z}R^X_V(F_n\uplus G)_{W}^{V} D_{H}^W=\bigwedge_{n\geq k}\bigvee_{V\atomof\X,W\atomof\Y}R_V^X(F_n)_{W}^{V} D_{H}^W
\\
&= \bigwedge_{n\geq k}\bigvee_{V\atomof\X}R_V^X (F_n)_{Y}^VD_{H}^{Y}=\bigwedge_{n\geq k}(R\circ F_n)^X_{Y}D_{H}^{Y}=\bigwedge_{n\geq k}(R\circ F_n)^X_Y\mathrm{proj}_{Y}
\\ & =(R\circ F_\infty)^X_{Y}\mathrm{proj}_{Y}=\bigvee_{V\atomof\X}R_V^X(F_\infty)^V_{Y}D_{H}^{Y}  = \bigvee_{V\atomof\X,W\atomof\Y_k}R_V^X(F_\infty)^V_{W}D_{H}^{W}\\
& =\bigvee_{V\atomof\X,W\atomof\Y\uplus\Z}R_V^X(F_\infty\uplus G)_{W}^{V} D_{H}^W = \big(R\circ (F_\infty\uplus G)\circ D\big)_{H}^{X}\\
& =(R\circ E_\infty)_{H}^{ X}=(R_\perp\circ E_\infty)_{H}^{X}.
\end{align*}
We conclude that $E_n\nearrow E_\infty$.

By Proposition \ref{prop:lift functor qPOS}, $J_\X$ is monotone.
It remains to show that $J_\X$ is Scott continuous. So let $E_1\sqsubseteq E_2\sqsubseteq\ldots:\H\to\X$ be a monotone sequence with limit $E_\infty$. 
 We need to show that $J_\X\circ E_n\nearrow J_\X\circ E_\infty$. For any $E:\H\to\X$, and $X\atomof\X_\perp$, using (\ref{eq:Jorderembedding}) from Proposition \ref{prop:lift functor qPOS}, we find
 \[ (R_\perp\circ J_X\circ E)_{H}^X=(J_\X\circ R\circ E)_{H}^X=\bigvee_{Y\atomof\X_\perp}(J_\X)_Y^X(R\circ E)_{H}^Y=\begin{cases}(R\circ E)_{H}^X, & X\atomof\X\\
     0, & X=\CC_\perp. 
    \end{cases},\]
    hence we obtain
        $\left(\bigwedge_{n\in\NN}R_\perp\circ J_\X\circ E_n\right)_{H}^{\CC_\perp}=\bigwedge_{n\in\NN}(R_\perp\circ J_\X\circ E_n)_{H}^{\CC_\perp}=0=(R_\perp\circ J_X\circ E_\infty)_{H}^{\CC_\perp},$
 whereas for $X\atomof\X$, we find
\begin{align*}\left(\bigwedge_{n\in\NN}R_\perp\circ J_\X\circ E_n\right)_{H}^X & =\bigwedge_{n\in\NN}(R_\perp\circ J_\X\circ E_n)_{H}^X=\bigwedge_{n\in\NN}(R\circ E_n)_{H}^X=\left(\bigwedge_{n\in\NN}R\circ E_n\right)_{H}^X\\
& =(R\circ E_\infty)_{H}^X=(R_\perp\circ J_\X\circ E_\infty)_{H}^X,
\end{align*}
where we used $E_n\nearrow E_\infty$ in the penultimate equality. 
We conclude that $\bigwedge_{n\in\NN}R_\perp\circ J_\X\circ E_n=R_\perp\circ J_\X\circ E_\infty$,
hence indeed $J_\X\circ E_n\nearrow J_\X\circ E_\infty$.
\end{proof}

\begin{lemma}\label{lem:lift lemma 2}
Let $(\X,R)$ and $(\V,S)$ be quantum cpos, and let $K:\X_\perp\to\V$ be a monotone map. If $K\circ J_\X$ is Scott continuous, so is $K$.
\end{lemma}
\begin{proof}
Assume that $K\circ J_\X$ is Scott continuous. Let $\H$ be an atomic quantum set, and
let $E_1\sqsubseteq E_2\sqsubseteq\cdots:\H\to\X_\perp$ be a monotone sequence with limit $E_\infty$.
By Lemma \ref{lem:lift lemma} there is some $k\in\NN$ such that there is a decomposition $D:\H\to \Y\uplus \Z$, there are functions $F_k\sqsubseteq F_{k+1}\sqsubseteq\cdots\sqsubseteq F_\infty:\Y\to\X$ with $F_n\nearrow F_\infty$ for $n\geq k$, and there is a function $G:\Z\to\Q\{\CC_\perp\}$ such that $E_n=(F_n\uplus G)\circ D$ for $k\leq n\leq \infty$.
Since $F_n\nearrow F_\infty$ for $n\geq k$, it follows from the Scott continuity of $K\circ J_\X$ that $K\circ J_\X\circ  F_n\nearrow K\circ J_\X\circ F_\infty$ for $n\geq k$, i.e., 
\begin{equation}\label{eq:K is Scott continuous for lift}
\bigwedge_{n\geq k}S\circ K\circ J_\X\circ F_n=S\circ K\circ J_\X\circ F_\infty.
\end{equation}
Fix $n\in\{k,k+1,\ldots,\infty\}$. Then for each $X\atomof\X_\perp$, we have
\begin{align*} (E_n)_{H}^X 
& =
((F_n\uplus G)\circ D)_{H}^X
=
\bigvee_{W\atomof\Y\uplus \Z}(F_n\uplus G)_W^X D_{H}^W 
= 
\begin{cases}
\bigvee_{W\atomof\Y}(F_n)_W^XD_{H}^W, & X\atomof\X,\\
\bigvee_{W\atomof\Z}G_W^XD_{H}^W, & X=\CC_\perp\end{cases}\\
 & =  
\begin{cases}
(F_n)_Y^X\proj_Y, & X\atomof\X,\\
G_Z^{\CC_\perp}\proj_Z, & X=\CC_\perp, \end{cases}
\end{align*}
since $\Y$ and $\Z$ are subatomic, hence of the form $\Q\{Y\}$ and $\Q\{Z\}$, respectively, for (possibly zero-dimensional) Hilbert spaces $Y$ and $Z$. 

Using \cite[Lemma~A.5]{KLM20} in the third equality of the next calculation, we find for each $W\atomof\V$:
\begin{align*}
    (K\circ E_n)_H^W & =  \bigvee_{X\atomof\X_\perp} K_X^W\cdot(E_n)_H^X =\bigvee_{X\atomof\X}K_X^W\cdot(E_n)_H^X\vee K_{\CC_\perp}^W\cdot (E_n)_H^{\CC_\perp}\\
    & =\bigvee_{X\atomof\X}(K\circ J_\X)_X^W\cdot(E_n)_H^X\vee (K\circ J_{\CC_\perp})_{\CC_\perp}^W\cdot(E_n)_H^{\CC_\perp}\\
    & = \bigvee_{X\atomof\X}(K\circ J_\X)_X^W\cdot (F_{n})_Y^X\proj_Y\vee (K\circ J_{\CC_\perp})_{\CC_\perp}^W\cdot G_Z^{\CC_\perp}\proj_Z\\
    & = (K\circ J_\X\circ F_n)_Y^W\proj_Y\vee (K\circ J_{\CC_\perp}\circ G)_Z^{W}\proj_Z.
\end{align*}

Then for each $V\atomof\V$, we find
\begin{align*}
    (S\circ K\circ E_n)_H^V &= \bigvee_{W\atomof\V} S_W^V\cdot (K\circ E_n)_H^W\\
    & = \bigvee_{W\in\V}S_W^V\cdot\big((K\circ J_\X\circ F_n)_Y^W\proj_Y\vee (K\circ J_{\CC_\perp}\circ G)_Z^{W}\proj_Z\big)\\
    & = \bigvee_{W\in\V}S_W^V\cdot(K\circ J_\X\circ F_n)_Y^W\proj_Y\vee \bigvee_{W\in\V}S_W^V\cdot(K\circ J_{\CC_\perp}\circ G)_Z^{W}\proj_Z\\
    & = (S\circ K\circ J_\X\circ F_n)_Y^V\proj_Y\vee(S\circ K\circ J_{\CC_\perp}\circ G)_Z^{V}\proj_Z
\end{align*}

We note that two subspaces $N_1,N_2$ of the Hilbert space $L(H,V)$ are orthogonal if every operator in $x\in N_1$ is orthogonal to every operator $y\in N_2$, which is the case precisely when $\Tr(x^\dag y)=0$. We note that any operator $x$ in the subspace $(S\circ K\circ J_\X\circ F_n)_{Y}^V\proj_Y$ is of the form $a\circ\proj_Y$ for some $a\in (S\circ K\circ J_\X\circ F_n)_{Y}^V$, and every operator $y$ in $ (S\circ K\circ J_\X\circ G)_{Z}^V\proj_Z$ is of the form $b\circ\proj_Z$ for some $b\in  (S\circ K\circ J_\X\circ G)_{Z}^V$. Then
\[\Tr(x^\dag y)=\Tr(\proj_Y^\dag\circ a^\dag\circ b\circ\proj_Z)=\Tr(a^\dag\circ b\circ\proj_Z\circ\proj_Y^\dag)=0,\]
since by definition of a decomposition $Y$ and $Z$ are orthogonal subspaces of $H$. We conclude that $(S\circ K\circ J_\X\circ F_n)_{Y}^V\proj_Y$ and $ (S\circ K\circ J_{V}\circ G)_Z^{\CC_\perp}\proj_Z$ are orthogonal to each other. 
As a consequence, the subspaces are said to be compatible with each other in the orthomodular sense (cf. \cite[Definition 1.2.1]{PtakPulmannova91}). Using (\ref{eq:K is Scott continuous for lift}) in the second equality of the next calculation, we obtain:
\begin{align*}
    \bigwedge_{n\in\NN} (S\circ K\circ J_\X\circ F_n)_{Y}^V\proj_Y & = \left(\bigwedge_{n\in\NN} S\circ K\circ J_\X\circ  F_n\right)_{Y}^V\proj_Y = (S\circ K\circ J_\X\circ F_\infty)_{Y}^V\proj_Y,
\end{align*}
which is also orthogonal to $(S\circ K\circ J_{\CC_\perp}\circ G)_Z^{V}\proj_Z$ (since the previous analysis was for all $n\in\{k,k+1,\ldots,\infty\}$), we are allowed to apply the Foulis-Holland Theorem (see for instance \cite[Proposition 1.3.8]{PtakPulmannova91}) to obtain the second equality in the next calculation:  
\begin{align*}
\bigwedge_{n\in\NN}(S\circ K\circ E_n)_H^V & =     \bigwedge_{n\in\NN}\big( (S\circ K \circ J_\X\circ F_n)_{Y}^V\proj_Y  \vee (S\circ K\circ J_{\CC_\perp}\circ G)_Z^{V}\proj_Z\big)\\
    & = \left(\bigwedge_{n\in\NN} (S\circ K\circ J_\X\circ F_n)_{Y}^V\proj_Y\right) \vee (S\circ K\circ J_{\CC_\perp}\circ G)_Z^{V}\proj_Z\\
    & = (S\circ K\circ J_\X\circ F_\infty)_{Y}^V\proj_Y\vee (S\circ K\circ J_{\CC_\perp}\circ G)_Z^{V}\proj_Z\\
    & = (S\circ K\circ E_\infty)_H^V,
\end{align*}
where we used (\ref{eq:K is Scott continuous for lift}) in the penultimate equality.
We conclude that $S_\perp\circ K\circ E_n\nearrow S_\perp\circ K\circ E_\infty$, so $K$ is Scott continuous.
\end{proof}

\begin{lemma}\label{lem:monad maps-J}
Let $(\X,R)$ be a quantum cpo. Then $J_\X:(\X,R)\to(\X,R)_\perp$ is a Scott continuous order embedding such that 
\begin{equation}\label{eq:lift condition}
J_\X\circ R=R_\perp\circ J_\X.
\end{equation} 
\end{lemma}
\begin{proof}
Let $J_\X:\X\to\X_\perp$ be the inclusion. Then for $X\atomof\X$ and $Y\atomof\X_\perp$, we have
\[(J_\X\circ R)(X,Y)=\bigvee_{Z\atomof\X}J_\X(Z,Y)R(X,Z)=\begin{cases}
0, & Y=\CC_\perp,\\
R(X,Y), & Y\atomof\X,
\end{cases}\]
whereas \cite[Lemma~A.5]{KLM20} gives $(R_\perp\circ J_X)(X,Y)=R_\perp(X,Y),$
which yields (\ref{eq:lift condition}).
 Since $J_\X$ is injective, it follows from
 (\ref{eq:lift condition}) that
$ R=J_\X^\dag\circ J_\X\circ R=J_\X^\dag\circ R_\perp\circ J_\X, $
 so $J_\X$ is an order embedding. 
Scott continuity follows from Lemma \ref{lem:lift lemma}.
\end{proof}

\begin{lemma}\label{lem:monad maps-M}
Let $(\X,R)$ be a quantum cpo. Then 
$M_\X  :(\X,R)_{\perp\perp}\to(\X,R)_\perp$ is Scott continuous.
\end{lemma}
\begin{proof}
 Since we have to apply $(-)_\perp$, hence $\M$ twice, we will write $\X_\perp=\M\X=\X\uplus\Q\{\CC_1\}$ and $\X_{\perp\perp}=\M^2\X=\X\uplus\Q\{\CC_1\}\uplus\Q\{\CC_2\}$.
Let $X\atomof\M^2\X$ and $Y\atomof\M\X$. Then
\begin{align*}
    (M_\X\circ R_{\perp\perp})_X^Y & =\bigvee_{Z\atomof\M^2\X} (M_\X)_Z^Y\cdot (R_{\perp\perp})_X^Z=\delta_{Y,\CC_1}L(\CC_2,\CC_1)\cdot (R_{\perp\perp})_X^{\CC_2}\vee\bigvee_{Z\atomof\M\X}\delta_{Z,Y}(R_{\perp\perp})_{X}^Z\\
    & =\delta_{Y,\CC_1}\delta_{X,\CC_2}L(\CC_2,\CC_1)\cdot L(X,\CC_2)\vee (R_{\perp\perp})_X^Y=\delta_{Y,\CC_1}\delta_{X,\CC_2}L(\CC_2,\CC_1)\vee (R_{\perp\perp})_X^Y\\
    & =\begin{cases}
    L(\CC_2,\CC_1), & X=\CC_2,Y=\CC_1,\\
    R_{\perp}(X,Y), & X\atomof\M\X,\\
    0, & \text{otherwise}
    \end{cases}\\
     & \leq \begin{cases}
    L(\CC_1,Y)\cdot L(\CC_2,\CC_1), & X=\CC_2,\\
    R_{\perp}(X,Y), & X\atomof\M\X
    \end{cases}\\
& = \begin{cases}
    (R_\perp)_{\CC_1}^Y\cdot L(\CC_2,\CC_1), & X=\CC_2\\
    (R_\perp)_X^Y, &  X\atomof\M\X
    \end{cases}\\
    & = \bigvee_{Z\atomof\M\X}(R_\perp)_Z^Y\cdot (M_\X)_X^Z\\
    & = (R_\perp\circ M_\X)_X^Y,
\end{align*}
which shows that $M_\X$ is monotone. 

We note that since the $\X$-components of the multiplication and the unit of the monad $\M$ are $M_\X$ and $J_\X$, respectively, we have $M_\X\circ J_{\X_\perp}=I_{\X_\perp}$, which obviously is Scott continuous. Hence we can  apply Lemma \ref{lem:lift lemma 2}
to conclude that $M_\X$ is Scott continuous.
\end{proof}
 
\begin{lemma}\label{lem:monad maps-K}
Let $(\X,R)$ and $(\Y,S)$ be quantum posets. Then $K_{\X,\Y}:(\X,R)_\perp\times(\Y,S)_\perp\to(\X\times\Y,R\times S)_\perp$ is Scott continuous.
\end{lemma}
\begin{proof}
Since $K$ is a double strength for the monad $\M$ on $\qSet$, it follows from \cite[1.2]{Seal13} that $K\circ (J_\X\times J_\Y)=K\circ (H_\X\times H_\Y)=H_{\X\times\Y}=J_{\X\times\Y}$, which is Scott continuous by Lemma \ref{lem:monad maps-J}. Similar to the proof of Lemma \ref{lem:lift lemma 2}, one can show that the Scott continuity of $K\circ (J_\X\times J_\Y)$ implies that $K$ is both Scott continuous in the first and in the second variable, hence Scott continuous by Proposition \ref{prop:Scott continuous iff Scott continuous in both variables}.
\end{proof}

\begin{theorem}\label{thm:monoidal monad}
The triple $( (-)_\perp,M,H)$ is a symmetric monoidal monad on $\qCPO$ with double strength $K$.
\end{theorem}
\begin{proof}
By Theorem \ref{thm:maybe monad on qSet}, $(\M,M,H)$ is a symmetric monoidal monad on $\qSet$. Let $(\X,R)$ be a quantum cpo. Then $\M\X$ is the underlying quantum set of $(\X,R)_\perp$, and the $\X$-components $J_\X$ and $M_\X$ of $H$ and $M$ are Scott continuous by Lemmas \ref{lem:monad maps-J} and \ref{lem:monad maps-M}, whence $((-)_\perp,M,H)$ is a monad on $\qCPO$. Moreover, if $(\Y,S)$ is another quantum cpo, then the $(\X,\Y)$-component of the double strength $K$ of $(\M,M,H)$ is Scott continuous by Lemma \ref{lem:monad maps-K}, hence $((-)_\perp,M,H)$ is a symmetric monoidal monad on $\qCPO$ with double strength $K$.
\end{proof}

We denote the Kleisli category of $(-)_\perp$ by $\mathbf{Kl}$. So its objects are the objects of $\qCPO$, and morphisms $F\in\mathbf{Kl}(\X,\Y)$ if and only if $F\in\qCPO(\X,\Y_\perp)$. Composition $G\circ F$ of $F$ with $G\in\mathbf{Kl}(\Y,\Z)$ is given by $M_\Z\circ G_\perp\circ F$.
Furthermore, we define $\V:\qCPO\to\mathbf{Kl}$ to be the functor whose action on objects is given by $\V((\X,R))=(\X,R)_\perp$, and that acts on morphisms from $(\X,R)$ to $(\Y,S)$ in $\qCPO$ by $F\mapsto H_{(\Y,S)}\circ F$.

\begin{corollary}\label{cor:Kl is monoidal}

There exists a symmetric monoidal product $\odot$ on $\mathbf{Kl}$ such that the functor $\mathcal V:\qCPO\to\mathbf{Kl}$ is strict monoidal. In particular, we have $\X\odot\Y=\X\times\Y$ on objects, and for morphisms $F:\X\circlearrow\X'$ and $G:\Y\circlearrow\Y'$, we have that $F\odot G:\X\odot\Y\circlearrow\X'\odot\Y'$ is given by the following composition
\[  \X\times\Y\xrightarrow{F\times G}\X_\perp'\times\Y_\perp'\xrightarrow{K_{\X',\Y'}}(\X'\times\Y')_\perp,\]
so $F\odot G=K_{\X',\Y'}\circ (F\times G)$.
\end{corollary}
\begin{proof}Since $(-)_\perp$ is a symmetric monoidal monad on $\qCPO$ by Theorem \ref{thm:monoidal monad}, this follows directly from \cite{Seal13}*{Proposition 1.2.2} and its proof.
\end{proof}

\subsection{Pointed quantum cpos}

\begin{lemma}\label{lem:pointed order}
Let $(\X,R)$ be a quantum poset such for which there is a one-dimensional atom $X_\perp\atomof\X$ such that $R(X_\perp,X)=L(X_\perp,X)$ for each $X\atomof\X$. Then $R(X,X_\perp)=0$ for each $X\atomof\X$ such that $X\neq X_\perp$.
\end{lemma}
\begin{proof}
By a direct calculation we have \begin{align*} 0 & = I_\X(X,X_\perp)=(R\wedge R^\dag)(X,X_\perp)=R(X,X_\perp)\wedge R^\dag(X,X_\perp)=R(X,X_\perp)\wedge R(X_\perp,X)^\dag\\
& = R(X,X_\perp)\wedge L(X_\perp,X)^\dag=R(X,X_\perp)\wedge L(X,X_\perp)=R(X,X_\perp)\end{align*} for each $X\atomof\X$ such that $X\neq X_\perp$.
\end{proof}

\begin{lemma}\label{lem:pointed qcpo}
Let $(\X,R)$ be a quantum cpo. Then the following conditions are equivalent:
\begin{itemize}
    \item[(a)] There is a (necessarily unique) quantum cpo $(\Y,S)$ such that $(\X,R)=(\Y,S)_\perp$;
    \item[(b)] there exists a Scott continuous function $B_\X:\mathbf 1\to\X$ such that for each atomic quantum set $\H$ the composition $B_\X\circ !_\H$ is the least element of $\qSet(\H,\X)$;
    \item[(c)] There exists a (necessarily unique) atom $\CC_\perp\atomof\X$ such that $R(\CC_\perp,X)=L(\CC_\perp,X)$ for each $X\atomof\X$.
\end{itemize}
The function $B_\X$ in (b) and the atom $\CC_\perp$ in (c) are related to each other via $B_\X(\CC,X)=\delta_{X,\CC_\perp}L(\CC,X)$.  
\end{lemma}
\begin{proof}
Assume that $(\X,R)=(\Y,S)_\perp$ for some quantum cpo $(\Y,S)$.
We define $B_\X$ as the function $B_\X(\CC,X)=\delta_{X,\CC_\perp}L(\CC,X)$. Let $\H$ be atomic. We have $B_\X\circ !_\H(H,X)=B_\X(\CC,X)\cdot L(H,\CC)=\delta_{X,\CC_\perp}L(H,X)$, hence $(R\circ B_\X\circ !_\H)(H,Y)=(S_\perp\circ B_\X\circ !_\H)(H,Y)=\bigvee_{X\atomof\X}S_\perp(X,Y)\cdot \delta_{X,\CC_\perp}L(H,X)=L(\CC_\perp,Y)\cdot L(H,\CC_\perp)=L(H,Y)$. So, for each $F\in\qSet(\H,\X)$, we have $R\circ F\leq R\circ B_\X\circ !_{\H}$, expressing that $B_\X\circ !_\H$ is the least element of $\qSet(\H,\X)$.

Now assume that there exists a function $B_\X:\mathbf 1\to\X$ satisfying the conditions in (b). By Lemma \ref{lem:b is order iso} and the paragraph preceding it, there is a one-dimensional atom $\CC_\perp\atomof\X$ such that $B_\X(\CC,X)=\delta_{X,\CC_\perp}L(\CC,X)$ for each $X\atomof\X$, which is Scott continuous by Example \ref{ex:function with trivially ordered domain is Scott continuous}.

Then for any other atom $X\atomof\X$, we have that $B_\X\circ !_{\Q\{X\}}$ is smaller than or equal to the embedding $J_X:\Q\{X\}\to\X$ in $\qSet(\Q\{X\},\X)$. Hence $J_X\leq R\circ B_\X\circ !_{\Q\{X\}}$, so 
\begin{align*}
\CC 1_X & =J_X(X,X)\leq (R\circ B_\X\circ !_{\Q\{X\}})(X,X)=\bigvee_{Y\atomof\X}R(Y,X)\cdot B_\X(\CC,Y)\cdot !_\X(X,\CC)\\
& =R(\CC_\perp,X)\cdot L(\CC,\CC_\perp)\cdot L(X,\CC)=R(\CC_\perp,X)\cdot L(X,\CC_\perp).
\end{align*}
Multiplying both sides on the right with $L(\CC_\perp,X)$ yields
\[L(\CC_\perp,X)\leq R(\CC_\perp,X)\cdot L(X,\CC_\perp)\cdot L(\CC_\perp,X)=R(\CC_\perp,X)\cdot L(\CC_\perp,\CC_\perp)=R(\CC_\perp,X),\] forcing $R(\CC_\perp,X)=L(\CC_\perp,X).$

Finally, to show that (c) implies (a), we first note that $R(X,\CC_\perp)=0$ by Lemma \ref{lem:pointed order} for each atom $X\neq\CC_\perp$. 
Let $\Y$ be the quantum set obtained from $\X$ by removing $\CC_\perp$. We define the order $S$ on $\Y$ as the relative order, i.e., $S=J_\Y^\dag\circ R\circ J_\Y$. We show that $(\Y,S)$ is a quantum cpo: Let $\H$ be an atomic quantum set and let $K_1\sqsubseteq K_2\sqsubseteq\cdots:\H\to \Y$ be a monotonically ascending sequence of functions. Let $G_n=J_{\Y}\circ K_n$ for each $n\in\NN$. Since $J_{\Y}$ by definition of the order on $\Y$ is an order embedding, hence monotone, it follows from Lemma \ref{lem:left multiplication by Scott continuous function is Scott continuous} that $G_1\sqsubseteq G_2\sqsubseteq \cdots:\H\to\X$ is a monotonically ascending sequence, which therefore has a limit $G_\infty$. 
For any $G:\H\to\X$, we have
\[\ran G=\{X\atomof\X:G(H',X)\neq 0 \text{ for some }H'\atomof\H\}=\{X\atomof\X:G(H,X)\neq 0\}.\]
Moreover, by Lemma \ref{lem:factors through a subset} we have $\ran G\subseteq \Y$ if and only if $G$ factors via $J_{\Y}$. Thus $G(H,X_\perp)=0$ if and only if $G=J_{\Y}\circ K$ for some $K:\H\to \Y$
Furthermore, we have
\[(R\circ G)(H,X_\perp)=\bigvee_{X\atomof \X}R(X,X_\perp)\cdot G(H,X_\perp) = R(X_\perp,X_\perp)\cdot G(H,X_\perp),\]
where we used that $R(X,\CC_\perp)=0$ for each atom $X\neq\CC_\perp$ in the last equality. Since $R(X_\perp,X_\perp)\geq I_{\X}(X,X)\geq\CC 1_X$, it follows that $(R\circ G)(H,X_\perp)=0$ if and only if $G(H,X_\perp)=0$ if and only if $G=J_{\Y}\circ K$ for some $K:\H\to \Y$. 
Since $G_n\nearrow G_\infty$, we have $(R\circ G_\infty)(H,X_\infty)=\bigwedge_{n\in\NN}(R\circ G_n)(H,X_\perp)=0,$
hence $G_\infty=J_{\Y}\circ K_\infty$ for some $K_\infty:\H\to \Y$. Then $G_n\nearrow G_\infty$ translates to $\bigwedge_{n\in \NN} R\circ J_{\Y}\circ K_n=\bigwedge_{n\in\NN}R\circ G_n=R\circ G_\infty=R\circ J_{\Y}\circ K_\infty.$
It then follows from \cite[Proposition~A.6]{KLM20} that
\begin{align*} \bigwedge_{n\in\NN}S\circ K_n& =\bigwedge_{n\in\NN}J_{\Y
}^\dag\circ R\circ J_{\Y}\circ K_n=J_{\Y}^\dag\circ\bigwedge_{n\in\NN} R\circ J_{\Y}\circ K_n\\
& =J_{\Y}^\dag\circ R\circ J_{\Y}\circ K_\infty=S\circ K_\infty,\end{align*} 
which shows that $K_n\nearrow K_\infty$, so $\Y$ is indeed a quantum cpo. Clearly $(-)_\perp$ is injective on objects, hence $(\Y,S)$ is unique.
\end{proof}

\begin{definition}\label{def:pointed quantum cpo}
We call a quantum cpo $(\X,R)$ \emph{pointed} if it satisfies one (and hence all) of the conditions in Lemma \ref{lem:pointed qcpo}. We say that a Scott continuous map $F:(\X,R)\to(\Y,S)$ between pointed quantum cpos is \emph{strict} if $F\circ B_\X=B_\Y$. The subcategory of $\qCPO$ consisting of pointed quantum cpos and strict Scott continuous maps is denoted by $\qCPO_{\perp!}$.
\end{definition}

\begin{lemma}\label{lem:pointed qcpo 2}
Let $(\X,R)$ be a pointed quantum cpo. Then for each quantum set $\Y$, the function $B_{\Y,\X}:=B_\X\circ !_\Y$ is Scott continuous, and it is the least element of $\qSet(\Y,\X)$. Moreover, if $\Z$ is another quantum set, and $F:\Z\to\Y$ is a function, then $B_{\Y,\X}\circ F=B_{\Z,\X}$.
\end{lemma}
\begin{proof}
Let $K:\Y\to\X$ be a function. Let $X\atomof\X$ and $Y\atomof\X$, and let $J_Y:\Q\{Y\}\to\Y$ be the embedding. Lemma \ref{lem:pointed qcpo} assures that $B_\X\circ !_{\Q\{Y\}}$ is the least element of $\qSet(\Q\{Y\},\X)$, hence $B_{\Y,\X}\circ J_Y=B_\X\circ !_{\Y}\circ J_Y=B_\X\circ 1_{\Q\{Y\}}\leq K\circ J_Y$, so  $K(Y,X)=(K\circ J_Y)(X,Y)\leq (R\circ B_\X\circ !_\Y)(Y,X)$, showing that $B_\X\circ !_\Y\sqsubseteq K$, so $B_{\Y,\X}$ is indeed the least element of $\qSet(\Y,\X)$. Since it is the composition of Scott continuous functions, it is Scott continuous. Finally, we have $B_{\Y,\X}\circ F=B_{\X}\circ !_{\Y}\circ F=B_\X\circ !_\Z=B_{\Z,\X}$.
\end{proof}

\begin{lemma}\label{lem:property of strict function}
   Let $(\X,R)$ and $(\Y,S)$ be pointed quantum cpos and let $F:\X\to\Y$ be a Scott continuous function. Then $F$ is strict if and only if $F(\CC_\perp,Y)=\delta_{Y,\CC_\perp}L(\CC_\perp,Y)$ for each $Y\atomof\Y$. 
\end{lemma}
\begin{proof}
Let $F$ be strict. Then for each $Y\atomof\Y$, we have
\[ \delta_{Y,\CC_\perp}L(\CC,Y)=B_\Y(\CC,Y)=(F\circ B_\X)(\CC,Y)=\bigvee_{X\atomof\X}F(X,Y)\cdot B_\X(\CC,X)=F(\CC_\perp,Y)\cdot L(\CC,\CC_\perp),\]
hence
\[ F(\CC_\perp,Y)=F(\CC_\perp,Y)\cdot L(\CC,\CC_\perp)\cdot L(\CC_\perp,C)=\delta_{Y,\CC_\perp}L(\CC,Y)\cdot L(\CC_\perp,\CC)=\delta_{Y,\CC_\perp}L(\CC_\perp,Y).\]
Conversely, if $F(\CC_\perp,Y)=\delta_{Y,\CC_\perp}L(\CC_\perp,Y)$ for each $Y\atomof\Y$, we have by definition of $B_\X$:
\begin{align*}(F\circ B_\X)(\CC,Y) & =\bigvee_{X\atomof\X}F(X,Y)\cdot B_\X(\CC,X)=F(\CC_\perp,Y)\cdot L(\CC,\CC_\perp)\\
& =\delta_{Y,\CC_\perp}L(\CC_\perp,Y)\cdot L(\CC,\CC_\perp)=\delta_{Y,\CC_\perp}L(\CC,Y)=B_\Y(\CC,Y).
\end{align*}
We conclude that $F$ is strict.
\end{proof}

\begin{theorem}
    The category $\qCPO_{\perp!}$ is equivalent to $\mathbf{Kl}$.
\end{theorem}
\begin{proof}
We can identify $\mathbf{Kl}$ with the category of free $(-)_\perp$-algebras. 

By Definition \ref{def:pointed quantum cpo} and Lemma \ref{lem:pointed qcpo} any object of $\qCPO_{\perp!}$ is of the form $(\X,R)_\perp$ for some unique quantum cpo $(\X,R)$. 
So, let $(\X,R)_\perp$ and $(\X,S)_\perp$ be two pointed quantum cpos and let $F:(\X,R)_\perp\to (\Y,S)_\perp$ be a strict Scott continuous map. We show that $F:(\X_\perp,M_\X)\to(\Y_\perp,M_\Y)$ is a morphism of $(-)_\perp$-algebras, where  $M_\X:\X_{\perp\perp}\to\X_\perp$ is given in Theorem \ref{thm:maybe monad on qSet}. So, we need to show that $F\circ M_\X=M_\Y\circ F_\perp$. 
For any quantum cpo $\Z$, we write $\At(\Z_\perp)=\At(\Z)\cup\{\CC_{\perp}\}$ and $\At(\Z_{\perp\perp})=\At(\Z_\perp)\cup\{\CC_{\perp_2}\}$. Then the only nonzero components of $M_\Z$ are $M_\Z(Z,Z)=\CC 1_Z$ for $Z\neq\CC_{\perp_2}$ and $M_\Z(\CC_{\perp_2},\CC_{\perp})=L(\CC_{\perp_2},\CC_{\perp})$. Let $X\atomof\X_{\perp\perp}$ and $Y\atomof\Y_\perp$. If $X\neq\CC_{\perp_2}$, we have 
\begin{align*}
   ( F\circ M_\X)(X,Y)  &  = \bigvee_{Z\atomof\X_\perp}F(Z,Y)\cdot M_\X(X,Z)=F(X,Y)= \bigvee_{Z\atomof\Y_\perp}M_\Y(Z,Y)\cdot F(X,Z)\\
    & = \bigvee_{Z\atomof\Y_{\perp\perp}}M_\Y(Z,Y)\cdot F_\perp(X,Z) =     M_\Y\circ F_\perp(X,Y) .
\end{align*}
If $X=\CC_{\perp_2}$, then using Lemma \ref{lem:property of strict function}, we find:
\begin{align*}
 (   F\circ M_\X)(X,Y)  & = \bigvee_{Z\atomof\X_\perp}F(Z,Y)\cdot M_\X(\CC_{\perp_2},Z)=F(\CC_{\perp},Y)\cdot L(\CC_{\perp_2},\CC_{\perp})\\
    & = \delta_{Y,\CC_\perp}L(\CC_{\perp},Y)\cdot L(\CC_{\perp_2},\CC_{\perp})=\delta_{Y,\CC_\perp}L(\CC_{\perp_2},Y)\\
    & = \delta_{Y,\CC_\perp}L(\CC_{\perp_2},Y)\cdot L(\CC_{\perp_2},\CC_{\perp_2})  = M_\Y(\CC_{\perp_2},Y)\cdot L(\CC_{\perp_2},\CC_{\perp_2})\\
    & = \bigvee_{Z\atomof\Y_{\perp\perp}}M_\Y(Z,Y)\cdot F_\perp(\CC_{\perp_2},Z)=    M_\Y\circ F_\perp(X,Y). 
\end{align*}
We conclude that $M_\Y\circ F_\perp=F\circ M_\X$, so $F$ is a morphism of $(-)_\perp$-algebras.

Conversely, assume that $F$ is a morphism of $(-)_\perp$-algebras.
For any quantum cpo $(\Z,T)$, and any $Z\atomof\Z_\perp$, we calculate:
\begin{align*}
 (   M_\Z\circ B_{\Z_\perp})(\CC,Z) & = \bigvee_{Z'\atomof\Z_{\perp\perp}}M_\Z(Z',Z )\cdot\delta_{Z',\CC_\perp}L(\CC,Z)\\
    & = M_\Z(\CC_\perp,Z)\cdot L(\CC_,\CC_\perp) = \delta_{Z,\CC_\perp}L(\CC_\perp,Z)\cdot L(\CC,\CC_\perp)=\delta_{Z,\CC_\perp}L(\CC,Z).
\end{align*}
Hence, we find for each $Y\atomof\Y_\perp$:
\begin{align*}
    \delta_{Y,\CC_\perp}L(\CC,Y) & = (M_\Y\circ B_{\Y_\perp})(\CC,Y) = M_\Y\circ (F_\perp\circ B_{\X_\perp})(\CC,Y) =(F\circ M_\X\circ B_{\X_\perp})(\CC,Y)\\
    & = \bigvee_{X\atomof\X_\perp}F(X,Y)\cdot ( M_\X\circ B_{\X_\perp})(\CC,X) = F(\CC_\perp,Y)\cdot L(\CC,\CC_\perp),
\end{align*}
where we used the definition of a strict map in the second equality, and the fact that $F$ is a $(-)_\perp$-algebra morphism in the third equality. Hence,
\[ F(\CC_\perp,Y)=F(\CC_\perp,Y)\cdot L(\CC,\CC_\perp)\cdot L(\CC_\perp,\CC)=\delta_{Y,\CC_\perp}L(\CC,\CC_Y)\cdot L(\CC_\perp,\CC)=\delta_{Y,\CC_\perp}L(\CC_\perp,Y).\]
It now follows from Lemma \ref{lem:property of strict function} that $F$ is strict. We conclude that every object of $\qCPO_{\perp!}$ is of the form $(\X,R)_\perp$ for some unique quantum cpo $(\X,R)$, and similarly, the underlying quantum cpo of any free $(-)_\perp$-algebra $(\X_\perp,M_\X)$ is $(\X,R)_\perp$ for some unique quantum cpo $(\X,R)$. Moreover, any Scott continuous map $F:(\X,R)_\perp\to (\Y,S)_\perp$ between pointed quantum cpos is strict if and only if it is an $(-)_\perp$-algebra morphism $(\X_\perp,M_\X)\to(\Y_\perp,M_\Y)$. It now easily follows that $\qCPO_{\perp!}$ and $\mathbf{Kl}$ are equivalent.
\end{proof}

\subsection{Monoidal closure}
Recall that $\qCPO$ is monoidal closed with inner hom $[-,-]_\uparrow$ and with evaluation function $\Eval_\uparrow$ (cf. Theorem \ref{thm:qCPO is monoidal closed}). In this section we show that $\mathbf{Kl}$ is monoidal closed, too.

\begin{proposition}\label{prop:monoidal closure and lift}
Let $(\X,R)$ be a quantum cpo, and let $(\Y,S)$ be a pointed quantum cpo. Then $[\X,\Y]_\uparrow$ is pointed, too.
\end{proposition}
\begin{proof}
By Lemma \ref{lem:pointed qcpo 2}, there is a Scott continuous function $B_{\X,\Y}:\mathbf \X\to \Y$ that is the least element of $\qSet(\X,\Y)$. The same lemma assures that $B_{\mathbf 1\times\X,\Y}=B_{\X,\Y}\circ L_\X:\mathbf 1\times \X\to\Y$ is Scott continuous, and the least element of $\qSet(\mathbf 1\times\X,\Y)$, where $L_\X$ denotes the left unitor of $\qCPO$. By the monoidal closure of $\qCPO$ there is a Scott continuous function $B:\mathbf 1\to[\X,\Y]_\uparrow$ such that the following diagram commutes:
\[\begin{tikzcd}
\mathbf 1\times \X\ar{dr}{B_{\mathbf 1\times\X,\Y}}\ar{d}[swap]{B\times I_\X} & \\
\phantom{}[\X,\Y]_\uparrow\times\X\ar{r}[swap]{\Eval_\uparrow} & \Y
\end{tikzcd}\]

We next show that for each atomic quantum set $\H$, the map $B\circ !_{\H}$ is the least element of $\qSet(\H,[\X,\Y]_\uparrow)$ as by Lemma \ref{lem:pointed qcpo} this would assure that $[\X,\Y]_\uparrow$ is pointed. By Lemma \ref{lem:pointed qcpo 2}, we have $B_{\H\times\X,\Y}=B_{\mathbf 1\times\X,\Y}\circ (!_{\H}\times I_\X)$. Thus consider the following diagram:

\[\begin{tikzcd}
\H\times\X\ar[ddr,bend left=55,"B_{\H\times\X,\Y}",""{name=B,above}]\ar{d}[swap]{!_{\H}\times I_\X} &\\
\mathbf 1\times \X\ar{dr}{B_{\mathbf 1\times\X,\Y}}\ar{d}[swap]{B\times I_\X} & \\
\phantom{}[\X,\Y]_\uparrow\times\X\ar{r}[swap]{\Eval_\uparrow} & \Y
\end{tikzcd}\]

It follows from Lemma \ref{lem:pointed qcpo 2} that $B_{\H\times \X,\Y}$ is the least element of $\qCPO(\H\times\X,\Y)\subseteq \qSet(\H\times\X,\Y)$. Since $\qCPO$ is order enriched, currying is an order isomorphism, hence $B\circ !_{\H}$ is the least element of $\qCPO(\H,[\X,\Y]_\uparrow)$. By Example \ref{ex:function with trivially ordered domain is Scott continuous}, it follows that any function $K:\H\to[\X,\Y]_\uparrow$ is Scott continuous, hence $\qCPO(\H,[\X,\Y]_\uparrow)=\qSet(\H,[\X,\Y]_\uparrow)$. We conclude that $B\circ !_{\H}$ is indeed the least element of $\qSet(\H,[\X,\Y]_\uparrow)$, i.e., $B=B_{[\X,\Y]_\uparrow}$.
 \end{proof}

By Proposition \ref{prop:monoidal closure and lift}, we know that $[\X,\Y_\perp]_\uparrow$ is a pointed quantum cpo for any two quantum cpos $\X$ and $\Y$, and by Lemma \ref{lem:pointed qcpo} there is a unique quantum cpo $\{\X,\Y\}$ such that $\{\X,\Y\}_\perp=[\X,\Y_\perp]_\uparrow$. Recall the conventions about Kleisli categories at the beginning of this section. We thus have a morphism $E:\{\X,\Y\}\odot\X\circlearrow\Y$ given by the composition \[\{\X,\Y\}\times\X\xrightarrow{ H_{\{\X,\Y\}}\times I_\X}[\X,\Y_\perp]_\uparrow\times \X\xrightarrow{\Eval_\uparrow}\Y_\perp\] in $\qCPO$.

\begin{lemma}\label{lem:Kleisli closure}
Let $\X$, $\Y$ and $\Z$ be quantum cpos and let $G:\Z\circlearrow\{\X,\Y\}$ be a morphism in $\mathbf{Kl}$, which is a morphism $G:\Z\to[\X,\Y_\perp]_\uparrow$ in $\qCPO$. Then $E\bullet (G\odot H_\X)=\Eval_\uparrow\circ (G\times I_X)$.
\end{lemma}
\begin{proof}
Consider the following commuting diagram:
\[\begin{tikzcd}
\{\X,\Y\}_\perp\times\X\ar{rr}{I_{\{\X,\Y\}_\perp}\times H_\X}\ar{rrdd}[swap]{H_{\{\X,\Y\}_\perp}\times H_\X} && \{\X,\Y\}_\perp\times\X_\perp\ar{rr}{K_{\{\X,\Y\},\X}}\ar{dd}{H_{\{\X,\Y\}_\perp}\times I_{\X_\perp}} && (\{\X,\Y\}\times \X)_\perp \ar{dd}{(H_{\{\X,\Y\}}\times I_\X)_\perp} \\
\\
&& \{\X,\Y\}_{\perp\perp}\times \X_\perp\ar{rr}[swap]{K_{\{\X,\Y\}_\perp,\X}} &&  (\{\X,\Y\}_\perp\times\X)_\perp
\end{tikzcd}\]
Here, the square commutes by naturality of $K$. Since the lift is a monoidal monad, we have $K_{\{\X,\Y\}_\perp,\X}\circ (H_{\{\X,\Y\}_\perp}\times H_\X)=H_{\{\X,\Y\}_\perp\times\X}$, whence
\[ (H_{\{\X,\Y\}}\times I_\X)_\perp\circ K_{\{\X,\Y\},\X}\circ (I_{\{\X,\Y\}_\perp}\times H_\X)=H_{\{\X,\Y\}_\perp\times\X}\]
As a consequence, we obtain 
\[M_\Y\circ (\Eval_\uparrow)_\perp\circ (H_{\{X,\Y\}}\times I_\X)_\perp\circ K_{\{\X,\Y\},\X}\circ (I_{\{\X,\Y\}_\perp}\times H_\X)=M_\Y\circ (\Eval_\uparrow)_\perp\circ H_{\{\X,\Y\}_\perp\times\X}=\Eval_\uparrow,\]
where the last equality follows since $\Y_\perp$, the codomain of $\Eval_\uparrow$ is lifted, and for arbitrary morphisms $F:\V\to\W_\perp$ with a lifted codomain we have $M_\W\circ F_\perp\circ H_\V=F$, which follows the following commuting diagram:
\[ \begin{tikzcd}
\V\ar{r}{H_{\V}}\ar{d}[swap]{F} &\V_\perp\ar{r}{F_\perp} \ar{d}[swap]{F_\perp} & \W_{\perp\perp}\ar{r}{M_\W} & \W_\perp  \\
\W_\perp\ar{r}[swap]{H_{\W_\perp}} & \W_{\perp\perp}\ar{ru}[swap]{I_{\W_{\perp\perp}}} & 
\end{tikzcd} \]
where the square commutes by naturality of $H$. By definition of a monad, we have $M_\W\circ H_{\W_\perp}=I_{\W_\perp}$, from which indeed $M_\W\circ F_\perp\circ H_\V=F$ follows.
Thus we conclude
\begin{equation*}
    \Eval_\uparrow=M_\Y\circ (\Eval_\uparrow)_\perp\circ (H_{\{X,\Y\}}\times I_\X)_\perp\circ K_{\{\X,\Y\},\X}\circ (I_{\{\X,\Y\}_\perp}\times H_\X),
\end{equation*}
whence
\begin{align*}
    \Eval_\uparrow\circ (G\times I_\X) & = M_\Y\circ (\Eval_\uparrow)_\perp\circ (H_{\{X,\Y\}}\times I_\X)_\perp\circ K_{\{\X,\Y\},\X}\circ (I_{\{\X,\Y\}_\perp}\times H_\X)\circ (G\times I_\X)\\
    & = M_\Y\circ (\Eval_\uparrow)_\perp\circ (H_{\{X,\Y\}}\times I_\X)_\perp\circ K_{\{\X,\Y\},\X}\circ (G\times H_\X)\\
    & = M_\Y\circ (\Eval_\uparrow \circ H_{\{X,\Y\}}\times I_\X)_\perp\circ (G\odot H_\X)\\
     & = M_\Y\circ E_\perp\circ (G\odot H_\X)\\
     & = E\bullet (G\odot H_\X).\qedhere
\end{align*}
\end{proof}

\begin{theorem}\label{thm:qCPObs is monoidal closed}The Kleisli category $\mathbf{Kl}$ of the lift monad is symmmetric monoidal closed.
\end{theorem}
\begin{proof}By Corollary \ref{cor:Kl is monoidal}
$\mathbf{Kl}$ is a symmetric monoidal category, so we only have to show that $\mathbf{Kl}$ is monoidal closed.
Let $F:\Z\odot\X\circlearrow\Y$ be a morphism in $\mathbf{Kl}$. Then $F$ is a morphism $\Z\times\X\to\Y_\perp$ in $\qCPO$. By the universal property of $\Eval_\perp$, there is a unique morphism $G:\Z\to[\X,\Y_\perp]_\uparrow$ such that $\Eval_\uparrow\circ (G\circ I_\X)=F$. Now, $G$ is a morphism $\Z\circlearrow\{\X,\Y\}$ in $\mathbf{Kl}$, and by Lemma \ref{lem:Kleisli closure}, it follows that $G$ is the unique morphism in $\mathbf{Kl}$ such that $E\bullet(G\odot H_\X)=F$. Since $H_\X$ is the identity morphism on $\X$, this shows the assignment $\Y\mapsto \{\X,\Y\}$ is the right adjoint of $\Z\mapsto\Z\odot\X$, i.e., $\mathbf{Kl}$ is monoidal closed. 
\end{proof}

\subsection{Completeness and coproducts}

\begin{lemma}\label{lem:order embedding induces order embedding on homsets}
Let $(\X,R)$ and $(\Y,S)$ be quantum posets and let $F:\X\to\Y$ be an order embedding, i.e., $R=F^\dag\circ S\circ F$. If $\Z$ is a quantum set, then the map $\qSet(\X,F):\qSet(\Z,\X)\to\qSet(\Z,\Y)$, $G\mapsto F\circ G$ is an order embedding. 
\end{lemma}
\begin{proof}
Let $G_1,G_2:\Z\to\X$ be functions. Since $F$ is monotone, its left action is monotone \cite{KLM20}*{Lemma 4.4}, whence $G_1\sqsubseteq G_2$ implies $F\circ G_1\sqsubseteq F\circ G_2$. For the converse, assume that $F\circ G_1\sqsubseteq F\circ G_2$. Then $S\circ F\circ G_2\leq S\circ F\circ G_1$, so $R\circ G_2=F^\dag\circ S\circ F\circ G_2\leq F^\dag\circ S\circ F\circ G_1=R\circ G_1$, i.e., $G_1\sqsubseteq G_2$.
\end{proof}

\begin{proposition}\label{prop:Kleisli and EM category}
The canonical embedding $\mathbf{Kl}\to\mathbf{EM}$, $\X\mapsto (\X_\perp,M_\X)$ of the Kleisli category of the monad $(-)_\perp$ on $\qCPO$ into the Eilenberg-Moore category $\cat{EM}$ of $(-)_\perp$ is a weak equivalence of categories.
\end{proposition}
\begin{proof}
It is sufficient to show that each $(-)_\perp$-algebra $(\X,F)$ is free, i.e., equal to $(\Y_\perp,M_\Y)$ for some quantum cpo $\Y$. By definition of an $(-)_\perp$-algebra $F:\X_\perp\to\X$ is a Scott continuous map satisying $F\circ H_\X=I_\X$, where we recall that $H_\X:\X\to\X_\perp$ is the unit of the monad $(-)_\perp$, and is equal to the embedding $J_\X:\X\to\X_\perp$.

We first show that $\X$ is pointed, for which we consider the map $K:\X_\perp\to\X_\perp$ given by $K=H_\X\circ F$, which is Scott continuous since it is the composition of Scott continuous maps. Let $\H$ be an atomic quantum set. By Theorem \ref{thm:qCPO enriched over CPO},  $\qCPO(\H,\X_\perp)$ is a cpo, which is equal to $\qSet(\H,\X_\perp)$ by Example \ref{ex:function with trivially ordered domain is Scott continuous}. It follows from Lemma \ref{lem:pointed qcpo 2} that $\qCPO(\H,\X_\perp)$ is a pointed cpo with a least element $B_{\H,\X_\perp}$ that equals $B_\X\circ !_\H$. Note that $\mathbf 1$ is atomic, and $B_{\mathbf 1,\X_\perp}=B_{\X_\perp}$.

Since $K$ is Scott continuous, it follows from Lemma \ref{lem:left multiplication by Scott continuous function is Scott continuous} that the map $\qCPO(\H,\X_\perp)\to\qCPO(\H,\X_\perp)$, $G\mapsto K\circ G$ is Scott continuous, hence Kleene's Fixpoint Theorem (see for instance \cite[Theorem 2.1.19]{abramskyjung:domaintheory}) assures that this map has a least fixpoint $G_\H$ given by $\bigvee_{n\in\NN}K^n\circ B_{\H,\X_\perp}$, where the supremum is taken with respect to the order $\sqsubseteq$ on $\qCPO(\H,\X_\perp)$.
By Lemma \ref{lem:right multiplication is Scott continuous}, we have
$G_\H=\bigvee_{n\in\NN}K^n\circ B_{\H,\X_\perp}=\bigvee_{n\in\NN}K^n\circ B_{\X_\perp}\circ !_{\H}=\bigvee_{n\in\NN}K^n\circ B_{\mathbf 1,\X_\perp}\circ !_{\H}=G_{\mathbf 1}\circ !_{\H}$.

We claim that $F\circ G_\H$ is the least element of $\qSet(\H,\X)$. Let $N:\H\to\X$ be a function. Then $H_\X\circ N$ is a function $\H\to\X_\perp$, and we have
$K^0\circ B_{\H,\X_\perp}=B_{\H,\X_\perp}\sqsubseteq H_\X\circ N$. Assume that $K^n\circ B_{\H,\X_\perp}\sqsubseteq H_\X\circ N$. Then $K^{n+1}\circ B_{\H,\X_\perp}\sqsubseteq K\circ H_\X\circ N=H_\X\circ F\circ H_\X\circ N=H_\X\circ N$, since $F\circ H_\X=I_\X$. We conclude that $K^n\circ B_{\H,\X_\perp}\sqsubseteq H_\X\circ N$ for each $n=0,1,2,\ldots$, hence also $G_\H=\bigvee_{n\in\NN}K^n\circ B_{\H,\X_\perp}\sqsubseteq H_\X\circ N$.
Note that by definition of $R_\perp$, we have $R_\perp(X,Y)=R(X,Y)$ for each $X,Y\atomof\X$, which translates to $J_\X^\dag\circ R_\perp\circ J_\X=R$, so $J_\X$ is an order embedding.
We have $H_\X\circ F\circ G_\H=K\circ G_\H=G_\H\sqsubseteq H_\X\circ N$. Since $H_\X=J_\X$, we have $J_\X\circ F\circ G_\H\sqsubseteq J_\X\circ N$, which implies $F\circ G_\H\sqsubseteq N$ by Lemma \ref{lem:order embedding induces order embedding on homsets}.
We conclude that since $G_\H=G_{\mathbf 1}\circ !_{\H}$, we have a function $F\circ G_\mathbf 1:\mathbf 1\to\X$, such that $F\circ G_1\circ !_\H$ is the least element of $\qSet(\H,\X)$ for each atomic quantum set $\H$. It follows from Lemma \ref{lem:pointed qcpo} that $(\X,R)$ is a pointed quantum cpo, so $(\X,R)=(\Y,S)_\perp$ for some quantum cpo $(\Y,S)$. The same lemma assures the existence of a one-dimensional atom $X_\perp\atomof\X$ such that $R(X_\perp,X)=L(X_\perp,X)$ and $(F\circ G_{\mathbf 1})(\CC,X)=\delta_{X_\perp,X}L(\CC,X) $ for each $X\atomof\X$.

Next we show that $F=M_\Y$. Since $F\circ H_\X=I_\X$, and $H_\X=J_\X$, it follows that $F(X,Y)=\delta_{X,Y}\CC 1_X$ for each $X,Y\atomof\X$. It remains to show that $F(\CC_\perp,X_\perp)=L(\CC_\perp,X_\perp)$. Since $R(X_\perp,X)=L(X_\perp,X)$ for each $X\atomof\X_\perp$, it follows from Lemma \ref{lem:pointed order} that $R(X,X_\perp)=0$ for each atom $X\neq\CC_\perp$. Hence $(R\circ F)(\CC_\perp,X_\perp)=\bigvee_{X\atomof\X}R(X,X_\perp)\cdot F(\CC_\perp,X)=R(X_\perp,X_\perp)\cdot F(\CC_\perp,X_\perp)=L(X_\perp,X_\perp)\cdot F(\CC_\perp,X_\perp)=F(\CC_\perp,X_\perp)$. Since $F$ is monotone, we have $F\circ R_\perp\leq R\circ F$, hence 
$L(\CC_\perp,X_\perp)=L(X_\perp,X_\perp)\cdot L(\CC_\perp,X_\perp)=F(X_\perp,X_\perp)\cdot R_\perp(\CC_\perp,X_\perp)\leq \bigvee_{X\atomof\X_\perp}F(X,X_\perp)\cdot R(\CC_\perp,X)=(F\circ R_\perp)(\CC_\perp,X_\perp)\leq (R\circ F)(\CC_\perp,X_\perp)=F(\CC_\perp,X_\perp)$, which forces $F(\CC_\perp,X_\perp)=L(\CC_\perp,X_\perp)$.

We conclude that $(\X,F)=(\Y_\perp,M_\Y)$, so the embedding of $\mathbf{Kl}$ into $\mathbf{EM}$ is surjective on all objects, hence $\mathbf{Kl}$ is isomorphic to $\mathbf{EM}$.
\end{proof}

\begin{theorem}\label{thm:Kl is complete and has coproducts}
The Kleisli category $\mathbf{Kl}$ of the monad $(-)_\perp$ on $\qCPO$ is complete and has all coproducts.
\end{theorem}
\begin{proof}
Since $\qCPO$ has all coproducts (cf. Theorem \ref{thm:qCPO has all coproducts}), it follows from \cite[Proposition 2.2]{Szigeti} that $\mathbf{Kl}$ also has all coproducts. By Proposition \ref{prop:Kleisli and EM category}, there is an isomorphism of categories $\mathbf{Kl}\to\mathbf{EM}$, which trivially has a right adjoint: its inverse. Since $\qCPO$ is complete (cf. Theorem \ref{thm:qCPO has all limits}), it follows now from  \cite[Theorem 3.1]{Szigeti} that $\mathbf{Kl}$ is also complete.
\end{proof}

\subsection{$\CPO$-algebraically compact categories of quantum cpos}\label{sec:algebraic-compactness}
In this section, we show that the equivalent categories $\qCPO_{\perp!}$ and $\mathbf{Kl}$ are $\CPO$-algebraically compact. This notion, which we will define below, is nearly sufficient to model recursion; only a few additional mild conditions are needed to obtain computational adequacy.

\begin{definition}
A category $\mathbf C$ enriched over a symmetric monoidal category $\mathbf V$ is called $\mathbf V$-\emph{algebraically compact} if each $\mathbf V$-endofunctor $T$ on $\mathbf C$ has an initial $T$-algebra $\omega:T\Omega\to \Omega$ such that $\omega^{-1}:\Omega\to T\Omega$ is a final $T$-coalgebra.
\end{definition}

We next need the following definition:
\begin{definition}
    Let $\mathbf C$ be a $\CPO$-enriched category.
    \begin{itemize}
        \item Given two objects $A$ and $B$ of $\mathbf C$,  two morphisms $e:A\to B$ and $p:B\to A$ are said to form an \emph{embedding-projection pair} or an \emph{e-p pair} if $p\circ e=1_A$ and $e\circ p\leq 1_B$. The morphism $e$ and $p$ are called an \emph{embedding} and a \emph{projection}, respectively.
        \item $\mathbf C$ has an \emph{e-initial} object if it has an initial object $0$ such that any morphism with source $0$ is an embedding. \item $\mathbf C$ has a \emph{p-terminal} object if it has a terminal object $1$ such that any morphism with target $1$ is a projection.
        \item 
 $\mathbf C$ is said to have an \emph{ep-zero} object if it has an e-initial object that is p-terminal.
 \item $\mathbf C$ has \emph{$\omega$-colimits over embeddings} if every $\omega$-diagram in $\mathbf C$ whose connecting morphisms are embeddings has a colimit.
 \item $\mathbf C$ has \emph{$\omega$-limits over projections} if every $\omega$-diagram in $\mathbf C$ whose connecting morphisms are projections has a limit.
    \end{itemize}
\end{definition}

\begin{lemma}$\mathbf{Kl}$ is enriched over $\CPO_{\perp!}$.
\end{lemma}
\begin{proof}
Let $\X$ and $\X'$ be quantum cpos. Then it follows from Theorem \ref{thm:qCPO enriched over CPO} and Lemma \ref{lem:pointed qcpo 2} that $\mathbf{Kl}(\X,\X')=\qCPO(\X,\X'_\perp)$ is a pointed cpo with least element $B_{\X,\X'_\perp}$. Let $\Y$ and $\Y'$ be two other quantum cpos. By Corollary \ref{cor:Kl is monoidal} the monoidal product $F\odot G$ of two morphisms $F:\X\circlearrow \X'$ and $G:\Y\circlearrow\Y'$ in $\mathbf{Kl}$ is given by $K_{\X'\times \Y'}\circ( F\times G)$. Since $K_{\X',\Y'}$ is Scott continuous by Lemma \ref{lem:monad maps-K} and $(-\times-)$ is Scott continuous by Lemma \ref{lem:tensor product is Scott continuous bifunctor}, it follows that $(-\odot-)$ is Scott continuous. A short calculation shows that for any two quantum cpos $\Z$ and $\Z'$, we have $B_{\Z,\Z'_\perp}(Z,Z')=\delta_{Z',\CC_\perp}L(Z,Z')$ for each $Z\atomof\Z$ and $Z'\atomof\Z'_\perp$. Hence for each $X\otimes Y\atomof\X\times\Y$ and $Z\atomof(\X'\times\Y')_\perp$, we have
\begin{align*}(B_{\X,\X'_\perp}\odot B_{\Y,\Y'_\perp})(X\otimes Y,Z)  & = \bigvee_{X'\otimes Y'\atomof\X'_\perp\times\Y'_\perp}K_{\X',\Y'}(X'\otimes Y',Z)\cdot (B_{\X,\X'_\perp}\times B_{\Y,\Y'_\perp})(X\otimes Y,X'\otimes Y')\\
& = K_{\X',\Y'}(\CC_\perp\otimes \CC_\perp,Z)\cdot (L(X,\CC_\perp)\otimes L(Y,\CC_\perp))\\
& = \delta_{Z,\CC_\perp}L(\CC_\perp\otimes\CC_\perp,Z)\cdot L(X\otimes Y,\CC_\perp\otimes\CC_\perp)\\
& = \delta_{Z,\CC_\perp}L(X\otimes Y,Z)= B_{\X\times\Y,(\X'\times\Y')_\perp}(X\otimes Y,Z)
\end{align*}
Hence, $(-\odot-)$ is strict. Let $\X$, $\Y$ and $\Z$ be quantum cpos. Then the map $\mathbf{Kl}(\Y,\Z)\times \mathbf{Kl}(\X,\Y)\to\mathbf{Kl}(\X,\Z)$, $(G,F)\mapsto G\bullet F$ is the map $\qCPO(\Y,\Z_\perp)\times\qCPO(\X,\Y_\perp)\to\qCPO(\X,\Z_\perp)$, $(G,F)\mapsto M_\Z\circ \M G\circ F$. This map is Scott continuous because the map $G\mapsto \M G$ is Scott continuous by Theorem \ref{thm:monoidal monad}, hence the map $(G,F)\mapsto \M G\circ F$ is Scott continuous by Lemma \ref{lem:right multiplication is Scott continuous}, and since $M_\Z$ is Scott continuous by Lemma \ref{lem:monad maps-M}, it also follows from Lemma \ref{lem:left multiplication by Scott continuous function is Scott continuous} that $(G,F)\mapsto M_\Z\circ \M G\circ F$ is Scott continuous. We verify that $G\bullet B_{\X,\Y_\perp}=B_{\Y,\Z_\perp}\bullet F=B_{\X,\Z_\perp}$. We write $\CC_1$ for the augmented bottom in $\Z_\perp$ and $\CC_\perp$ for the augmented bottom in $\Z_{\perp\perp}$. Then for $X\atomof\X$ and $Z\atomof\Z_{\perp\perp}$, we have
\[ (\M G\circ B_{\X,\Y_\perp})(X,Z)=(G\uplus I_\mathbf 1)(\CC_\perp,Z)\cdot L(X,\CC_\perp)=\begin{cases}L(X,\CC_2), & Z=\CC_\perp,\\
0, & Z\neq\CC_\perp,
\end{cases}\]
hence for each $X\atomof\X$ and $Z\atomof\Z_\perp$, we have
\begin{align*}(G\bullet B_{\X,\Y_\perp})(X,Z) & = (M_\Z\circ\M G\circ B_{\X,\Y_\perp})(X,Z) =M_\Z(\CC_\perp,Z)\cdot L(X,\CC_\perp)\\
& =\delta_{Z,\CC_1}L(\CC_\perp,Z)\cdot L(X,\CC_\perp)
 =\delta_{Z,\CC_1}L(X,Z)=B_{\X,\Z_\perp}(X,Z).\end{align*}
Moreover, for $Z\atomof\Z_\perp$ and $Y\atomof\Y_\perp$, we have
\begin{align*}
    (M_\Z\circ \M B_{\Y,\Z_\perp})(Y,Z) & = (M_\Z\circ (B_{\Y,\Z_\perp}\uplus I_\mathbf 1))(Y,Z) =\bigvee_{W\atomof\Z_{\perp\perp}}M_\Z(W,Z)\cdot (B_{\Y,\Z_\perp}\uplus I_\mathbf 1)(Y,W)\\
    & = \begin{cases} \bigvee_{W\atomof\Z_\perp}M_\Z(W,Z)\cdot B_{\Y,Z_\perp}(Y,W), & Y\neq\CC_\perp,\\
    M_\Z(\CC_\perp,Z)\cdot I_\mathbf 1(\CC_\perp,\CC_\perp), & Y=\CC_\perp,
    \end{cases}\\
    & = \begin{cases}
        M_\Z(\CC_1,Z)\cdot L(Y,\CC_1), & Y\neq \CC_\perp,\\
         M_\Z(\CC_\perp,Z), & Y=\CC_\perp,
    \end{cases}\\
    & =\delta_{Z,\CC_1}L(Y,Z)=B_{\Y_\perp,\Z_\perp}(Y,Z).
\end{align*}
So, $M_\Z\circ \M B_{\Y,\Z_\perp}=B_{\Y_\perp,\Z_\perp}$. It now follows from Lemma \ref{lem:pointed qcpo 2} that
\[ B_{\Y,\Z_\perp}\bullet F=M_\Z\circ \M B_{\Y,\Z_\perp}\circ F=B_{\Y_\perp,\Z_\perp}\circ F=B_{\X,\Z_\perp}. \]
It follows that $\mathbf{Kl}$ is enriched over $\CPO_{\perp!}$.
\end{proof}

\begin{lemma}\label{lem:ep-zero}
$\mathbf{Kl}$ has an ep-zero object.
\end{lemma}
\begin{proof}
    Note that $H_\X=J_\X$ is the identity on $\X$ in $\mathbf{Kl}$. The empty quantum set $\emptyset$ is the initial object of $\qCPO$, hence also of $\mathbf{Kl}$. Given any $\X\in\mathbf{Kl}$, we denote the initial map $\emptyset\circlearrow\X$ by $E_\X$, which is a morphism $\emptyset\to\X_\perp$ in $\qCPO$. 
    
    $\emptyset$ is also the terminal object of $\mathbf{Kl}$, since $\mathbf{Kl}(\X,\emptyset)=\qCPO(\X,\emptyset_\perp)=\qCPO(\X,\mathbf 1)=1$. Hence, we are done if we can show that $E_\X$ and the terminal map $!_\X:\X\to\mathbf 1$ in $\qCPO$ form an e-p pair. Since $!_\X\bullet E_\X$ is a morphism $\emptyset\circlearrow\emptyset$, and since $\emptyset$ is initial in $\mathbf{Kl}$, the only element in $\mathbf{Kl}(\emptyset,\emptyset)$ is $J_\emptyset$, so $!_\X\bullet E_\X=J_\emptyset$. 
We have $E_\X\bullet !_\X=M_\X\circ (E_\X)_\perp\circ !_\X$. It easily follows from the action of the lift monad on morphisms that $(E_\X)_\perp=B_{\X_{\perp\perp}}$. Furthermore, if we denote $\X_\perp=\X\uplus Q\{\CC_1\}$ and $\X_{\perp\perp}=\X\uplus\Q\{\CC_2\}$, we have $M_\X\circ B_{\X_{\perp\perp}}(\CC,X)=M_\X(\CC_2,X)\cdot B_{\X_{\perp\perp}}(\CC,\CC_2)=\delta_{X,\CC_1}L(\CC_2,X)\cdot L(\CC,\CC_2)=\delta_{X,\CC_1} L(\CC,X)=B_{\X_\perp}(\CC,X)$ for each $X\atomof\X_\perp$. Using Lemma \ref{lem:pointed qcpo 2}, we now obtain $E_\X\bullet !_\X=M_\X\circ B_{\X_{\perp\perp}}\circ !_\X=B_{\X_\perp}\circ !_\X=B_{\X,\X_\perp}$, which is the least element of $\qCPO(\X,\X_\perp)=\mathbf{Kl}(\X,\X)$. Hence, $E_\X\bullet !_\X=B_{\X_\perp}\sqsubseteq J_\X$. We conclude that $E_\X$ and $!_\X$ indeed form an e-p pair.
\end{proof}

\begin{lemma}\label{lem:omega-colimits}
    $\mathbf{Kl}$ has all $\omega$-colimits over embeddings.
\end{lemma}
\begin{proof}
 By Theorem \ref{thm:Kl is complete and has coproducts},   $\mathbf{Kl}$ is complete, hence it has all $\omega$-limits over projections. By the limit-colimit coincidence theorem for $\CPO$-enriched categories (see for instance \cite[Theorem 2]{smyth-plotkin:domain-equations}), it follows that $\mathbf{Kl}$ also has all $\omega$-colimits over embeddings.
\end{proof}

\begin{theorem}\label{thm:algebraically compactness}
The equivalent categories $\qCPO_{\perp!}$ and $\mathbf{Kl}$ are $\CPO$-algebraically compact.
\end{theorem}
\begin{proof}
By Lemmas \ref{lem:ep-zero} and \ref{lem:omega-colimits}, it follows that $\mathbf{Kl}$ has an ep-zero object and all $\omega$-colimits over embeddings. 
\cite[Corollary 7.2.4]{fiore-thesis} states that these two properties are sufficient to guarantee that a $\CPO$-category is $\CPO$-algebraically compact.
\end{proof}

\subsection{A categorical model of ECLNL and LNL-FPC}
In this section, we describe how our results can be used to form a denotational model for two quantum programming languages, ECLNL \cite{stringdiagramssemantics} and LNL-FPC \cite{lnl-fpc-lmcs}, both of which support (classically-controlled) recursion. We focus on these languages because they were introduced together, and because there is a characterization of the properties that a category must possess in order to be a model for these languages. Hence, we only have to verify that the model we present using quantum cpos has these properties.

Both ECLNL and LNL-FPC are based on Rios and Selinger's Proto-Quipper-M, introduced in \cite{pqm}, which is a \emph{circuit description language}, i.e., it can be used to generate quantum circuits in a structured way. Moreover, completed circuits can be used as data, which for instance can be stored in variables, and on which meta-operations can be performed. The original version of Proto-Quipper-M did not include constructs for recursion, state preparation or measurement, but the addition of the last two features is relatively straightforward. In that case the language has two runtimes: circuit generation time and circuit execution time. The language ECLNL is an extension of Proto-Quipper-M with recursive terms (but still without state preparation and measurement). LNL-FPC can be regarded both as the circuit-free fragment of Proto-Quipper-M extended with recursive types, and as an extension of Plotkin's FPC with linear types. FPC itself is regarded as a blueprint of a higher-order programming language with recursive types. Just as in FPC, type-level recursion in LNL-FPC induces term-level recursion.

It is generally accepted that models of quantum programming languages should utilize a \emph{call-by-value} methodology. Languages utilizing call-by-value specify a family of \emph{values} (typically of ground type) that do have any syntactic reductions. We recall that a model is \emph{sound} if whenever a term $m$ reduces to a value $v$, then $m$ and $v$ have the same denotation. In this case, $m$ is said to \emph{halt}. A term whose denotation is $\bot$ is said to \emph{diverge}, and \emph{computational adequacy} specifies that if a term $m$ does not diverge, then $m$ halts. Utilizing the fact that $\mathbf{Kl}$ is $\mathbf{CPO}$-algebraically compact, we show that $\mathbf{KL}$ is sound for both ECLNL and LNL-FPC. We also show $\mathbf{Kl}$ is computationally adequate for the circuit-free fragment of LNL-FPC.  Since ECLNL is a quantum circuit description language, showing that it can be modelled by $\mathbf{Kl}$ demonstrates that quantum cpos support recursion during the circuit generation stage of a quantum computation.

\begin{proposition}\label{prop:Kl model}
    The functor $\V\circ`(-):\CPO\to\mathbf{Kl}$ forms the left adjoint of a linear/non-linear model (cf. Definition \ref{def:lnlmodel}) in which both the linear and the non-linear category have all coproducts.
\end{proposition}
\begin{proof}
    We establish an adjunction with $\CPO$ as follows: we first note that Theorem \ref{thm:CPO-qCPO adjunction} assures that the functor $`(-):\CPO\to\qCPO$ has a right adjoint. Moreover, it is easy to verify that $`(-)$ is strong monoidal. Furthermore, by Corollary \ref{cor:Kl is monoidal}, the functor $\V:\qCPO\to\mathbf{Kl}$ is strict monoidal, and it has a right adjoint by construction of Kleisli categories. As a consequence, the functor $\V\circ`(-):\CPO\to\mathbf{Kl}$ is strong monoidal and has a right adjoint, with which it forms a linear/non-linear adjunction. Since $\CPO$ is cocomplete, it has all coproducts. The existence of coproducts in $\mathbf{Kl}$ follows from Theorem \ref{thm:Kl is complete and has coproducts}.
\end{proof}

\begin{theorem}
    The linear/non-linear model in Proposition \ref{prop:Kl model} is a sound categorical model for ECLNL.
\end{theorem}
\begin{proof}
    The proposition expresses that the model is a so-called \emph{CLNL}-model as defined in \cite[Definition 2]{stringdiagramssemantics}. By \cite[Theorem 2]{stringdiagramssemantics}, all categorical data in this model is enriched over $\CPO$. By Theorem \ref{thm:algebraically compactness}, the comonad induced by the adjunction in the model is algebraically compact. Hence, the linear/non-linear model is a sound categorical model for Proto-Quipper-M extended with recursive terms by \cite[Theorem 7]{stringdiagramssemantics}.
\end{proof}

In order to show that the model in Proposition \ref{prop:Kl model} is computationally adequate for LNL-FPC, we need one more concept:

\begin{definition}
   Let $(\mathbf C,\otimes, I)$ be a symmetric monoidal $\CPO_{\perp!}$-category. Then we say that the monoidal product $\otimes$ on $\mathbf C$ \emph{reflects the order}  if for any two pairs of parallel morphisms $f_1, g_1: I \to A$ and $f_2, g_2 : I \to B$ such that $f_1 \neq  \perp \neq g_1$ and $f_2 \neq \perp\neq  g_2$, we have that $f_1\otimes f_2\leq  g_1 \otimes g_2$ implies $f_1\leq g_1$ and $f_2\leq  g_2$.
\end{definition}

\begin{lemma}\label{lem:reflecting the order}The monoidal product $\odot$ on $\mathbf{Kl}$ reflects the order.
\end{lemma}
\begin{proof}
Let $(\X,R)$ and $(\Y,S)$ be quantum cpos, and let $F_1,F_2\in\mathbf{Kl}(\mathbf 1,\X)$ and $G_1,G_2\in\mathbf{Kl}(\mathbf 1,\Y)$ such that $F_1\neq B_{\X_\perp}$ and $G_1\neq B_{\Y_\perp}$. We need to show that $F_1\odot G_1\sqsubseteq F_2\odot G_2$ implies $F_1\sqsubseteq F_2$ and $G_1\sqsubseteq G_2$.
Note that $F_i\odot G_i= K_{\X,\Y}\circ (F_i\times G_i)$ for $i=1,2$. For simplicity, we write $K$ instead of $K_{\X,\Y}$. For any quantum set $\Z$, we write $\Z_1$ to denote its subset of one-dimensional atoms.
By definition of the range (cf. \cite[Definition~3.2]{KLM20}) and by Lemma \ref{lem:b is order iso} and its preceding paragraph, we have $\ran F_i\subseteq(\X_\perp)_1$, and $\ran G_i\subseteq (\Y_\perp)_1$. It follows from Lemma \ref{lem:factors through a subset} that there are functions $F_i':\mathbf 1\to(\X_\perp)_1$ and $G_i':\mathbf 1\to(\Y_\perp)_1$ such that $F_i=J_{(\X_\perp)_1}\circ F_i'$ and $G_i=J_{(\Y_\perp)_1}\circ G_i'$. Note that all atoms of $(\X_\perp)_1\times(\Y_\perp)_1$
are of the form $X\otimes Y$ for $X\atomof(\X_\perp)_1$ and $Y\atomof(\Y_\perp)_1$, hence one-dimensional. By definition of $K$ and by definition of the range of a function, it follows that the restriction $K\circ (J_{(\X_\perp)_1}\times J_{(\Y_\perp)_1})$ of $K$ to $(\X_\perp)_1\times(\Y_\perp)_1$ is contained in $((\X\times\Y)_\perp)_1$, hence again by Lemma \ref{lem:factors through a subset} it follows that $K\circ J_{(\X_\perp)_1\times(\Y_\perp)_1}=J_{((\X\times\Y)_\perp)_1}\circ K'$ for some function $K':(\X_\perp)_1\times(\Y_\perp)_1\to((\X\times\Y)_\perp)_1$.  Hence we obtain the following commuting diagram:
\[  \begin{tikzcd}
\mathbf 1\times\mathbf 1\ar{r}{F_i\times G_i}\ar{rd}[swap]{F'_i\times G'_i} & \X_\perp\times\Y_\perp \ar{r}{K} & (\X\times\Y)_\perp\\
& (\X_\perp)_1\times(\Y_\perp)_1\ar{u}[swap]{J_{(\X_\perp)_1}\times J_{(\Y_\perp)_1}}\ar{r}{K'} & ((\X\times\Y)_\perp)_1\ar{u}[swap]{J_{((\X\times\Y)_\perp)_1}}
\end{tikzcd} \]
Since $J_{((\X\times\Y)_\perp)_1}$ is an order embedding, it follows from Lemma \ref{lem:order embedding induces order embedding on homsets} and $K\circ (F_1\times G_1)=F_1\odot G_1\sqsubseteq F_2\odot G_2=K\circ (F_2\times G_2)$ that $K'\circ (F_1'\times G_1')\sqsubseteq K'\circ (F_2'\times G_2')$. Since each quantum cpo $\Z$ consisting of one-dimensional atoms can be identified with $`Z$, where it is clear that $\Z$ is pointed if and only if $Z$ is pointed, and the functor $`(-):\CPO\to\qCPO$ is fully faithful by Theorem \ref{thm:CPO-qCPO adjunction}, it follows that we can identify $\mathbf 1$ with $`1$, $(\X_\perp)_1$ and $(\Y_\perp)_1$ with $`(X_\perp)$ and $`(Y_\perp)$ for some posets $X$ and $Y$, and $F_i'$ and $G_i'$ with $`f_i$ for some functions $f_i:1\to X_\perp$ and $g_i:1\to Y_\perp$. Furthermore, $((\X\times\Y)_\perp)_1$ corresponds to $`((X\times Y)_\perp)$ and $K'$ corresponds to $`k$, where $k:X_\perp\times Y_\perp\to (X\times Y)_\perp$ is the usual double strength on $\CPO$, i.e., $k(x,y)=(x,y)$ if $x\neq \perp\neq y$, and $\perp$ otherwise. $F_1\neq B_{\X_\perp}$ corresponds to $f_1\neq\perp$ and similarly $G_1\neq B_{\Y_\perp}$ corresponds to $g_1\neq\perp$.
Moreover, it is easy to verify that $`(-):\CPO\to\qCPO$ preserves products, hence  $K'\circ (F_1'\times G_1')\sqsubseteq K'\circ (F_2'\times G_2')$ implies  $k\circ (f_1\times g_1)\sqsubseteq k\circ (f_2\times g_2)$. Since $f_1\neq\perp g_1$, it is easy to verify that this implies $f_1\sqsubseteq f_2$ and $g_1\sqsubseteq g_2$, whence $F_1\sqsubseteq F_2$ and $G_1\sqsubseteq G_2$.
\end{proof}

We can now formulate the main theorem of this section:
\begin{theorem}
The linear/non-linear model in Proposition \ref{prop:Kl model} is a computationally adequate $\CPO$-LNL model, i.e., a computationally adequate model for LNL-FPC.
\end{theorem}
\begin{proof}
Any model satisfying the conditions of \cite[Definitions 5.13 \& 7.1]{lnl-fpc-lmcs} is called a \emph{computationally adequate $\CPO$-LNL model}, which is computationally adequate for LNL-FPC by \cite[Theorem 7.10]{lnl-fpc-lmcs}.
By Proposition \ref{prop:Kl model}, the first three properties of \cite[Definition 5.13]{lnl-fpc-lmcs} are fulfilled, where \cite[Theorem 2]{stringdiagramssemantics} is used to guarantee that all categorical structure is enriched over $\CPO$. The last condition in \cite[Definition 5.13]{lnl-fpc-lmcs} follows from Lemmas \ref{lem:ep-zero} and \ref{lem:omega-colimits}. The parts about the category $\mathbf C$ in the last condition are automatically fulfilled, since we take $\mathbf C=\CPO$, which is cocomplete, so it in particular has all $\omega$-colimits over pre-embeddings with respect to $\V\circ `(-).$ Preservation  of these colimits by $(-\times-)$ follows from the fact $\CPO$ is cartesian closed, so both $X\times(-)$ and $(-)\times X$ preserve all colimits for any cpo $X$.

The first condition of \cite[Definition 7.1]{lnl-fpc-lmcs} states that the monoidal unit of $\mathbf{Kl}$ should not be equal to its initial object, which clearly is the case, since $\mathbf 1$ is the monoidal unit of $\mathbf{Kl}$, whereas the zero object is the empty quantum poset. The second and last condition of \cite[Definition 7.1]{lnl-fpc-lmcs} is that the monoidal product on $\mathbf{Kl}$ reflects the order, which is the case by Lemma \ref{lem:reflecting the order}.
\end{proof}

\section{Conclusion}
By internalizing cpos in the category $\mathbf{qRel}$ of quantum sets and binary relations, we obtained a noncommutative generalization of cpos, called quantum cpos. We investigated the properties of the category $\mathbf{qCPO}$ of quantum cpos and Scott continuous maps and showed there is a lift monad whose Kleisli category $\mathbf{Kl}$ is equivalent to the category $\qCPO_{\perp!}$ of pointed quantum cpos with strict Scott continuous maps. We also established that $\mathbf{Kl}$ is $\mathbf{CPO}$-algebraically compact, a feature commonly shared by models that support recursive types.

Many works on denotational semantics employ \emph{directed complete partial orders (dcpos)} instead of cpos. A definition of a quantum dcpo $(\X,R)$ can be obtained by replacing the sequences $K_1\sqsubseteq K_2\sqsubseteq\cdots$ of functions $\H\to\X$ in Definition \ref{def:quantum cpo} by a directed set in $\qSet(\H,\X)$. We believe the proofs in the present article can be modified in a straightforward manner to obtain analogous results for quantum dcpos, but we are unsure whether a category of quantum dcpos would have any advantage over $\qCPO$.

We illustrated the utility of our results by showing the Kleisli category $\mathbf{Kl}$ provides a denotational model for two  programming languages, ECLNL and LNL-FPC.  Utilizing the fact that $\mathbf{Kl}$ is $\mathbf{CPO}$-algebraically compact, we showed that $\mathbf{KL}$ is sound for both ECLNL and LNL-FPC, and that $\mathbf{Kl}$ is computationally adequate for the circuit-free fragment of LNL-FPC. 

Other models for ECLNL and LNL-FPC preceded ours (cf.~\cite{stringdiagramssemantics,lnl-fpc-lmcs}), but to our knowledge, $\mathbf{Kl}$ is the only model that is simultaneously sound for both languages, and also computationally adequate for LNL-FPC. We believe $\mathbf{Kl}$ could be computationally adequate for ECLNL, by adapting the proof using logical relations that was used for LNL-FPC to the setting of ECLNL. All proofs we know of for establishing computational adequacy rely on \emph{formal relations} in which the monoidal product reflects the order, and $\mathbf{Kl}$ satisfies this property. The fact that $\mathbf{Kl}$ satisfies this property at linear types also leads us to believe $\mathbf{Kl}$ could be computationally adequate for an extension of LNL-FPC with circuits, or equivalently, an extension of Proto-Quipper-M with recursive types.

As mentioned in Section \ref{sec:modeling physical systems}, quantum circuit description languages with measurement and state preparation typically have two runtimes: \emph{circuit generation time}, during which a quantum circuit is constructed, and \emph{circuit execution time}, during which the qubit states are prepared, the circuit is applied, and the qubits are measured.
In the current paradigm of quantum computing, recursion is controlled classically. This is related to the fact that recursion in quantum circuit description languages takes place during circuit generation time. As a consequence, measurement and state preparation, which are operations that typically take place during circuit execution time, are less relevant for 
understanding recursion in quantum computing according to the current paradigm. The fact that $\mathbf{Kl}$ is a model of ECLNL, a circuit description language with recursive terms, demonstrates that quantum cpos support recursion at circuit generation time.

Neither state preparation nor measurement were included in the original formulation of Proto-Quipper-M, or in its related languages ECLNL and LNL-FPC. However, both operations can be described in terms of quantum sets, where state preparation requires the use of one of the probabilistic monads $\D$ or $\S$ on $\qSet$, as mentioned in Section \ref{sec:modeling physical systems}. In collaboration with Xiaodong Jia and Vladimir Zamdzhiev, the authors used the monad $\S$ to model a hybrid quantum programming language whose classical subsystem is recursively typed and whose quantum subsystem is a first-order type system with recursive terms \cite{JKLMZ}. The resulting language handles quantum data and classical data separately, allowing recursion to be defined only for classical data types.

\emph{Dynamic lifting} allows for the interweaving of circuit generation time and circuit execution time. We expect that $\qSet$ is a model for Proto-Quipper-Dyn~\cite{pq-dyn2,pq-dyn}, the extension of Proto-Quipper-M with dynamic lifting, where the latter operation is modeled by the monad $\S$ of Example \ref{ex:subdistributions-monad}. 

What a quantum domain theory based on quantum cpos presently lacks is, above all, a quantum probabilistic power domain monad, i.e.,  probabilistic monad on $\qCPO$, just like $\S$ is a probabilistic monad on $\qSet$. We expect that such a monad on $\qCPO$ will suffice to provide models in terms of quantum cpos for various quantum programming languages with classically-controlled recursion, such as the quantum lambda calculus, its extension with recursive types Quantum FPC, and a possible extension of Proto-Quipper-Dyn with recursive types. 

Perhaps more importantly, we expect that such a monad will provide the right setting to study the possibility of implementing recursion with quantum control, for the following reasons. Firstly, such a monad allows us to describe higher-order impure functions in the quantum cpo framework, analogue to the description of such functions in the quantum sets framework via the monad $\S$ on $\qSet$. Since the quantum cpo framework is noncommutative by construction, we expect that it can describe the quantum interactions between the higher-order functions. Secondly, such a monad combines in a compatible way the order of quantum cpos with the CP-L\"owner order of the von Neumann algebras associated with the quantum sets underlying quantum cpos. We recall that the order of quantum cpos is used to describe recursion at circuit generation time, whereas the L\"owner order describes the state of completion of a computation at circuit execution time. Thus, combining both orders likely supports recursion at both runtimes. Moreover, such a combination might support a single runtime in which circuit generation and circuit execution are no longer sequential, but exhibit an indefinite causal order, as in case of the quantum switch, and more generally, in quantum control flow. 

Constructing such a quantum probabilistic power domain is ongoing research. Classically, probabilistic effects in the context of recursion are described by Jones's probabilistic power domain, a monad based on valuations over domains \cite{Jones90}. According to the principles of noncommutative geometry, quantizing Jones's monad - or an appropriate variant - might yield a suitable probabilistic monad on $\qCPO$. However, quantizing valuations is difficult due to the modularity property of valuations, which would require an understanding how to quantize addition and equality. This is highly nontrivial, which makes finding a quantum probabilistic power domain challenging.

Finally, valuations on classical domains are defined on the lattice of Scott-open sets. 
But, the standard notion of quantum topological spaces is a noncommutative generalization of locally compact Hausdorff spaces, i.e., $C^*$-algebras, which generalize locally compact Hausdorff spaces. This is problematic for domains where only the Lawson topology is Hausdorff, while the Scott topology (which is typically $T_0$) underpins most applications to computation. We plan to investigate the quantization of topological spaces beyond locally compact Hausdorff spaces, via discrete quantization. One step in this direction is the introduction of quantum suplattices, which are quantized versions of complete lattices \cite{qSup}. Since topological methods also play a role in finding alternative constructions of probabilistic power domains \cites{JiaMislove,commutative-monads}, we hope that such a notion of quantum topological spaces may yield a different path towards a suitable quantum probabilistic power domain.

\section*{Acknowledgements}
We thank Xiaodong Jia and Sander Uijlen for valuable conversations about this work. We are especially grateful to Vladimir Zamdzhiev for his insights, which significantly enhanced our understanding of abstract models of quantum computation. Finally, we thank the anonymous reviewers for their helpful comments, which improved the text. This work was partially funded by the AFOSR under the MURI grant
number FA9550-16-1-0082 entitled ``Semantics, Formal Reasoning, and Tool Support for Quantum Programming''.
The second author is currently supported by the Austrian Science Fund (FWF) [Project Number: PAT6443523, DOI: 10.55776/PAT6443523].

\newpage


\begin{bibdiv}
\begin{biblist}

\bib{AbramskyCoecke08}{book}{
      author={Abramsky, S.},
      author={Coecke, B.},
       title={Categorical quantum mechanics},
        date={2008},
}

\bib{abramskyjung:domaintheory}{incollection}{
      author={Abramsky, S.},
      author={Jung, A.},
       title={Handbook of logic in computer science (vol. 3)},
        date={1994},
   publisher={Oxford University Press},
     address={Oxford, UK},
       pages={1\ndash 168},
         url={http://dl.acm.org/citation.cfm?id=218742.218744},
}

\bib{adameketall:joyofcats}{book}{
      author={Ad\'amek, J.},
      author={Herrlich, H.},
      author={Strecker, G.E.},
       title={Abstract and {C}oncrete {C}ategories: {T}he {J}oy of {C}ats},
   publisher={John Wiley and Sons, Inc},
        date={1990},
         url={http://katmat.math.uni-bremen.de/acc/acc.pdf},
}

\bib{benton}{inproceedings}{
author={Benton, P. N.},
editor={Pacholski, Leszek
and Tiuryn, Jerzy},
title={A mixed linear and non-linear logic: Proofs, terms and models},
booktitle={Computer Science Logic},
date={1995},
publisher={Springer Verlag},
address={Berlin, Heidelberg},
pages={121--135},
isbn={978-3-540-49404-1},
}

\bib{BCSW}{article}{
      author={Betti, R.},
      author={Carboni, A.},
      author={Street, R.},
      author={Walters, R.},
       title={Variation through enrichment},
        date={1983},
     journal={Journal of Pure and Applied Algebra},
      volume={29},
       pages={109\ndash 127},
}

\bib{Blackadar}{book}{
author={Blackadar, B.},
title = {Operator Algebras},
publisher = {Springer Verlag},
date ={2006}
}

\bib{borceux:handbook1}{book}{
      author={Borceux, F.},
       title={Handbook of categorical algebra 1: Basic category theory},
   publisher={Cambridge University Press},
        date={1994},
}


\bib{quantum-switch}{article}{  title = {Quantum computations without definite causal structure},
  author = {Chiribella, G.},
  author ={D'Ariano, G.M.},
  author = {Perinotti, P.},
  author = {Valiron, B.},
  journal = {Phys. Rev. A},
  volume = {88},
  issue = {2},
  pages = {022318},
  numpages = {15},
  date = {2013},
  month = {Aug},
  publisher = {American Physical Society},
  url = {https://link.aps.org/doi/10.1103/PhysRevA.88.022318}
}


\bib{cho:semantics}{article}{
      author={Cho, K.},
       title={Semantics for a quantum programming language by operator
  algebras},
        date={2016},
     journal={New Generation Comput.},
      volume={34},
      number={1-2},
       pages={25\ndash 68},
}

\bib{ChoWesterbaan16}{article}{
      author={Cho, K.},
      author={Westerbaan, A.},
       title={{Von Neumann Algebras form a Model for the Quantum Lambda
  Calculus}},
        date={2016},
      eprint={1603:02113},
}

\bib{coeckekissinger}{book}{
      author={Coecke, B.},
      author={Kissinger, A.},
       title={Picturing quantum processes: A first course in quantum theory and
  diagrammatic reasoning},
   publisher={Cambridge University Press},
        date={2017},
}

\bib{connes:ncg}{book}{
      author={Connes, A.},
       title={Noncommutative geometry},
   publisher={Academic Press},
        date={1994},
}

\bib{duanseveriniwinter}{article}{
      author={Duan, R.},
      author={Severini, S.},
      author={Winter, A.},
       title={{Zero-error communication via quantum channels, noncommutative
  graphs, and a quantum Lov\'asz number}},
        date={2013},
     journal={IEEE Trans. Inform. Theory},
      volume={59},
       pages={1164\ndash 1174},
}

\bib{EffrosRuan}{book}{
author={Effros, E.G.},
author = {Ruan, Z.-J.},
title={Theory of Operator Spaces},
publisher={Oxford University Press},
date={2000}
}

\bib{eklund2018semigroups}{book}{
      author={Eklund, P.},
      author={García, J.~Gutiérrez},
      author={Höhle, U.},
      author={Kortelainen, J.},
       title={Semigroups in complete lattices: Quantales, modules and related
  topics},
      series={Developments in Mathematics},
   publisher={Springer},
     address={Cham},
        date={2018},
      volume={54},
}

\bib{fiore-thesis}{thesis}{
      author={Fiore, M.P.},
       title={Axiomatic domain theory in categories of partial maps},
        type={Ph.D. Thesis},
        date={1994},
}

\bib{pq-dyn2}{article}{
      author={Fu, P.},
      author={Kishida, K.},
      author={Ross, N.J.},
      author={Selinger, P.},
       title={{A biset-enriched categorical model for Proto-Quipper with
  dynamic lifting}},
        date={2022},
     journal={Proceedings 19th International Conference on Quantum Physics and
  Logic},
       pages={302--342},
         url={https://arxiv.org/abs/2204.13039},
}

\bib{pq-dyn}{article}{
      author={Fu, P.},
      author={Kishida, K.},
      author={Ross, N.J.},
      author={Selinger, P.},
       title={Proto-quipper with dynamic lifting},
        date={2023},
     journal={Proc. ACM Program. Lang.},
      volume={7},
      number={POPL},
         url={https://doi.org/10.1145/3571204},
}


\bib{girard-linear}{article}{
      author={Girard, J.-Y.},
       title={Linear logic},
        date={1987},
     journal={Theoretical Computer Science},
      volume={50},
       pages={1 \ndash  101},
}

\bib{grover}{inproceedings}{
      author={Grover, L.K.},
       title={A fast quantum mechanical algorithm for database search},
        date={1996},
   booktitle={{STOC} (symposium on the theory of computing)},
}

\bib{haag-kastler}{article}{
author ={Haag, R.},
author = {Kastler,D.},
title={An Algebraic Approach to Quantum Field Theory},
journal= {J. Math. Phys.},
volume= {5},
pages ={848--861}
date ={1964}
}

\bib{heunenvicary}{book}{
      author={Heunen, C.},
      author={Vicary, J.},
       title={Categories for quantum theory, an introduction},
   publisher={Oxford University Press},
        date={2019},
}

\bib{Heymans-Stubbe}{article}{
      author={Heymans, H.},
      author={Stubbe, I.},
       title={{Grothendieck quantaloids for allegories of enriched
  categories}},
        date={2012},
     journal={Bulletin of the Belgian Mathematical Society - Simon Stevin},
      volume={19},
      number={5},
       pages={859 \ndash  888},
         url={https://doi.org/10.36045/bbms/1354031554},
}

\bib{huotstaton:qptheory}{incollection}{
      author={Huot, M.},
      author={Staton, S.},
       title={Universal Properties in Quantum Theory},
        date={2019},
   booktitle={Proceedings of QPL 2018},
       pages={213\ndash 223},
     publisher={EPTCS},
       volume={287},
         url={https://doi.org/10.48550/arXiv.1901.10117},
}

\bib{qSup}{inproceedings}{
      author={Jen\v{c}a, G.},
      author={Lindenhovius, B.},
       title={Quantum Suplattices},
        date={2023},
     booktitle={Proceedings of {QPL} 2023},
     series={EPTCS},
      volume={384},
       pages={58\ndash 74},
       url={https://doi.org/10.48550/arXiv.2308.16495},
}


\bib{JKLMZ}{article}{
      author={Jia, X.},
      author={Kornell, A.},
      author={Lindenhovius, B.},
      author={Mislove, M.},
      author={Zamdzhiev, V.},
       title={Semantics for variational quantum programming},
        date={2022jan},
     journal={Proc. ACM Program. Lang.},
      volume={6},
      number={POPL},
         url={https://doi.org/10.1145/3498687},
}

\bib{commutative-monads}{inproceedings}{
      author={Jia, X.},
      author={Lindenhovius, B.},
      author={Mislove, M.},
      author={Zamdzhiev, V.},
       title={Commutative monads for probabilistic programming languages},
        date={2021},
   booktitle={36th annual acm/ieee symposium on logic in computer science
  (lics)},
         url={https://ieeexplore.ieee.org/document/9470611},
}

\bib{JiaMislove}{article}{
      author={Jia, X.},
      author={Mislove, M.},
       title={{Completing simple valuations in $K$-categories}},
        date={2022},
     journal={Topology and its Applications},
      volume={318},
}

\bib{Jones90}{article}{
      author={Jones, C.},
       title={Probabilistic Nondeterminism},
        date={1990},
     journal={Ph.D. dissertation, University of Edinburgh},
}

\bib{knill96}{report}{
      author={Knill, E.},
       title={Conventions for Quantum Pseudocode},
       publisher={LANL report LAUR-96-2724}
        date={1996},
       pages={1 \ndash  13},
       url={https://arxiv.org/pdf/2211.02559},
}



\bib{Kornell17}{article}{
      author={Kornell, A.},
       title={Quantum Collections},
        date={2017},
     journal={Int. J. Math.},
      volume={28}
}

\bib{Kornell18}{article}{
      author={Kornell, A.},
       title={Quantum sets},
        date={2020},
     journal={J. Math. Phys.},
      volume={61},
      pages={102202}
}

\bib{Kornell-Discrete-I}{article}{
      author={Kornell, A.},
       title={{Discrete quantum structures I: Quantum predicate logic}},
        date={2024},
     journal={J. Noncommut. Geom.},
     volume ={18},
      number={1},
       pages={337\ndash 382},
}

\bib{Kornell-Discrete-II}{article}{
      author={Kornell, A.},
       title={{Discrete quantum structures II: Examples}},
        date={2024},
     journal={J. Noncommut. Geom.},
     volume={18},
      number={2},
       pages={411\ndash 450},
}

\bib{KLM20}{article}{
      author={Kornell, A.},
      author={Lindenhovius, B.},
      author={Mislove, M.},
       title={{A category of quantum posets}},
        date={2022},
     journal={{Indagationes Mathematicae}},
      volume={33},
       pages={1137\ndash 1171},
}

\bib{kuperbergweaver:quantummetrics}{book}{
      author={Kuperberg, G.},
      author={Weaver, N.},
       title={A {V}on {N}eumann {A}lgebra {A}pproach to {Q}uantum {M}etrics:
  {Q}uantum {R}elations},
      series={Memoirs of the American Mathematical Society},
   publisher={American Mathematical Society},
        date={2012},
}

\bib{maclane}{book}{
      author={MacLane, S.},
       title={Categories for the working mathematician (2nd ed.)},
   publisher={Springer},
        date={1998},
}

\bib{free-exponential-modality}{inproceedings}{
author={Melliès, P.A.},
author={Tabareau, N.},
author={Tasson, C.},
title={An Explicit Formula for the Free Exponential Modality of Linear Logic},
date={2009},
booktitle={{ICALP} 2009},
editor={Albers, S.},
editor={Marchetti-Spaccamela, A.},
editor={Matias, Y.},
editor={Nikoletseas, S.},
editor={Thomas, W.},
publisher={{Springer}},
series={Lecture Notes in Computer Science},
volume={5556},
location={Berlin, Heidelberg},
pages={247\ndash 260},
url={https://doi.org/10.1007/978-3-642-02930-1_21},
}


\bib{levy-cbpv}{book}{
      author={Levy, P.B.},
       title={Call-by-push-value: A functional/imperative synthesis},
   publisher={Springer},
        date={2004},
}


\bib{lnl-fpc-lmcs}{article}{
      author={Lindenhovius, B.},
      author={Mislove, M.},
      author={Zamdzhiev, V.},
       title={{LNL-FPC: The Linear/Non-linear Fixpoint Calculus}},
        date={2021},
     journal={{Logical Methods in Computer Science}},
      volume={17},
       pages={9:1–\ndash 9:61},
}

\bib{stringdiagramssemantics}{incollection}{
      author={Lindenhovius, B.},
      author={Mislove, M.},
      author={Zamdzhiev, V.},
       title={Semantics for a lambda calculus for string diagrams},
        date={2023},
   booktitle={{Samson Abramsky on Logic and Structure in Computer Science and
  Beyond}},
      editor={Palmigiano, A.},
      editor={Sadrzadeh, M.},
   publisher={Springer},
}

\bib{moggi}{article}{
      author={Moggi, E.},
       title={Computational lambda-calculus and monads},
        date={1989},
     journal={Proceedings of the Fourth Annual Symposium on Logic in Computer
  Science},
       pages={14\ndash 23},
}

\bib{moggi1991}{article}{
      author={Moggi, E.},
       title={Notions of computation and monads},
        date={1991},
        ISSN={0890-5401},
     journal={Information and Computation},
      volume={93},
      number={1},
       pages={55\ndash 92},
  url={https://www.sciencedirect.com/science/article/pii/0890540191900524},
        note={Selections from 1989 IEEE Symposium on Logic in Computer
  Science},
}

\bib{NAKAGAWA1989563}{incollection}{
      author={Nakagawa, R.},
       title={Chapter 14 categorical topology},
        date={1989},
   booktitle={Topics in general topology},
      editor={Morita, Kiiti},
      editor={iti Nagata, Jun},
      series={North-Holland Mathematical Library},
      volume={41},
   publisher={Elsevier},
       pages={563 \ndash  623},
  url={http://www.sciencedirect.com/science/article/pii/S0924650908701600},
}

\bib{Olson}{article}{
      author={Olson, M.~P.},
       title={The self-adjoint operators of a von neumann algebra form a
  conditionally complete lattice},
        date={1971},
     journal={Proc. Amer. Math. Soc.},
      volume={28},
}

\bib{quant-semantics}{inproceedings}{
      author={Pagani, M.},
      author={Selinger, P.},
      author={Valiron, B.},
       title={Applying quantitative semantics to higher-order quantum
  computing},
        date={2014},
   booktitle={{POPL} 2014},
      editor={Jagannathan, Suresh},
      editor={Sewell, Peter},
   publisher={{ACM}},
       pages={647\ndash 658},
         url={https://doi.org/10.1145/2535838.2535879},
}

\bib{PPRZ19}{article}{
      author={P\'echoux, R.},
      author={Perdrix, S.},
      author={Rennela, M.},
      author={Zamdzhiev, V.},
       title={Quantum programming with inductive datatypes: Causality and
  affine type theory},
        date={2020},
       pages={562\ndash 581},
}

\bib{PtakPulmannova91}{book}{
      author={Pt\'ak, P.},
      author={Pulmannov\'a, S.},
       title={Orthomodular structures as quantum logics},
      series={Fundamental Theories of Physics},
   publisher={Kluwer Academic Publishers},
        date={1991},
}

\bib{rennela:domains}{inproceedings}{
      author={Rennela, M.},
       title={Towards a quantum domain theory: order-enrichment and fixpoints
  in {W}*-algebras},
        date={2014},
   booktitle={Mfps 30},
      series={Electronic Notes in Theoretical Computer Science},
      volume={308},
       pages={289\ndash 307},
}


\bib{Riehl}{book}{
      author={Riehl, E.},
       title={Category Theory in Context},
        date={2016},
   publisher={Dover Publications},
}



\bib{pqm}{inproceedings}{
author    = {F. Rios and
P. Selinger},
editor    = {Bob Coecke and
Aleks Kissinger},
title     = {A categorical model for a quantum circuit description language},
booktitle = {Proceedings 14th International Conference on Quantum Physics and Logic,
{QPL} 2017, Nijmegen, The Netherlands, 3-7 July 2017.},
series    = {{EPTCS}},
volume    = {266},
pages     = {164--178},
date      = {2017},
url       = {https://doi.org/10.4204/EPTCS.266.11},
}

\bib{Seal13}{article}{
      author={Seal, G.J.},
       title={Tensors, monads and actions},
        date={2013},
     journal={Theory and Applications of Categories},
      volume={28},
      number={15},
       pages={403\ndash 434},
         url={https://arxiv.org/abs/1205.0101},
}


\bib{selingervaliron:quantumlambda}{article}{
      author={Selinger, P.},
      author={Valiron, B.},
       title={A lambda calculus for quantum computation with classical control},
        date={2006},
     journal={Mathematical Structures in Computer Science},
      volume={16},
      number={3},
       pages={527\ndash 552},
}

\bib{shor}{article}{
      author={Shor, P.W.},
       title={Polynomial-time algorithms for prime factorization and discrete
  logarithms on a quantum computer},
        date={1999},
     journal={SIAM Review},
      volume={41},
      number={2},
       pages={303\ndash 332},
}

\bib{smyth-plotkin:domain-equations}{article}{
      author={Smyth, M.B.},
      author={Plotkin, G.D.},
       title={The category-theoretic solution of recursive domain equations},
        date={1982},
     journal={Siam J. Comput.},
}

\bib{Szigeti}{article}{
      author={Szigeti, J.},
       title={{On limits and colimits in the Kleisli category}},
        date={1983},
     journal={Cahiers de topologie et g'eom'etrie diff'erentielle
  cat'egoriques},
      volume={24},
      number={4},
       pages={381\ndash 391},
}



\bib{QuantumFPC}{inproceedings}{
  author       = {Takeshi Tsukada and Kazuyuki Asada},
  title        = {Enriched Presheaf Model of Quantum FPC},
  journal      = {Proceedings of the ACM on Programming Languages},
  volume       = {8},
  number       = {POPL},
  date         = {2024},
  articleno    = {13},
  pages        = {362--392},
  publisher    = {ACM},
}

\bib{Weaver10}{article}{
      author={Weaver, N.},
       title={Quantum relations},
        date={2012},
     journal={Mem. Amer. Math. Soc.},
      volume={215},
}

\bib{Weaver19}{article}{
      author={Weaver, N.},
       title={Hereditarily antisymmetric operator algebras},
        date={2019},
     journal={J. Inst. Math. Jussieu},
}

\bib{Weaver21}{article}{
      author={Weaver, N.},
       title={Quantum graphs as quantum relations},
        date={2021},
     journal={J. Geom. Anal.},
      volume={31},
       pages={9090–\ndash 9112},
}

\bib{Westerbaan16}{inproceedings}{
      author={Westerbaan, A.},
       title={Quantum programs as {K}leisli {M}aps},
        date={2016},
   booktitle={{QPL 2016: Proceedings 13th International Conference on Quantum
  Physics and Logic}},
}

\bib{Westerbaanthesis}{thesis}{
    author={A. Westerbaan},
    title={{The Category of Von Neumann Algebras, PhD Thesis}},
    school = {Radboud University, {NL}},
    publisher={GVO drukkers & vormgevers B.V.},
    date={2019},
}


\bib{Zhao-Fan}{article}{
      author={Zhao, D.},
      author={Fan, T.},
       title={Dcpo-completion of posets},
        date={2010},
        ISSN={0304-3975},
     journal={Theoretical Computer Science},
      volume={411},
      number={22},
       pages={2167\ndash 2173},
  url={https://www.sciencedirect.com/science/article/pii/S0304397510001155},
}


\end{biblist}
\end{bibdiv}
\end{document}